\documentclass[11.2pt, twoside]{Thesis} % The default font size and one-sided printing (no margin offsets)

\usepackage{mathrsfs}
\usepackage{stmaryrd}
\usepackage{hyperref}
\usepackage{multirow}
\usepackage{pdfpages}
\usepackage{amsmath}
\newcommand{\triplenorm}{\ensuremath{| \! | \! |}}
\usepackage{amssymb}
\usepackage{tikz}
\usetikzlibrary{3d,calc}
\usepackage{color}
\usepackage[numbers]{natbib}
\usepackage{amsfonts}
\usepackage{textcomp}
\usepackage{url}
\usepackage{amsbsy}
\usepackage[utf8]{inputenc}
\usepackage{bbold}
\usepackage{accents}

\usepackage{amsthm}
\theoremstyle{plain}
\newtheorem*{theorem*}{Theorem}
\newtheorem*{lemma*}{Lemma}
\newtheorem*{proposition*}{Proposition}
\newtheorem*{definition*}{Definition}

\newtheorem{conjecture}{Conjecture}

\title{\ttitle} % Defines the thesis title - don't touch this

\begin{document}
\changefont{cmr}{m}{n}

\frontmatter % Use roman page numbering style (i, ii, iii, iv...) for the pre-content pages

\setstretch{1.15} % Line spacing of 1.3

% Define the page headers using the FancyHdr package and set up for one-sided printing
\fancyhead{} % Clears all page headers and footers
\rhead{\thepage} % Sets the right side header to show the page number
\lhead{} % Clears the left side page header

\pagestyle{fancy} % Finally, use the "fancy" page style to implement the FancyHdr headers

\newcommand{\HRule}{\rule{\linewidth}{0.5mm}} % New command to make the lines in the title page

% PDF meta-data
\hypersetup{pdftitle={\ttitle}}
\hypersetup{pdfsubject=\subjectname}
\hypersetup{pdfauthor=\authornames}
\hypersetup{pdfkeywords=\keywordnames}

%----------------------------------------------------------------------------------------
%	TITLE PAGE
%----------------------------------------------------------------------------------------

\begin{titlepage}
\hspace{-2.5cm}\begin{minipage}{\textwidth}
\includegraphics[scale=.815]{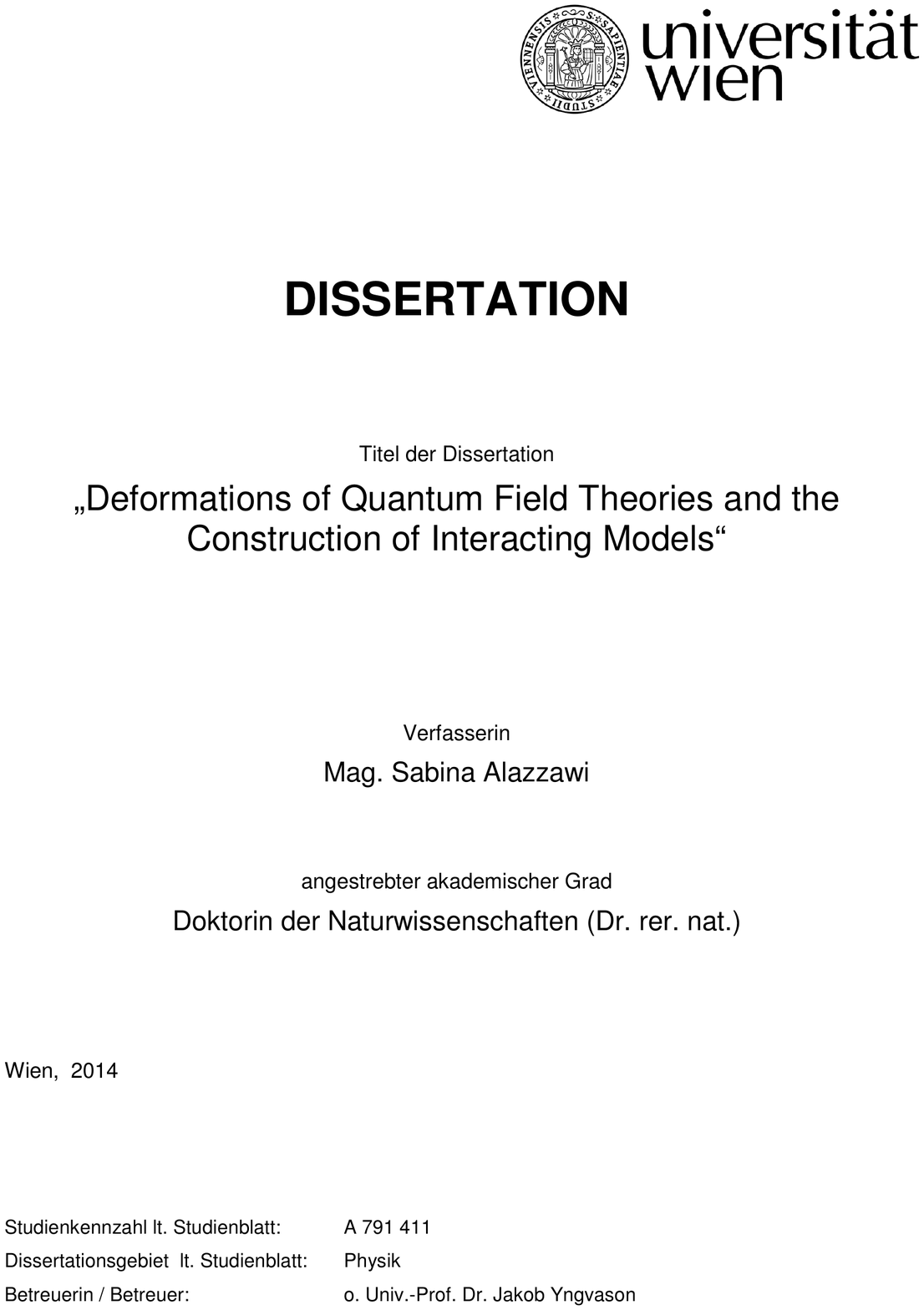}
\end{minipage}

%\includegraphics[\textwidth]{Titelseite1.jpg}

%\begin{center}

%\textsc{\LARGE \univname}\\[1.5cm] % University name
%\textsc{\Large Doctoral Thesis}\\[0.5cm] % Thesis type

%\HRule \\[0.4cm] % Horizontal line
%{\huge \bfseries \ttitle}\\[0.4cm] % Thesis title
%\HRule \\[1.5cm] % Horizontal line
 
%\begin{minipage}{0.4\textwidth}
%\begin{flushleft} \large
%\emph{Author:}\\

%\end{flushleft}
%\end{minipage}
%\begin{minipage}{0.4\textwidth}
%\begin{flushright} \large
%\emph{Supervisor:} \\

%\end{flushright}
%\end{minipage}\\[3cm]
 
%\large \textit{A thesis submitted in fulfilment of the requirements\\ for the degree of \degreename}\\[0.3cm] % University requirement text
%\textit{in the}\\[0.4cm]
%\groupname\\\deptname\\[2cm] % Research group name and department name
 
%{\large \today}\\[4cm] % Date
%\includegraphics{Logo} % University/department logo - uncomment to place it
 
%\vfill
%\end{center}

\end{titlepage}

\pagestyle{fancy} % The page style headers have been "empty" all this time, now use the "fancy" headers as defined before to bring them back

\lhead{\emph{Contents}} % Set the left side page header to "Contents"
\tableofcontents % Write out the Table of Contents
%----------------------------------------------------------------------------------------
%	QUOTATION PAGE
%----------------------------------------------------------------------------------------

%\pagestyle{empty} % No headers or footers for the following pages

%\null\vfill % Add some space to move the quote down the page a bit

%\textit{``Thanks to my solid academic training, today I can write hundreds of words on virtually any topic without possessing a shred of information, which is how I got a good job in journalism."}

%\begin{flushright}
%Dave Barry
%\end{flushright}

%\vfill\vfill\vfill\vfill\vfill\vfill\null % Add some space at the bottom to position the quote just right

%\clearpage % Start a new page

%----------------------------------------------------------------------------------------
%	DEDICATION
%----------------------------------------------------------------------------------------

%\setstretch{1.3} % Return the line spacing back to 1.3

%\pagestyle{empty} % Page style needs to be empty for this page

%\dedicatory{For/Dedicated to/To my\ldots} % Dedication text

%\addtocontents{toc}{\vspace{2em}} % Add a gap in the Contents, for aesthetics

%----------------------------------------------------------------------------------------
%	THESIS CONTENT - CHAPTERS
%----------------------------------------------------------------------------------------

\mainmatter % Begin numeric (1,2,3...) page numbering
%----------------------------------------------------------------------------------------
	%ABSTRACT PAGE
%----------------------------------------------------------------------------------------

\addtotoc{Abstract} % Add the "Abstract" page entry to the Contents

\abstract{\addtocontents{toc}{\vspace{1em}} % Add a gap in the Contents, for aesthetics

%The Thesis Abstract is written here (and usually kept to just this page). The page is kept centered vertically so can expand into the blank space above the title too\ldots
%}
\pagestyle{empty}
\LARGE\textbf{{Abstract}}\\\\
\normalsize{}
The subject of this thesis is the rigorous construction of quantum field theoretic models with nontrivial interaction. For this task techniques available in the framework of Algebraic Quantum Field Theory are applied and two different approaches are discussed.\par
On the one hand, an inverse scattering problem is considered. A given scattering matrix is thereby taken as the starting point of the construction. In two spacetime dimensions one may work with factorizing scattering matrices which exhibit a simple structure. The particle spectrum taken into account involves an arbitrary number of massive particle species which transform under some global gauge group. It is a known fact that auxiliary fields with weakened localization, namely in wedges, can be constructed. In the main part of this thesis the more involved transition to \textit{local} theories is shown by means of operator algebraic methods. Concretely, we make use of the so-called modular nuclearity condition. To this end, we investigate certain maps from the wedge algebras, generated by the auxiliary fields, to the considered Hilbert space. Under a very plausible conjecture it is shown that these maps are nuclear, which implies the nontriviality of algebras associated with bounded regions in the sense that the Reeh-Schlieder property holds. This construction method yields a large class of integrable models with factorizing S-matrices in two spacetime dimensions, complying with localization in bounded regions above a minimal size. The constructed family contains, for example, the multifaceted $O(N)$-invariant nonlinear $\sigma$-models.\par
On the other hand, deformation techniques constitute a method of construction which may be applied in arbitrary spacetime dimensions. This approach starts from a known quantum field theoretic model which is subjected to a certain modification. Here, concretely, the model of a scalar massive Fermion was deformed. It is shown that the correspondingly emerging models are based on fields with weakened localization properties again with regard to wedges. Due to this remnant of locality scattering theory can be applied and the two-particle S-matrix can be computed. The resulting scattering matrix depends on the deformation and has a very simple structure, not allowing for particle production nor momentum transfer in scattering processes. However, it differs from the S-matrix of the initial model. By restricting the spacetime dimension to two, it is shown that the considered deformation method yields a large class of integrable models with factorizing S-matrices which, moreover, comply with localization in bounded regions above a minimal size. Among the integrable models arising by deformation is the famous Sinh-Gordon model.\\

\newpage

\LARGE\textbf{{Zusammenfassung}}\\\\
\normalsize{}
Gegenstand dieser Arbeit ist die rigorose Konstruktion von quantenfeldtheoretischen Modellen mit nicht-trivialer Wechselwirkung. Dazu werden Techniken aus dem Rahmen der Algebraischen Quantenfeldtheorie angewandt und zwei verschiedene Verfahren diskutiert.\par
Zum einen wird ein inverses Streuproblem betrachtet. Dabei ist eine vorgegebene Streumatrix der Ausgangspunkt der Konstruktion. In zwei Raum-Zeit-Dimensionen kann dazu mit faktorisierenden Streumatrizen, welche eine einfache Struktur aufweisen, gearbeitet werden. Das betrachtete Teilchenspektrum schlie\ss t eine beliebige Zahl an massiven Teilchensorten ein, welche sich unter einer beliebigen globalen Eichgruppe transformieren. Wie bekannt, k\"onnen Hilfsfelder mit abgeschw\"achter Lokalisierung, und zwar in Keilgebieten, konstruiert werden. Im Hauptteil dieser Arbeit wird der weitaus aufwendigere \"Ubergang zu \textit{lokalen} Theorien unter Verwendung von operatoralgebraischen Methoden gezeigt. Konkret wird die sogenannte modulare Nuklearit\"atsbedingung herangezogen. Dazu werden gewisse Abbildungen von der aus den Hilfsfeldern generierten Keilalgebra in den betrachteten Hilbertraum untersucht. Es wird unter einer sehr plausiblen Vermutung gezeigt, dass diese Abbildungen nuklear sind. Dieser Sachverhalt wiederum impliziert die Nichttrivialit\"at von Algebren, welche mit beschr\"ankten Gebieten assoziiert werden, in dem Sinne, dass die Reeh-Schlieder Eigenschaft gilt. Dieses Konstruktionsverfahren f\"uhrt zu einer gro\ss en Klasse von integrablen Modellen mit faktorisierenden S-Matrizen in zwei Raumzeit-Dimensionen, welche mit Lokalisierung in beschr\"ankten Gebieten oberhalb einer Mindestgr\"o\ss e kompatibel sind. In die konstruierte Klasse fallen beispielsweise die facettenreichen $O(N)$-invarianten nicht-linearen $\sigma$-Modelle.\par
Zum anderen stellen Deformationsverfahren eine Konstruktionsmethode dar, welche in beliebigen Raumzeit-Dimensionen anwendbar ist. In diesem Zugang wird ein bekanntes quantenfeldtheoretisches Modell als Ausgangspunkt betrachtet, welches einer gewissen Modifikation unterzogen wird. Konkret wurde hier das Modell eines skalaren massiven Fermions deformiert. Es wird gezeigt, dass die entsprechend hervorgehenden Modelle auf Feldern mit abgeschw\"achten Lokalisierungseigenschaften, wiederum in Bezug auf Keile, basieren. Aufgrund dieser Restlokalit\"at kann Streutheorie angewendet und die Zwei-Teilchen-S-Matrix bestimmt werden. Die resultierende Streumatrix ist abh\"angig von der Deformation und hat eine sehr einfache Struktur, welche keine Teilchenerzeugung oder Impuls\"ubertrag in Sto\ss prozessen zul\"asst. Sie unterscheidet sich jedoch von jener des Ausgangsmodells. In Einschr\"ankung auf zwei Raumzeit-Dimensionen wird gezeigt, dass die betrachtete Deformationsmethode zu einer gro\ss en Klasse von integrablen Modellen mit faktorisierenden S-Matrizen f\"uhrt, welche dar\"uber hinaus im Einklang mit Lokalisierung in beschr\"ankten Gebieten oberhalb einer Mindestgr\"o\ss e sind. Zu den aus der Deformation hervorgehenden integrablen Modellen z\"ahlt beispielsweise das renommierte Sinh-Gordon Modell.

\clearpage % Start a new page

\pagestyle{fancy} % Return the page headers back to the "fancy" style

% Include the chapters of the thesis as separate files from the Chapters folder
% Uncomment the lines as you write the chapters

% Chapter 1

\chapter{Introduction} % Main chapter title
\fancyhead[LE,RO]{\thepage}
\fancyhead[LO]{}
\fancyhead[RE]{Chapter 1. \emph{Introduction}}
\renewcommand{\chaptermark}[1]{ \markboth{#1}{} }
\renewcommand{\sectionmark}[1]{ \markright{#1}{} }
\label{Chapter1} % For referencing the chapter elsewhere, use \ref{Chapter1} 

%\lhead{Chapter 1. \emph{Introduction}} % This is for the header on each page - perhaps a shortened title

%----------------------------------------------------------------------------------------

The quest for a consistent description of high energy particle physics led to the development of quantum field theory. From its beginnings in the first half of the 20th century to the present day quantum field theory as the unification of quantum mechanics and special relativity has asserted itself as the most successful theory for explaining the observed microscopic phenomena. In the scope of perturbative calculations, numerical predictions of the theory are in excellent agreement with the experiments. Nevertheless, one is faced with serious difficulties when seeking for a mathematically rigorous construction of an interacting quantum field theory model compatible with relativistic covariance, positivity of the energy and causality, i.e. locality. This problem has been attacked by theoretical physicists from various directions. For instance, in the framework of constructive quantum field theory \cite{glimm1981quantum, jaffe2000constructive} nontrivial models complying with the axiomatic approaches to quantum field theory, formulated by Wightman \cite{streater2000pct, Jost} and Haag and Kastler \cite{Araki99, Haag} respectively, were constructed in two and three spacetime dimensions. Thereby, Hamiltonian strategies, given by Jaffe \cite{jaffe1966existence} and Lanford \cite{Lanford} respectively, and a functional integral approach, developed by Symanzik \cite{symanzik1964modified}, proved successful in renormalizing certain models defined in terms of a Lagrangian in a rigorous way. Despite these achievements, results in four spacetime dimensions are still missing.\par
Many interesting models with polynomial self-interaction in two spacetime dimensions constructed earlier by methods of constructive quantum field theory arose in a very recent approach developed by Barata, J\"akel and Mund \cite{barata2013p}. It is based on Tomita-Takesaki theory \cite{takesaki1970tomita, takesaki2003theory} and, hence, does not rely on any Lagrangian formulation.\par
Other methods that have proven to be useful in the construction of particularly two-dimensional quantum field theories without recourse
to Lagrangians or perturbation theory are due to inverse scattering techniques. In the so-called bootstrap form-factor program models are defined in terms of a given scattering matrix. To this end, particularly factorizing collision operators are taken into account due to their simple structure and since they can be specified explicitly. Such S-matrices are known to appear in the context of integrable models like the Sinh-Gordon or the $O(N)$-invariant sigma models \cite{AAR}.\par
In the form factor program one aims at calculating the Wightman $n$-point functions of a theory associated with a given factorizing S-matrix. For this task, matrix elements of local field operators in scattering states, called form factors, are investigated. The special form of the considered scattering operator as well as the assumed properties of the local fields such as locality or covariance yield a number of constraints on the form factors. Their explicit computation can be achieved by solving these conditions, which was, indeed, established for many models \cite{babujian2011n}.\par
The Wightman $n$-point functions of local fields are then given by infinite series of integrals over form factors. The difficulty, thereby, is to control the convergence of these sums. Although this is possible in a few special cases \cite{babujian2006form}, this problem remains open to a large extent.\par
More recent inverse scattering techniques yield a large class of integrable models in two spacetime dimensions \cite{L08}. This approach takes place in the algebraic framework of quantum field theory \cite{Haag}. It is based on the possibility of explicitly constructing quantum fields with weakened localization properties by exploiting the crossing symmetry of the scattering matrix \cite{schroer1999modular}, again considered to be of factorizing type. These auxiliary objects are localized in Minkowski space in wedge shaped regions and are referred to as polarization-free generators due to their simple momentum-space properties. Starting from the construction of these wedge-local operators, one can prove the \textit{existence} of local fields by making use of certain operator-algebraic techniques \cite{BL4}. This procedure has been carried out in the case of a particle spectrum consisting only of a single species of neutral massive particles by Lechner in \cite{L08, erratum}. The constructed models are integrable and asymptotically complete, hence solve the inverse scattering problem. Among them are the Sinh-Gordon model and the scaling Ising model, which are usually realized in terms of a Lagrangian. More interestingly, however, this nonperturbative construction yields a large class of integrable models to which a Lagrangian formulation is not known.\par
On the other hand, interesting models such as the $O(N)$-invariant nonlinear sigma models, which describe a single species of neutral massive particles with an \textit{internal degree of freedom} and are accessible by perturbative renormalization in $1+1$ dimensions \cite{AAR}, do not fit into the framework of the latter approach. This is due to the simple particle spectrum considered, consisting of just one species of neutral massive particles \textit{without} any internal degree of freedom. In view of the achievements of \cite{L08, erratum} for this special situation, it is, therefore, natural to raise the question whether the employed methods can be generalized to the case of more general particle spectra, opening up the possibility to construct, for instance, the $O(N)$ nonlinear $\sigma$-models in a rigorous way for the first time. In this thesis we look into this important issue. We take, in particular, a particle spectrum into account which involves an arbitrary number of massive particle species which transform under some global gauge group. The starting point of the construction of local theories is a scattering \textit{matrix} which factorizes into a product of a number of two-particle S-matrices, i.e. which is of factorizing type. In a paper by Lechner and Sch\"utzenhofer \cite{LS} the first intermediate step of constructing wedge-local field algebras by means of auxiliary quantum fields was shown to be feasible for the considered setting. The second more involved step, the transition to \textit{local} algebras associated with \textit{bounded} regions in Minkowski space, such as double cones, on the other hand, shall be discussed in this thesis.\par
Since in $1+1$ dimensions double cones are intersections of two opposite wedges, it can be inferred that the local algebras are obtained by intersecting wedge algebras. However, it is not straightforward to decide whether or not these intersections are trivial. Fortunately, there are powerful tools available in the operator-algebraic framework of quantum field theory which allow for the clarification of this problem. We shall argue that the so-called \textit{modular nuclearity condition}, requiring the nuclearity of certain maps, is best suited for our purposes. Its verification implies the cyclicity of the vacuum for the local algebras, and, hence, their nontriviality. To establish this result, we shall proceed similarly as in the special case considered in \cite{L08} where the modular nuclearity condition also played a crucial role. Unfortunately, we discovered in the last stages of this thesis a mistake in the proof of \cite[Proposition 4.4 $b)$]{L08}. As a consequence the main results of \cite{L08}, namely Theorems 5.6 and 5.8, do not hold true by the applied arguments. However, there has been progress in reestablishing the claimed verification of the modular nuclearity condition. In fact as shown in \cite{erratum}, Theorem 5.8 can presently be proven in a \textit{weaker} sense, namely such that the arising models comply with localization in bounded regions \textit{above a minimal size}\footnote{In the original version of the corresponding theorem the localization regions were of \textit{arbitrary} size.}. The techniques leading to this great result can, nevertheless, be applied only to a certain class of models with factorizing S-matrices. A much larger family of models was, originally, covered by Theorem 5.6 which so far could not be repaired.\par
These problems also reflect in the more general situation of a richer particle spectrum considered here. For the proof of the modular nuclearity condition we shall, namely, rely on the existence of a certain map for which we have plausible arguments but have no complete proof yet. However, we emphasize that due to the results obtained so far with regard to this mapping there is strong indication for the establishment of a complete proof in the near future, in particular, in connection with the $O(N)$-invariant nonlinear $\sigma$-models.\\

The rigorous construction methods presented by the previous discussion do not give any results in spacetime dimensions greater than \textit{three}. In the context of recently developed deformation procedures \cite{grosse2007wedge, buchholz2011warped}, however, generalizations of the latter inverse scattering approach to higher dimensions are, indeed, possible as the analysis of \cite{GL} shows. The strategy, thereby, is to start from a well-known model and modify it in a suitable way such that the basic properties of quantum field theory like locality or covariance are preserved. The deformation techniques allow for the nonperturbative construction of new quantum field theoretic models with nontrivial interaction. The corresponding field operators are not localized in bounded regions but are wedge-local. Due to this remnant of locality, scattering theory can be applied and the two-particle scattering matrix can consistently be determined \cite{BBS01}. The structure of the resulting scattering operator is very simple. It does not allow for particle creation or momentum transfer in collision processes of particles. However, the great achievement of this approach is that a deformation of a trivial, i.e. interaction-free, theory on $d\geq 2$-dimensional spacetime gives rise to a theory which admits nontrivial scattering as already first examples show \cite{grosse2007wedge, buchholz2008warped}.\par
Building on these achievements, in the second part of this thesis we are concerned with the construction of new nontrivial models in higher spacetime dimensions by means of deformations schemes.\par
Thus, in this thesis we are, on the one hand, concerned with the rigorous construction of \textit{local} theories on two-dimensional Minkowski space, by applying inverse scattering methods. On the other hand, we are interested in constructing models in \textit{higher} dimensions by means of deformation procedures.\par
This thesis is organized as follows. In the next chapter we introduce the general framework. This enables us simultaneously to present the notation used throughout the thesis and to recall the basic concepts of Algebraic Quantum Field Theory necessary for our purposes.\par
In Chapter \ref{Chapter3} we are concerned with the rigorous construction of nontrivial quantum field theoretical models in two spacetime dimensions by means of inverse scattering techniques. We start by clarifying the considered particle spectra which involve several massive particle species carrying arbitrary charges and collect further notation. Moreover, we specify the properties of a factorizing S-matrix, which is the starting point of our construction, and introduce a convenient Hilbert space. The subsequent section in that chapter deals with the construction of local quantum field theoretic models. To this end, we first review the results on the wedge-local auxiliary fields obtained in \cite{LS} and the wedge-local algebras they generate. Next, we prove that the Bisognano-Wichmann property holds true in our setting, constituting an important step with regard to the verification of the modular nuclearity condition. Having established this result, we discuss the transition from wedge-local to local algebras. We further introduce the modular nuclearity condition and clarify its significance for our approach. The main part of Chapter \ref{Chapter3} is concerned with the verification of this condition, which can be established under a very plausible conjecture, implying the existence of local fields in the constructed models. In the last sections to that chapter the physical properties including asymptotic completeness, already shown to hold true in \cite{LS}, of the constructed models are discussed, and concrete examples are illustrated. The most prominent models fitting in this inverse scattering approach are the $O(N)$-invariant nonlinear $\sigma$-models to which we devote Chapter \ref{Chapter4}.\par
In Chapter \ref{Chapter5}, on the other hand, we construct nontrivial models via deformation methods in $d\geq 2$ spacetime dimensions. To this end, we first introduce this approach by reviewing its development. This presentation is followed by specifying the particular model, namely that of a scalar massive Fermion, which constitutes the starting point of our construction. Its deformation is carried out and the properties of the emerging models are discussed. The construction gives rise to theories based on wedge-local fields which constitute examples of so-called tempered polarization-free generators. This property allows for the consistent computation of the two-particle scattering matrix which is shown to differ from the one of the initial, undeformed model in a nontrivial way and depends, furthermore, on the deformation. Due to its simple form the effects of the deformation can only be uncovered in special arrangements such as time delay experiments, however, effects like momentum transfer or particle production cannot be expected to be found in the deformed models.\par
In Chapter \ref{Chapter6} we investigate the deformed models in the special case of two spacetime dimensions. We shall show that under this restriction the deformation of a scalar massive Fermion does not only give rise to a wedge-local model but also to certain \textit{local} theories, complying with localization in bounded regions above a minimal size, constructed previously by inverse scattering methods. This result establishes, thereby, a connection between the two construction approaches presented in this thesis.\par
The family of integrable models with factorizing S-matrices emerging by the applied deformation procedures contains, in particular, the famous Sinh-Gordon model.\par
Finally, we shall discuss our findings and open questions in Chapter \ref{Chapter7}, including an outlook. This completes the main text.\\\\
Appendix \ref{AppendixA} collects auxiliary results needed, in particular, in Chapter \ref{Chapter3}. Appendix \ref{AppendixB}, on the other hand, covers some mathematical background material.
\\\\
Most of the content of Chapters \ref{Chapter5} and \ref{Chapter6} has been published in \cite{alazzawi2013deformations}.

% Chapter2

\chapter{Preliminaries} % Main chapter title

\label{Chapter2} % Change X to a consecutive number; for referencing this chapter elsewhere, use \ref{ChapterX}
\fancyhead[LE,RO]{\thepage}
\fancyhead[LO]{\thesection. \emph{\rightmark}}
\fancyhead[RE]{Chapter 2. \emph{Preliminaries}}
\renewcommand{\chaptermark}[1]{ \markboth{#1}{} }
\renewcommand{\sectionmark}[1]{ \markright{#1}{} }
%\lhead{Chapter 2. \emph{Preliminaries}} % Change X to a consecutive number; this is for the header on each page - perhaps a shortened title
In this thesis we use the system of units in which the speed of light $c$ as well as Planck's constant $\hbar$ are set to be equal to one.
\section{Minkowski Spacetime}
To set the stage, in this section we collect the basic notations and conventions regarding geometrical aspects that shall be used throughout this thesis.\\

The considered spacetime is the Minkowskian of $d\geq 2$ dimensions with coordinates \mbox{$x=(x^0,\vec{x})\in\mathbb{R}\times\mathbb{R}^{d-1}$}, equipped with the metric
\begin{equation}
x\cdot y=x^0y^0-\vec{x}\vec{y}=x^0y^0-\sum_{i=1}^{d-1}x^iy^i,\qquad\forall x,y\in\mathbb{R}^d.
\end{equation}
Minkowski spacetime is divided into
subregions called spacelike, timelike and lightlike according to $x\cdot x<0$,
$x\cdot x>0$, and $x\cdot x=0$ respectively. The group of isometries is given by the Poincaré group $\mathcal{P}$. It is the semidirect product $\mathbb{R}^d\rtimes\mathcal{L}$ with $\mathbb{R}^d$ corresponding to the spacetime translations and $\mathcal{L}$ the full homogeneous Lorentz group. Respectively,
\begin{equation}
(a,\Lambda)\cdot(a',\Lambda')=(a+\Lambda a',\Lambda\Lambda'),\qquad a,a'\in\mathbb{R}^d,\,\, \Lambda,\Lambda'\in\mathcal{L}.
\end{equation}
Moreover, $\mathcal{P}$ is simply connected and admits a splitting into connected components
\begin{equation}\label{poinc}
\mathcal{P}=\mathcal{P}_+^\uparrow\cup\mathcal{P}_-^\uparrow\cup\mathcal{P}_+^\downarrow\cup\mathcal{P}_-^\downarrow,
\end{equation}
where $\mathcal{P}_+^\uparrow$ denotes the proper orthochronous part, consisting of elements $g$ which preserve time orientation, symbolized by $\uparrow$, and for which \mbox{$\det g=+1$}, represented by the sign $+$. Correspondingly, $\downarrow$ stands for group elements changing time orientation and the sign $-$ for those with determinant $-1$. For the following purposes the proper part $\mathcal{P}_+=\mathcal{P}_+^\uparrow\cup\mathcal{P}_+^\downarrow$ shall be of main concern.

\subsection{Wedges and Double Cones}\label{wedges}
In this thesis certain regions in $\mathbb{R}^d$ are of special interest and are introduced in this section. We agree upon the convention to work with open regions only. Denoting by $\mathfrak{R}'$ the spacelike (causal) complement of $\mathfrak{R}\subset\mathbb{R}^d$, defined as the interior of the set $\{x\in\mathbb{R}^d:(x-y)^2<0,\,\,\forall y\in\overline{\mathfrak{R}}\}$, then $\mathfrak{R}$ is said to be causally complete if and only if $\mathfrak{R}=\mathfrak{R}''$. Of great importance is a particular class of causally complete subregions of Minkowski space, namely so-called \textit{wedges}. Each wedge is a Poincaré transform of
\begin{equation}\label{wedge}
W_R:=\{x\in\mathbb{R}^d:x_1>|x_0|\},
\end{equation}
referred to as the \textit{right} wedge with causal complement $W_R'=-W_R=:W_L$, also called the \textit{left} wedge. The set of all wedges is denoted by $\mathcal{W}=\mathcal{P}W_R$. Note that the boost
\begin{equation}\label{boost}
\Lambda_{W_R}:\mathbb{R}\ni t\rightarrow\Lambda_{W_R}(t):=\begin{pmatrix}
\cosh(2\pi t)&\sinh(2\pi t)& 0& \dots& 0\\
\sinh(2\pi t)&\cosh(2\pi t)& 0& \dots& 0\\
0& 0& 1&\dots& 0\\
\vdots &\vdots&\vdots&\ddots&\vdots\\
0& 0& 0&\dots& 1
\end{pmatrix}\in\mathcal{L}_+^\uparrow
\end{equation}
preserves the wedge $W_R$. Moreover, the reflection across the edge of $W_R$, denoted by $j_{W_R}\in\mathcal{P}_+$, acts according to
\begin{equation}\label{reflection}
j_{W_R} (x_0,x_1,\dots,x_{d-1})=(-x_0,-x_1,\dots,x_{d-1}).
\end{equation}
Of further interest are certain bounded regions in $\mathbb{R}^d$, namely open \textit{double cones} $\mathcal{O}$. A double cone is defined to be the intersection of the causal future of a point $x$ with the causal past of a point $y$ to the future of $x$. \begin{figure}[h]
\begin{center}
\begin{minipage}{3cm}
\begin{tikzpicture}
\begin{scope}[->]
\draw (.5,-1)--(.5,1)node[anchor=north]{};
\draw (.5,-1)--(2.5,-1) node[anchor=east]{};
\end{scope}
        \shade[left color=blue!5!white,right color=blue!40!white,opacity=0.3] (-1,0) arc (180:0:1.5cm and 0.3cm) -- (.5,-2) -- cycle;
        \draw (-1,0) arc (180:360:1.5cm and 0.3cm);
        \draw (-1,0) arc (180:0:1.5cm and 0.3cm);
        \draw (-1,0)--(.5,-2)--(2,0);
    \draw[fill] (.7,1.3) node{\footnotesize{time}};
     \draw[fill] (2.7,-0.7) node{\footnotesize{space}};
    %\fill[black] (.5,-2+1) circle (.25ex);
    \end{tikzpicture}
\end{minipage}\hspace{1.7cm}
\begin{minipage}{3cm}
\begin{tikzpicture}
        \shade[left color=blue!5!white,right color=blue!40!white,opacity=0.3] (-1,0) arc (180:360:1.5cm and 0.3cm) -- (.5,2) -- cycle;
        \draw (-1,0) arc (180:360:1.5cm and 0.3cm);
        \draw[dashed] (-1,0) arc (180:0:1.5cm and 0.3cm);
        \draw (-1,0)--(.5,2)--(2,0);
    %\draw[fill] (.7,2.1) node{\footnotesize{$y$}};
    %\fill[black] (.5,2) circle (.25ex);
\end{tikzpicture}
\end{minipage}\hspace{1.7cm}
\begin{minipage}{3cm}
\begin{tikzpicture}
        \shade[left color=blue!5!white,right color=blue!40!white,opacity=0.3] (-5/8,.5) arc (180:0:1.12cm and 0.3cm) -- (.5,-1) -- cycle;
        \draw[gray] (-1,1) arc (180:360:1.5cm and 0.3cm);
        \draw[gray, dashed] (-1,1) arc (180:0:1.5cm and 0.3cm);
 \draw[gray] (-1,1) arc (180:111:1.5cm and 0.3cm);
  \draw[gray] (2,1) arc (0:65:1.5cm and 0.3cm);
        \draw[gray] (-1,1)--(-5/8,.5);
\draw[blue] (-.03,-.29)--(.5,-1);
\draw[blue] (.5,-1)--(1.03,-.29);
\draw[gray] (1+5/8,.5)--(2,1);
\draw[blue, densely dashed] (1.03,-.29)--(1+5/8,.5);
\draw[blue, densely dashed] (-.03,-.29)--(-5/8,.5);
    \draw[fill] (.7,-1.1) node{\footnotesize{$x$}};
    \fill[black] (.5,-1) circle (.25ex);
           \shade[left color=blue!5!white,right color=blue!40!white,opacity=0.3] (-5/8,.5) arc (180:360:1.12cm and 0.3cm) -- (.5,2) -- cycle;
            \draw[gray] (-1,0) arc (180:360:1.5cm and 0.3cm);
    \draw[blue] (-5/8,.5) arc (180:360:1.12cm and 0.3cm);
    \draw[blue, densely dashed] (-5/8,.5) arc (180:0:1.12cm and 0.3cm);
            \draw[gray, dashed] (-1,0) arc (180:0:1.5cm and 0.3cm);
            \draw[gray] (-1,0)--(-5/8,.5);
 \draw[gray, densely dashed](-5/8,.5)--(-5/8+.2,.5+.8/3);
\draw[blue, densely dashed](-5/8,.5)--(-5/8+.2,.5+.8/3);
            \draw[blue] (-5/8+.2,.5+.8/3)--(.5,2);
   \draw[blue] (.5,2)--(-5/8+.2+1.85,.5+.8/3);
  \draw[blue, densely dashed] (-5/8+.2+1.85,.5+.8/3)--(1+5/8,.5);
  \draw[gray] (1+5/8,.5)--(2,0);
        \draw[fill] (.7,2.1) node{\footnotesize{$y$}};
        \fill[black] (.5,2) circle (.25ex);
\end{tikzpicture}
\end{minipage}
\end{center}
\caption{Three-dimensional illustration of a forward lightcone (left), backward lightcone (middle) and a double cone (right).}
\end{figure}
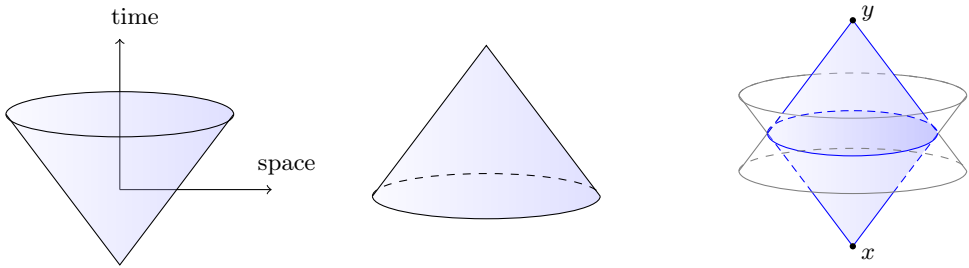
In other words, a double cone is a non-empty intersection of a forward lightcone $V_x^+$ with a backward lightcone $V_y^-$, i.e. $V_x^+\cap V_y^-$ with $y\in V_x^+$. On the other hand, double cones arise also from a suitable intersection of wedges, namely
\begin{equation}\label{o}
\mathcal{O}=\bigcap_{W\supset\mathcal{O}}W,\qquad W\in\mathcal{W}.
\end{equation}
Moreover, the set $\mathcal{W}$ is causally separating for double cones. That is, for every pair of spacelike separated double cones $\mathcal{O}_1$ and $\mathcal{O}_2$ there exists a $W\in\mathcal{W}$, such that
\begin{equation}
\mathcal{O}_1\subset W\subset\mathcal{O}_2'.
\end{equation}
For further purposes, we shall denote the set of all double cones by $\mathscr{O}$.
\subsection{$1+1$ Dimensions}\label{wedges1+1}
At several places we shall restrict our attention to two spacetime dimensions. To this end, we may specialize the notions just introduced to this setting.\par
In case $d=2$ the set $\mathcal{W}$ consists of translates of either $W_R$ or $W_L$, that is,
\begin{equation}\label{setwedges2}
\mathcal{W}=\{W_L+x:x\in\mathbb{R}^2\}\cup\{W_R+x:x\in\mathbb{R}^2\}.
\end{equation}
Note that both $W_R$ and $W_L$ are invariant under the action of the boosts (\ref{boost}).
\begin{figure}[h]
\begin{center}
    \begin{tikzpicture}[scale=1]
                \begin{scope}%[transparent]    %current
                \draw[fill] (-2,1)  node{\small{$W_L$}};
                \shade[left color=blue!5!white,right color=blue!40!white,opacity=0.5] (0,0)--(-1.5,0)--(-1.5,1.5) -- (0,0); 
                    \shade[left color=blue!5!white,right color=blue!40!white,opacity=0.5] (0,0)--(-1.5,0)--(-1.5,-1.5) -- (0,0);
                \end{scope} 
        \begin{scope}%[transparent]    %current
                \draw[fill] (2,1) node{\small{$W_R$}};
        \shade[right color=blue!5!white,left color=blue!40!white,opacity=0.5] (0,0)--(1.5,0)--(1.5,1.5) -- (0,0); 
        \shade[right color=blue!5!white,left color=blue!40!white,opacity=0.5] (0,0)--(1.5,0)--(1.5,-1.5) -- (0,0);
        \begin{scope}[->]
            \draw (-2,0) -- (2,0) node[anchor=north] {};
            \draw (0,-2) -- (0,2) node[anchor=east] {};
        \end{scope}
        \draw[gray] (-1.5,-1.5) -- (1.5,1.5);
        \draw[gray] (-1.5,1.5) -- (1.5,-1.5);
       \end{scope}
       \draw[fill] (2.2,.3) node{\small{$x^1$}};
       \draw[fill] (.3,2.1) node{\small{$x^0$}};
    \end{tikzpicture} 
    \end{center}
    \caption{The left and the right wedge.}
\end{figure}
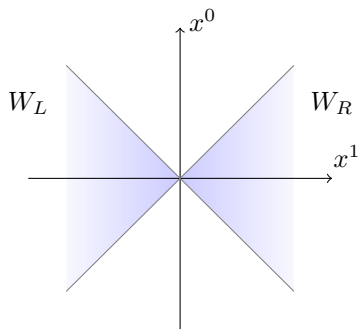

The definition of double cones, stated above for $d\geq 2$, is in two spacetime dimensions equivalently formulated as the non-empty intersection of a right wedge with a left wedge. More precisely,
\begin{equation}\label{double cone}
\mathcal{O}_{x,y}:=(W_R+x)\cap(W_L+y),\qquad y-x\in W_R.
\end{equation}
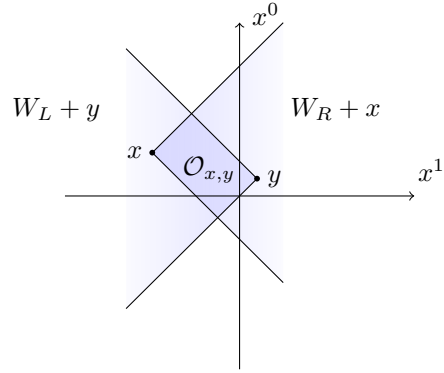
\begin{figure}[h]
 \begin{center}
        \begin{tikzpicture}[scale=1.15]
         \shade[left color=blue!5!white,right color=blue!40!white,opacity=0.3] (-1.3,-1.3)--(.2,.2)--(-1.3,1.7)--cycle;
    \shade[right color=blue!5!white,left color=blue!40!white,opacity=0.3](.5,-1)--(-1,.5)--(.5,2)--cycle;
                        \draw[fill] (-2.1,1) node{\small{$W_L+y$}};
            \draw (-1.3,-1.3) -- (0.2,0.2);
            \draw (0.2,0.2) -- (-1.3,1.7);
                \draw[fill] (1.1,1) node{\small{$W_R+x$}};
                \draw[fill] (-.35,.3) node{\small{$\mathcal{O}_{x,y}$}};
\draw[fill] (-1.2,.5)node{\small{$x$}};
\draw[fill] (.4,.2)node{\small{$y$}};
  \fill[black] (.2,.2) circle (.2ex);
    \fill[black] (-1,.5) circle (.2ex);
            \draw (-1,.5) -- (.5,-1);
            \draw (-1,.5) -- (.5,2);
\begin{scope}[->]
            \draw (-2,0) -- (2,0) node[anchor=north] {};
            \draw (0,-2) -- (0,2) node[anchor=east] {};
        \end{scope}
       \draw[fill] (2.2,.3) node{\small{$x^1$}};
       \draw[fill] (.3,2.1) node{\small{$x^0$}};
        \end{tikzpicture} 
        \end{center}
\caption{A double cone as the intersection of wedges in two spacetime dimensions.}\label{doubleFig}
\end{figure}

\section{Algebraic Approach to Quantum Field Theory}\label{AQFT}
Algebraic Quantum Field Theory \cite{Araki99, Haag, halvorson2006algebraic, horuzhy1990introduction}, synonymously Local Quantum Physics, was originally proposed by Rudolf Haag \cite{haag1955quantum, haag1958quantum} in the mid 1950's as a concept by which scattering of particles can be understood as a consequence of the principle of locality, an expression of Einstein causality in relativistic quantum theory. Later on, a mathematically precise description was established by Araki, Haag and Kastler. The main mathematical methods underlying local quantum physics arise from the theory of operator algebras which are at the basis of this approach. Specifically, a \textit{net} of operator algebras, i.e. a family of operator algebras labeled by regions of spacetime, is of central significance. For our purposes the focus is on Minkowski space $\mathbb{R}^d$, $d\geq 2$. However, the formalism of Algebraic Quantum Field Theory can also be applied to many other spacetimes.\par
We refrain from giving a detailed introduction to this topic and rather limit ourselves to notions and concepts necessary for our purposes. The references cited above provide a thorough overview for further reading.\par
A model in local quantum physics is characterized by a family of algebras $\mathcal{F}(\mathcal{O})$, generated by operators which are localized in the spacetime region $\mathcal{O}\subset\mathbb{R}^d$ and act on a separable Hilbert space $\mathscr{H}$. The algebras $\mathcal{F}(\mathcal{O})\subset\mathcal{B}(\mathscr{H})$ are taken to be $*$-algebras that are closed in the weak operator topology, that is, they are von Neumann algebras. The correspondence
\begin{equation}
\mathbb{R}^d\supset\mathcal{O}\mapsto\mathcal{F}(\mathcal{O})\subset\mathcal{B}(\mathscr{H}),
\end{equation}
together with the requirement of
\begin{itemize}
\item[i)] \textit{Isotony}: $\qquad\mathcal{F}(\mathcal{O}_1)\subset\mathcal{F}(\mathcal{O}_2) \qquad\text{if}\qquad\mathcal{O}_1\subset\mathcal{O}_2,$
\end{itemize}
constitutes a \textit{field net}. The $C^*$-inductive limit for $\mathcal{O}\rightarrow\mathbb{R}^d$ of this directed system of von Neumann algebras is called the \textit{quasilocal field algebra} $\mathcal{F}$, \cite{kadison1986fundamentals}.\par

Besides the von Neumann algebras $\mathcal{F}(\mathcal{O})$, we have a unitary, strongly continuous representation $U$ of the identity component of the Poincaré group $\mathcal{P}_+^\uparrow$ acting on $\mathscr{H}$, which satisfies the
\begin{itemize}
\item[ii)] \textit{Spectrum condition}: The joint spectrum of the energy-momentum operators $P=(P^0,\vec{P})$, the generators of the translations $U(a,\mathbb{1}):=U(a)=e^{i P\cdot a}$, is restricted to the closed forward light cone $\overline{V}_+=\{p\in\mathbb{R}^d:p^0\geq|\vec{p}|\}$.
\end{itemize}
The action of the relativistic symmetries on the net $\{\mathcal{F}(\mathcal{O})\}_{\mathcal{O}\in\mathbb{R}^d}$ is further required to incorporate
\begin{itemize}
\item[iii)]\textit{Covariance}: $\,\,\, \sigma_{(a,\Lambda)}(\mathcal{F}(\mathcal{O}))=U(a,\Lambda)\mathcal{F}(\mathcal{O})U(a,\Lambda)^{-1}=\mathcal{F}(\Lambda\mathcal{O}+a),\quad (a,\Lambda)\in\mathcal{P}_+^\uparrow,$
\end{itemize}
where $\Lambda\mathcal{O}+a=\{\Lambda x+a:x\in\mathcal{O}\}$ and $\sigma_{(a,\Lambda)}$ denotes the automorphisms of $\mathcal{F}$ induced by $U$.\par
Additionally, we assume the existence of a compact Lie group $G$ (the gauge group) and a faithful, strongly continuous unitary representation $V$ of it which induces automorphisms $\alpha_g$, $g\in G$, of $\mathcal{F}$
\begin{equation}
V(g) A V(g)^{-1}=\alpha_g(A),\qquad A\in\mathcal{F}.
\end{equation}
The representations $U$ and $V$ commute and the $\alpha_g$ are assumed to respect the local structure, that is
\begin{itemize}
\item[iv)] \textit{Inner symmetry}: The automorphisms $\alpha_g$ leave each $\mathcal{F}(\mathcal{O})$ globally fixed, i.e.
\begin{equation}
\alpha_g(\mathcal{F}(\mathcal{O}))=\mathcal{F}(\mathcal{O}),\qquad\forall\, g\in G.
\end{equation}
\end{itemize}
In the center of $G$ a Bose-Fermi operator $k$ of order 2 is assumed to exist, giving rise to a decomposition of $A\in\mathcal{F}(\mathcal{O})$ into a bosonic $(+)$ and a fermionic $(-)$ part, i.e. $A=A_++A_-$, where
\begin{equation}
A_\pm:=\frac{1}{2}\left(A\pm\alpha_k(A)\right).
\end{equation}
Thereby $A_+$ is even whereas $A_-$ is odd under the adjoint action of the unitary operator $V(k)$ which further fulfills $V(k)=V(k)^*=V(k)^{-1}$.\par
The principle of causality, the relativistic prohibition of the existence of superluminal signals, is implemented in the framework by demanding for $A\in\mathcal{F}(\mathcal{O}_1)$ and $B\in\mathcal{F}(\mathcal{O}_2)$
\begin{equation}\label{locality}
[A_+,B_+]=[A_+,B_-]=[A_-,B_+]=\{A_-,B_-\}=0,\qquad \text{if}\quad\mathcal{O}_1\subset\mathcal{O}'_2.
\end{equation}
Introducing the twist operation
\begin{equation}
\mathcal{F}(\mathcal{O})^t:=Z\mathcal{F}(\mathcal{O})Z^*,
\end{equation}
with
\begin{equation}\label{twist}
Z:=\frac{1+iV(k)}{1+i},
\end{equation}
then the locality postulate (\ref{locality}) can equivalently be formulated in the following way.
\begin{itemize}
\item[v)] \textit{Twisted locality}:
\begin{equation}\label{locali}
\mathcal{F}(\mathcal{O}_1)^t\subset \mathcal{F}(\mathcal{O}_2)',\qquad \text{if}\quad\mathcal{O}_1\subset\mathcal{O}'_2,
\end{equation}
\end{itemize}
where $\mathcal{F}(\mathcal{O})'$ denotes the \textit{commutant} of $\mathcal{F}(\mathcal{O})$, i.e.
\begin{equation}\label{commutant}
\mathcal{F}(\mathcal{O})':=\{B\in\mathcal{B}(\mathscr{H}):\,[A,B]=0,\,\,\forall A\in\mathcal{F}(\mathcal{O})\}.
\end{equation}
As can be easily verified, see also \cite{foit1983abstract}, we have
\begin{equation}
\mathcal{F}(\mathcal{O})^{tt}=\mathcal{F}(\mathcal{O}),
\end{equation}
since $Z^2=V(k)$, and, moreover,
\begin{equation}
\mathcal{F}(\mathcal{O})^{t\prime} =\mathcal{F}(\mathcal{O})^{\prime t},
\end{equation}
as a consequence of the unitarity of $Z$. A sharpened version of the locality condition is that of \textit{twisted Haag duality}
\begin{equation}\label{haagDuality}
\mathcal{F}(\mathcal{O})^t=\mathcal{F}(\mathcal{O}')',
\end{equation}
which is particularly known to hold \cite{bisognano1976duality} if the region $\mathcal{O}$ in (\ref{haagDuality}) is a wedge and the net is generated by finite-component Wightman fields \cite{streater2000pct}. In general, however, this is not the case for bounded regions.\par
The field net $\{\mathcal{F}(\mathcal{O})\}_{\mathcal{O}\in\mathbb{R}^d}$ contains, in particular, a \textit{subnet} $\{\mathcal{A}(\mathcal{O})\}_{\mathcal{O}\in\mathbb{R}^d}$ of local \textit{observables}. Each $\mathcal{A}(\mathcal{O})$ is defined as the set of fixed points under the action of the global gauge group $G$, that is,
\begin{eqnarray}
\mathcal{A}(\mathcal{O})&=& \mathcal{F}(\mathcal{O})\cap V(G)'\nonumber\\
&=&
\{B\in\mathcal{F}(\mathcal{O}):\alpha_g(B)=B,\,\forall g\in G\}.
\end{eqnarray}
Hence the observable algebra fulfills (untwisted) locality. Moreover, the self-adjoint elements of $\mathcal{A}(\mathcal{O})$ are local relative to the fields, that is, $\mathcal{A}(\mathcal{O}_1)$ and $\mathcal{F}(\mathcal{O}_2)$ commute element-wise if $\mathcal{O}_1\subset\mathcal{O}'_2$. They further correspond to physical properties of the system which may be measured in $\mathcal{O}$. Note that in the theory of superselection sectors \cite{doplicher1971local, doplicher1974local, doplicher1990there} the initial point is the net of observables whereas the gauge group and the field net are derived objects.\par
Furthermore, there exists a $U$- and $V$-invariant, normalized vector $\Omega\in\mathscr{H}$, the physical vacuum state which is unique up to a phase factor. One demands
\begin{itemize}
\item[vi)] \textit{Cyclicity of the vacuum}: $\qquad\bigcup_{\mathcal{O}}\mathcal{F}(\mathcal{O})\Omega$ is dense in $\mathscr{H}$.
\end{itemize}
If the field net $\{\mathcal{F}(\mathcal{O})\}_{\mathcal{O}\in\mathbb{R}^d}$ is generated by Wightman fields then $\Omega$ is a cyclic vector for every single field algebra $\mathcal{F}(\mathcal{O})$ with $\mathcal{O}\subset\mathbb R^d$ being open (Reeh-Schlieder Theorem). This, however, cannot always be expected for nets that are defined without association to Wightman fields. In general, the generation of the algebras $\mathcal{F}(\mathcal{O})$ from quantum field operators is a nontrivial problem, cf. e.g. \cite{driessler1986connection, borchers1992quantum}. Nevertheless, under the additional assumption of
\textit{weak additivity}, that is $$\left(\bigcup_{x\in\mathbb{R}^d}\mathcal{F}(\mathcal{O}_0+x)\right)''=\left(\bigcup_{\mathcal{O}\subset\mathbb{R}^d}\mathcal{F}(\mathcal{O})\right)''$$
for every fixed open set $\mathcal{O}_0$, which particularly is fulfilled if the net $\{\mathcal{F}(\mathcal{O})\}_{\mathcal{O}\in\mathbb{R}^d}$ emerges from Wightman fields, the Reeh-Schlieder Theorem can be proven to hold also in a more general context, see e.g. \cite{Araki99, yngvason2014localization}. That is, under the assumption of weak additivity, we have
\begin{equation}
\overline{\mathcal{F}(\mathcal{O})\Omega}=\mathscr{H},\qquad\mathcal{O}\subset\mathbb{R}^d\,\,\text{open}.
\end{equation}
For regions $\mathcal{O}$ with non-empty causal complement $\mathcal{O}'$, it then follows from the Reeh-Schlieder property that the vacuum $\Omega$ is a separating vector of $\mathcal{F}(\mathcal{O})$ for every $\mathcal{O}$, i.e. $A\Omega=0$ for $A\in\mathcal{F}(\mathcal{O})$ implies $A=0$.\\

The factorial decomposition of the representation $V$ of the gauge group $G$ splits the Hilbert space $\mathscr{H}$ into a
direct sum of orthogonal subspaces $\mathscr{H}_\xi$. The reduced unitary representation
$V_\xi$ of $G$ on $\mathscr{H}_\xi$ is for each $\xi$ factorial, that is, the von Neumann algebra generated by the operators $\{V(g)|_{\mathscr{H}_\xi}:g\in G\}$ is a factor. $V_\xi$ decomposes further (in general non-uniquely) into a direct sum of unitarily equivalent irreducible representations of $G$. The character $\xi$ hence corresponds to a unitary equivalence class of irreducible representations contained in $V$ and is in one-to-one correspondence with the charge quantum numbers. The subspace $\mathscr{H}_\xi$, also referred to as (charge) sector, is generated by those vectors in $\mathscr{H}$ which transform according to this equivalence class $\xi$. So, the factorial decomposition of $V$ is given by
\begin{equation}
\mathscr{H}=\bigoplus_{\xi\in\Sigma}\mathscr{H}_\xi,
\end{equation}
with $\Sigma$ the set of equivalence classes of irreducible unitary representations contained in $V$.\par
From a physical point of view the above assumed structure, regarding gauge invariance of the first kind, does not capture all superselection sectors if long-range forces are present. In Quantum Electrodynamics, for instance, we may have to take a finer decomposition labeled by ``infrared clouds'' \cite{buchholz1982physical} into account. Our assumptions are, however, consistent and sufficient in a theory with only short-range interaction.\\

The remarkable feature of the present approach is that the full physical interpretation of a theory is encoded in the corresponding net $\mathcal{O}\mapsto\mathcal{F}(\mathcal{O})$ of algebras, complying with the above conditions. In order to analyze e.g. the particle spectrum or collision cross sections, no further specification of operators in $\mathcal{F}(\mathcal{O})$ is necessary aside from their localization.\par
Throughout this thesis we shall refer to a net $\mathcal{O}\mapsto\mathcal{F}(\mathcal{O})$ satisfying the properties i), iii), iv), v) and vi) as \textit{local field net} on $\mathbb{R}^d$. With regard to measurements, localization regions are preferably \textit{bounded}. Hence, in the following we shall be particularly interested in algebras associated with compact regions, such as double cones, and refer to them as \textit{local} algebras. Correspondingly, their elements shall be called \textit{local} operators.\\

% Chapter3

\chapter{Inverse Scattering Approach} % Main chapter title

\label{Chapter3} % Change X to a consecutive number; for referencing this chapter elsewhere, use \ref{ChapterX}
\fancyhead[LE,RO]{\thepage}
\fancyhead[LO]{\thesection. \emph{\rightmark}}
\fancyhead[RE]{Chapter 3. \emph{Inverse Scattering Approach}}
\renewcommand{\chaptermark}[1]{ \markboth{#1}{} }
\renewcommand{\sectionmark}[1]{ \markright{#1}{} }
%\lhead{Chapter 3. \emph{Inverse Scattering Approach}} % Change X to a consecutive number; this is for the header on each page - perhaps a shortened title

In this chapter we are concerned with the explicit construction of nontrivial quantum field theoretic models with particle spectra that contain several massive particle species carrying arbitrary charges. The construction takes place within the framework of Local Quantum Physics and is formulated as an inverse scattering problem. The starting point is, therefore, a prescribed S-matrix which in general can be a rather complicated object. In spacetime dimensions greater than two no-go theorems \cite{aks1965proof, coleman1967all} imply that simplifications such as the exclusion of particle production from scattering processes result in the triviality of the S-matrix. In two spacetime dimensions, however, there do exist ``manageable'' S-matrices describing nontrivial interaction, namely those of factorizing type. The factorization of the multi-particle S-matrix into the product of a number of two-particle ones provides a drastic simplification we shall take advantage of. In this chapter the spacetime dimension is, therefore, restricted to \textit{two}.\par
Considering a factorizing S-matrix as given, the first step is to construct wedge-local field algebras by means of auxiliary quantum fields. This task was already carried out in \cite{LS} and shall be reviewed in the following. In a second, more involved step we discuss the possibility of proceeding from these algebras, corresponding to infinitely expanded regions, to a net of \textit{local} field algebras that are associated with bounded regions in Minkowski spacetime.\par
The techniques applied in that procedure result from an analysis that benefits, first of all, from the fact that wedge algebras are thoroughly studied objects in Algebraic Quantum Field Theory. This is, in particular, due to the interpretation of their modular operators as unitary representations of specific Poincaré transformations \cite{BisWich, bisognano1976duality, Bor, buchholz2000geometric}. Based upon this geometric action, families of wedge algebras can be constructed in an algebraic framework \cite{BGL, Mund}. It was, however, Bert Schroer's fundamental insight \cite{schroer1997modular, schroer1999modular} that wedge algebras, complying with the principle of locality, also arise in a setting were a factorizing S-matrix in two spacetime dimensions is taken as the starting point of the construction, initiating thereby an interesting development. A complete construction of these algebras for a simple class of S-matrices was established in \cite{lechner2003polarization}. The transition from field operators localized in wedges to such localized in compact spacetime regions, e.g. double cones, is a rather complicated task. Since any double cone in two dimensions can be obtained by a suitable intersection of wedges, the desired local operators are necessarily elements of the intersection of wedge algebras. The difficulty now is to prove the nontriviality of these intersections. Buchholz and Lechner derived in \cite{BL4} a simple condition, referred to as the \textit{modular nuclearity condition}, by means of which the nontriviality of the double cone algebras can be inferred. At the basis of their investigation is the analysis of spectral properties of the modular operators, which replaces the algebraic problem of checking locality by computing relative commutants \cite{schroer2000modular}. Their arguments are closely linked to the so-called \textit{split property} of wedge algebras. In fact, the modular nuclearity condition implies the split property for wedges which in turn yields the cyclicity of the vacuum for the local algebras, i.e. the nontriviality of the intersections mentioned above. Thus, in order to prove that the inverse scattering point of view we pursue here gives rise to nontrivial local theories, we shall make use of this powerful condition by verifying it for the models at hand. This task does, however, rely on a certain conjecture which at the moment has only been shown to hold true in model theories with certain scalar valued S-matrices \cite{L08, erratum}.\par
This chapter is organized as follows. In the first section we specify the particle spectrum of the models to be constructed and collect some notation. Then the properties of a factorizing S-matrix, which is taken as an input for our construction, are specified. We continue our preparations by constructing a convenient Hilbert space. The subsequent section is concerned with the construction of local quantum field theoretic models by verifying the modular nuclearity condition under a certain conjecture. That section is followed by a discussion of the physical properties of the constructed models and concrete examples are illustrated. The most prominent models fitting in this inverse scattering approach are the $O(N)$-invariant nonlinear $\sigma$-models to which we devote Chapter \ref{Chapter4}.

\section{Framework}
The starting point of our construction is the specification of a certain particle spectrum. This task benefits from Wigner's pioneering analysis of relativistic symmetries in quantum theory \cite{wigner1939unitary}, which was undertaken in 1939. Accordingly, any relativistic formulation of quantum
theory should at least involve a Hilbert space $\mathscr{H}$ of state vectors and a unitary representation $U$ of the proper, orthochronous Poincaré group $\mathcal{P}_+^\uparrow$ on $\mathscr{H}$. Irreducible positive energy representations were thereby identified by Wigner with relativistic particle states corresponding to a certain mass and spin. Recall that a positive energy representation is one that complies with the stability requirement of positive energy in all Lorentz frames, implying that the joint spectrum of the generators \mbox{$P=(P^0,P^1)$} of the representation $U(a,\mathbb{1}):=U(a)=e^{iP\cdot a}$ of the translations is contained in the closed forward light cone $\overline{V}_+=\{p\in\mathbb{R}^2:p^0\geq|p^1|\}$. The observables $P^0$, the energy, and $P^1$, the momentum, give rise to the mass operator $M=[(P^0)^2-(P^1)^2]^{1/2}$ which in a specific irreducible positive energy representation of $\mathcal{P}_+^\uparrow$ has a sharp eigenvalue. Fixing the mass spectrum of the theory, therefore, amounts to choosing a certain unitary equivalence class of the representation $U$. The choice of the representation space $\mathscr{H}$, on the other hand, is to a large extent solely a matter of convenience.\par
By using this usual description of particles, we exclude theories in which long-range forces are present. For those theories allow for particles whose mass fluctuates due to the presence of other excitations and therefore cannot be described by eigenstates of the mass operator. A particular example for such a theory constitutes quantum electrodynamics where charged particles are inevitably accompanied by soft photons \cite{schroer1963infrateilchen}.\par
Concretely, we specify the single particle mass and charge spectra of the models to be constructed as follows. First of all, we consider a compact Lie group as the global gauge group $G$. Equivalence classes $q$ of unitary irreducible representations of $G$ are in one-to-one correspondence with the charge quantum numbers. We shall here be concerned with a finite subset $\mathcal{Q}$ of such quantum numbers. For the sake of concreteness, the analysis is restricted to massive stable theories, which is the case if the restriction of the mass operator to each sector has positive isolated eigenvalues. For simplicity, we limit the number of isolated mass shells in each sector to one, that is, to each charge $q$ there is exactly one mass $m(q)>0$. Our results, however, can be shown to hold also if finitely many isolated mass values are considered in each sector. The mass gap of the theory, i.e. $\min\{m(q)>0:q\in\mathcal{Q}\}$, shall be denoted by $m_\circ$.\par
Since we are working on two-dimensional Minkowski space, we may parameterize the upper mass shell $H_{m(q)}^+=\{((p^2+m(q)^2)^{1/2},p):p\in\mathbb{R}\}$ by the rapidity $\theta$, that is
\begin{equation}\label{p}
p_{m(q)}(\theta):=m(q)
\begin{pmatrix}
\cosh\theta\\
\sinh\theta
\end{pmatrix}
,\qquad\theta\in\mathbb{R}.
\end{equation}
The one-particle Hilbert space $\mathscr{H}_{1}$ can, therefore, be identified with $L^2(\mathbb{R},d\theta)\otimes\mathcal{K}$, where $\mathcal{K}$ is a $D$-dimensional Hilbert space with $D<\infty$. In particular, $\mathscr{H}_1$ decomposes into subspaces of fixed charge $q\in\mathcal{Q}$ and mass $m(q)$, namely
\begin{equation}\label{hilbert1}
\mathscr{H}_1=\bigoplus_{q\in\mathcal{Q}}L^2(\mathbb{R},d\theta)\otimes\mathcal{K}_q.
\end{equation}
The spacetime symmetries, i.e. the proper orthochronous Poincaré group $\mathcal{P}_+^\uparrow$, act on $\mathscr{H}_1$ by means of the unitary, strongly continuous representation
\begin{equation}\label{U1}
U_1(a,t):=\bigoplus_{q\in\mathcal{Q}}\left(U_{1,m(q)}(a,t)\otimes id_{\mathcal{K}_q}\right),\qquad (a,t)\in\mathcal{P}_+^\uparrow,
\end{equation}
which satisfies the relativistic spectrum condition. In particular, $U_{1,m(q)}(a,t)$ is irreducible and we have
\begin{equation}
\left(U_{1,m(q)}(a,t)\psi\right)(\theta):=e^{ip_{m(q)}(\theta)\cdot a}\,\psi(\theta-2\pi t),\qquad \psi\in L^2(\mathbb{R},d\theta)\otimes\mathcal{K}_q.
\end{equation}
Recall that in two spacetime dimensions it is convenient to consider pairs $(a,t)\in\mathcal{P}_+^\uparrow$ consisting of a translation $a\in\mathbb{R}^2$ and a boost $\Lambda(t)$ with rapidity $t\in\mathbb{R}$, i.e.
\begin{equation}
\Lambda(t):=\left(
\begin{array}{cc}
\mathrm{cosh}(2\pi t) & \mathrm{sinh}(2\pi t)\\
\mathrm{sinh}(2\pi t) & \mathrm{cosh}(2\pi t)
\end{array}\right).
\end{equation}
The global gauge group $G$, on the other hand, is represented on $\mathscr{H}_1$ by unitaries
\begin{equation}\label{V1}
V_1(g):=\bigoplus_{q\in\mathcal{Q}}\left(id_{L^2(\mathbb{R},d\theta)}\otimes V_{1,q}(g)\right),\qquad g\in G,
\end{equation}
and acts on $\mathcal{K}_q$ via the irreducible representation $V_{1,q}$. It is further obvious that $V_1$ commutes with $U_1$.\\

It will be useful to consider an orthonormal basis for each $\mathcal{K}_q$. Then their direct sum, denoted by $\{e^\alpha:\alpha=1,\dots,D\}$, constitutes an orthonormal basis in $\mathcal{K}$. Each index $\alpha$ thus corresponds to a certain charge $q_{[\alpha]}$ and mass $m_{[\alpha]}:=m(q_{[\alpha]})$, and $\theta\mapsto \psi^\alpha(\theta)$ denotes the respective component of a vector $\psi\in\mathscr{H}_1$. With regard to charge conjugation, the involution $\alpha\mapsto \overline{\alpha}$, corresponding to a permutation of $\{1,\dots,D\}$, such that $q_{[\overline{\alpha}]}=\overline{q_{[\alpha]}}$ does the job.\par
Furthermore, let $(\cdot,\cdot)$ denote the scalar product in $\mathcal{K}$. We define for vectors $v\in\mathcal{K}^{\otimes n}$ and tensors $M:\mathcal{K}^{\otimes m}\rightarrow\mathcal{K}^{\otimes n}$, $m,n\in\mathbb{N}$,
\begin{eqnarray}
v^{\alpha_1\dots\alpha_n}&:=&(e^{\alpha_1}\otimes\cdots\otimes e^{\alpha_n},v)\\
M^{\alpha_1\dots\alpha_n}_{\beta_1\dots\beta_m}&:=&(e^{\alpha_1}\otimes\cdots\otimes e^{\alpha_n},Me^{\beta_1}\otimes\cdots\otimes e^{\beta_m}).
\end{eqnarray}
Then, for $M:\mathcal{K}^{\otimes m}\rightarrow\mathcal{K}^{\otimes n}$ and $N:\mathcal{K}^{\otimes l}\rightarrow\mathcal{K}^{\otimes m}$, $l\in\mathbb{N}$,
\begin{equation}\label{MN}
\left(M\cdot N\right)^{\alpha_{1}\dots\alpha_n}_{\beta_{1}\dots\beta_{l}}=\sum_{\gamma_1\dots\gamma_m}M^{\alpha_{1}\dots\alpha_n}_{\gamma_1\dots\gamma_m}N^{\gamma_1\dots\gamma_m}_{\beta_{1}\dots\beta_{l}}.
\end{equation}
Moreover, for $M\in\mathcal{B}(\mathcal{K}^{\otimes 2})$ and $n\geq 2$ a useful notation will be
\begin{equation}\label{kurznotation}
M_{n,k}:=1_{k-1}\otimes M\otimes 1_{n-k-1},\qquad k=1,\dots,n-1,
\end{equation}
where $1_j$ denotes the identity on $\mathcal{K}^{\otimes j}$ and $M_{n,k}\in\mathcal{B}(\mathcal{K}^{\otimes n})$. We, further, represent a tensor $M\in\mathcal{B}(\mathcal{K}^{\otimes m},\mathcal{K}^{\otimes n})$ by means of the basis tensors $e^{\alpha_1}\otimes\cdots\otimes e^{\alpha_n}\otimes e^*_{\beta_1}\otimes\cdots\otimes e^*_{\beta_m}$, namely
\begin{equation}
M=\sum_{\boldsymbol{\alpha},\boldsymbol{\beta}}\,M^{\alpha_1\dots\alpha_n}_{\beta_1\dots\beta_m}\,e^{\alpha_1}\otimes\cdots\otimes e^{\alpha_n}\otimes e^*_{\beta_1}\otimes\cdots\otimes e^*_{\beta_m},
\end{equation}
where $\{e^*_\beta:\beta=1,\dots,D\}$ is the dual basis of $\{e^\beta:\beta=1,\dots,D\}$ and $\boldsymbol{\alpha}:=(\alpha_1,\dots,\alpha_n)$. Using this notation we introduce for $1\leq k,l\leq m$ the following operation
\begin{equation}\label{trace}
\begin{aligned}
\hspace{-1cm}&\hspace{-1cm}\text{Tr}\,^k\,_l:\mathcal{B}(\mathcal{K}^{\otimes m},\mathcal{K}^{\otimes n})\rightarrow\mathcal{B}(\mathcal{K}^{\otimes (m-2)},\mathcal{K}^{\otimes n}),\\
& M\mapsto \sum_{\boldsymbol{\alpha},\boldsymbol{\beta}}\delta^{\beta_k\overline{\beta}_{m-l+1}}\,M^{\alpha_1\,\,\,\,\,\,\dots\,\,\,\,\,\,\alpha_n}_{\beta_1\dots\beta_k\dots\beta_{m-l+1}\dots\beta_m}\\
\hspace{.5cm}&\hspace{.5cm}\times\, e^{\alpha_1}\otimes\dots\otimes e^{\alpha_n}\otimes e^*_{\beta_1}\otimes\cdots\widehat{e^*_{\beta_k}}\otimes\cdots\otimes\widehat{e^*_{\beta_{m-l+1}}}\otimes\cdots\otimes e^*_{\beta_m},
\end{aligned}
\end{equation}
where the hat indicates omission of the corresponding factors. The symbol Tr$\,^k\,_l$ shall not be confused with the conventional trace, denoted by Tr, of a tensor which yields a scalar quantity. The presence of the indices $k$ and $l$ stresses this fact. We further define
\begin{equation}
\text{Tr}\,^{k_1,\dots,k_i}_{l_1,\dots,l_i}(M):=\text{Tr}\,^{k_1}\,_{l_1}\left(\text{Tr}\,^{k_2}\,_{l_2}\left(\dots\left(\text{Tr}\,^{k_i}\,_{l_i}\left(M\right)\right)\right)\right)
\end{equation}
for suitable $k_i,l_i$ and $M$.\par
For later purposes, we agree upon the following convention. Namely, considering some sequence $\{b_k\}_{1\leq k\leq n}$, $n\in\mathbb{N}$, of arbitrary algebraic objects $b_k$ not necessarily commuting, then
\begin{equation}\label{convention}
\prod_{k=1}^{n}b_k:=b_1\cdot b_2\cdots b_{n-1}\cdot b_n.
\end{equation}
We shall use this convention of ordered products throughout this thesis without further mentioning.\\

By second quantization, we have natural representations $\widehat{U}:=\bigoplus_{n=0}^\infty U_1^{\otimes n}$ of $\mathcal{P}_+^\uparrow$ and $\widehat{V}:=\bigoplus_{n=0}^\infty V_1^{\otimes n}$ of $G$ on the unsymmetrized Fock space
\begin{equation}\label{fockUnsy}
\widehat{\mathscr{H}}:=\bigoplus_{n=0}^\infty\mathscr{H}_1^{\otimes n}\simeq\bigoplus_{n=0}^\infty(L^2(\mathbb{R}^n,d^n\theta)\otimes\mathcal{K}^{\otimes n})
\end{equation}
over the single particle space $\mathscr{H}_1$.
In particular, we have for $\Psi\in\widehat{\mathscr{H}}$
\begin{equation}\label{actionPoincare}
\left(\widehat{U}(a,t)\Psi\right)_n^{\boldsymbol{\alpha}}(\boldsymbol{\theta})=e^{i\sum_{k=1}^{n}p_{m_{[\alpha_k]}}(\theta_k)\cdot a}\Psi_n^{\boldsymbol{\alpha}}(\theta_1-2\pi t,\dots,\theta_n-2\pi t),\qquad(a,t)\in\mathcal{P}_+^\uparrow.
\end{equation}
As $\widehat{\mathscr{H}}$ is quite a large space, one, usually, chooses to work on an appropriate subspace of it. Standard, well-known examples of subspaces are the totally symmetric Bose Fock space $\mathscr{H}^+=\bigoplus_{n=0}^\infty(L^2(\mathbb{R}^n,d^n\theta)\otimes\mathcal{K}^{\otimes n})^+$ or the totally antisymmetric Fermi Fock space $\mathscr{H}^-=\bigoplus_{n=0}^\infty(L^2(\mathbb{R}^n,d^n\theta)\otimes\mathcal{K}^{\otimes n})^-$, where the $+$ denotes total symmetrization and $-$ total antisymmetrization respectively. The analysis of \cite{DocL, L08}, however, suggests that for our purposes a more convenient choice is a subspace with a symmetrization which is obtained by means of a factorizing S-matrix, an essential input to our construction.

\section{Factorizing S-Matrices in Two Dimensions}\label{Sectionfactorizing}
To begin with, we first recall the basic formalism of the S-matrix theory, see e.g. \cite{iagolnitzer1978s, Araki99}.\par
Considering a particle spectrum as described in the previous section, Haag-Ruelle-Hepp scattering theory \cite{Araki99, hepp1965} reconciled with the charge structure being present \cite{doplicher1974local} can be applied and multiparticle collision states can be computed\footnote{Note that the assumption of isolated mass shells in each sector is not necessary for this task, for there are methods \cite{dybalski2005haag} by means of which scattering states can be calculated also for more general mass spectra.}. Assuming for the sake of clarity for the moment that the theory under consideration is purely bosonic, then, there are isometries $W_{\text{in}}$ and $W_{\text{out}}$, the generalized M\o ller operators, which map the Bose Fock space $\mathscr{H}^+$ onto certain subspaces $\mathscr{H}_{\text{in/out}}$ of the full Hilbert space $\mathscr{H}$ of the theory. The spaces $\mathscr{H}_{\text{in}}$ and $\mathscr{H}_{\text{out}}$ are generated by vectors which can be interpreted as in-, respectively, outgoing configurations of noninteracting particles. The S-matrix $\widehat{\textbf{S}}$ is now introduced as the operator $\widehat{\textbf{S}}:=W_{\text{in}}W_{\text{out}}\,^*:\mathscr{H}_{\text{out}}\rightarrow\mathscr{H}_{\text{in}}$ that maps outgoing into incoming scattering states. It is often, however, convenient to define the S-matrix as an operator on $\mathscr{H}^+$, that is, to consider
\begin{equation}\label{S}
\mathbf{S}:=W_{\text{out}}\,^*W_{\text{in}}:\mathscr{H}^+\rightarrow\mathscr{H}^+.
\end{equation}
Note that we have $\mathbf{S}=W_{\text{out}}\,^*W_{\text{in}}=W_{\text{in}}\,^*\,\widehat{\textbf{S}}\,W_{\text{in}}=W_{\text{out}}\,^*\,\widehat{\textbf{S}}\,W_{\text{out}}$.\par
If all states in the physical Hilbert space $\mathscr{H}$ can be interpreted as configurations of particles, i.e. $\mathscr{H}_{\text{in}}=\mathscr{H}_{\text{out}}=\mathscr{H}$, the theory has the property of asymptotic completeness. In this case the operator $\textbf{S}$ is unitary on $\mathscr{H}^+$ and $\widehat{\textbf{S}}$ is on $\mathscr{H}$ respectively.\par
The problem of a complete particle interpretation in relativistic quantum field theory is a matter of active research, see e.g. \cite{dybalski2011asymptotic, dybalski2013criterion, dybalski2014asymptotic}. Due to several conceptual and technical difficulties, the only non-trivial class of models known to be asymptotically complete are two-dimensional theories with factorizing S-matrices \cite{L08, tanimoto2014construction}. In fact, as is shown in \cite{LS} under a certain assumption, see also Section \ref{AC}, the models to be constructed here do also belong to this class if a certain conjecture, to be specified later on, holds true.\par
Considering \mbox{$\mathscr{H}^+=\bigoplus_{n=0}^\infty(L^2(\mathbb{R}^n,d^n\theta)\otimes\mathcal{K}^{\otimes n})^+$} as above, we find for the S-matrix elements
\begin{equation}
\langle\Psi,\textbf{S}\,\Phi\rangle=\sum_{n,m=0}^{\infty}\int d^n\boldsymbol{\theta}\int d^m\boldsymbol{\theta}'\Big(\Psi_n(\boldsymbol{\theta}),S_{n,m}(\boldsymbol{\theta};\boldsymbol{\theta}')\cdot\Phi_m(\boldsymbol{\theta}')\Big)_{\mathcal{K}^{\otimes n}},
\end{equation}
where $\boldsymbol{\theta}:=(\theta_1,\dots,\theta_n)$, and $\Psi_n\in\mathscr{H}_n^+$ and $\Phi_m\in\mathscr{H}_m^+$ are the wavefunctions of the respective asymptotic states. The kernels $S_{n,m}$ are tempered distributions on $\mathscr{S}(\mathbb{R}^{n+m})\otimes\mathcal{K}^{\otimes(n+m)}$. Their form is determined by constraints such as energy-momentum conservation or covariance. Further properties like crossing symmetry and hermitian analyticity, which are related to the analytic features of the collision operator \cite{edenSMatrix, iagolnitzer1993scattering}, restrict, in particular, the structure of the two-particle S-matrix elements $S_{2,2}$. In fact, these two-particle kernels are thoroughly studied objects in the literature, see e.g. \cite{dorey1997exact, edenSMatrix, iagolnitzer1993scattering}. On the other hand, due to their general complexity less is known about the higher S-matrix elements $S_{n,m}$, $n,m>2$. As already mentioned in the introduction to this chapter, imposing certain simplifications on the S-matrix necessarily leads to its triviality at least in spacetime dimensions greater than two. This follows from several no-go theorems. For instance, if only elastic scattering is required, i.e. $S_{n,m}=0$ for $n\neq m$, then it can be proven that in such a theory there is no interaction at all \cite{aks1965proof}. As shown by Coleman and Mandula in 1967 \cite{coleman1967all}, the S-matrix has also to be trivial if there are conserved quantities in scattering processes which transform as higher Lorentz tensors. These triviality theorems, however, do not apply in $1+1$ dimensions where ``simple'' but nontrivial S-matrices exist. Collision operators of this type have first been discovered in a nonrelativistic setting where the scattering of particles, interacting through a $\delta$-potential, was considered \cite{berezin1964schrodinger, mcguire1964study}. They further were found to appear in quantized versions of completely integrable classical field
theories such as the $O(N)$-invariant nonlinear $\sigma$-models \cite{Zam78, zamolodchikov1979factorized}. Since completely integrable models admit an infinite number of conservation laws, the dynamics is severely restricted. The particle number is, in particular, a conserved quantity in scattering processes, even though the dynamics is fully relativistic. The structure of the collision operator of such a theory is, clearly, limited to a large extent. In fact, the S-matrix has to be consistent with \cite{dorey1997exact}
\begin{itemize}
\item no particle production, i.e. $S_{n,m}=0$ for $n\neq m$,
\item factorization of the kernels $S_{n,n}$ into a product of $S_{2,2}$ kernels,
\item equality of the sets of incoming and outgoing rapidities, i.e.\\ $S_{n,n}(\theta_1,\dots,\theta_n; \theta_1',\dots,\theta_n')=0$ unless $\{\theta_1,\dots,\theta_n\}=\{\theta_1',\dots,\theta_n'\}$,
\item no interaction of particles with different masses.
\end{itemize}
Due to the second property, scattering operators complying with these special properties are called factorizing S-matrices. It is, then, natural to focus the attention on the two-particle collision operator. The energy-momentum conservation law for scattering process of two incoming particles of types $\alpha, \beta$ with rapidities $\theta_1,\theta_2$ and two outgoing particles $\gamma,\delta$ with rapidities $\theta_1',\theta_2'$ yields the appearance of the following delta distributions in
\begin{equation}
\begin{aligned}
\left(S_{2,2}\right)^{\alpha\beta}_{\gamma\delta}(\theta_1,\theta_2;\theta_1',\theta_2')&=\frac{1}{2}\left(\delta(\theta_1-\theta_1')\delta(\theta_2-\theta_2')+\delta(\theta_1-\theta_2')\delta(\theta_2-\theta_1')\right)S_2(\theta_1,\theta_2)^{\alpha\beta}_{\gamma\delta},\\
S_2(\theta_1,\theta_2)^{\alpha\beta}_{\gamma\delta} &=\left\{\begin{array}{cc}
S^{\beta\alpha}_{\gamma\delta}(\theta_1-\theta_2),& \quad \theta_1>\theta_2\\
S^{\alpha\beta}_{\delta\gamma}(\theta_2-\theta_1),& \quad \theta_2\geq \theta_1.
\end{array}\right.
\end{aligned}
\end{equation}
By Lorentz invariance, the scattering amplitude $S(|\theta|)$ depends only on the difference of rapidities $\theta=\theta_1-\theta_2$. Furthermore, standard S-matrix features, such as unitarity, covariance or the Yang-Baxter equation restrict the structure of $S$. We shall in the following be mainly concerned with this quantity. It is, therefore, convenient to refer to it as S-matrix for short.  With regard to the position of the indices of $S$, different conventions appear throughout the literature. Our choice is the same as in \cite{LS}.\par
The particular set $\mathcal{S}$ of S-matrices which are at the basis of the present construction of nontrivial models is specified in the subsequent definition.
\begin{definition}\label{S-matrixDefinition}
A continuous bounded function $S:\{\zeta\in\mathbb{C}:0\leq\text{Im}\,\zeta\leq\pi\}\rightarrow\mathcal{B}(\mathcal{K}\otimes\mathcal{K} )$, analytic in $\{\zeta\in\mathbb{C}:0<\text{Im}\,\zeta<\pi\}$, is referred to as S-matrix if for $\theta,\theta'\in\mathbb{R}$ and $\alpha,\beta,\gamma,\delta\in\{1,\dots,D\}$ the following properties are fulfilled:
\begin{enumerate}
\item[1.)]  Unitarity:\hspace{.7cm} $$S(\theta)^*=S(\theta)^{-1}.$$
\item[2.)] Hermitian analyticity:\hspace{.7cm} $$S(\theta)^{-1}=S(-\theta).$$
\item[3.)] Yang-Baxter equation:\\
$$(S(\theta)\otimes 1_1)(1_1\otimes S(\theta+\theta'))(S(\theta')\otimes 1_1)=(1_1\otimes S(\theta'))(S(\theta+\theta')\otimes 1_1)(1_1\otimes S(\theta)).$$
\item[4.)] Crossing symmetry:\hspace{.7cm} $$S^{\alpha\beta}_{\gamma\delta}(i\pi-\theta)=S^{\overline{\gamma}\alpha}_{\delta\overline{\beta}}(\theta).$$
\item[5.)] PCT invariance:\hspace{.7cm} $$S^{\alpha\beta}_{\gamma\delta}(\theta)=S^{\overline{\delta}\overline{\gamma}}_{\overline{\beta}\overline{\alpha}}(\theta).$$
\item[6.)] Translational invariance\footnote{Although not obvious at first sight, this is, indeed, the right condition for translational invariance of $S$. It becomes obvious by taking the action of the translations, cf. (\ref{actionPoincare}), and (\ref{reprpermugroup}) stated below into account.}:\hspace{.7cm} $$S^{\alpha\beta}_{\gamma\delta}(\theta)=0 \quad\text{if}\quad m_{[\alpha]}\neq m_{[\delta]}\quad \text{or}\quad m_{[\beta]}\neq m_{[\gamma]}.$$
\item[7.)] Gauge invariance:\hspace{.7cm} $$[S(\theta),V_1(g)\otimes V_1(g)]=0\qquad g\in G.$$
\end{enumerate}
\end{definition}
The set $\mathcal{S}$ has so far been explicitly determined only for a class of \textit{neutral} particles with the same mass $m>0$ \cite{DocL}. In this case the Yang-Baxter equation, translational invariance and gauge invariance are trivially fulfilled. Apart from this scalar case, other special solutions for the set of constraints \ref{S-matrixDefinition} were found for e.g. $O(N)$ $\sigma$-models \cite{AAR,Zam78}.\par
In conclusion, we may point out that, despite the severe limitations on the interaction in models governed by factorizing S-matrices, there do exist observable effects in scattering processes. In particular, the nonconstant nature of the phase shift of $S$ gives rise to the appearance of time delays.\par

\section{S-Symmetric Fock Space}\label{Fock}
As indicated earlier, for our purposes it is practical to choose the Hilbert space of the theory to be S-symmetric. Such a Fock space is constructed by introducing an $S$-dependent action $D_n$ of the permutation group $\mathfrak{S}_n$ of $n$ elements on $\mathscr{H}_1^{\otimes n}$ \cite{DocL}. That is, we put
\begin{equation}\label{reprpermugroup}
\left(D_n(\tau_k)\Psi_n\right)(\theta_1,\dots,\theta_n):=S(\theta_{k+1}-\theta_k)_{n,k}\Psi_n(\theta_1,\dots,\theta_{k+1},\theta_k,\dots,\theta_n),\quad\Psi_n\in\mathscr{H}_1^{\otimes n},
\end{equation}
with $\tau_k\in\mathfrak{S}_n$, $k=1,\dots,n-1$, being the transposition that exchanges $k$ and $k+1$. Accordingly, for arbitrary permutations $\pi\in\mathfrak{S}_n$ there exists a unitary tensor $S^\pi_n:\mathbb{R}^n\rightarrow\mathcal{U}(\mathcal{K}^{\otimes n})$ such that
\begin{equation}\label{tensor}
\left(D_n(\pi)\Psi_n\right)(\boldsymbol{\theta})=S^\pi_n(\boldsymbol{\theta})\Psi_n(\theta_{\pi(1)},\dots,\theta_{\pi(n)}),\qquad\Psi_n\in\mathscr{H}_1^{\otimes n}.
\end{equation}
Obviously, $S^{\tau_k}_n(\boldsymbol{\theta})=S(\theta_{k+1}-\theta_k)_{n,k}$ holds. Due to the properties of the S-matrix, particularly due to the features 1.) to 3.) in Definition \ref{S-matrixDefinition}, we have
\begin{lemma}[\cite{LM}]
$D_n$ is a unitary representation of the permutation group $\mathfrak{S}_n$ on $\mathscr{H}_1^{\otimes n}$.
\end{lemma}
Moreover, the mean over $D_n$,
\begin{equation}\label{Pn}
P_n:=\frac{1}{n!}\sum_{\pi\in\mathfrak{S}_n}D_n(\pi),
\end{equation}
is the orthogonal projection onto the $D_n$-invariant subspace of $\mathscr{H}_1^{\otimes n}$ \cite{DocL}. The S-symmetrized Fock space $\mathscr{H}$ over $\mathscr{H}_1$ is then defined by
\begin{equation}\label{s}
\mathscr{H}:=\mathbb{C}\oplus\bigoplus_{n=1}^\infty\mathscr{H}_n,\qquad\mathscr{H}_n:=P_n\mathscr{H}_1^{\otimes n},
\end{equation}
where $\mathbb{C}$ consists of multiples of the vacuum vector $\Omega$. Thus, vectors $\Psi_n\in\mathscr{H}_n$ satisfy
\begin{equation}\label{sum}
\Psi_n^{\boldsymbol{\alpha}}(\boldsymbol{\theta})=S^{\alpha_k\alpha_{k+1}}_{\beta_{k+1}\beta_k}(\theta_{k+1}-\theta_k)\Psi_n^{\alpha_1\dots\alpha_{k-1}\beta_{k+1}\beta_k\alpha_{k+2}\dots\alpha_n}(\theta_1,\dots,\theta_{k+1},\theta_k,\dots,\theta_n),
\end{equation}
where $\boldsymbol{\alpha}=(\alpha_1,\dots,\alpha_n)$ and $\boldsymbol{\theta}=(\theta_1,\dots,\theta_n)$. Note that the summation convention is used in (\ref{sum}) and in the following. Moreover, for elements $\Psi=(\Psi_0,\Psi_1,\dots)\in\mathscr{H}$ we have $\|\Psi\|^2=\sum_{n=0}^{\infty}\int d^n\boldsymbol{\theta}\overline{\Psi_n^{\boldsymbol{\alpha}}(\boldsymbol{\theta})}\Psi_n^{\boldsymbol{\alpha}}(\boldsymbol{\theta})<\infty$.\par
For further purposes we introduce the particle number operator $N$ on $\mathscr{H}$ by
\begin{equation}\label{opN}
(N\Psi)_n:=n\Psi_n,
\end{equation}
for vectors with $\sum_n n^2\|\Psi_n\|^2<\infty$, and refer to the dense subspace $\mathcal{D}\subset\mathscr{H}$, consisting of terminating sequences $(\Psi_0,\Psi_1,\dots,\Psi_n,0,\dots)$, as the subspace of finite particle number.\par
We denote the restrictions of the representations $\widehat{U}=\bigoplus_{n=0}^\infty U_1^{\otimes n}$ of $\mathcal{P}_+^\uparrow$ and $\widehat{V}=\bigoplus_{n=0}^\infty V_1^{\otimes n}$ of $G$ on the unsymmetrized Fock space $\widehat{\mathscr{H}}$, (\ref{fockUnsy}), to the subspace $\mathscr{H}$ by
\begin{equation}\label{rep}
U:=\widehat{U}\big|_{\mathscr{H}},\qquad V:=\widehat{V}\big|_{\mathscr{H}}.
\end{equation}
Clearly, $U$ is a strongly continuous positive energy representation of $\mathcal{P}_+^\uparrow$, with up to a phase unique invariant unit vector $\Omega$, legitimizing thereby the interpretation of the latter as the physical vacuum state. The PCT operator $J$ on $\mathscr{H}$ is further defined by
\begin{equation}
\left(J\Psi\right)_n^{\boldsymbol{\alpha}}(\boldsymbol{\theta}):=\overline{\Psi_n^{\overline{\alpha}_n\dots\overline{\alpha}_1}(\theta_n,\dots,\theta_1)}, \qquad \Psi\in\mathscr{H}.
\end{equation}
It is the antiunitary involution which extends $U$ to a representation of the proper Poincaré group $\mathcal{P}_+$ as stated in the following Lemma.
\begin{lemma}\cite{schutzenhofer2011multi}\label{extendedU} The operator $J$ is an antiunitary involution. Moreover, let $j(x)=-x$ be the space-time reflection defined in (\ref{reflection}). Then, the assignment
\begin{equation}
U(j):=J
\end{equation}
extends $U$ to a representation of the proper Poincaré group $\mathcal{P}_+$ on $\mathscr{H}$.
\end{lemma}
The PCT operator $J$ further commutes with the representation $V$ as shown in \cite[Lemma 2.3.]{LS}. For our purposes it is not relevant to extend $U$ to the full Poincaré group $\mathcal{P}$. This is related to the fact that we shall not construct models with S-matrices being invariant under the symmetries of parity and time reflection separately.
\begin{lemma}\label{eigenschaften}
Let $n\in\mathbb{N}$, then PCT invariance and crossing symmetry of the S-matrix yield
\begin{itemize}
\item[a)]
\begin{equation}\label{PCT}
\left[\prod_{k=a}^{b}S(\theta_k)_{n,k}\right]^{\alpha_1\dots\alpha_{n}}_{\beta_1\dots\beta_{n}}=\left[\prod_{k=b}^{a}S(\theta_k)_{n,n-k}\right]^{\overline{\beta}_{n}\dots\overline{\beta}_1}_{\overline{\alpha}_{n}\dots\overline{\alpha}_{1}},\qquad 1\leq a\leq b\leq n-1,
\end{equation}
\end{itemize}
and
\begin{itemize}
\item[b)]
\begin{equation}\label{crossing}
\left[\prod_{k=1}^{n}S(i\pi-\theta_k)_{n+1,k}\right]^{\alpha_1\dots\alpha_{n+1}}_{\beta_1\dots\beta_{n+1}}=\left[\prod_{k=1}^{n}S(\theta_k)_{n+1,n+1-k}\right]^{\overline{\beta}_{n}\dots\overline{\beta}_1\alpha_1}_{\beta_{n+1}\overline{\alpha}_{n+1}\dots\overline{\alpha}_{2}},
\end{equation}
\end{itemize}
respectively.
\end{lemma}
\begin{proof}
Considering first part $a)$, we proceed by induction in $b$ with $a$ fixed. Then, for the base case $b=a$ the identity
\begin{equation*}
\left[S(\theta)_{n,a}\right]^{\alpha_1\dots\alpha_{n}}_{\beta_1\dots\beta{n}}=\left[S(\theta)_{n,n-a}\right]_{\overline{\alpha}_n\dots\overline{\alpha}_{1}}^{\overline{\beta}_n\dots\overline{\beta}_{1}}
\end{equation*}
holds, since
\begin{eqnarray*}
\hspace{-.3cm}&&\hspace{-.3cm}\left(JD_n(\tau_j)\Psi_n\right)^{\boldsymbol{\alpha}}(\boldsymbol{\theta})\\
&&= \overline{S_{\overline{\beta}_{n-j}\overline{\beta}_{n-j+1}}^{\overline{\alpha}_{n-j+1}\overline{\alpha}_{n-j}}(\theta_{n-j}-\theta_{n-j+1})\Psi_n^{\overline{\alpha}_n\dots\overline{\beta}_{n-j}{\overline{\beta}_{n-j+1}\dots\overline{\alpha}_{1}}}(\theta_n,\dots,\theta_{n-j},\theta_{n-j+1},\dots,\theta_1)}\\
&&= S_{\beta_{n-j+1}\beta_{n-j}}^{\alpha_{n-j}\alpha_{n-j+1}}(\theta_{n-j+1}-\theta_{n-j})\overline{\Psi_n^{\overline{\alpha}_n\dots\overline{\beta}_{n-j}{\overline{\beta}_{n-j+1}\dots\overline{\alpha}_{1}}}(\theta_n,\dots,\theta_{n-j},\theta_{n-j+1},\dots,\theta_1)}\\
&&=\left(D_n(\tau_{n-j})J\Psi_n\right)^{\boldsymbol{\alpha}}(\boldsymbol{\theta}),
\end{eqnarray*} with $\Psi_n\in\mathscr{H}_1^{\otimes n}$, yields
\begin{equation}\label{PCT1}
\Big(S(\theta)_{n,j}\Big)^{\alpha_1\dots\alpha_n}_{\beta_1\dots\beta_n}=\Big(S(\theta)_{n,n-j}\Big)^{\overline{\beta}_n\dots\overline{\beta}_1}_{\overline{\alpha}_n\dots\overline{\alpha}_1},\qquad j=1,\dots,n-1.
\end{equation}
Turning now to the inductive step $b\rightarrow b+1$, we compute
\begin{eqnarray*}
\left[\prod_{k=a}^{b+1}S(\theta_k)_{n,k}\right]^{\alpha_1\dots\alpha_{n}}_{\beta_1\dots\beta{n}}=\left[\prod_{k=a}^{b}S(\theta_k)_{n,k}\right]^{\alpha_1\dots\alpha_{n}}_{\xi_1\dots\xi_{n}}\left[S(\theta_{b+1})_{n,b+1}\right]^{\xi_1\dots\xi_{n}}_{\beta_1\dots\beta{n}},
\end{eqnarray*}
which by the inductive hypothesis and by (\ref{PCT1}) becomes
\begin{eqnarray*}
\left[\prod_{k=a}^{b+1}S(\theta_k)_{n,k}\right]^{\alpha_1\dots\alpha_{n}}_{\beta_1\dots\beta{n}}&=&\left[\prod_{k=b}^{a}S(\theta_k)_{n,n-k}\right]^{\overline{\xi}_{n}\dots\overline{\xi}_1}_{\overline{\alpha}_{n}\dots\overline{\alpha}_{1}}\left[S(\theta_{b+1})_{n,n-b-1}\right]_{\overline{\xi}_{n}\dots\overline{\xi}_1}^{\overline{\beta}_{n}\dots\overline{\beta}_1}\\
&=&\left[S(\theta_{b+1})_{n,n-b-1}\right]_{\overline{\xi}_{n}\dots\overline{\xi}_1}^{\overline{\beta}_{n}\dots\overline{\beta}_1}\left[\prod_{k=b}^{a}S(\theta_k)_{n,n-k}\right]^{\overline{\xi}_{n}\dots\overline{\xi}_1}_{\overline{\alpha}_{n}\dots\overline{\alpha}_{1}}\\
&=&\left[\prod_{k=b+1}^{a}S(\theta_k)_{n,n-k}\right]^{\overline{\beta}_{n}\dots\overline{\beta}_1}_{\overline{\alpha}_{n}\dots\overline{\alpha}_{1}},
\end{eqnarray*}
proving statement $a)$. In order to show part $b)$, we also proceed by induction, namely in $n$. For $n=1$ the claim follows directly from the crossing symmetry of the S-matrix. Letting $n\rightarrow n+1$, gives
\begin{eqnarray*}
\hspace{-.5cm}&&\hspace{-.5cm}\left[\prod_{k=1}^{n+1}S(i\pi-\theta_k)_{n+2,k}\right]^{\alpha_1\dots\alpha_{n+2}}_{\beta_1\dots\beta_{n+2}}\\
&&=\left[\prod_{k=1}^{n}S(i\pi-\theta_k)_{n+2,k}\right]^{\alpha_1\dots\alpha_{n+2}}_{\xi_1\dots\xi_{n+2}}\left[S(i\pi-\theta_{n+1})_{n+2,n+1}\right]^{\xi_1\dots\xi_{n+2}}_{\beta_1\dots\beta_{n+2}}\\
&&=\left[\prod_{k=1}^{n}S(i\pi-\theta_k)_{n+1,k}\right]^{\alpha_1\dots\alpha_{n+1}}_{\xi_1\dots\xi_{n+1}}\delta^{\alpha_{n+2}}_{\xi_{n+2}}\prod_{j=1}^{n}\delta^{\xi_j}_{\beta_j}\left[S(i\pi-\theta_{n+1})_{n+2,n+1}\right]^{\xi_{n+1}\xi_{n+2}}_{\beta_{n+1}\beta_{n+2}}\\
&&=\left[\prod_{k=1}^{n}S(\theta_k)_{n+1,n+1-k}\right]^{\overline{\beta}_n\dots\overline{\beta}_1\alpha_1}_{\xi_{n+1}\overline{\alpha}_{n+1}\dots\overline{\alpha}_{2}}\left[S(\theta_{n+1})\right]^{\overline{\beta}_{n+1}\xi_{n+1}}_{\beta_{n+2}\overline{\alpha}_{n+2}}\\
&&=\left[\prod_{k=1}^{n}S(\theta_k)_{n+2,n+2-k}\right]^{\overline{\beta}_{n+1}\dots\overline{\beta}_1\alpha_1}_{\xi_{1}\dots\xi_{n+2}}\left[S(\theta_{n+1})_{n+2,1}\right]^{\xi_{1}\dots\xi_{n+2}}_{\beta_{n+2}\overline{\alpha}_{n+2}\dots\overline{\alpha}_{2}}\\
&&=\left[\prod_{k=1}^{n+1}S(\theta_k)_{n+2,n+2-k}\right]^{\overline{\beta}_{n+1}\dots\overline{\beta}_1\alpha_1}_{\beta_{n+2}\overline{\alpha}_{n+2}\dots\overline{\alpha}_{2}}.
\end{eqnarray*}
\end{proof}
On the dense subspace $\mathcal{D}\subset\mathscr{H}$ of vectors with finite particle number we define creation and annihilation operators $z^\dagger(\theta)$, $z(\theta)$ in the sense of operator-valued distributions. To this end, recall their unsymmetrized counterparts $\widehat{a}^\dagger(\varphi),\widehat{a}(\varphi)$, $\varphi\in\mathscr{H}_1$, on $\widehat{\mathscr{H}}$, subject to
\begin{eqnarray}\label{cr}
\widehat{a}^\dagger(\varphi)\psi_1\otimes\cdots\otimes\psi_n&:=&\sqrt{n+1}\,\varphi\otimes\psi_1\otimes\cdots\otimes\psi_n,\qquad \psi_1,\dots,\psi_n\in\mathscr{H}_1,\\
\widehat{a}(\varphi)\psi_1\otimes\cdots\otimes\psi_n&:=&\sqrt{n}\,\langle\varphi,\psi_1\rangle\psi_2\otimes\cdots\otimes\psi_n,\qquad \widehat{a}(\varphi)\Omega:=0.
\end{eqnarray}
Linear and continuous extension yields densely defined operators which on the subspace of finite particle number fulfill $\widehat{a}(\varphi)^*\supset \widehat{a}^\dagger(\varphi)$. Their projections onto the S-symmetrized $\mathscr{H}$ give rise to
\begin{equation}\label{erzVern}
z^\dagger(\varphi):=P\widehat{a}^\dagger(\varphi)P,\qquad z(\varphi):=P\widehat{a}(\varphi)P,\qquad\varphi\in\mathscr{H}_1,
\end{equation}
by means of the orthogonal projection $P:\widehat{\mathscr{H}}\rightarrow\mathscr{H}$. We relate to these operators the distributions $z^{\dagger}_\alpha(\theta)$ and $z_\alpha(\theta)$ by
\begin{equation}
z^\dagger(\varphi)=\int d\theta\, z^\dagger_\alpha(\theta)\varphi^\alpha(\theta),\qquad z(\varphi)=\int d\theta \,z_\alpha(\theta)\overline{\varphi^\alpha(\theta)}.
\end{equation}
\begin{proposition}[\cite{LS}] Let $\varphi\in\mathscr{H}_1$ and $\Psi\in\mathcal{D}$.
\begin{itemize}
\item[i)] The operators $z^\dagger(\varphi)$ and $z(\varphi)$ act explicitly according to
\begin{subequations}
\begin{eqnarray}
\hspace{-.4cm}&&\hspace{-.4cm}\left(z(\varphi)\Psi\right)_n^{\boldsymbol{\alpha}}(\boldsymbol{\theta})=\sqrt{n+1}\int d\theta' \overline{\varphi^\beta(\theta')}\Psi^{\beta\boldsymbol{\alpha}}_{n+1}(\theta',\boldsymbol{\theta}),\\
\hspace{-.6cm}&&\hspace{-.6cm}\left(z^\dagger(\varphi)\Psi\right)_n(\boldsymbol{\theta})=\frac{1}{\sqrt{n}}\sum_{k=1}^{n}S_n^{\sigma_k}(\boldsymbol{\theta})\left(\varphi(\theta_k)\otimes\Psi_{n-1}(\theta_1,\dots,\hat{\theta}_k,\dots,\theta_n)\right),\\
\hspace{-.4cm}&&\hspace{-.4cm}\quad z(\varphi)\Omega=0,\qquad\left(z^\dagger(\varphi)\Psi\right)_0=0,
\end{eqnarray}
\end{subequations}
where $\sigma_k:=\tau_{k-1}\tau_{k-2}\cdots\tau_1\in\mathfrak{S}_n$ with $\sigma_1:=$ id and $\hat{\theta}_k$ denotes omission of the variable $\theta_k$.
\item[ii)]
\begin{equation}
z(\varphi)^*\supset z^\dagger(\varphi)
\end{equation}
\item[iii)]\begin{equation}\label{numberBounds}
\|z(\varphi)\Psi\|\leq \|\varphi\|\,\|N^{1/2}\Psi\|,\qquad\|z^\dagger(\varphi)\Psi\|\leq \|\varphi\|\,\|(N+1)^{1/2}\Psi\|,
\end{equation}
with $N$ the particle number operator (\ref{opN}).
\item[iv)] The distributional kernels $z^{\dagger}_\alpha(\theta)$ and $z_\alpha(\theta)$ satisfy
\begin{subequations}\label{exchange}
\begin{eqnarray}
z_\alpha(\theta)z_\beta(\theta')&=& S_{\delta\gamma}^{\beta\alpha}(\theta-\theta')z_\gamma(\theta')z_\delta(\theta),\\
z^\dagger_\alpha(\theta)z^\dagger_\beta(\theta')&=& S^{\gamma\delta}_{\alpha\beta}(\theta-\theta')z^\dagger_\gamma(\theta')z^\dagger_\delta(\theta),\\
z_\alpha(\theta)z^\dagger_\beta(\theta')&=& S^{\alpha\gamma}_{\beta\delta}(\theta'-\theta)z_\gamma^\dagger(\theta')z_\delta(\theta)+\delta^{\alpha\beta}\delta(\theta-\theta')\cdot 1.
\end{eqnarray}
\end{subequations}
\end{itemize}
\end{proposition}
Note that the operators $z^\dagger(\varphi)$ and $z(\varphi)$ are in general unbounded. Boundedness, however, follows in case of a constant S-matrix $S^{\alpha\beta}_{\gamma\eta}(\theta)=-\delta^\alpha_\eta\delta^\beta_\gamma$. In fact, the exchange relations stated under item $iv)$ in the previous proposition coincide for such an S-matrix with the canonical anticommutation relations (CAR). Similarly, one recovers the canonical commutation relations (CCR) if $S^{\alpha\beta}_{\gamma\eta}(\theta)=+\delta^\alpha_\eta\delta^\beta_\gamma$. For generic S-matrices the operators $z^\dagger(\varphi)$, $z(\varphi)$ form a representation of the so-called Zamolodchikov-Faddeev algebra \cite{zamolodchikov1979factorized}, commonly used in the context of integrable quantum field theories, see e.g. \cite{Smir92}.

\section{Construction of Models with Factorizing S-Matrices}\label{Construction}
\subsection{Wedge-Local Fields}\label{SecWedgeLocal}
The preparations made in the previous sections allow for the explicit construction of wedge local fields as shown in \cite{LS}. These auxiliary operators play an important role in our analysis on the existence of \textit{local} fields in the present approach. We shall, therefore, review in this section those results which are of particular interest for our purposes.\par
Motivated by the free theory with S-matrix $S^{\alpha\beta}_{\gamma\eta}(\theta)=+\delta^\alpha_\eta\delta^\beta_\gamma$, one introduces by means of the creation and annihilation operators $z^\dagger(\varphi)$ and $z(\varphi)$, defined in (\ref{erzVern}), for theories with a given $S\in\mathcal{S}$ a field on $\mathcal{D}$ via
\begin{equation}\label{field1}
\phi(f):=z^\dagger(f^+)+z(Jf^-),\qquad f\in\mathscr{S}(\mathbb{R}^2)\otimes\mathcal{K},
\end{equation}
where
\begin{equation}
f^{\pm,\alpha}(\theta):=\widetilde{f}^\alpha(\pm p_{m_{[\alpha]}}(\theta))=\frac{1}{2\pi}\int d^2x\,e^{\pm ip_{m_{[\alpha]}}(\theta)\cdot x}f^\alpha(x),\qquad \theta\in\mathbb{R}.
\end{equation}
Since, obviously, $f^{\pm,\alpha}\in L^2(\mathbb{R},d\theta)$ for $f^\alpha\in\mathscr{S}(\mathbb{R}^2)$, the functions $f^\pm$ may be considered as vectors in $\mathscr{H}_1$. The operators (\ref{field1}) are related to the distributions
\begin{eqnarray}\label{fieldkernel}
\phi_\alpha(x)=\int d\theta\,\left(z^\dagger_\alpha(\theta)\,e^{ip_{m_{[\alpha]}}(\theta)\cdot x}+z_{\overline{\alpha}}(\theta)\,e^{-ip_{m_{[\alpha]}}(\theta)\cdot x}\right),
\end{eqnarray}
by
\begin{equation}
\phi(f)=\int d^2x\, \phi_\alpha(x) f^\alpha(x),\qquad f\in\mathscr{S}(\mathbb{R}^2)\otimes\mathcal{K}.
\end{equation}
It is, furthermore, useful to introduce a second auxiliary field, given by
\begin{equation}\label{field2}
\phi'(f):=Jz^\dagger(Jf^+)J+Jz(f^-)J,\qquad f\in\mathscr{S}(\mathbb{R}^2)\otimes\mathcal{K}.
\end{equation}
The following theorem establishes most of the well-known properties of Wightman fields for $\phi$ and $\phi'$. However, as stated below, these fields are in general nonlocal. Nevertheless, they can be interpreted as being localized in wedge regions $W\subset\mathbb{R}^2$, cf. Section \ref{wedges1+1}. A mathematical motivation for this interpretation is the crossing symmetry of the S-matrix as it is reminiscent of the Kubo-Martin-Schwinger (KMS) condition for the vacuum state on an algebra of wedge-local observables with respect to the boost group. In fact, the relation stated under item $viii)$ in the subsequent theorem depends severely on the crossing symmetry of the S-matrix.
\begin{theorem}[\cite{LS}]\label{theo}
Let $f\in\mathscr{S}(\mathbb{R}^2)\otimes\mathcal{K}$ and $\Psi\in\mathcal{D}$.
\begin{itemize}
\item[i)] The map $f\mapsto\phi(f)\Psi$ is linear and continuous.
\item[ii)] Define $(f^*)^\alpha(x):=\overline{f^{\overline{\alpha}}(x)}$, then
\begin{equation}
\phi(f)^*\supset\phi(f^*).
\end{equation}
\item[iii)] Each vector in $\mathcal{D}$ is entire analytic for $\phi(f)$. If $f=f^*$, then $\phi(f)$ is essentially self-adjoint on $\mathcal{D}$.
\item[iv)] $\phi(f)$ transforms covariantly under $\mathcal{P}_+^\uparrow$ and $G$, that is,
\begin{equation}
\begin{aligned}
U(a,t)\phi(f)U(a,t)^{-1}=\phi(f_{(a,t)}),&\qquad f_{(a,t)}(x):=f(\Lambda(t)^{-1}(x-a)),\,\, (a,t)\in\mathcal{P}_+^\uparrow,\\
V(g)\phi(f)V(g^{-1})=\phi(V_1(g)f),&\qquad (V_1(g)f)(x):=V_1(g)f(x),\qquad g\in G.
\end{aligned}
\end{equation}
\item[v)] Let $j(x):=-x$, then
\begin{equation}
\begin{aligned}
J\phi (f)J=\phi'(f_{(j)}),&\qquad f_{(j)}^\alpha(x):=\overline{f^{\overline{\alpha}}(-x)},\\
J\phi' (f)J=\phi(f_{(j)}),&
\end{aligned}
\end{equation}
\item[vi)] For any open set $\mathcal{O}\subset\mathbb{R}^2$, the subspace
\begin{equation}
\mathcal{D}_{\mathcal{O}}:=\text{span}\,\{\phi(f_1)\cdots\phi(f_n)\Omega:f_1,\dots,f_n\in\mathscr{S}(\mathcal{O})\otimes\mathcal{K},n\in\mathbb{N}_0\}
\end{equation}
is dense in $\mathscr{H}$. That is, $\Omega$ is cyclic for the field $\phi$.
\item[vii)] The field $\phi$ is local if and only if $S^{\alpha\beta}_{\gamma\eta}(\theta)=\delta^\alpha_\eta\delta^\beta_\gamma$.
\end{itemize}
Statements $i)-iv)$ and $vi),vii)$ hold also true if $\phi$ is replaced by $\phi'$. Let, moreover, $f\in\mathscr{S}(W_R+a)\otimes\mathcal{K}$, $g\in\mathscr{S}(W_L+a)\otimes\mathcal{K}$ and $a\in\mathbb{R}^2$, then
\begin{itemize}
\item[viii)]
\begin{equation}\label{commutatorwedge}
[\phi'(f),\phi(g)]\Psi=0,\qquad\Psi\in\mathcal{D},
\end{equation}
that is, the fields $\phi$ and $\phi'$, are relatively wedge-local.
\end{itemize}
\end{theorem}
The properties stated in the previous theorem allow for interpreting the auxiliary fields $\phi$ and $\phi'$ as being localized in wedges. One can, namely, assign $\phi(g)$ to the localization region $(W_L+\text{supp}\,g)''$ and $\phi'(f)$ to the localization region $(W_R+\text{supp}\,f)''$ in agreement with covariance and causality by properties $iv)$ and $viii)$ of Theorem \ref{theo}. Thereby, $(W_R+\text{supp}\,f)''$ is the smallest right wedge that contains supp$f$ which, on the other hand, is defined as the smallest subset of $\mathbb{R}^2$ containing supp$f_\alpha$ for all $\alpha\in\{1,\dots,D\}$. Building on these wedge-local fields $\phi$ and $\phi'$, we now aim at proving the existence of a \textit{local} quantum field theory with S-matrix $S$. To this end, it is advantageous to employ an operator-algebraic formulation of the already constructed models. This is due to certain techniques available in the algebraic framework which lead to a practical way for checking locality. In particular, the so-called modular nuclearity condition, as introduced later on, is of essential importance in the present analysis.\par
An algebraic description of the models at hand is obtained by considering the von Neumann algebras generated by $\phi$ and $\phi'$, i.e. with $x\in\mathbb{R}^2$
\begin{subequations}\label{algebras}
\begin{eqnarray}
\mathcal{F}(W_L+x)&:=&\{e^{i\overline{\phi (f)}}:f=f^*\in\mathscr{S}(W_L+x)\otimes \mathcal{K}\}'',\\
\mathcal{F}(W_R+x)&:=&\{e^{i\overline{\phi'(f)}}:f=f^*\in\mathscr{S}(W_R+x)\otimes \mathcal{K}\}''.
\end{eqnarray}
\end{subequations}
This definition and the properties of the fields $\phi$ and $\phi'$ yield
\begin{proposition}[\cite{LS}]\label{PropWedgeAlgebra}
Let $S\in\mathcal{S}$ and $\mathcal{F}(W)$, \mbox{$W\in\mathcal{W}$}, be defined as in (\ref{algebras}). Then, $\{\mathcal{F}(W)\}_{W\in\mathcal{W}}$ is a local field net, cf. Section \ref{AQFT}, of von Neumann algebras, transforming covariantly under the adjoint action of the extended representation $U$, cf. Lemma \ref{extendedU}, of the proper Poincaré group $\mathcal{P}_+$. Furthermore, locality is fulfilled (without twisting) and the vacuum vector $\Omega$ is cyclic and separating for each $\mathcal{F}(W)$, $W\in\mathcal{W}$.
\end{proposition}
This result opens up the possibility to employ the powerful methods available in the context of Algebraic Quantum Field Theory as explained in the subsequent sections.\par
Note that, due to the fact that the net $\{\mathcal{F}(W)\}_{W\in\mathcal{W}}$ constructed by means of the auxiliary fields $\phi$ and $\phi'$ satisfies locality without twisting, we are dealing here with a purely bosonic theory. This is a particular consequence of the commutativity of the fields $\phi$ and $\phi'$ at spacelike distances. It is, in principle, possible to include Fermi fields in the framework by constructing anticommuting auxiliary objects. We do, however, not discuss this case here.\par
We remark further that $\phi$ and $\phi'$ constitute examples of so-called temperate polar-ization-free generators \cite{schroer1999modular, BBS01}. This fact was shown in \cite[Theorem 4.5.]{schutzenhofer2011multi} and is due to the following properties. First, these fields generate only single particle states when applied to the vacuum $\Omega$. Second, they are localized in wedges. Last, $\phi$ and $\phi'$ admit a certain continuous and bounded behavior with regard to Poincaré transformations. We shall come back to this topic in Chapter \ref{Chapter5} where such operators will be used to compute two-particle scattering amplitudes in Haag-Ruelle collision theory. For the present context such a calculation was done in \cite{schutzenhofer2011multi}.
\subsection{The Bisognano-Wichmann Property}
The cyclicity and separability of the vacuum $\Omega$ for $\mathcal{F}(W)$, $W\in\mathcal{W}$, established in Proposition \ref{PropWedgeAlgebra}, allows for the application of the Tomita-Takesaki theory \cite{takesaki1970tomita, takesaki2003theory}. Considering, in particular, the pair $(\mathcal{F}(W_R),\Omega)$, we denote by $\Delta^{it}$, $t\in\mathbb{R}$, and $\tilde{J}$ the associated modular unitaries and modular conjugation respectively. Whereas the adjoint action of the modular group $\Delta^{it}$ leaves the algebra $\mathcal{F}(W_R)$ invariant, the antiunitary involution $\tilde{J}$ maps the algebra
into its commutant $\mathcal{F}(W_R)'$. Bisognano and Wichmann showed that these operators have a geometric interpretation if the algebra $\mathcal{F}(W_R)$ is generated by Wightman fields \cite{BisWich, bisognano1976duality}. Namely, the modular group coincides with the representation of the Lorentz-boosts of the wedge and the modular conjugation is related to the PCT operator $J$. With regard to our intension to make use of the modular nuclearity condition which involves the maps (\ref{modNucCon}) below and, in particular, the modular operator $\Delta$, it is desirable to prove these remarkable features also in the present context. As stated in the subsequent Proposition \ref{BisognWich}, it turns out that this is indeed the case.\par
Before presenting the corresponding results, we recall \cite{BisWich} that if $f\in\mathscr{S}(\mathbb{R}^2)\otimes\mathcal{K}$ has compact support in $W_R$ then $f^+\in\mathscr{H}_1$ lies in the domain of the positive self-adjoint operator $U(0,\tfrac{i}{2})$ and
\begin{equation}
(U(0,\tfrac{i}{2})f^+)(\theta)=f^+(\theta-i\pi)=f^-(\theta).
\end{equation}
Moreover, since $(f^{*})^\alpha(x)=\overline{f^{\overline{\alpha}}(x)}$, we have
\begin{equation}
(Jf^{\pm})^\alpha(\theta)=(f^{*})^{\mp,\alpha}(\theta).
\end{equation}
Thus, for compactly supported $f\in\mathscr{S}(W_R)\otimes\mathcal{K}$,
\begin{equation}\label{oneparticletomita}
JU(0,\tfrac{i}{2})\phi'(f)\Omega=JU(0,\tfrac{i}{2})f^+=(f^*)^+=\phi'(f^*)\Omega=\phi'(f)^*\Omega=S\phi'(f)\Omega,
\end{equation}
with $S$ being the Tomita operator of the pair $(\mathcal{F}(W_R),\Omega)$, with polar
decomposition $S=\tilde{J}\Delta^{1/2}$. This calculation already suggests a one particle version of the Bisognano-Wichmann property which, as shown below, indeed holds true.
\begin{proposition}\label{BisognWich}
\textit{Let} $\{\mathcal{F}(W)\}_{W\in\mathcal{W}}$ \textit{be the net of von Neumann algebras defined in (\ref{algebras}). Then,}
\begin{description}
    \item[i)] the Bisognano-Wichmann property holds, that is, \textit{the modular operator} $\Delta$ \textit{and modular conjugation} $\tilde{J}$ \textit{associated with the pair} $(\mathcal{F}(W_R),\Omega)$ \textit{are given by}
\begin{eqnarray}
\Delta^{it}&=&U(0,-t),\qquad t\in\mathbb{R},\\
\tilde{J}&=&J,
\end{eqnarray}
\item[ii)] moreover, \textit{Haag-duality},
\begin{equation}
\mathcal{F}(W)'=\mathcal{F}(W'),\qquad W\in\mathcal{W},
\end{equation}
holds.
\end{description}
\end{proposition}
\begin{proof}
We start the proof by showing that the operator $A:=\overline{\phi'(f)}$, with $f\in\mathscr{S}(W_R)\otimes \mathcal{K}$ of compact support, is affiliated with $\mathcal{F}(W_R)$. To this end, consider $\Psi_0\in\mathcal{D}$, $\Psi\in \mathrm{dom}\, A$ and $g\in\mathscr{S}(W_R')\otimes \mathcal{K}$ with $g=g^*$. Using the property that any vector in $\mathcal{D}$ is an analytic vector for $\phi'(f)$ and $\phi(g)$ \cite{LS} it follows that $\Psi_0$ and $A^*\Psi_0$ are analytic vectors for $B:=\overline{\phi(g)}$ because the particle number is only finitely changed by $A^*$. Moreover, we have \cite{LS}
\begin{equation}
[A^*,B^p]|_{\mathcal{D}}=0,\qquad\forall\, p\in\mathbb{N}_0.
\end{equation}
Therefore,
\begin{equation}
\langle\Psi_0,e^{iB}A\Psi\rangle=\sum_{p=0}^\infty\frac{i^p}{p!}\langle A^*B^p\Psi_0,\Psi\rangle=\sum_{p=0}^\infty\frac{i^p}{p!}\langle B^pA^*\Psi_0,\Psi\rangle
=\langle\Psi_0,Ae^{iB}\Psi\rangle,
\end{equation}
which implies $e^{iB}A\Psi=Ae^{iB}\Psi$ since $\mathcal{D}\subset\mathscr{H}$ is dense. This relation still holds if one replaces $e^{iB}$ by any element in the $*$-algebra $\mathcal{A}$ generated by $e^{i\overline{\phi(g)}}$ with $g\in\mathscr{S}(W_R')\otimes \mathcal{K}$ and $g=g^*$. However, any $D'\in\mathcal{F}(W_R)'$ is a weak limit of a sequence $D_n'\in\mathcal{A}$ and the identity $D_n'A\Psi=AD_n'\Psi$ is preserved under weak limits. Thus, $A$ is affiliated with $\mathcal{F}(W_R)$.\par
Since $A$ is a closed, densely defined unbounded operator between Hilbert spaces, it has a unique polar decomposition $A=Y|A|$. Due to the affiliation of $A$ with $\mathcal{F}(W_R)$, we have $Y,E_n|A|\in\mathcal{F}(W_R)$, $n\in\mathbb{N}$, where $E_n$ are the spectral projections of $|A|$ onto the spectrum in the interval $[0,n]$. Consider further the Tomita operator $S$ of the pair $(\mathcal{F}(W_R),\Omega)$. Then it follows from $SYE_n|A|\Omega=|A|E_nY^*\Omega$, the strong convergence $E_n\rightarrow 1$ as $n\rightarrow\infty$ and the closedness of $S$ that $A\Omega$ is in the domain of $S$ and $SA\Omega=A^*\Omega$. As dom$\,\Delta^{1/2}=\text{dom}\,S$, $A\Omega$ is also in the domain of $\Delta^{1/2}$. By the same arguments also $AF\Omega$ with any $F\in\mathcal{F}(W_R)$ is in the domain of $S$ and $\Delta^{1/2}$ respectively.\par
By modular theory, the Tomita operator of the pair $(\mathcal{F}(W_R)',\Omega)$ is $S^*$. Using the same arguments as above, one can show that the operator $\overline{\phi(f')}$, with $f'\in\mathscr{S}(W_R')\otimes\mathcal{K}$ of compact support, is affiliated with $\mathcal{F}(W_R)'$ and $S^*\phi(f')\Omega=\phi(f')^*\Omega$.\par
Next, we prove the statement on the modular operator. It follows from Equation (\ref{oneparticletomita}) that
\begin{equation}\label{inclusion}
S^1_{\text{geo}}\subset SE^{(1)},
\end{equation}
where $S^1_{\text{geo}}:=JU(0,\tfrac{i}{2})E^{(1)}$ and $E^{(1)}$ is the projection onto $\mathscr{H}_1$. To prove (\ref{inclusion}), define $\mathcal{D}_0:=\left\lbrace \Psi\in\mathscr{H}_1: \Psi=\phi'(f)\Omega,\, \text{supp}\,f\in W_R\,\,\text{compact} \right\rbrace$ which is dense in $\mathscr{H}_1$. In \cite{BGL} it is shown that $S^1_{\text{geo}}$ is a closed operator which is densely defined on $\mathscr{H}_1$. Moreover, dom$\,S^1_{\text{geo}}=\text{dom}\,U(0,\tfrac{i}{2})=K+iK$, where $K$ is a real subspace of $\mathscr{H}_1$. Clearly, we have $\mathcal{D}_0\subset\text{dom}\,S^1_{\text{geo}}$. Therefore, there exist sequences $\left\lbrace h_n\right\rbrace _{n\geq 1}$ and $\left\lbrace k_n\right\rbrace _{n\geq 1}$ with $h_n,k_n\in \mathcal{D}_0$ corresponding to real functions $f=f^*$, which converge to $h,k\in K$ such that $h_n+ik_n\overset{n}{\longrightarrow} h+ik\in\text{dom}\,S^1_{\text{geo}}$. In particular, $S^1_{\text{geo}}(h_n+ik_n)=h_n-ik_n\overset{n}{\longrightarrow}h-ik$, which by the closedness of $S^1_{\text{geo}}$ gives $S^1_{\text{geo}}(h+ik)=h-ik$. Thus, from (\ref{oneparticletomita}), the closedness of $S$ and the fact that $Sf^+=f^+$ for $f^*=f$ it follows that $S^1_{\text{geo}}(h+ik)=S(h+ik)$, for all $h+ik\in\text{dom}\,S^1_{\text{geo}}$, proving (\ref{inclusion}).\par
To show the opposite inclusion, namely $SE^{(1)}\subset S^1_{\text{geo}}$, note that $E^{(1)}=\sum_{q\in\mathcal{Q}}E_{m_q}$, where the $E_{m_q}$, $q\in\mathcal{Q}$, are spectral projections of the mass operator $\sqrt{P^2}$. By a theorem of Borchers \cite{Bor1} the modular group $\Delta^{it}$ and modular conjugation $\tilde{J}$ have geometrically the correct commutation relations with the translations, which is a crucial step towards the Bisognano-Wichmann property. In particular, the exchange relations of $\Delta^{it}$ and $\tilde{J}$ with the mass operator $\sqrt{P^2}$ imply that $S$ commutes with $\sqrt{P^2}$ and consequently with $E^{(1)}$\footnote{Since $E^{(1)}:\text{dom}\,S\rightarrow\text{dom}\,S$, $E^{(1)}$ and $S$ commute strongly on vectors in dom$\,S$.} \cite{Mund}. We proceed by defining $S(W_R'):=JSJ$ and $S^1_{\text{geo}}(W_R'):=JS^1_{\text{geo}}J$ respectively. Then, it follows from locality\footnote{$\mathcal{F}(W_R)\subset\mathcal{F}(W_R')'$} and modular theory\footnote{If $S_0$ is the Tomita operator of the pair $(\mathcal{F}(W),\Omega)$, then $S_0^*$ is the Tomita operator corresponding to $(\mathcal{F}(W)',\Omega)$.} that $S\subset S(W_R')^*$. Since $S(W_R')$ commutes with $E^{(1)}$ we have together with Equation (\ref{inclusion})
\begin{equation}
SE^{(1)}\subset (S(W_R')E^{(1)})^*\subset S^1_{\text{geo}}(W_R')^*=S^1_{\text{geo}},
\end{equation}
where the last equality follows from $JU(0,\tfrac{i}{2})=U(0,-\tfrac{i}{2})J$ by the anti-unitarity of $J$. Thus, we have shown
\begin{eqnarray}\label{polar}
SE^{(1)}=JU(0,\tfrac{i}{2})E^{(1)},
\end{eqnarray}
As the polar decomposition of a closed operator is unique and $S$ is closed, we obtain from (\ref{polar}) a one particle version of the Bisognano-Wichmann theorem, namely $\Delta^{1/2}E^{(1)}=U(0,\tfrac{i}{2})E^{(1)}$ and $\tilde{J}E^{(1)}=JE^{(1)}$. In particular, we have
\begin{equation}\label{oneparticleBisoWich}
\Delta^{it}E^{(1)}=U(0,-t)E^{(1)},\qquad\tilde{J}E^{(1)}=JE^{(1)}.
\end{equation}
This result can now be used to prove the equality of the modular operator $\Delta^{it}$ with $U(0,-t)$. To this end, define $L(t):=U(0,-t)\Delta^{-it}$, $t\in\mathbb{R}$, and $\overline{\phi'_t(f)}:=L(t)\overline{\phi'(f)}L(t)^{-1}$ with $f\in\mathscr{S}(W_R)\otimes\mathcal{K}$ of compact support. Since $L(t)\mathcal{F}(W_R)L(t)^{-1}\subset\mathcal{F}(W_R)$, $\overline{\phi'_t(f)}$ is affiliated with $\mathcal{F}(W_R)$ because $A=\overline{\phi'(f)}$ is. Hence, $\overline{\phi'_t(f)}$ commutes with elements in $\mathcal{F}(W_R)'$. Making use of the preceding result (\ref{oneparticleBisoWich}) and $L(t)^{-1}\Omega=\Omega$ we have with $A'\in\mathcal{F}(W_R)'$
$$\phi'_t(f)A'\Omega=A'\phi'_t(f)\Omega=A'L(t)\phi'(f)\Omega=A'L(t)E^{(1)}\phi'(f)\Omega=A'\phi'(f)\Omega=\phi'(f)A'\Omega,$$
since $\phi'(f)\Omega\in\mathscr{H}_1$. That is,
\begin{equation}
(\phi'_t(f)-\phi'(f))A'\Omega=0, \qquad A'\in\mathcal{F}(W_R)',
\end{equation}
for any $f$ as above. We will show below that $\mathcal{F}(W_R)'\Omega$ is a core for $\phi'(f)$ and $\phi'_{t}(f)$, giving the result $\phi'(f)=\phi'_t(f)$. Hence $U(0,-t)\Delta^{-it}$ acts trivially on $\mathcal{F}(W_R)$ and since $\Omega$ is cyclic for $\mathcal{F}(W_R)$, it follows that $U(0,-t)\Delta^{-it}=1$, $t\in\mathbb{R}$.\par
Similarly, one proves the equality of the modular conjugation $\tilde{J}$ with the TCP operator $J$. Firstly, we have $\tilde{J}\mathcal{F}(W_R)\tilde{J}=\mathcal{F}(W_R)'$ by modular theory and secondly $J\mathcal{F}(W_R)J=\mathcal{F}(W_L)$ by definition. Moreover, for an arbitrary wedge $W$ we have $\mathcal{F}(W')\subset\mathcal{F}(W)'$ \cite{LS}. Therefore, it is obvious that $\tilde{J}J\mathcal{F}(W_R)(\tilde{J}J)^{-1}\subset\mathcal{F}(W_R)$. Hence, the operator $\overline{\phi'_I(f)}:=I\overline{\phi'(f)}I^{-1}$ with $I:= \tilde{J}J$ and $f\in\mathscr{S}(W_R)\otimes\mathcal{K}$ is affiliated with $\mathcal{F}(W_R)$. By the same arguments as above we find
\begin{equation}
(\phi'_I(f)-\phi'(f))A'\Omega=0, \qquad A'\in\mathcal{F}(W_R)'.
\end{equation}
We will show below that $\mathcal{F}(W_R)'\Omega$ is also a core for $\phi'_I(f)$, which implies  $\phi'_I(f)=\phi'(f)$ and consequently $I=1$. Haag-duality, statement $ii)$, then follows easily from $\mathcal{F}(W_R)'=\tilde{J}\mathcal{F}(W_R)\tilde{J}=J\mathcal{F}(W_R)J=\mathcal{F}(W_L)$ and covariance.\par
In order to prove that $\mathcal{F}(W_L)\Omega$ is a core for the operators $\phi'(f)$, $\phi'_I(f)$, and $\phi'_{t}(f)$ we use the particle number bounds (\ref{numberBounds}) and consider $P_0$, the positive generator of time translations, i.e. we have the following inequalities
\begin{eqnarray}
||\phi'(f)\Psi||&\leq& \left(||f^+||+||f^-||\right)\,||(N+1)^{1/2}\Psi||,\qquad\Psi\in\mathcal{D}\\
m_\circ(N+1)&\leq& P_0+m_\circ 1,
\end{eqnarray}
with $m_\circ>0$ the mass gap of the theory. Therefore, for $\Psi\in\mathcal{D}\cap\mathcal{D}_{P_0}$, where $\mathcal{D}_{P_0}$ is the domain of $P_0$, we have
\begin{equation}
||\phi'(f)\Psi||\leq m_\circ^{-1/2}(||f^+||+||f^-||)||(P_0+m_\circ 1)^{1/2}\Psi||.
\end{equation}
Using standard arguments it follows from this estimate that any core for $P_0$ is also a core for $\phi'(f)$. Since $I$ and $L$ also commute with all translations and therefore in particular with the time translations, this property is also true for $\phi'_I(f)$ and $\phi'_{t}(f)$. The remaining part is to show that $\mathcal{F}(W_L)\Omega\cap\mathcal{D}_{P_0}$ is a core for $P_0$. To this end, note that $\mathcal{F}(W_L)\Omega$ is mapped into itself by all translations $U(y)$ with $y\in W_L$. Moreover, let $h$ be some test function such that supp$\,h\subset W_L$ then, since $\Omega$ is invariant under translations, we have $\widetilde{h}(P)\mathcal{F}(W_L)\Omega\subset\mathcal{F}(W_L)\Omega\cap\mathcal{D}_{P_0}$ with $\widetilde{h}(P)=\int dy\, h(y)U(y)$. However, there exist test functions $h$ such that $\widetilde{h}(P)$ is invertible, i.e. Ker$(\widetilde{h}(P))=\{0\}$, and hence $\mathcal{F}(W_L)\Omega\cap\mathcal{D}_{P_0}$ is not empty. Therefore, we have $(P_0\pm i)\widetilde{h}(P)\mathcal{F}(W_L)\Omega\subset(P_0\pm i)(\mathcal{F}(W_L)\Omega\cap\mathcal{D}_{P_0})$ and as $\mathcal{F}(W_L)\Omega$ is dense in $\mathscr{H}$ also the intersection $\mathcal{F}(W_L)\Omega\cap\mathcal{D}_{P_0}$ is dense. This proves that $P_0$ is essentially self-adjoint on $\mathcal{F}(W_L)\Omega\cap\mathcal{D}_{P_0}$ \cite[Cor. to Thm. VIII.3]{RS1} which therefore is a core for $P_0$. This finishes the proof.
\end{proof}
Note that in case $\dim\mathcal{K}=1$, the Poincaré group acts irreducibly on $\mathscr{H}_1$ and, therefore, the Bisognano-Wichmann property can be verified without establishing first a single particle version \cite{BL4}. It is, however, an important intermediate step in the situation of a richer particle spectrum, with $\dim\mathcal{K}>1$, where this irreducibility feature is no longer present.

\subsection{From Wedge-Local to Local Algebras - The Modular Nuclearity Condition}\label{ModSection}
In the previous discussion we introduced certain models in terms of wedge-local quantum fields $\phi$, $\phi'$ and proceeded to the corresponding von Neumann algebras $\mathcal{F}(W)$, $W\in\mathcal{W}$, (\ref{algebras}) they generate. In local quantum physics, however, we are interested in local field algebras $\mathcal{F}(\mathcal{O})$, associated with bounded regions $\mathcal{O}\subset\mathbb{R}^2$. In order to proceed from the wedge algebras to a net of local algebras, we first notice that the causal complement of a double cone
\begin{equation}
\mathcal{O}_{x,y}=(W_R+x)\cap(W_L+y),\qquad y-x\in W_R,
\end{equation}
consists of the two disconnected components $W_L+x$ and $W_R+y$, cf. Figure \ref{doubleFig}. Considering an operator $A\in\mathcal{B}(\mathscr{H})$ associated to the region $\mathcal{O}_{x,y}$, then causality requires that operations in $\mathcal{O}_{x,y}'=(W_L+x)\cup(W_R+y)$ do not influence $A$ and, therefore, it has to commute with all operators in $\mathcal{F}(W_L+x)$ and $\mathcal{F}(W_R+y)$. Consequently, $A$ belongs to the von Neumann algebra
\begin{equation}\label{vNDouble}
\mathcal{F}(\mathcal{O}_{x,y}):=\left(\mathcal{F}(W_L+x)\vee\mathcal{F}(W_R+y)\right)'=\mathcal{F}(W_L+x)'\cap\mathcal{F}(W_R+y)',
\end{equation}
which, with regard to compatibility with locality, actually is the maximal possible choice. Recall that $\mathcal{F}(W_L+x)\vee\mathcal{F}(W_R+y):=\left(\mathcal{F}(W_L+x)\cup\mathcal{F}(W_R+y)\right)''$ denotes the smallest von Neumann algebra containing both $\mathcal{F}(W_L+x)$ and $\mathcal{F}(W_R+y)$.
Definition (\ref{vNDouble}) extends to arbitrary bounded open regions $\mathfrak{R}\subset\mathbb{R}^2$ by additivity, that is,
\begin{equation}\label{additivity}
\mathcal{F}(\mathfrak{R}):=\bigvee_{\mathfrak{R}\supset\mathcal{O}\in\mathscr{O}}\mathcal{F}(\mathcal{O}),
\end{equation}
where $\mathscr{O}:=\{\mathcal{O}_{x,y}:y-x\in W_R\}$ denotes the set of all double cones in $\mathbb{R}^2$. This construction indeed yields a local net $\{\mathcal{F}(\mathcal{O})\}_{\mathcal{O}\subset\mathbb{R}^2}$ which inherits its basic properties from those of the wedge algebras (\ref{algebras}). More precisely, the following results hold true.
\begin{proposition}\label{PropDoubleCones}
Consider the algebras $\mathcal{F}(\mathcal{O})$ defined in (\ref{vNDouble}, \ref{additivity}). Then, for bounded open $\mathcal{O}_1,\mathcal{O}_2\subset\mathbb{R}^2$ we have
\begin{itemize}
\item isotony: $\mathcal{F}(\mathcal{O}_1)\subset\mathcal{F}(\mathcal{O}_2)$ for $\mathcal{O}_1\subset\mathcal{O}_2$,
\item covariance: $U(a,t)\mathcal{F}(\mathcal{O})U(a,t)^{-1}=\mathcal{F}(\Lambda(t)\mathcal{O}+a)$, $(a,t)\in\mathcal{P}_+$,
\item gauge symmetry: $V(g)\mathcal{F}(\mathcal{O})V(g)^{-1}=\mathcal{F}(\mathcal{O})$, $g\in G$,
\item locality: $\mathcal{F}(\mathcal{O}_1)\subset\mathcal{F}(\mathcal{O}_2)'$ for $\mathcal{O}_1\subset\mathcal{O}_2'$,
\end{itemize}
that is, $\{\mathcal{F}(\mathcal{O})\}_{\mathcal{O}\subset\mathbb{R}^2}$ is a local net of von Neumann algebras.
\end{proposition}
A similar result was proven for the special case of a particle spectrum with only one type of neutral massive particles, i.e. in case the dimension of $\mathcal{K}$ is restricted to one. The proof of Proposition \ref{PropDoubleCones} requires, in fact, only trivial adjustments. For the convenience of the reader we, nevertheless, give a complete proof.
\begin{proof} Let $\mathcal{O}_{x_1,y_1}\subset\mathcal{O}_{x_2,y_2}$ be an inclusion of double cones. Then, we have $W_R+y_2\subset W_R+y_1$ and $W_L+x_2\subset W_L+x_1$. Hence,
$$\mathcal{F}(\mathcal{O}_{x_1,y_1})=\mathcal{F}(W_L+x_1)'\cap\mathcal{F}(W_R+y_1)'\subset\mathcal{F}(W_L+x_2)'\cap\mathcal{F}(W_R+y_2)'=\mathcal{F}(\mathcal{O}_{x_2,y_2}),$$
implying together with Definition (\ref{additivity}) isotony, i.e. the map $\mathcal{O}\mapsto\mathcal{F}(\mathcal{O})$ is a \textit{net} of von Neumann algebras.\par
With regard to the covariance claim, let $\mathcal{O}_{x,y}\in\mathscr{O}$ and $g\in\mathcal{P}_+$. Then, 
\begin{eqnarray*}
U(g)\mathcal{F}(\mathcal{O}_{x,y})U(g)^{-1}&=&U(g)\mathcal{F}(W_L+x)'U(g)^{-1}\cap U(g)\mathcal{F}(W_R+y)'U(g)^{-1}\\
&=&\mathcal{F}(g(W_L+x))'\cap\mathcal{F}(g(W_R+y))'\\
&=&\mathcal{F}(g\mathcal{O}).
\end{eqnarray*}
Taking again Definition (\ref{additivity}) into account, the covariance property follows. In an analogous manner, the transformation behavior under inner symmetries $g\in G$ follows directly from definition.\par
It remains to show locality. To this end, we consider two spacelike separated double cones $\mathcal{O}_{x_1,y_1}\subset\mathcal{O}_{x_2,y_2}'$. In this case, we have either $\mathcal{O}_{x_1,y_1}\subset W_R+y_2$ and $\mathcal{F}(\mathcal{O}_{x_1,y_1})\subset\mathcal{F}(W_R+y_2)$, or $\mathcal{O}_{x_1,y_1}\subset W_L+x_2$ and $\mathcal{F}(\mathcal{O}_{x_1,y_1})\subset\mathcal{F}(W_L+x_2)$. Since $\mathcal{F}(\mathcal{O}_{x_2,y_2})=\left(\mathcal{F}(W_L+x_2)\vee\mathcal{F}(W_R+y_2)\right)'$, we have $\mathcal{F}(\mathcal{O}_{x_1,y_1})\subset\mathcal{F}(\mathcal{O}_{x_2,y_2})'$. By the same argument, the locality property carries over to arbitrary spacelike separated bounded regions $\mathfrak{R}_1\subset \mathfrak{R}_2'$.
\end{proof}
The constructed net of local field algebras $\mathcal{F}(\mathcal{O})$ characterizes the model theory associated to the S-matrix $S$, initially taken into account. Thereby, the problem of finding explicit expressions for local quantum fields underlying the model is circumvented. In the present approach these operators are, namely, rather defined in an indirect manner as elements of the intersections (\ref{vNDouble}). The interesting question on the concrete form of the local fields was looked into by Bostelmann and Cadamuro \cite{cadamuro2012characterization, bostelmann2013operator, bostelmann2014characterization, bostelmanntowards}. Their analysis, however, is restricted to the scalar case, i.e. to a particle spectrum consisting of only a single type of neutral massive particles. It is based on a certain series expansion which is a generalization of Araki's expansion of bounded operators on Fock space in case of the free field theory \cite{araki1963lattice} to integrable models in $1+1$ dimensions. Localization properties of these operators, represented by a series, then reflect in analyticity features of certain expansion coefficients. The characterization of the local observables in terms of the latter was carried out in \cite{bostelmann2014characterization}. An explicit construction of local operators within this approach was thereby given for a model with scalar S-matrix $S=-1$ \cite{cadamuro2012characterization}.\par
In this thesis we are concerned with showing the \textit{existence} of local quantum fields in a model with scattering operator determined by $S\in\mathcal{S}$. This is the case if the local algebras $\mathcal{F}(\mathcal{O})$ do not consist only of multiples of the identity, that is, if the intersections of wedge algebras (\ref{vNDouble}) are nontrivial, and if further the Reeh-Schlieder property holds. Namely, in case the local algebras $\mathcal{F}(\mathcal{O})$, corresponding to these intersections, have the vacuum $\Omega$ as a cyclic vector\footnote{In local theories with gauge charges \cite{buchholz1982locality} the observable, gauge invariant subalgebras $\mathcal{A}(\mathcal{O})$ of $\mathcal{F}(\mathcal{O})$ would have $\Omega$ as a cyclic vector on the charge zero subspace.}, their nontriviality and, in particular, the existence of a local field theory is implied.\par
In order to check the validity of the Reeh-Schlieder property for bounded regions, we benefit from a model-independent analysis carried out in \cite{BL4} and \cite{L08} respectively where the maps
\begin{equation}\label{modNucCon}
\Xi(x):\mathcal{F}(W_R)\rightarrow\mathscr{H},\qquad \Xi(x)A:=\Delta^{1/4}U(x)A\Omega,\qquad x\in W_R,
\end{equation}
are of central significance. Building on an earlier result by Buchholz, d'Antoni and Longo \cite{BDL90} for an inclusion of von Neumann factors on a Hilbert space, the nuclearity\footnote{Appendix \ref{AppendixB} provides supporting material with regard to nuclear maps.} of the maps (\ref{modNucCon}) for any $x\in W_R$ can be shown to imply that the net $\mathcal{F}$ (\ref{algebras}) satisfies the (distal) split property for wedges \cite{BL4}, a powerful feature as shall be explained in more detail in the sequel. Recall first from \cite{muger1998superselection} that a local field net $\mathcal{F}$ is said to have the split property for wedges if the inclusion $\mathcal{F}(W_1)\subset\mathcal{F}(W_2)$, $W_1,W_2\in\mathcal{W}$, is split whenever the closure of $W_1$ is contained in the interior of $W_2$ (in symbols $W_1\subset\subset W_2$). The notion ``distal'' refers to split inclusions for pairs of regions $W_1$ and $W_2$ with a sufficiently large inner distance. The split property \cite{DL84} of the inclusion $\mathcal{F}(W_1)\subset\mathcal{F}(W_2)$ thereby means that there exists a type I factor $\mathcal{N}$ such that
\begin{equation*}
\mathcal{F}(W_1)\subset\mathcal{N}\subset\mathcal{F}(W_2).
\end{equation*}
It further corresponds to a form of statistical independence between $\mathcal{F}(W_1)$ and $\mathcal{F}(W_2)'$. In particular, if $\mathcal{F}(W_1)$ and $\mathcal{F}(W_2)$ are factors and if there exists a cyclic and separating vector for $\mathcal{F}(W_1),$ $\mathcal{F}(W_2)$ and the relative commutant $\mathcal{F}(W_1)'\cap\mathcal{F}(W_2)$ then \cite{d1983interpolation, DL84} the split property of the inclusion $\mathcal{F}(W_1)\subset\mathcal{F}(W_2)$ for $W_1\subset\subset W_2$ is equivalent to the existence of a unitary $T:\mathscr{H}\rightarrow\mathscr{H}\otimes\mathscr{H}$ implementing an isomorphism between $\mathcal{F}(W_1)\vee\mathcal{F}(W_2)'$ and $\mathcal{F}(W_1)\otimes\mathcal{F}(W_2)'$,
\begin{equation*}
TA_1A_2'T^*=A_1\otimes A_2',\qquad A_1\in\mathcal{F}(W_1),\,\,A_2'\in\mathcal{F}(W_2)'.
\end{equation*}
Indeed, by a result established in \cite{longo1979notes, driessler1975comments} the wedge algebras $\mathcal{F}(W)$, $W\in\mathcal{W}$, are factors of type III$_1$ according to the classification of Connes \cite{connes1973classification}. Moreover, as shown in \cite{BL4} the second requirement yielding the above isomorphism is also met. These facts imply a rather simple structure of the local algebras, since $\mathcal{F}(W_1'\cap W_2)=\mathcal{F}(W_1')\cap\mathcal{F}(W_2)=(\mathcal{F}(W_1)\vee\mathcal{F}(W_2)')'$ can be realized as a tensor product of wedge algebras, namely $\mathcal{F}(W_1')\otimes\mathcal{F}(W_2)$, on $\mathscr{H}\otimes\mathscr{H}$, provided the inclusion $\mathcal{F}(W_1)\subset\mathcal{F}(W_2)$ is split for $W_1\subset\subset W_2$. In addition, the nontriviality of $\mathcal{F}(W_1'\cap W_2)$ can be inferred.\par
The split property for wedges implies further that the double cone algebras $\mathcal{F}(\mathcal{O})$, $\mathcal{O}\in\mathscr{O}$, (\ref{vNDouble}) are isomorphic to the hyperfinite type III$_1$ factor \cite{BL4}. Moreover, it follows that the set of cyclic and separating vectors for the local algebras is a dense $G_\delta$ set in $\mathcal{H}$ \cite{dixmier1971vecteurs}. Although a priori not obvious, the analysis of \cite{L08} shows that the vacuum vector $\Omega$ does also belong to this set, establishing the Reeh-Schlieder property for double cones.\par
The significance of the split property for our construction is, therefore, made clear and one may look for a convenient method to prove this feature for wedge algebras. As indicated above the modular nuclearity condition, i.e. the requirement of the nuclearity of the maps (\ref{modNucCon}), for instance, implies the split property. There are, however, other possible approaches, e.g. \cite{buchholz1974product, buchholz1990nuclear, BW86}. Nevertheless, in the present context the maps $\Xi(x)$ take a concrete form by the Bisognano-Wichmann theorem, constituting an advantage in the investigation of their properties. In addition, by taking the modular nuclearity condition into account, we do not have to be concerned with the construction of interpolation type I factors, as is the case in a direct verification of the split property. A strong indication for choosing the modular nuclearity condition as the best suited method for our purposes provides the scalar case, discussed in \cite{L08, erratum}, where this condition was successfully verified\footnote{\label{footnote}As we noticed in the last stages of this thesis, the construction carried out in \cite{L08} is, in fact, incomplete. Fortunately by Erratum \cite{erratum}, the arguments can partially, however presently not to the full extent, be repaired. The reestablishment of the previously claimed results is currently being investigated.}. Due to these considerations, in Section \ref{VerifyingNuclearity} we aim at verifying the nuclearity of the maps $\Xi(x)$ for the models at hand. The implications we shall draw from the fulfillment of this property are summarized in the following theorem.
\begin{theorem}[\cite{BL4,L08}]\label{NuclearityCondition}
Let the modular nuclearity condition be satisfied, i.e. assume the maps $\Xi(x)$ (\ref{modNucCon}) are nuclear for $x\in W_R$. Then, the inclusion $\mathcal{F}(W_R+x)\subset\mathcal{F}(W_R)$ is (distal) split. Moreover, for $\mathcal{O}_{0,x}\in\mathscr{O}$,
\begin{itemize}
\item $\mathcal{F}(\mathcal{O}_{0,x})$ is isomorphic to the hyperfinite type III$_1$ factor,
\item  $\mathcal{F}(\mathcal{O}_{0,x})$ has the vacuum $\Omega$ as a cyclic vector.
\end{itemize}
\end{theorem}
Clearly, since any double cone $\mathcal{O}_{a,b}$, $b-a\in W_R$, can be transformed to $\mathcal{O}_{0,x}$ by a boost and a translation, the above stated results corresponding to $\mathcal{O}_{0,x}$ also hold for $\mathcal{F}(\mathcal{O}_{a,b})$ by covariance. By analogous arguments it also follows that the net $\mathcal{F}$ has the (distal) split property for wedges. In addition, further features such as Haag duality for double cones or weak additivity \cite{L08, muger1998superselection} may be inferred from the modular nuclearity condition. Since we are mainly interested in the Reeh-Schlieder property we do not discuss these consequences in more detail. However, we note that the modular nuclearity condition, and hence the split property, implies the compactness of the global gauge group $G$ \cite[Theorem 3.1]{DL84}.

\subsection{Verifying The Modular Nuclearity Condition}\label{VerifyingNuclearity}
In the previous section we discussed powerful techniques available in the literature by means of which the existence of local interacting models with prescribed S-matrix \mbox{$S\in\mathcal{S}$} may be proven. As explained above the verification of the nuclearity of the maps $\Xi(x)$ (\ref{modNucCon}) is thereby of central significance. The strategy followed here for this task is motivated by a similar construction carried out for the case of a particle spectrum consisting of only a single species of neutral massive particles \cite{DocL, L08}. However, the results obtained in this special case, dealing, in particular, only with scalar-valued S-matrices, do not generalize readily to the present setting of a richer particle spectrum. Nevertheless, our analysis shall show that the methods of the scalar case turn out to be effective also in the situation at hand, at least to a great extent.\par
We start our investigation of the maps $\Xi(x)$ (\ref{modNucCon}) by decomposing them into a series
\begin{equation}\label{xi}
\Xi(x)=\sum_{n=0}^{\infty}\Xi_n(x),
\end{equation}
where, since $\Delta^{1/4}$ and $U(x)$ commute with the particle number operator $N$,
\begin{equation}
\Xi_n(x):\mathcal{F}(W_R)\rightarrow\mathscr{H}_n,\qquad \Xi_n(x)A:=P_n\Xi(x) A=\Delta^{1/4}U(x)(A\Omega)_n,\qquad x\in W_R,
\end{equation}
with $(A\Omega)_n:=P_nA\Omega\in\mathscr{H}_n$, $n\in\mathbb{N}_0$. For $\Xi(x)$ to be nuclear we have to prove that all $\Xi_n(x)$ are nuclear with summable nuclear norms
\begin{equation}
\sum_{n=0}^{\infty}||\Xi_n(x)||_1<\infty.
\end{equation}
Namely, in that case (\ref{xi}) converges in nuclear norm, which, therefore, yields that $\Xi(x)$ is in the Banach space $\left(\mathcal{N}(\mathcal{F}(W_R),\mathscr{H}),\|\cdot\|_1\right)$ of nuclear maps.\par
A great facilitation in the analysis of the nuclearity of the maps $\Xi(x)$ is provided by the Bisognano-Wichmann property. For it allows for a concrete formulation of the rather abstract modular nuclearity condition. In particular, it follows from Proposition \ref{BisognWich} that the functions $\Xi_n(x)A$ have in terms of analytic continuation the concrete form
\begin{equation}
(\Xi_n(x)A)^{\boldsymbol{\alpha}}(\boldsymbol{\theta})=\prod_{k=1}^{n}e^{m_{[\alpha_k]}(x_0\sinh\theta_k-x_1\cosh\theta_k)}\cdot(A\Omega)_n^{\boldsymbol{\alpha}}(\theta_1-\tfrac{i\pi}{2},\dots,\theta_n-\tfrac{i\pi}{2}).
\end{equation}
A further simplification is obtained by the obvious fact that for $y\in W_R$ there exists a $t\in\mathbb{R}$ such that
\begin{equation*}
\begin{pmatrix}
\cosh t&\sinh t\\
\sinh t&\cosh t
\end{pmatrix}
\begin{pmatrix}
y^0\\y^1
\end{pmatrix}=
\begin{pmatrix}
0\\y^1
\end{pmatrix}
=:\begin{pmatrix}
0\\s
\end{pmatrix},\qquad s>0.
\end{equation*}
Since the modular operator $\Delta$ of $(\mathcal{F}(W_R),\Omega)$ commutes with the boosts $U(0,t)$ and in view of the invariance of $\Omega$ as well as $\mathcal{F}(W_R)$ under these unitaries, nuclearity of $\Xi(0,s)$ would imply that $\Xi(y)$ was also a nuclear map with the same nuclear norm. Thus, without loss of generality, we shall from now on consider $\Xi(x)=\Xi(0,s):=\Xi(s)$ with $s>0$, giving
\begin{equation}\label{concrete}
(\Xi_n(s)A)^{\boldsymbol{\alpha}}(\boldsymbol{\theta})=\prod_{k=1}^{n}e^{-m_{[\alpha_k]}s\cosh\theta_k}\cdot(A\Omega)_n^{\boldsymbol{\alpha}}(\theta_1-\tfrac{i\pi}{2},\dots,\theta_n-\tfrac{i\pi}{2}).
\end{equation}
As indicated above this concrete form suggests that the properties of the maps $\Xi(s)$ may be deduced from analytic features of the functions $(A\Omega)^{\boldsymbol{\alpha}}_n$, $A\in\mathcal{F}(W_R)$. In order to derive information about their analytic structure, we may make use of their localization. To this end, we consider the time zero fields of $\phi$, namely
\begin{equation}
\varphi_\alpha(x_1):=\sqrt{2\pi}\phi_\alpha(0,x_1),\qquad \pi_\alpha(x_1):=\sqrt{2\pi}(\partial_0\phi)_\alpha(0,x_1),\qquad x_1\in\mathbb{R},
\end{equation}
understood in the sense of operator-valued distributions. For test functions \mbox{$f\in\mathscr{S}(\mathbb{R})\otimes\mathcal{K}$}, we have with regard to the Definition (\ref{field1}) of $\phi$
\begin{eqnarray}\label{zero}
\varphi(f)&=&z^\dagger(\hat{f})+z(J\hat{f}_-),\\
\pi(f)&=&i\left(z^\dagger(\omega\hat{f})-z(\omega J\hat{f}_-)\right),
\end{eqnarray}
where $\hat{f}^{\alpha}(\theta):=\widetilde{f}^\alpha(m_{[\alpha]} \sinh\theta)$, $\hat{f}_-^{\alpha}(\theta):=\widetilde{f}^\alpha(-m_{[\alpha]} \sinh\theta)$, and
\begin{equation}
\left(\omega\,\Phi\right)_1^\alpha(\theta):=\omega_{[\alpha]}(\theta)\,\Phi_1^\alpha(\theta),\qquad \omega_{[\alpha]}(\theta):=m_{[\alpha]}\cosh\theta,\qquad\Phi\in\mathcal{D}_\omega\subset\mathscr{H}_1.
\end{equation}
The operators $\varphi(f)$ and $\pi(f)$ are well-defined on $\mathcal{D}\subset\mathscr{H}$, the subspace of finite particle number. Moreover, $\varphi(f^*)\subset\varphi(f)^*$, $\pi(f^*)\subset\pi(f)^*$ and by analogous arguments to those yielding (\ref{commutatorwedge}) one can prove that for $A\in\mathcal{F}(W_R)$ and supp$\,f \subset\mathbb{R}_-$
\begin{equation}
[\varphi(f),A]\Psi=0,\qquad[\pi(f),A]\Psi=0,\qquad\Psi\in\mathcal{D}.
\end{equation}
By means of the locality properties of $A\in\mathcal{F}(W_R)$, a certain Hardy space structure can be established, as shown below. That is, we may come across analytic functions $h^{\boldsymbol{\alpha}}:\mathcal{T}_{\mathcal{C}}\rightarrow\mathbb{C}$ in a Hardy space $H^2(\mathcal{T}_{\mathcal{C}})$, where $\mathcal{T}_{\mathcal{C}}:=\mathbb{R}^n+i\mathcal{C}\subset\mathbb{C}^n$ is a tube based on an open convex domain $\mathcal{C}\subset\mathbb{R}^n$. For each $\boldsymbol{\lambda}\in\mathcal{C}$ the maps $h_{\boldsymbol{\lambda}}:\boldsymbol{\theta}\mapsto h(\boldsymbol{\theta}+i\boldsymbol{\lambda})$ are elements of $L^2(\mathbb{R}^n)\otimes\mathcal{K}^{\otimes n}$ with finite norms denoted by $\|\cdot\|_2$. Correspondingly, their finite Hardy norms are given by
\begin{equation}
\triplenorm h\triplenorm:=\sup_{\boldsymbol{\lambda}\in\mathcal{C}}\|h_{\boldsymbol{\lambda}}\|_2=\sup_{\boldsymbol{\lambda}\in\mathcal{C}}\left(\sum_{\boldsymbol{\alpha}}\int\limits_{\mathbb{R}^n}d^n\boldsymbol{\theta}\left|h^{\boldsymbol{\alpha}}(\boldsymbol{\theta}+i\boldsymbol{\lambda})\right|^2\right)^{1/2}<\infty.
\end{equation}
In Appendix \ref{hardyAppendix} some useful properties concerning Hardy spaces on tube domains are collected.\par
For further purposes, let $S(a,b)$ denote a certain open strip region in the complex plane, namely
\begin{equation}\label{strip}
S(a,b):=\{\zeta\in\mathbb{C}:a<\text{Im}\,\zeta<b\}.
\end{equation}
Based on the localization of $A$ in the right wedge, we obtain the following generalization of a result established in the scalar case \cite{L08}.

\begin{lemma}\label{basicLemma}
\textit{Consider for} $A\in\mathcal{F}(W_R)$, $n_1,n_2\in\mathbb{N}_0$, \textit{and} $\Psi_i\in\mathcal{H}_{n_i}$, $i=1,2$, \textit{the functionals }$K,K^\dagger:\mathscr{S}(\mathbb{R})\otimes\mathcal{K}\rightarrow\mathbb{C}$
\begin{equation}\label{functionalC}
K(\hat{f}):=\langle\Psi_1,[z(\hat{f}),A]\Psi_2\rangle,\qquad K^\dagger(\hat{f}):=\langle\Psi_1,[z^\dagger(\hat{f}),A]\Psi_2\rangle,
\end{equation}
\textit{with} $\hat{f}^\alpha(\theta):=\widetilde{f}^\alpha(m_{[\alpha]}\sinh\theta)$. \textit{Then, there exists a function} $\hat{K}\in H^2(S(-\pi,0))\otimes\mathcal{K}$ \textit{whose boundary values satisfy}
\begin{equation}
K(\hat{f})=\sum_\alpha\int d\theta\, \hat{K}_{\alpha}(\theta) \overline{\hat{f}^\alpha(\theta)},\qquad
K^\dagger(\hat{f})=-\sum_\alpha\int d\theta\, \hat{K}_{\overline{\alpha}}(\theta-i\pi) \hat{f}^\alpha(\theta),
\end{equation}
\textit{and whose Hardy norm is bounded by}
\begin{equation}\label{hardynorm}
\triplenorm \hat{K}\triplenorm\leq \left((n_1+1)^{1/2}+(n_2+1)^{1/2} \right)\|\Psi_1\|\,\|\Psi_2\|\,\|A\|.
\end{equation}
\end{lemma}
\begin{proof}
Firstly, with regard to the particle number bounds (\ref{numberBounds}) and the Cauchy-Schwarz inequality, we have
\begin{eqnarray*}
\left|K(\hat{f})\right|&\leq& \|z^\dagger(\hat{f})\Psi_1\|\|A\Psi_2\|+\|A^*\Psi_1\|\|z(\hat{f})\Psi_2\|\\
&\leq&\left(\sqrt{n_1+1}+\sqrt{n_2}\right)\|\Psi_1\|\|\Psi_2\|\|A\|\,\|\hat{f}\|,\\
\left|K^\dagger(\hat{f})\right|&\leq&\left(\sqrt{n_1}+\sqrt{n_2+1}\right)\|\Psi_1\|\|\Psi_2\|\|A\|\,\|\hat{f}\|.
\end{eqnarray*}
That is, by application of Riesz' Lemma \cite[Theorem II.4]{RS1}, it follows that the distributions $K$ and $K^\dagger$ are given by integration against functions in $\mathscr{H}_1=L^2(\mathbb{R})\otimes\mathcal{K}$ denoted by  $\hat{K}^{\#}\in\{\hat{K},\,\hat{K}^\dagger\}$. Their norms are bounded by
\begin{equation}\label{esti1}
\|\hat{K}^{\#}\|_2\leq\left(\sqrt{n_1+1}+\sqrt{n_2+1}\right)\|\Psi_1\|\|\Psi_2\|\|A\|.
\end{equation}
Next, we show by means of the time zero fields (\ref{zero}) the claimed analytic structure. To this end, consider the functionals $K_\pm:\mathscr{S}(\mathbb{R})\otimes\mathcal{K}\rightarrow\mathbb{C}$,
\begin{equation}
K_-(f):=\langle\Psi_1,[\varphi(f),A]\Psi_2\rangle,\qquad K_+(f):=\langle\Psi_1,[\pi(f),A]\Psi_2\rangle,
\end{equation}
which are related to $K^{\#}$ via
\begin{eqnarray}
z(\hat{f})&=&\dfrac{1}{2}\left(\varphi(f^*)+i\pi(\omega^{-1}f^*)\right),\\ z^\dagger(\hat{f})&=&\dfrac{1}{2}\left(\varphi(f)-i\pi(\omega^{-1}f)\right).
\end{eqnarray}
Note that $\widehat{f^*}^\alpha(\theta)=(J\hat{f}_-)^\alpha(\theta)$. Moreover, since $K_\pm$ vanishes for supp$\,f\subset\mathbb{R}_-$ it follows that supp$\,K_\pm\subset\mathbb{R}_+$. Thus, there exist functions $p\mapsto\widetilde{K}_{\pm}(p)$ which are analytic in the lower half plane and have polynomial bounds at the real boundary and at infinity such that the Fourier transforms of $K_\pm$ are their distributional boundary values \cite[Thm. IX.16]{RS2}. Taking into account that $\sinh$ maps the strip $S(-\pi,0)$ to the lower half plane, we have that
\begin{equation}\label{Cpm}
\hat{K}_{+,\alpha}(\theta):=\widetilde{K}_{+,\alpha}(m_{[\alpha]} \sinh\theta),\qquad \hat{K}_{-,\alpha}(\theta):=m_{[\alpha]} \cosh\theta\,\widetilde{K}_{-,\alpha}(m_{[\alpha]} \sinh\theta),
\end{equation}
are analytic in this strip. Hence, it follows that
\begin{eqnarray*}
K(\hat{f})&=&\dfrac{1}{2}\left(K_-(f^*)+iK_+(\omega^{-1}f^*)\right)=\dfrac{1}{2}\int dp\left(\widetilde{K}_{-,\alpha}(p)+i\omega(p)^{-1}\widetilde{K}_{+,\alpha}(p)\right)\overline{\widetilde{f}^{\overline{\alpha}}(p)}\\
&=&\dfrac{1}{2}\int d\theta \left(\hat{K}_{-,\overline{\alpha}}(\theta)+i\hat{K}_{+,\overline{\alpha}}(\theta)\right)\overline{\hat{f}^{\alpha}(\theta)}
=\int d\theta\, \hat{K}_{\alpha}(\theta)\overline{\hat{f}^\alpha(\theta)},
\end{eqnarray*}
and
\begin{eqnarray*}
K^\dagger(\hat{f})=\dfrac{1}{2}\left(K_-(f)-iK_+(\omega^{-1}f)\right)
&=&\dfrac{1}{2}\int d\theta \left(\hat{K}_{-,\alpha}(-\theta)-i\hat{K}_{+,\alpha}(-\theta)\right)\hat{f}^\alpha(\theta)\\
&=&
\int d\theta\, \hat{K}^\dagger_\alpha(\theta)\hat{f}^\alpha(\theta),
\end{eqnarray*}
respectively. These equations, moreover, yield
\begin{equation}
\hat{K}_{\overline{\alpha}}(\theta)=\dfrac{1}{2}\left(\hat{K}_{-,\alpha}(\theta)+i\hat{K}_{+,\alpha}(\theta)\right),\qquad \hat{K}_\alpha^\dagger(\theta)=\dfrac{1}{2}\left(\hat{K}_{-,\alpha}(-\theta)-i\hat{K}_{+,\alpha}(-\theta)\right),
\end{equation}
and, thus, the analyticity of $\theta\mapsto \hat{K}_\alpha(\theta)$ in the strip $S(-\pi,0)$. Furthermore, it follows that the boundary values of $\hat{K}_\pm$ also exist as functions in $L^2(\mathbb{R})\otimes\mathcal{K}$. Since $\hat{K}_{\pm,\alpha}(\theta-i\pi)=\pm\hat{K}_{\pm,\alpha}(-\theta)$ for $\theta\in\mathbb{R}$ by (\ref{Cpm}), we have $\hat{K}_\alpha^\dagger(\theta)=-\hat{K}_{\overline{\alpha}}(\theta-i\pi)$.\par
It, therefore, remains to prove that the function $\hat{K}$ is an element of the Hardy space $H^2(S(-\pi,0))\otimes\mathcal{K}$. For that purpose, we consider $\hat{K}_\alpha^{(s)}(\zeta):=e^{-im_{[\alpha]}s\sinh\zeta}\hat{K}_\alpha(\zeta)$, with $s>0$, which clearly is analytic in the strip $S(-\pi,0)$. The identity
\begin{equation}\label{esti2}
\left|\hat{K}^{(s)}_{-\lambda,\alpha}(\theta)\right|=\frac{1}{2}e^{-m_{[\alpha]}s\sin\lambda\cosh\theta}\left|\hat{K}_{-,\overline{\alpha}}(\theta-i\lambda)+i\hat{K}_{+,\overline{\alpha}}(\theta-i\lambda)\right|
\end{equation}
yields that $\hat{K}^{(s)}_{-\lambda,\alpha}\in L^2(\mathbb{R})$ for all $\lambda\in[0,\pi]$ and $s>0$, since $\theta\mapsto\hat{K}_{\pm,\alpha}(\theta-i\lambda)$ is bounded by polynomials in $\cosh\theta$ for $\theta\rightarrow\infty$ and $0<\lambda<\pi$. Noting that $\|\hat{K}^{(0)}_{0/-\pi}\|_2=\|\hat{K}^{(s)}_{0/-\pi}\|_2$ and with regard to (\ref{esti1}), the three lines theorem may be applied and we arrive at
\begin{equation}\label{esti3}
\|\hat{K}^{(s)}_{-\lambda}\|_2\leq\left(\sqrt{n_1+1}+\sqrt{n_2+1}\right)\|\Psi_1\|\|\Psi_2\|\|A\|,\qquad 0\leq\lambda\leq\pi.
\end{equation}
Since (\ref{esti2}) is monotonically increasing for $s\rightarrow 0$, it follows that the uniform bound (\ref{esti3}) holds also for $\hat{K}_{-\lambda}=\hat{K}_{-\lambda}^{(0)}$, with $0\leq\lambda\leq \pi$. Thus,
\begin{equation*}
\hat{K}\in H^2(S(-\pi,0))\otimes\mathcal{K},\qquad\triplenorm \hat{K}\triplenorm\leq \left((n_1+1)^{1/2}+(n_2+1)^{1/2} \right)\|\Psi_1\|\,\|\Psi_2\|\,\|A\|.
\end{equation*}
\end{proof}
Putting $\Psi_1=\Psi_2=\Omega$ in (\ref{functionalC}), we have
\begin{equation}
\int d\theta  \overline{\hat{f}_\alpha(\theta)}\hat{K}_{\alpha}(\theta)=\langle\Omega,z(\hat{f})A\Omega\rangle=\langle\hat{f},A\Omega\rangle=\int d\theta\overline{\hat{f}_\alpha(\theta)}(A\Omega)^\alpha_1(\theta),
\end{equation}
which implies that the single particle functions corresponding to operators localized in $W_R$ are boundary values of functions in $H^2(S(-\pi,0))\otimes\mathcal{K}$ with norms $\triplenorm(A\Omega)_1\triplenorm\leq 2\|A\|$. In this special case, however, it is obvious that the bound can actually be improved to $\triplenorm(A\Omega)_1\triplenorm\leq \|A\|$.\par
From this observation one can infer that the map $\Xi_1(s)$ is nuclear. Namely, it is possible to consider $\Xi_1(s)$ as the composition of two maps denoted by $\Upsilon_1$ and $X_1(s)$ as the following diagram explains.
    \begin{center}
        \begin{tikzpicture}
        \begin{scope}%[transparent]    %current
                \draw[fill] (0.3,2.2) node{$\mathcal{F}(W_R)$};
\draw[fill] (-1,0) node{$H^2(S(-\pi,0))\otimes\mathcal{K}$};
\draw[fill] (2.8,0) node{$\mathscr{H}_1$};
\draw[fill] (2.2,1.3) node{$\Xi_1(s)$};
\draw[fill] (-1,1.3) node{$\Upsilon_1$};
\draw[fill] (1.7,-.3) node{$X_1(s)$};
        \begin{scope}[->]
            \draw (0,2) -- (-1,0.5) node[anchor=south] {};
            \draw (.8,0) -- (2.4,0) node[anchor=east] {};
            \draw (0.7,2) -- (2.4,0.2) node[anchor=east] {};
        \end{scope}
        \end{scope}
        \end{tikzpicture} 
        \end{center}
 The map $\Upsilon_1$ from $\mathcal{F}(W_R)$ to $H^2(S(-\pi,0))\otimes\mathcal{K}$ acts according to $\Upsilon_1A:=(A\Omega)_1$, whereas $X_1(s):H^2(S(-\pi,0))\otimes\mathcal{K}\rightarrow\mathscr{H}_1$ according to $(A\Omega)_1\mapsto\Xi_1(s)A$. We have already shown that $\Upsilon_1$ is bounded as a linear map between the Banach spaces $(\mathcal{F}(W_R),\|\cdot\|_{\mathcal{B}(\mathscr{H})})$ and $(H^2(S(-\pi,0))\otimes\mathcal{K},\triplenorm\cdot\triplenorm)$. Comparing with Equation (\ref{concrete}) we find for $X_1(s)$ the following explicit action
\begin{equation}
(X_1(s)h)^\alpha(\theta):=e^{-sm_{[\alpha]} \cosh\theta}h^\alpha(\theta-\tfrac{i\pi}{2}).
\end{equation}
Since any $h^\alpha\in H^2(S(-\pi,0))$ is an analytic function which, moreover, has the property that $h_\lambda(\theta)\rightarrow 0$ uniform in $\lambda$ for $|\theta|\rightarrow\infty$ if $-\pi\leq\lambda\leq 0$, and that $h_\lambda$ converges in the norm topology of $L^2(\mathbb{R})$ as $\lambda$ approaches the boundary of $S(-\pi,0)$, cf. Appendix \ref{hardyAppendix}, we may consider its Cauchy integral over a closed curve $\gamma$ around the point $\theta-\tfrac{i\pi}{2}$ and then deform $\gamma$ to the boundary of the strip $S(-\pi,0)$. This gives the identity
\begin{equation}
h^\alpha(\theta-\tfrac{i\pi}{2})=\frac{1}{2\pi i}\oint_\gamma d\zeta'\frac{h^\alpha(\zeta')}{\zeta'-(\theta-\tfrac{i\pi}{2})}=\frac{1}{2\pi i}\int_\mathbb{R}d\theta'\left(\frac{h^\alpha(\theta'-i\pi)}{\theta'-\theta-\tfrac{i\pi}{2}}-\frac{h^\alpha(\theta')}{\theta'-\theta+\tfrac{i\pi}{2}}\right).
\end{equation}
Introducing integral operators $T_{s,\pm}$ on $L^2(\mathbb{R})\otimes\mathcal{K}$ defined by the kernels
\begin{equation}
T^{[\alpha]}_{s,\pm}(\theta,\theta'):=\frac{1}{\pi i}\frac{e^{-sm_{[\alpha]} \cosh\theta}}{\theta'-\theta\pm \tfrac{i\pi}{2}},
\end{equation}
we obtain
\begin{equation}
(X_1(s)h)^\alpha(\theta)=\frac{1}{2}(T_{s,-}\,h_{-\pi}-T_{s,+}h_{0})^\alpha(\theta).
\end{equation}
It can be shown \cite[Appendix B.2]{DocL} that $T_{s,\pm}$ are trace class operators on $L^2(\mathbb{R})\otimes\mathcal{K}$ for any $s>0$, with trace norms $\|T_{s,\pm}\|_1\leq D\|T^{[\alpha_\circ]}_{s,\pm}\|_1\leq c(s,m_{[\alpha_\circ]}:=m_{\circ})<\infty$, cf. (\ref{tracenorm}). The boundedness of the maps $h_\alpha\mapsto h_{0,\alpha}$ and $h_\alpha\mapsto h_{-\pi,\alpha}$ from $H^2(S(-\pi,0))$ to $L^2(\mathbb{R})$, with norms not exceeding one, implies that
\begin{equation}
\|X_1(s)\|_1\leq c(s,m_{\circ}),
\end{equation}
resulting in the nuclearity of $\Xi_1(s)$, cf. Lemma \ref{propertiesNuclearMaps}.
\begin{lemma}\label{oneparticle}
In a model with S-matrix $S\in\mathcal{S}$ the map $\Xi_1(s):\mathcal{F}(W_R)\rightarrow\mathscr{H}_1$, $A\mapsto \Delta^{1/4}U(\undertilde{s})(A\Omega)_1$, $\undertilde{s}:=(0,s)$, is nuclear for any $s>0$.
\end{lemma}
In order to prove nuclearity of the maps $\Xi_n(s)$ with $n>1$, we may proceed in the same manner as in the case $n=1$. Therefore, we start with the investigation of analytic and boundedness properties of $(A\Omega)_n^{\boldsymbol{\alpha}}$, $n>1$, by means of Lemma \ref{basicLemma}. Since
$$\sqrt{n!}\,(A\Omega)_n^{\boldsymbol{\alpha}}(\boldsymbol{\theta})=\langle z_{\alpha_1
}^\dagger(\theta_1)\cdots z_{\alpha_n
}^\dagger(\theta_n)\Omega,A\Omega\rangle=\langle z_{\alpha_2
}^\dagger(\theta_2)\cdots z_{\alpha_n
}^\dagger(\theta_n)\Omega,[z_{\alpha_1
}(\theta_1),A]\Omega\rangle,$$
Lemma \ref{basicLemma} yields for $A\in\mathcal{F}(W_R)$ analyticity of the expression above in the variable $\theta_1$ in the strip $S(-\pi,0)$, with boundary value at Im$(\theta_1)=-\pi$ given by
\begin{multline}\label{Rechnung}
\langle z_{\alpha_2
}^\dagger(\theta_2)\cdots z_{\alpha_n
}^\dagger(\theta_n)\Omega,[A,z^\dagger_{\overline{\alpha}_1}
(\theta_1)]\Omega\rangle\\
=\langle z_{\alpha_2
}^\dagger(\theta_2)\cdots z_{\alpha_n
}^\dagger(\theta_n)\Omega,Az^\dagger_{\overline{\alpha}_1}
(\theta_1)\Omega\rangle-\langle z_{\overline{\alpha}_1}
(\theta_1) z_{\alpha_2
}^\dagger(\theta_2)\cdots z_{\alpha_n
}^\dagger(\theta_n)\Omega,A\Omega\rangle.
\end{multline}
This calculation suggests the general analysis of matrix elements of the form
\begin{equation*}
\langle z_{\alpha_{k+1}}^\dagger(\theta_{k+1})\cdots z_{\alpha_{n}}^\dagger(\theta_{n})\Omega,Az_{\overline{\alpha}_{k}}^\dagger(\theta_{k})\cdots z_{\overline{\alpha}_{1}}^\dagger(\theta_{1})\Omega\rangle, \qquad A\in\mathcal{F}(W_R),
\end{equation*}
and, in particular, the investigation of ``contracted'' matrix elements, corresponding to the second term in Equation (\ref{Rechnung}). These contractions between the variables $\theta_{k+1},\dots,\theta_n$ and $\theta_k,\dots,\theta_1$ on the respective sides of the scalar product result from repeated application of the exchange relations (\ref{exchange}).\par
For the description of such contracted matrix elements we adopt the notation introduced in \cite{DocL}, as follows. The set of contractions $\mathscr{C}_{n,k}$ is defined to be the power set of $\left\lbrace k+1,\dots,n\right\rbrace\times\left\lbrace 1,\dots,k \right\rbrace$, with integers $0\leq k\leq n$, hence $\mathscr{C}_{n,0}=\mathscr{C}_{n,n}=\emptyset$. Furthermore, a contraction $C\in\mathscr{C}_{n,k}$ is parametrized by
\begin{itemize}
\item an ordered set of ``right'' indices $1\leq r_1<\dots<r_{|C|}\leq k$,
\item an unordered set of pairwise different ``left'' indices $k+1\leq l_1,\dots, l_{|C|}\leq n$, and
\item a permutation of $\left\lbrace l_1,\dots, l_{|C|} \right\rbrace$,
\end{itemize}
giving pairs $(l_i,r_i)\in C$. The index $|C|\leq \text{min}\{k,n-k\}$ denotes the length of the contraction, i.e. the number of pairs.\par
With these notations we define for a $C\in\mathscr{C}_{n,k}$ the corresponding contracted matrix element as follows
\begin{multline}\label{conmatrixelements}
\langle A\rangle_C^{\alpha_{k+1}\,\stackrel{\widehat{l}}{\dots}\,\alpha_n\,\alpha_1\,\stackrel{\widehat{r}}{\dots}\,\alpha_k}(\theta_1,\stackrel{\widehat{r}}{\dots},\theta_k,\theta_{k+1},,\stackrel{\widehat{l}}{\dots},\theta_n)\\
:=\langle z_{\alpha_{k+1}}^\dagger(\theta_{k+1})\stackrel{\widehat{l}}{\cdots} z_{\alpha_{n}}^\dagger(\theta_{n})\Omega,A\,z_{\overline{\alpha}_{k}}^\dagger(\theta_{k})\stackrel{\widehat{r}}{\cdots} z_{\overline{\alpha}_{1}}^\dagger(\theta_{1})\Omega \rangle,
\end{multline}
where the symbols $\widehat{l}$ and $\widehat{r}$ indicate the omission of $z^\dagger(\theta_{l_i})$, $z^\dagger(\theta_{r_i})$, $\alpha_{r_i/l_i}$ and $\theta_{r_i/l_i}$ with $i=1,\dots,|C|$. Note, in particular, the order of the indices. Having regard to the particle number bounds (\ref{numberBounds}) and $A$ being a bounded operator, the contracted matrix elements (\ref{conmatrixelements}) are well-defined tempered distributions on $\mathscr{S}(\mathbb{R}^{n-2|C|})\otimes\mathcal{K}^{\otimes (n-2|C|)}$. In particular, for functions $F\in L^2(\mathbb{R}^{n-k-|C|})\otimes\mathcal{K}^{\otimes (n-k-|C|)}$ and $G\in L^2(\mathbb{R}^{k-|C|})\otimes\mathcal{K}^{\otimes (k-|C|)}$ we find by means of the Cauchy-Schwarz inequality the bounds
\begin{eqnarray}\label{estimate1}
|\langle A\rangle_C\left(F\otimes G\right)|&=&\Big|\int d^{n-2|C|}\underline{\theta}\Big(\overline{(F\otimes G)(\underline{\theta})},\langle A\rangle_C(\underline{\theta})\Big)\Big|\nonumber\\
&=&\Big|\int d^{n-2|C|}\underline{\theta}\sum_{\underline{\alpha}}\left(F\otimes G\right)^{\underline{\alpha}}(\underline{\theta})\,\langle A\rangle_C^{\underline{\alpha}}(\underline{\theta})\Big|\\&\leq& \sqrt{(n-k-|C|)!}\sqrt{(k-|C|)!}\,||F||\,||G||\,||A||.\nonumber
\end{eqnarray}
Note that $\underline{\theta}$ and $\underline{\alpha}$ are $(n-2|C|)$-tuples.\par
Returning now to our example (\ref{Rechnung}) we have, in particular,
\begin{eqnarray}\label{sample}
\hspace{-.8cm}&&\hspace{-.8cm}\langle z_{\alpha_2
}^\dagger(\theta_{2})\cdots z_{\alpha_n
}^\dagger(\theta_{n})\Omega,[A,z^\dagger_{\overline{\alpha}_1}(\theta_{1})
]\Omega\rangle\nonumber \\
&&=\langle z_{\alpha_2
}^\dagger(\theta_{2})\cdots z_{\alpha_n
}^\dagger(\theta_{n})\Omega,Az^\dagger_{\overline{\alpha}_1}(\theta_{1})
\Omega\rangle
-\sum_{l=2}^{n}\delta(\theta_l-\theta_1)\delta^{\alpha_l\overline{\xi_{l-1}}}\,\delta^{\overline{\alpha}_1\overline{\xi}_1}\prod_{m=2}^{l-1}S^{\alpha_m\overline{\xi}_m}_{\overline{\xi}_{m-1}\beta_m}(\theta_1-\theta_m)\nonumber\\
&&\hspace{.8cm}\times\langle z_{\beta_2
}^\dagger(\theta_{2})\cdots z_{\beta_{l-1}
}^\dagger(\theta_{l-1}) z_{\alpha_{l+1}
}^\dagger(\theta_{l+1}) \cdots z_{\alpha_n
}^\dagger(\theta_{n})\Omega,A
\Omega\rangle\nonumber\\
&&=\sum_{C\in\mathscr{C}_{n,1}}(-1)^{|C|}\,\hat{\delta}_C^{\boldsymbol{\alpha'}}(\boldsymbol{\theta})\,\hat{S}_C^{\boldsymbol{\beta'}}(\boldsymbol{\theta})\,\langle A\rangle_C^{\underline{\gamma}}(\underline{\theta}),
\end{eqnarray}
where the exchange relations (\ref{exchange}) were used and the factors $\hat{\delta}_{C}^{\boldsymbol{\alpha'}}$ and $\hat{S}_{C}^{\boldsymbol{\beta'}}$ are short hand for the product of the delta distribution with the Kronecker deltas and the product of S-matrices respectively. We shall in the following refer to such kernels as completely contracted matrix elements and denote them by $\langle A\rangle^{\text{con}}_{n,1}$ for the case $k=1$. Since we are interested in investigating analytic and boundedness properties of $(A\Omega)_n^{\boldsymbol{\alpha}}$, $n>1$, we have to consider general $0\leq k\leq n$. From the computation performed above it is obvious that, in contrast to the single particle case, the underlying S-matrix has an important impact on the analytic structure, and hence on the nuclearity properties of the maps $\Xi_n(s)$. It turns out that it is necessary to restrict to a certain subclass $\mathcal{S}_0\subset\mathcal{S}$ of S-matrices, in order to establish the desired result\footnote{This intermediate step was also necessary in the special case of scalar field theories \cite{L08}.}. We define the subfamily $\mathcal{S}_0$ as follows.
\begin{definition}\label{regularS}
The subset $\mathcal{S}_0\subset\mathcal{S}$ consists of those S-matrices which extend to bounded and analytic functions on enlarged strips $S(-\kappa,\pi+\kappa)$, with $0<\kappa\leq\kappa(S)$ and
\begin{equation}\label{kappa}
\kappa(S):=\max\{\text{Im}\,\theta:\theta\in S(0,\tfrac{\pi}{2}),\,\,\det S(\theta)\neq0\}.
\end{equation}
Furthermore, these S-matrices satisfy
\begin{equation}\label{Skappabounded}
\|S\|_\kappa:=\sup\{\|S(\zeta)\|:\zeta\in\overline{S(-\kappa,\pi+\kappa)}\}<\infty,\qquad 0<\kappa\leq\kappa(S).
\end{equation}
The family $\mathcal{S}_0$ is referred to as the set of regular S-matrices.
\end{definition}
Note that regular S-matrices are smooth on the real line and $\|\partial^n_\theta S(\theta)\|\leq c_n$ for all $\theta\in\mathbb{R}$, $n\in\mathbb{N}_0$ and constants $c_n$, depending on $S$, as can be derived by application of Cauchy's integral formula.\par
As the concrete calculation (\ref{sample}) already illustrates, certain products of S-matrices arise. Depending on the contraction $C$, they can become rather complex. Even though it is possible to find an explicit formula for these products of S-matrices\footnote{Unpublished notes by the author.}, it is more convenient to take the underlying action of the permutation group into account. Therefore, using the notation introduced above, we associate to a contraction $C\in\mathscr{C}_{n,k}$ the permutations
\begin{equation}\label{PiContraction}
\pi_{\rho}:=\prod_{i=1}^{|C|}\tau_{r_i-i+1}\cdot\tau_{r_i-i+2}\cdots\tau_{k-i},\qquad
\pi_{\lambda}:=\prod_{i=1}^{|C|}\tau_{l_i+u_i-1}\cdot\tau_{l_i+u_i-2}\cdots\tau_{k+i},
\end{equation}
with $\pi_{\rho/\lambda}\in\mathfrak{S}_n$ and
where $0\leq u_i\leq |C|-1$ is defined by $$u_i:=\text{card}\{l_j\in\{l_1,\dots,l_{|C|}\}:l_j>l_i,\, j<i\}.$$
It should be noted that the products representing $\pi_\rho$ and $\pi_\lambda$ respectively are understood as ordered ones, namely in the sense of (\ref{convention}). Alternatively, we may use Cauchy's two line notation to represent these permutations, that is,

\begin{eqnarray}\label{PiRL}
\begin{aligned}
\pi_\rho&=
\left(
\begin{array}{ccccccccccccccc}
1&&&\dots&&&& k& k+1&&&\dots&&& n \\
1&\dots&\widehat{r}&\dots& k& r_{|C|}&\dots & r_1& k+1&&&\dots&&& n
\end{array}\right),\\
\pi_\lambda&=
\left(
\begin{array}{ccccccccccccccc}
1&&&\dots&&&& k& k+1&&&\dots&&& n \\
1&&&\dots&&&& k& l_1&\dots & l_{|C|}& k+1&\widehat{l}&\dots& n
\end{array}\right),
\end{aligned}
\end{eqnarray}
where $\widehat{r}$ and $\widehat{l}$ indicate the omission of the elements $r_i$ and $l_i$, $i=1,\dots,|C|$, respectively. We further put \begin{equation}\label{piC}
\pi_C:=\pi_{\rho}\cdot\pi_{\lambda}.
\end{equation}
Let, moreover, $\pi_a^b\in\mathfrak{S}_n$ be the permutation moving the element $a$, $1\leq a\leq n$, to $b$, $1\leq b\leq n$, defined by
\begin{equation}\label{schiebePermut}
\begin{aligned}
\pi_a^b&:=\prod_{j=a}^{b-1}\tau_j,\qquad a<b,\\
\pi_a^b&:=\prod_{j=a-1}^{b}\tau_j,\qquad a>b.
\end{aligned}
\end{equation}
In terms of the two line notation we have
\begin{equation}\label{schiebepermu}
\pi_a^b=\left\{
\begin{array}{cc}
\left(
\begin{array}{cccccccccccc}
1& &\dots &    a & &\dots& & b & & \dots& & n\\
1&\dots& a-1& a+1 & &\dots& b& a& b+1&\dots& & n
\end{array}\right),& a<b,\\
\left(
\begin{array}{cccccccccccc}
1& &\dots &    b & &\dots& & a &  & \dots& & n\\
1&\dots& b-1& a & b &\dots& & a-1& a+1 &\dots& & n
\end{array}\right),& a>b.
\end{array}\right.
\end{equation}
Then, the permutations (\ref{PiContraction}) may be expressed by means of (\ref{schiebePermut}) as follows
\begin{eqnarray}\label{permuShort}
\pi_{\rho}=\prod_{i=1}^{|C|}\pi_{r_i-i+1}^{k-i+1},\qquad\pi_{\lambda}=\prod_{i=1}^{|C|}\pi_{l_i+u_i}^{k+i},
\end{eqnarray}
with $u_i$ as above.
\begin{lemma}\label{shortLemma}
The permutations $\pi_\rho$ and $\pi_\lambda$, defined in (\ref{PiContraction}), commute.
\end{lemma}
\begin{proof}
The statement can be read off directly from (\ref{PiRL}). Alternatively, with regard to the explicit expressions (\ref{PiContraction}), we note that for transpositions $\tau_i,\tau_j\in\mathfrak{S}_n$, $i,j=1,\dots,n-1$, we have $\tau_i\cdot\tau_j=\tau_j\cdot\tau_i$ in case $|i-j|>1$, likewise yielding
$$\pi_\rho\cdot\pi_\lambda=\pi_\lambda\cdot\pi_\rho.$$
\end{proof}
In the following, we shall consider certain relations between contractions. Of particular interest in this context is the differentiation between contractions $C\in\mathscr{C}_{n,k}$ which do not contract $k+1$, that is, satisfy $k+1\notin\{ l_1,\dots,l_{|C|}\}$, and those which do, i.e. fulfill $k+1\in\{ l_1,\dots,l_{|C|}\}$. The sets corresponding to this distinction of cases are denoted by $\hat{\mathscr{C}}_{n,k}$ and $\check{\mathscr{C}}_{n,k}$ respectively.  Their disjoint union $\hat{\mathscr{C}}_{n,k}\sqcup\check{\mathscr{C}}_{n,k}$ is the set of all contractions, namely $\mathscr{C}_{n,k}$. Moreover, a contraction $C'\in\check{\mathscr{C}}_{n,k}$ is always given by $C'=C\cup\{(k+1,r)\}$, where $C\in\hat{\mathscr{C}}_{n,k}$ and $r\notin\{r_1,\dots,r_{|C|}\}$, and hence $|C'|=|C|+1$.\par
In a similar fashion, contractions $C''\in\check{\check{\mathscr{C}}}_{n,k+1}$, contracting $k+1$ as a ``right'' index, are unions of the form $C''=\{(l,k+1)\}\cup \widetilde{C}$, with $\widetilde{C}\in\hat{\hat{\mathscr{C}}}_{n,k+1}$, not contracting $k+1$ as a right index, and $l\notin\{l_1,\dots,l_{|\widetilde{C}|}\}$.\par
Correspondingly, we denote the above defined permutations associated with e.g. $C'$ by $\pi_{\rho'}$, $\pi_{\lambda'}$ and $\pi_{C'}$ respectively, and use analogous notation in case of $C''$ and $\widetilde{C}$.
\begin{lemma}\label{LemmaVorbereitung01}
Let $C\in\hat{\mathscr{C}}_{n,k}$, where $\hat{\mathscr{C}}_{n,k}\subset\mathscr{C}_{n,k}$ denotes the set of contractions $C\in\mathscr{C}_{n,k}$ for which $k+1\notin\{ l_1,\dots,l_{|C|}\}$, and let $C'\in\check{\mathscr{C}}_{n,k}$, with $\check{\mathscr{C}}_{n,k}\subset\mathscr{C}_{n,k}$ the set of contractions $C'\in\mathscr{C}_{n,k}$ for which $k+1\in\{ l'_1,\dots,l'_{|C'|}\}$, such that $C'=C\cup\{(k+1,r)\}$ with $r\notin\{r_1,\dots,r_{|C|}\}$, then we have
\begin{eqnarray}
\pi_{\rho'}&=&\pi_\rho\cdot\pi_{r-v_r}^{k-v_r},\\
\pi_{\lambda'}&=&\pi_\lambda\cdot \pi_{k+1+|C|}^{k+1+v_r},
\end{eqnarray}
hence,
\begin{equation}
\pi_{C'}=\pi_C\cdot\pi_{r-v_r}^{k-v_r}\cdot\pi_{k+1+|C|}^{k+1+v_r},
\end{equation}
with $v_r:=$card$\{r_i\in\{r_1,\dots,r_{|C|}\}:r_i<r\}$.
\end{lemma}
\begin{proof}
We present two alternatives to prove this statement. For both cases we first note that $|C'|=|C|+1$, $\{r'_1,\dots,r'_{|C'|}\}=\{r_1,\dots,r_{v_r},r,r_{v_r+1}\dots,r_{|C|}\}$ and $\{l'_1,\dots,l'_{|C'|}\}=\{l_1,\dots,l_{v_r},k+1,l_{v_r+1},\dots,l_{|C|}\}$, with $0\leq v_r\leq |C|$.\par
Putting $\boldsymbol{r^\#}=(r^\#_1,\dots,r^\#_{|C|^\#})$, $r_i^\#\in\{r_i,r_i'\}$, then, one may simply read off
\begin{eqnarray*}
\pi_{\rho'}&=&
\left(
\begin{array}{ccccccccccccccc}
1&&&\dots&&&& k& k+1&&&\dots&&& n \\
1&\dots&\widehat{\boldsymbol{r}'}&\dots& k& r'_{|C'|}&\dots & r'_1& k+1&&&\dots&&& n
\end{array}\right)\\
&=&
\left(
\begin{array}{ccccccccccccccc}
1&&&&\dots&&&&& k& k+1&&\dots && n \\
1&\dots&\widehat{\boldsymbol{r}},\widehat{r}&\dots& k& r_{|C|}&\dots & r_{v_r+1}\,\,r\,\,r_{v_r}&\dots & r_1& k+1&&\dots&& n
\end{array}\right)
\end{eqnarray*}
that $\pi_{\rho'}=\pi_\rho\cdot\pi_{r-v_r}^{k-v_r}$. Similar arguments apply to $\pi_{\lambda'}$, already proving the claim.\par
As in the proof of Lemma \ref{shortLemma} we also want to make this statement plausible with regard to the explicit expressions of $\pi_{\rho'}$ and $\pi_{\lambda'}$ in terms of transpositions, cf. (\ref{PiContraction}). Since this is a rather lengthy task, the corresponding calculations can be found in Appendix \ref{LemmaVor01Appendix}.
\end{proof}

\begin{lemma}\label{LemmaVorbereitung02}
Let $\widetilde{C}\in\hat{\hat{\mathscr{C}}}_{n,k+1}$, where $\hat{\hat{\mathscr{C}}}_{n,k+1}\subset\mathscr{C}_{n,k+1}$ denotes the set of contractions \mbox{$\widetilde{C}\in\mathscr{C}_{n,k+1}$} for which $k+1\notin \{\widetilde{r}_1,\dots,\widetilde{r}_{|\widetilde{C}|}\}$, and let $C''\in\check{\check{\mathscr{C}}}_{n,k+1}$, with $\check{\check{\mathscr{C}}}_{n,k+1}\subset\mathscr{C}_{n,k+1}$ the set of contractions $C''\in\mathscr{C}_{n,k+1}$ for which $k+1\in\{r''_1,\dots,r''_{|C''|}\}$, such that $C''=\{(l,k+1)\}\cup \widetilde{C}$ with $l\notin\{\widetilde{l}_1,\dots,\widetilde{l}_{|\widetilde{C}|}\}$, then we have
\begin{equation}
\pi_{C''}=\pi_{\widetilde{C}}\cdot\pi_{l+u_l}^{k+2+|\widetilde{C}|},
\end{equation}
with $0\leq u_l\leq |\widetilde{C}|$, defined by $u_l:=$card$\{\widetilde{l}_i\in\{\widetilde{l}_1,\dots,\widetilde{l}_{|\widetilde{C}|}\}:\widetilde{l}_i>l\}$.
\end{lemma}
\begin{proof}
The contractions $\widetilde{C}$ and $C''$ are first of all linked to each other by the relations $|C''|=|\widetilde{C}|+1$, $\{r''_1,\dots,r''_{|C''|}\}=\{\widetilde{r}_1,\dots,\widetilde{r}_{|\widetilde{C}|},k+1\}$ and $\{l''_1,\dots,l''_{|C''|}\}=\{\widetilde{l}_1,\dots,\widetilde{l}_{|\widetilde{C}|},l\}$. Therefore, it is straightforward to check that $\pi_{\rho''}=\pi_{\widetilde{\rho}}$. Moreover, we have
\begin{equation*}
\pi_{\lambda''}=\pi_{\widetilde{\lambda}}\cdot\left(\tau_{l+u_l-1}\cdots\tau_{k+2+|\widetilde{C}|}\right)=\pi_{\lambda}\cdot\pi_{l+u_l}^{k+2+|\widetilde{C}|},
\end{equation*}
yielding what is claimed.
\end{proof}

Having discussed certain properties of permutations associated to contractions, we proceed now to their representations on $\mathscr{H}_1^{\otimes n}$. According to (\ref{tensor}) we, therefore, have for $C\in\mathscr{C}_{n,k}$ and $\psi_n\in\mathscr{H}_1^{\otimes n}$
\begin{eqnarray}\label{PiCDef}
\left(D_n\left(\pi_C\right)\psi_n\right)(\boldsymbol{\theta})&=& S_n^{\pi_C}(\boldsymbol{\theta})\,\psi_n(\boldsymbol{\theta}^{\pi_C})\nonumber\\
&=& S_n^{\pi_C}(\boldsymbol{\theta})\,\psi_n(\theta_1,\stackrel{\widehat{\theta}_{r}}{\dots},\theta_k,\theta_{r_{|C|}},\dots,\theta_{r_1},\theta_{l_1},\dots,\theta_{l_{|C|}},\theta_{k+1},\stackrel{\widehat{\theta}_{l}}{\dots},\theta_n)\nonumber\\
&=& S_n^{\pi_\rho}(\theta_1,\dots,\theta_k)\cdot S_n^{\pi_\lambda}(\theta_{k+1},\dots,\theta_n)\,\psi_n(\boldsymbol{\theta}^{\pi_C}),
\end{eqnarray}
where $\widehat{\theta}_{r/l}$ stands for the omission of the variables $\theta_{r_i}$ and $\theta_{l_i}$, $i=1,\dots,|C|$, respectively. The last line, on the other hand, holds true since the permutations $\pi_\rho$ and $\pi_\lambda$ act on separate sets. Depending on the contraction $C$ at hand, the matrices $S_n^{\pi_{\rho/\lambda}}(\boldsymbol{\theta})\in\mathcal{U}(\mathcal{K}^{\otimes n})$ are besides S-matrices composed of Kronecker deltas as well. Extracting this particular structure, we may express these tensors in the following form
\begin{eqnarray}\label{SMatrix}
S_n^{\pi_{\rho}}(\boldsymbol{\theta})&=&1_{r_1-1}\otimes S^{\rho}(\theta_{r_1},\dots,\theta_k)\otimes 1_{n-k},\\
S_n^{\pi_{\lambda}}(\boldsymbol{\theta})&=&1_{k}\otimes S^{\lambda}(\theta_{k+1},\dots,\theta_{l^C_{\text{max}}})\otimes 1_{n-l^C_{\text{max}}},\\
l^C_{\text{max}}&:=&\max\,\{l_i:\,i=1,\dots,|C|\},
\end{eqnarray}
where $S^{\rho}(\theta_{r_1},\dots,\theta_k)\in\mathcal{U}(\mathcal{K}^{\otimes k-r_1+1})$ and $S^{\lambda}(\theta_{k+1},\dots,\theta_{l^C_{\text{max}}})\in\mathcal{U}(\mathcal{K}^{\otimes l^C_{\text{max}}-k})$ consist of certain products of S-matrices which appear for the permutations $\pi_\rho$ and $\pi_\lambda$ respectively. By means of these objects we define the unitaries
\begin{eqnarray}
S^R(\boldsymbol{\theta})&:=& 1_{n-k+r_1-1}\otimes S^{\rho}(\theta_{r_1},\dots,\theta_k),\\ S^L(\boldsymbol{\theta})&:=& S^{\lambda}(\theta_{k+1},\dots,\theta_{l^C_{\text{max}}})\otimes 1_{n-l^C_{\text{max}}+k},
\end{eqnarray}
which clearly commute. For further purposes we state the following shorthand notation
\begin{equation}
\delta_{a,b}:=\delta(\theta_a-\theta_b),\qquad
\end{equation}
and define, in addition,
\begin{equation}\label{delta}
\delta_C:=\prod_{i=1}^{|C|}\delta(\theta_{l_i}-\theta_{r_i})=\prod_{i=1}^{|C|}\delta_{l_i,r_i}.
\end{equation}
Finally, we introduce entirely contracted matrix elements of $A\in\mathcal{F}(W_R)$ by
\begin{equation}\label{contraction}
\langle A\rangle^{\text{con}}_{n,k}:=\sum_{C\in\mathscr{C}_{n,k}}(-1)^{|C|}\,\delta_C\,\text{Tr}\,_{1,\dots,|C|}^{1,\dots,|C|}\Big[S^L\cdot S^R\cdot \left( 1_{|C|}\otimes\langle A\rangle_C\otimes 1_{|C|}\right)\Big],
\end{equation}
understood in the sense of distributions. Here and the following a dot between tensors is always understood in the sense of (\ref{MN}). For the sake of clarity, we omitted the $\boldsymbol{\theta}$-dependence in (\ref{contraction}) and we shall do so in the following as well, as long as no confusion arises. For $S\in\mathcal{S}_0$ the entirely contracted matrix elements (\ref{contraction}) exist as tempered distributions on $\mathscr{S}(\mathbb{R}^n)\otimes\mathcal{K}^{\otimes n}$. Namely, since the distributions $\delta_C$ and $\langle A\rangle_C$ act on different variables, their product is well-defined. Moreover,  $S\in\mathcal{S}_0$ is smooth and has bounded derivatives on the real line. Hence, $\langle A\rangle^{\text{con}}_{n,k}$ is well-defined.\par
The entirely contracted matrix elements turn out to have useful analytic properties as stated in Lemma \ref{Lemma}. These properties are based on the following observations.
\begin{lemma}\label{LemmaVorbereitung}
Let $C\in\hat{\mathscr{C}}_{n,k}$, i.e. satisfying $k+1\notin \{l_1,\dots,l_{|C|}\}$ and $C'\in\check{\mathscr{C}}_{n,k}$, fulfilling $k+1\in\{l'_1,\dots,l'_{|C'|}\}$, such that $C'=C\cup\{(k+1,r)\}$ with $r\notin\{r_1,\dots,r_{|C|}\}$, then we have
\begin{multline}\label{umrechnung}
\delta_{C'}\,\text{Tr}\,_{1,\dots,|C'|}^{1,\dots,|C'|}\Big[S^{L'}\cdot S^{R'}\cdot \left( 1_{|C'|}\otimes\langle A\rangle_{C'}\otimes 1_{|C'|}\right)\Big]\\
=
\delta_C\,\delta_{r,k+1}\,\text{Tr}\,_{1,\dots,|C|+1}^{1,\dots,|C|+1}\Big[S^L\cdot S^R\cdot S^r\cdot \left( 1_{|C|+1}\otimes\langle A\rangle_{C\cup\{(k+1,r)\}}\otimes 1_{|C|+1}\right)\Big],
\end{multline}
where $S^r(\boldsymbol{\theta}):=1_{n-k+r-v_r-1}\otimes S_r(\boldsymbol{\theta})\otimes 1_{|C|}$ with $S_r\in\mathcal{U}(\mathcal{K}^{\otimes k-|C|-r+v_r+1})$ is defined via
\begin{multline}\label{DefSR}
\left(D_n(\pi_{\rho})D_n\left(\pi_{r-v_r}^{k-|C|}\right)\psi_n\right)(\boldsymbol{\theta})\\
\hspace{-2.3cm}=S_n^{\pi_\rho}(\boldsymbol{\theta})\cdot\left(1_{r-v_r-1}\otimes S_r(\boldsymbol{\theta})\otimes 1_{|C|+n-k}\right)\\
\times\psi_n(\theta_1,\stackrel{\widehat{\theta}_{r_1},\dots,\widehat{\theta}_r,\dots\widehat{\theta}_{r_{|C|}}}{\dots},\theta_k,\theta_r,\theta_{r_{|C|}},\dots,\theta_{r_1},\theta_{k+1},\dots,\theta_n),
\end{multline}
with $\psi_n\in\mathscr{H}_1^{\otimes n}$ and $v_r:=$card$\{r_i\in\{r_1,\dots,r_{|C|}\}:r_i<r\}$.
\end{lemma}
\begin{proof}
Note first that $|C'|=|C|+1$ and $\delta_{C'}=\delta_C\delta_{r,k+1}$ is an immediate consequence of Definition (\ref{delta}). Moreover, it follows from Lemma \ref{LemmaVorbereitung01} that
\begin{equation*}
\pi_{\rho'}=\pi_\rho\cdot\pi_{r-v_r}^{k-v_r}=\pi_\rho\cdot\pi_{r-v_r}^{k-|C|}\cdot\pi_{k-|C|}^{k-v_r},\qquad\pi_{\lambda'}=\pi_\lambda\cdot\pi_{k+1+|C|}^{k+1+v_r}.
\end{equation*}
Corresponding to these permutations, we have the commuting tensors
\begin{equation}\label{einsetztensor}
\begin{aligned}
S^{R'}&=\left( 1_{n-k+r_1-1}\otimes S^{\rho}\right)\cdot\left(1_{n-k+r-v_r-1}\otimes S_r\otimes 1_{|C|}\right)\cdot\left(1_{n-|C|-1}\otimes S'_{r}\otimes 1_{v_r}\right),\\
S^{L'}&=\left(S^{\lambda}\otimes 1_{n-l^C_{\text{max}}+k}\right)\cdot\left(1_{v_r}\otimes S'_l\otimes 1_{n-|C|-1}\right),
\end{aligned}
\end{equation}
with $S'_{l},S'_{r}\in\mathcal{U}(\mathcal{K}^{\otimes |C|+1-v_r})$. Taking a closer look at the expression
\begin{equation*}
\mathcal{E}_{C'}:=\delta_{C'}\,\text{Tr}\,_{1,\dots,|C'|}^{1,\dots,|C'|}\Big[S^{L'}\cdot S^{R'}\cdot \left( 1_{|C'|}\otimes\langle A\rangle_{C'}\otimes 1_{|C'|}\right)\Big],
\end{equation*}
one finds after execution of the trace operation that
\begin{eqnarray*}
\mathcal{E}_{C'}^{\alpha_1\dots\alpha_n}=\delta_C\delta_{r,k+1}\left(S^{L'}\cdot S^{R'}\right)^{\alpha_1\quad\dots\quad\alpha_n}_{\gamma_1\dots\gamma_{|C'|}\beta_1\dots\beta_{n-2|C'|}\overline{\gamma}_{|C'|}\dots\overline{\gamma}_{1}}\langle A\rangle_{C'}^{\beta_1\dots\beta_{n-2|C'|}}.
\end{eqnarray*}
Inserting (\ref{einsetztensor}) into
\begin{equation*}
\left(S^{L'}\cdot S^{R'}\right)^{\alpha_1\quad\dots\quad\alpha_n}_{\gamma_1\dots\gamma_{|C'|}\beta_1\dots\beta_{n-2|C'|}\overline{\gamma}_{|C'|}\dots\overline{\gamma}_{1}},
\end{equation*}
one comes, in particular, across the factors
\begin{eqnarray*}
\left(S'_l\right)_{\gamma_{v_r+1}\dots\gamma_{|C'|}}^{\varepsilon_{v_r+1}\dots\varepsilon_{|C'|}}\left(S_r'\right)_{\overline{\gamma}_{|C'|}\dots\overline{\gamma}_{v_r+1}}^{\eta_1\dots\eta_{|C'|-v_r}}.
\end{eqnarray*}
Since $S'_l$ is obtained from
\begin{multline*}
\left(D_n(\pi_{\lambda'})\psi_n\right)(\theta)=S_n^{\pi_{\lambda}}(\boldsymbol{\theta})\cdot\left(D_n\left(\pi_{k+1+|C|}^{k+1+v_r}\right)\psi_n\right)(\theta_1,\dots,\theta_k,\theta_{l_1},\dots,\theta_{l_{|C|}},\theta_{k+1},\dots,\theta_n),
\end{multline*}
we have, using the notation introduced in (\ref{kurznotation}), that
\begin{equation*}
\left(S'_l(\boldsymbol{\theta})\right)_{\gamma_{v_r+1}\dots\gamma_{|C'|}}^{\varepsilon_{v_r+1}\dots\varepsilon_{|C'|}}=\Bigg[ \prod_{i=1}^{|C|-v_r}S(\theta_{k+1}-\theta_{l_{|C|+1-i}})_{|C|+1-v_r,|C|+1-v_r-i}\Bigg]_{\gamma_{v_r+1}\dots\gamma_{|C'|}}^{\varepsilon_{v_r+1}\dots\varepsilon_{|C'|}}.
\end{equation*}
Analogously, one finds
\begin{eqnarray*}
\left(S_r'(\boldsymbol{\theta})\right)_{\overline{\gamma}_{|C'|}\dots\overline{\gamma}_{v_r+1}}^{\eta_1\dots\eta_{|C'|-v_r}}&=&\Bigg[\prod_{i=1}^{|C|-v_r}S(\theta_{r_{|C|+1-i}}-\theta_{r})_{|C|+1-v_r,i}\Bigg]_{\overline{\gamma}_{|C'|}\dots\overline{\gamma}_{v_r+1}}^{\eta_1\dots\eta_{|C'|-v_r}}\\
&=&\Bigg[\prod_{i=1}^{|C|-v_r}S(\theta_{r_{v_r+i}}-\theta_{r})_{|C|+1-v_r,i}\Bigg]^{\gamma_{v_r+1}\dots\gamma_{|C'|}}_{\overline{\eta}_{|C'|-v_r}\dots\overline{\eta}_1}
\end{eqnarray*}
where in the last line the property (\ref{PCT}) was used.
Hence
\begin{eqnarray*}
\delta_{C'}\left(S'_l(\boldsymbol{\theta})\right)_{\gamma_{v_r+1}\dots\gamma_{|C'|}}^{\varepsilon_{v_r+1}\dots\varepsilon_{|C'|}}\left(S_r'(\boldsymbol{\theta})\right)_{\overline{\gamma}_{|C'|}\dots\overline{\gamma}_{v_r+1}}^{\eta_1\dots\eta_{|C'|-v_r}}=\delta_{C'}\left(1_{|C'|-v_r}\right)_{\overline{\eta}_{|C'|-v_r}\dots\overline{\eta}_1}^{\varepsilon_{v_r+1}\dots\varepsilon_{|C'|}},
\end{eqnarray*}
which yields
\begin{eqnarray*}
\mathcal{E}_{C'}^{\alpha_1\dots\alpha_n}&=&\delta_{C'}\left(S^{L'}\cdot S^{R'}\right)^{\alpha_1\quad\dots\quad\alpha_n}_{\gamma_1\dots\gamma_{|C'|}\beta_1\dots\beta_{n-2|C'|}\overline{\gamma}_{|C'|}\dots\overline{\gamma}_{1}}\langle A\rangle_{C'}^{\beta_1\dots\beta_{n-2|C'|}}\\
&=&\delta_{C'}\left(S^{L}\cdot S^{R}\cdot S^r\right)^{\alpha_1\quad\dots\quad\alpha_n}_{\gamma_1\dots\gamma_{v_r}\varepsilon_{v_r+1}\dots\varepsilon_{|C'|}\beta_1\dots\beta_{n-2|C'|}\overline{\varepsilon}_{|C'|}\dots\overline{\varepsilon}_{v_r+1}\overline{\gamma}_{v_r}\dots\overline{\gamma}_{1}}\\
&&\times\langle A\rangle_{C'}^{\beta_1\dots\beta_{n-2|C'|}}.
\end{eqnarray*}
\end{proof}
An immediate consequence of Lemma \ref{LemmaVorbereitung02} is the following
\begin{corollary}\label{LemmaVorbereitung2}
Let $\widetilde{C}\in\hat{\hat{\mathscr{C}}}_{n,k+1}$, i.e. satisfying $k+1\notin \{r_1,\dots,r_{|\widetilde{C}|}\}$ and $C''\in\check{\check{\mathscr{C}}}_{n,k+1}$, fulfilling $k+1\in\{r''_1,\dots,r''_{|C''|}\}$, such that $C''=\{(l,k+1)\}\cup \widetilde{C}$ with $l\notin\{\widetilde{l}_1,\dots,\widetilde{l}_{|\widetilde{C}|}\}$, then we have
\begin{multline}\label{umrechnung2}
\delta_{C''}\,\text{Tr}\,_{1,\dots,|C''|}^{1,\dots,|C''|}\left[S^{L''}\cdot S^{R''} \cdot\left( 1_{|C''|}\otimes\langle A\rangle_{C''}\otimes 1_{|C''|}\right)\right]\\
= \delta_{\widetilde{C}}\,\delta_{l,k+1}\,\text{Tr}\,_{1,\dots,|\widetilde{C}|+1}^{1,\dots,|\widetilde{C}|+1}\left[S^{\widetilde{L}}\cdot S^l \cdot S^{\widetilde{R}} \cdot\left( 1_{|\widetilde{C}|+1}\otimes\langle A\rangle_{\widetilde{C}\cup\{(l,k+1)\}}\otimes 1_{|\widetilde{C}|+1}\right)\right],
\end{multline}
where $S^l(\boldsymbol{\theta}):=\left(1_{|\widetilde{C}|}\otimes S_l(\boldsymbol{\theta})\otimes 1_{n-l+k+1-u_l}\right)$ with $S_l\in\mathcal{U}(\mathcal{K}^{\otimes l-k-1+u_l-|\widetilde{C}|})$ is defined via
\begin{multline}\label{DefSl}
\left(D_n(\pi_{\widetilde{\lambda}})D_n\left(\pi_{l+u_l}^{k+2+|\widetilde{C}|}\right)\psi_n\right)(\boldsymbol{\theta})\\
=S_n^{\pi_{\widetilde{\lambda}}}(\boldsymbol{\theta})\cdot\left(1_{k+1+|\widetilde{C}|}\otimes S_l(\boldsymbol{\theta})\otimes 1_{n-l-u_l}\right)\\
\times\psi_n(\theta_1,\dots,\theta_{k+1},\theta_{\widetilde{l}_1},\dots,\theta_{\widetilde{l}_{|\widetilde{C}|}},\theta_l,\theta_{k+2},\stackrel{\widehat{\theta}_{\widetilde{l}_1},\dots,\widehat{\theta}_{l},\dots,\widehat{\theta}_{\widetilde{l}_{|C|}}}{\dots},\theta_n),
\end{multline}
with $u_l:=$card$\{\widetilde{l}_i\in\{\widetilde{l}_1,\dots,\widetilde{l}_{|\widetilde{C}|}\}:\widetilde{l}_i>l\}$.
\end{corollary}
\begin{proof}
Since by Lemma \ref{LemmaVorbereitung02} we have $\pi_{\rho''}=\pi_{\widetilde{\rho}}$ and $\pi_{\lambda''}=\pi_{\widetilde{\lambda}}\cdot\pi_{l+u_l}^{k+2+|\widetilde{C}|}$, the statement follows directly.
\end{proof}

In order to prove certain analytic properties of the entirely contracted matrix elements $\langle A\rangle^{\text{con}}_{n,k}$, we want to proceed as in the single particle case and apply, in particular, Lemma \ref{basicLemma}. Fortunately, as shown in Lemma \ref{rewri}, $\langle A\rangle^{\text{con}}_{n,k}$ can be rewritten in a form which is best suited for this discussion. To this end, using similar notation as in (\ref{conmatrixelements}), we introduce the following contracted matrix elements, namely
\begin{multline}\label{commu}
\langle [z_{k+1},A]\rangle_C^{\alpha_{k+1}\,\stackrel{\widehat{l}}{\dots}\,\alpha_n\,\alpha_1\,\stackrel{\widehat{r}}{\dots}\,\alpha_k}(\underline{\theta})\\:=\langle z_{\alpha_{k+2}}^\dagger(\theta_{k+2})\stackrel{\widehat{l}}{\cdots} z_{\alpha_{n}}^\dagger(\theta_{n})\Omega,[z_{\alpha_{k+1}}(\theta_{k+1}),A]z_{\overline{\alpha}_{k}}^\dagger(\theta_{k})\stackrel{\widehat{r}}{\cdots} z_{\overline{\alpha}_{1}}^\dagger(\theta_{1})\Omega \rangle,
\end{multline}
and
\begin{multline}\label{commu2}
\langle [A,z^\dagger_{k+1}]\rangle_{\widetilde{C}}^{\alpha_{k+2}\,\stackrel{\widehat{\widetilde{l}}}{\dots}\,\alpha_n\,\alpha_1\,\stackrel{\widehat{\widetilde{r}}}{\dots}\,\alpha_{k+1}}(\underline{\theta})\\:=\langle z_{\alpha_{k+2}}^\dagger(\theta_{k+2})\stackrel{\widehat{\widetilde{l}}}{\cdots} z_{\alpha_{n}}^\dagger(\theta_{n})\Omega,[A,z^\dagger_{\overline{\alpha}_{k+1}}(\theta_{k+1})]z_{\overline{\alpha}_{k}}^\dagger(\theta_{k})\stackrel{\widehat{\widetilde{r}}}{\cdots} z_{\overline{\alpha}_{1}}^\dagger(\theta_{1})\Omega \rangle,
\end{multline}
with $C\in\hat{\mathscr{C}}_{n,k}$ and $\widetilde{C}\in\hat{\hat{\mathscr{C}}}_{n,k+1}$ as before.
\begin{lemma}\label{rewri}
Consider $\hat{\mathscr{C}}_{n,k}\subset\mathscr{C}_{n,k}$ and $\hat{\hat{\mathscr{C}}}_{n,k+1}\subset\mathscr{C}_{n,k+1}$ as above. Then,
\begin{equation}\label{rewrite}
\langle A\rangle^{\text{con}}_{n,k}=\sum_{C\in\hat{\mathscr{C}}_{n,k}}(-1)^{|C|}\,\delta_C\,\text{Tr}\,_{1,\dots,|C|}^{1,\dots,|C|}\Big[S^L\cdot S^R\cdot \left( 1_{|C|}\otimes\langle [z_{k+1},A]\rangle_C\otimes 1_{|C|}\right)\Big],
\end{equation}
\begin{equation}\label{rewrite2}
\langle A\rangle^{\text{con}}_{n,k+1}=\sum_{\widetilde{C}\in\hat{\hat{\mathscr{C}}}_{n,k+1}}(-1)^{|\widetilde{C}|}\,\delta_{\widetilde{C}}\,\text{Tr}\,_{1,\dots,|\widetilde{C}|}^{1,\dots,|\widetilde{C}|}\Big[S^{\widetilde{L}}\cdot S^{\widetilde{R}}\cdot \left( 1_{|\widetilde{C}|}\otimes\langle [A,z^\dagger_{k+1}]\rangle_{\widetilde{C}}\otimes 1_{|\widetilde{C}|}\right)\Big].
\end{equation}
\end{lemma}
\begin{proof}
Since $\mathscr{C}_{n,k}=\hat{\mathscr{C}}_{n,k}\sqcup\check{\mathscr{C}}_{n,k}$, where $C'\in\check{\mathscr{C}}_{n,k}$ contracts $k+1$, we have
\begin{equation}\label{umf}
\begin{aligned}
\langle A\rangle^{\text{con}}_{n,k}&=\sum_{C\in\hat{\mathscr{C}}_{n,k}}(-1)^{|C|}\,\delta_C\,\text{Tr}\,_{1,\dots,|C|}^{1,\dots,|C|}\Big[S^L\cdot S^R\cdot \left( 1_{|C|}\otimes\langle A\rangle_C\otimes 1_{|C|}\right)\Big]\\
&\hspace{.2cm}+\sum_{C'\in\check{\mathscr{C}}_{n,k}}(-1)^{|C'|}\,\delta_{C'}\,\text{Tr}\,_{1,\dots,|C'|}^{1,\dots,|C'|}\Big[S^{L'}\cdot S^{R'}\cdot \left( 1_{|C'|}\otimes\langle A\rangle_{C'}\otimes 1_{|C'|}\right)\Big],\\
&=\sum_{C\in\hat{\mathscr{C}}_{n,k}}(-1)^{|C|}\,\delta_C\,\text{Tr}\,_{1,\dots,|C|}^{1,\dots,|C|}\Big[S^L\cdot S^R\cdot \left( 1_{|C|}\otimes\langle A\rangle_C\otimes 1_{|C|}\right)\Big]\\
&\hspace{.2cm}-\sum_{C\in\hat{\mathscr{C}}_{n,k}}\sum_{\stackrel{r=1}{r\neq r_j}}^{k}(-1)^{|C|}\delta_C\,\delta_{r,k+1}\\
&\hspace{.3cm}\times\text{Tr}\,_{1,\dots,|C|+1}^{1,\dots,|C|+1}\Big[S^L\cdot S^R\cdot S^r\cdot \left( 1_{|C|+1}\otimes\langle A\rangle_{C\cup\{(k+1,r)\}}\otimes 1_{|C|+1}\right)\Big],
\end{aligned}
\end{equation}
where the last equality is due to Lemma \ref{LemmaVorbereitung} and $r_j\in\{r_1,\dots,r_{|C|}\}$. Moreover, since a contraction $C'\in\check{\mathscr{C}}_{n,k}$ is related to a $C\in\hat{\mathscr{C}}_{n,k}$ by the union $C'=C\cup \{(k+1,r)\}$ with $r\neq r_j$, we have $\sum_{C\in\hat{\mathscr{C}}_{n,k}}\sum_{r=1,r\neq r_j}^{k}=\sum_{C'\in\check{\mathscr{C}}_{n,k}}$ and $|C'|=|C|+1$. On the other hand, one checks that by repeated application of the exchange relations (\ref{exchange}) the following identity holds for $C\in\hat{\mathscr{C}}_{n,k}$, namely
\begin{eqnarray*}
\hspace{-.8cm}&&\hspace{-.8cm}\langle [z_{k+1},A]\rangle_C^{\underline{\alpha}}\\
\hspace{-.4cm}&&\hspace{-.4cm}=\langle A\rangle_C^{\underline{\alpha}}-\sum_{\stackrel{r=1}{r\neq r_i}}^{k}\delta_{k+1,r}\Bigg\{\text{Tr}_{1}^1\,\Bigg[1_{n-|C|-k+r-v_r-1}\otimes \prod_{\stackrel{i=r+1}{i\neq r_j}}^{k}S(\theta_i-\theta_{k+1})_{k-|C|+v_r-r+1,i-r-w_i}\\
&&\hspace{.8cm}\times\left(1_1\otimes\langle A\rangle_{C\cup \{(k+1,r)\}}\otimes 1_1\right)\Bigg]\Bigg\}^{\underline{\alpha}}\nonumber,
\end{eqnarray*}
where $w_i:=$card$\{r_j\in\{r_1,\dots,r_{|C|}\}:r+1\leq r_j<i\}$. For the sake of clarity, we only stated the $\theta$-dependence for the occurring S-matrices. Recall further that $\underline{\alpha}:=(\alpha_1,\dots,\alpha_{n-2|C|})$. With regard to (\ref{DefSR}), the definition of $S_r(\boldsymbol{\theta})$, it follows immediately that
\begin{equation*}
S_r(\boldsymbol{\theta})=\prod_{\stackrel{i=r+1}{i\neq r_j}}^{k}S(\theta_i-\theta_{r})_{k-|C|+v_r-r+1,i-r-w_i}.
\end{equation*}
Hence
\begin{multline*}
\langle A\rangle_C^{\underline{\alpha}}=\langle [z_{k+1},A]\rangle_C^{\underline{\alpha}}\\
+\sum_{\stackrel{r=1}{r\neq r_j}}^{k}\delta_{k+1,r}\left\{\text{Tr}_{1}^1\,\Bigg[\left(1_{n-|C|-k+r-v_r-1}\otimes S_r\right)\cdot\left(1_1\otimes\langle A\rangle_{C\cup \{(k+1,r)\}}\otimes 1_1\right)\Bigg]\right\}^{\underline{\alpha}},
\end{multline*}
which inserted into (\ref{umf}) yields (\ref{rewrite}), since for $M\in\mathcal{B}(\mathcal{K}^{\otimes n},\mathcal{K}^{\otimes t})$ and $\left(1_a\otimes N\otimes 1_b\right)\in\mathcal{B}(\mathcal{K}^{\otimes m},\mathcal{K}^{\otimes n})$, $a,b,m,n,t\in\mathbb{N}$,
\begin{eqnarray}
M\cdot \left[1_a\otimes \text{Tr}^k\,_l\left(N\right)\otimes 1_b\right]&=& M\cdot \left[\text{Tr}^{k+a}\,_{l+b}\left(1_a\otimes N\otimes 1_b\right)\right]\\
&=& \text{Tr}^{k+a}\,_{l+b}\Big[M\cdot \left(1_a\otimes N\otimes 1_b\right)\Big],\qquad 1\leq k,l\leq m.\nonumber
\end{eqnarray}
In order to show (\ref{rewrite2}), we proceed similarly and find due to $\mathscr{C}_{n,k+1}=\hat{\hat{\mathscr{C}}}_{n,k+1}\sqcup\check{\check{\mathscr{C}}}_{n,k+1}$, where $C''\in\check{\check{\mathscr{C}}}_{n,k+1}$ contracts $k+1$, that
\begin{equation}\label{umf2}
\begin{aligned}
\langle A\rangle^{\text{con}}_{n,k+1}&=\sum_{\widetilde{C}\in\hat{\hat{\mathscr{C}}}_{n,k+1}}(-1)^{|\widetilde{C}|}\,\delta_{\widetilde{C}}\,\text{Tr}\,_{1,\dots,|\widetilde{C}|}^{1,\dots,|\widetilde{C}|}\Big[S^{\widetilde{L}}\cdot S^{\widetilde{R}}\cdot \left( 1_{|\widetilde{C}|}\otimes\langle A\rangle_{\widetilde{C}}\otimes 1_{|\widetilde{C}|}\right)\Big]\\
&\hspace{.2cm}+\sum_{C''\in\check{\check{\mathscr{C}}}_{n,k+1}}(-1)^{|C''|}\,\delta_{C''}\,\text{Tr}\,_{1,\dots,|C''|}^{1,\dots,|C''|}\Big[S^{L''}\cdot S^{R''}\cdot \left( 1_{|C''|}\otimes\langle A\rangle_{C''}\otimes 1_{|C''|}\right)\Big],\\
&=\sum_{\widetilde{C}\in\hat{\hat{\mathscr{C}}}_{n,k+1}}(-1)^{|\widetilde{C}|}\,\delta_{\widetilde{C}}\,\text{Tr}\,_{1,\dots,|\widetilde{C}|}^{1,\dots,|\widetilde{C}|}\Big[S^{\widetilde{L}}\cdot S^{\widetilde{R}}\cdot \left( 1_{|\widetilde{C}|}\otimes\langle A\rangle_{\widetilde{C}}\otimes 1_{|\widetilde{C}|}\right)\Big]\\
&\hspace{.2cm}-\sum_{\widetilde{C}\in\hat{\hat{\mathscr{C}}}_{n,k+1}}\sum_{\stackrel{l=k+2}{l\neq \widetilde{l}_j}}^{n}(-1)^{|\widetilde{C}|}\delta_{\widetilde{C}}\,\delta_{l,k+1}\\
&\hspace{.3cm}\times\text{Tr}\,_{1,\dots,|\widetilde{C}|+1}^{1,\dots,|\widetilde{C}|+1}\Big[S^{\widetilde{L}}\cdot S^l\cdot S^{\widetilde{R}}\cdot \left( 1_{|\widetilde{C}|+1}\otimes\langle A\rangle_{\widetilde{C}\cup\{(l,k+1)\}}\otimes 1_{|\widetilde{C}|+1}\right)\Big],
\end{aligned}
\end{equation}
where the last equality follows from Corollary \ref{LemmaVorbereitung2} and $\widetilde{l}_j\in\{\widetilde{l}_1,\dots,\widetilde{l}_{|\widetilde{C}|}\}$. Moreover, since a contraction $C''\in\check{\check{\mathscr{C}}}_{n,k+1}$ is related to a $\widetilde{C}\in\hat{\hat{\mathscr{C}}}_{n,k+1}$ by the union $C''=\{(l,k+1)\}\cup \widetilde{C}$ with $l\neq \widetilde{l}_j$, we have $\sum_{\widetilde{C}\in\hat{\hat{\mathscr{C}}}_{n,k+1}}\sum_{l=k+2,l\neq \widetilde{l}_j}^{n}=\sum_{C''\in\check{\check{\mathscr{C}}}_{n,k+1}}$ and $|C''|=|\widetilde{C}|+1$. Again, by repeated application of the exchange relations (\ref{exchange}), we obtain
\begin{eqnarray*}
\hspace{-.8cm}&&\hspace{-.8cm}\langle [A,z^\dagger_{k+1}]\rangle_{\widetilde{C}}^{\underline{\alpha}}\\
\hspace{-.4cm}&&\hspace{-.4cm}=\langle A \rangle_{\widetilde{C}}^{\underline{\alpha}}-\sum_{\stackrel{l=k+2}{l\neq \widetilde{l}_j}}^{n}\delta_{k+1,l}\,\Bigg\{\text{Tr}_{1}^{1}\,\Bigg[\prod_{\stackrel{i=l-1}{i\neq \widetilde{l}_j}}^{k+2}S(\theta_l-\theta_i)_{l-k-1+u_l-|\widetilde{C}|,i-k-1+u_l-|\widetilde{C}|+t_i}\otimes 1_{n-l-u_l+k+1-|\widetilde{C}|}\\
&&\hspace{1cm}\times\left(1_1\otimes\langle A\rangle_{\widetilde{C}\cup \{(l,k+1)\}}\otimes 1_1\right)\Bigg]\Bigg\}^{\underline{\alpha}}\nonumber,
\end{eqnarray*}
where $t_i:=$card$\{\widetilde{l}_j\in\{\widetilde{l}_1,\dots,\widetilde{l}_{|\widetilde{C}|}\}:i<\widetilde{l}_j\leq l-1\}$. With regard to (\ref{DefSl}), the definition of $S_l(\boldsymbol{\theta})$, we arrive at
\begin{equation*}
S_l(\boldsymbol{\theta})=\prod_{\stackrel{i=l-1}{i\neq \widetilde{l}_j}}^{k+2}S(\theta_l-\theta_i)_{l-k-1+u_l-|\widetilde{C}|,i-k-1+u_l-|\widetilde{C}|+t_i}.
\end{equation*}
Hence
\begin{multline*}
\langle A \rangle_{\widetilde{C}}^{\underline{\alpha}}=\langle [A,z^\dagger_{k+1}]\rangle_{\widetilde{C}}^{\underline{\alpha}}\\
+\sum_{\stackrel{l=k+2}{l\neq \widetilde{l}_j}}^{n}\delta_{k+1,l}\,\Bigg\{\text{Tr}_{1}^{1}\,\Bigg[S_l\otimes 1_{n-l-u_l+k+1-|\widetilde{C}|}\cdot\left(1_1\otimes\langle A\rangle_{\widetilde{C}\cup \{(l,k+1)\}}\otimes 1_1\right)\Bigg]\Bigg\}^{\underline{\alpha}},
\end{multline*}
implying (\ref{rewrite2}).
\end{proof}
Using this result, we now state a certain analyticity property of the completely contracted matrix elements.
\begin{lemma}\label{Lemma}
Let $0\leq k< n$, $S\in\mathcal{S}_0$ and $A\in\mathcal{F}(W_R)$, then $\langle A\rangle^{\rm{con}}_{n,k}(\theta_1,\dots,\theta_n)$ has an analytic continuation in the variable $\theta_{k+1}$ to the strip $S(-\pi,0)$. Moreover, at \mbox{Im$\,\theta_{k+1}=-\pi$} one finds the boundary value
\begin{equation}
\left(\langle A\rangle^{\text{con}}_{n,k}\right)^{\boldsymbol{\alpha}}(\theta_1,\dots,\theta_{k+1}-i\pi,\dots\theta_n)=\left(\langle A\rangle^{\text{con}}_{n,k+1}\right)^{\alpha_2\dots\alpha_n\alpha_1}(\theta_1,\dots,\theta_{k+1},\dots\theta_n).
\end{equation}
\end{lemma}
\begin{proof}
By Lemma \ref{rewri} we have
\begin{equation*}
\langle A\rangle^{\text{con}}_{n,k}=\sum_{C\in\hat{\mathscr{C}}_{n,k}}(-1)^{|C|}\,\delta_C\,\text{Tr}\,_{1,\dots,|C|}^{1,\dots,|C|}\Big[S^L\cdot S^R\cdot \left( 1_{|C|}\otimes\langle [z_{k+1},A]\rangle_C\otimes 1_{|C|}\right)\Big],
\end{equation*}
with $\hat{\mathscr{C}}_{n,k}\subset\mathscr{C}_{n,k}$ the set of all contractions $C\in\mathscr{C}_{n,k}$ for which $k+1\notin\{ l_1,\dots,l_{|C|}\}$. Considering a single term in this formula, which corresponds to a contraction $C\in\hat{\mathscr{C}}_{n,k}$, it is obvious that the delta distributions appearing in $\delta_C$ do not depend on $\theta_{k+1}$ since $C$ does not contract $k+1$. Due to Lemma \ref{basicLemma} the matrix element $\langle [z_{k+1},A]\rangle_C^{\underline{\beta}}$ has analytic continuation in $\theta_{k+1}$ to $S(-\pi,0)$. Its boundary value at Im$(\theta_{k+1})=-\pi$ is obtained by replacing the commutator $[z_{\beta_{1}}(\theta_{k+1}),A]$ by $[A,z^\dagger_{\overline{\beta}_{1}}(\theta_{k+1})]$, yielding $\langle [A,z^\dagger_{k+1}]\rangle_C^{\beta_2\dots\beta_{n-2|C|}\beta_1}$. The variable $\theta_{k+1}$ does further only appear in the product $S^{L}(\boldsymbol{\theta})$ of S-matrices. With regard to Definition (\ref{PiCDef}) this product contains factors of the form $S^{\alpha\beta}_{\gamma\delta}(\theta_{l_i}-\theta_{k+1})$ which can be analytically continued in $\theta_{k+1}$ into the strip $S(-\pi,0)$. Due to crossing symmetry their boundary values are given by $S^{\overline{\gamma}\alpha}_{\delta\overline{\beta}}(\theta_{k+1}-\theta_{l_i})$. In particular, one may decompose $\pi_\lambda$ (\ref{PiContraction}) according to
\begin{eqnarray*}
\pi_\lambda=\prod_{i=1}^{|C|}\tau_{l_i+u_i-1}\cdot\tau_{l_i+u_i-2}\cdots\tau_{k+i}&=&\prod_{i=1}^{|C|}\tau_{l_i+u_i-1}\cdot\tau_{l_i+u_i-2}\cdots\tau_{k+1+i}\cdot\prod_{j=1}^{|C|} \tau_{k+j}\\
&=&\prod_{i=1}^{|C|}\pi_{l_i+u_i}^{k+i+1}\cdot \pi_{k+1}^{k+1+|C|}.
\end{eqnarray*}
Correspondingly, one has
\begin{equation}\label{neueKontraktion}
S^{L}(\boldsymbol{\theta})=\left(1_1\otimes S^{\lambda\backslash\{k+1\}}(\theta_{k+2},\dots,\theta_{l^C_{\text{max}}}) \otimes 1_{n-l^C_{\text{max}}+k}\right)\cdot\left(\prod_{j=1}^{|C|}S(\theta_{l_j}-\theta_{k+1})_{n,j}\right),
\end{equation}
with $S^{\lambda\backslash\{k+1\}}\in\mathcal{U}(\mathcal{K}^{\otimes l^C_{\text{max}}-(k+1)})$. Taking a closer look at the $\theta_{k+1}$-dependent factor in $S^L(\boldsymbol{\theta})$, we have
\begin{eqnarray*}
\prod_{j=1}^{|C|}S(\theta_{l_j}-\theta_{k+1})_{n,j}=\left(\prod_{j=1}^{|C|}S(\theta_{l_j}-\theta_{k+1})_{|C|+1,j}\right)\otimes 1_{n-|C|-1}.
\end{eqnarray*}
Hence, proceeding similarly as in the proof of Lemma \ref{LemmaVorbereitung}, it follows from
\begin{eqnarray*}
\mathcal{E}_{C}^{\alpha_1\dots\alpha_n}(\boldsymbol{\theta})=\delta_C\left(S^{L}(\boldsymbol{\theta})\cdot S^{R}(\boldsymbol{\theta})\right)^{\alpha_1\quad\dots\quad\alpha_n}_{\gamma_1\dots\gamma_{|C|}\beta_1\dots\beta_{n-2|C|}\overline{\gamma}_{|C|}\dots\overline{\gamma}_{1}}\langle [z_{k+1},A]\rangle_{C}^{\beta_1\dots\beta_{n-2|C|}}(\underline{\theta}),
\end{eqnarray*}
that
\begin{eqnarray*}
\left(S^{L}\cdot S^{R}\right)^{\alpha_1\quad\dots\quad\alpha_n}_{\gamma_1\dots\gamma_{|C|}\beta_1\dots\beta_{n-2|C|}\overline{\gamma}_{|C|}\dots\overline{\gamma}_{1}}
\end{eqnarray*}
contains, in particular, the factors
\begin{eqnarray*}
\delta^{\alpha_1}_{\xi_1}\left[\prod_{j=1}^{|C|}S(\theta_{l_j}-\theta_{k+1})_{|C|+1,j}\right]^{\xi_1\dots\xi_{|C|+1}}_{\gamma_1\dots\gamma_{|C|}\beta_1}\left[S^\rho(\theta_{r_1},\dots,\theta_{k})\right]^{\alpha_{n-k+r_1}\dots\alpha_n}_{\beta_{n-|C|-k+r_1}\dots\beta_{n-2|C|}\overline{\gamma}_{|C|}\dots\overline{\gamma}_1}.
\end{eqnarray*}
Then, analytic continuation yields at $\theta_{k+1}-i\pi$
\begin{eqnarray*}
\hspace{-.5cm}&&\hspace{-.5cm}\delta^{\alpha_1}_{\xi_1}\delta_C\left[\prod_{j=1}^{|C|}S(\theta_{l_j}-\theta_{k+1}+i\pi)_{|C|+1,j}\right]^{\xi_1\dots\xi_{|C|+1}}_{\gamma_1\dots\gamma_{|C|}\beta_1}\left[S^\rho(\theta_{r_1},\dots,\theta_{k})\right]^{\alpha_{n-k+r_1}\dots\alpha_n}_{\beta_{n-|C|-k+r_1}\dots\beta_{n-2|C|}\overline{\gamma}_{|C|}\dots\overline{\gamma}_1}\\
&=&\delta^{\alpha_1}_{\xi_1}\delta_C\left[S^\rho(\theta_{r_1},\dots,\theta_{k})\right]^{\alpha_{n-k+r_1}\dots\alpha_n}_{\beta_{n-|C|-k+r_1}\dots\beta_{n-2|C|}\overline{\gamma}_{|C|}\dots\overline{\gamma}_1}\left[\prod_{j=1}^{|C|}S(\theta_{k+1}-\theta_{l_j})_{|C|+1,|C|+1-j}\right]^{\overline{\gamma}_{|C|}\dots\overline{\gamma}_1\xi_1}_{\beta_1\overline{\xi}_{|C|+1}\dots\overline{\xi}_2}\\
&=&\delta_C\left[S^\rho(\theta_{r_1},\dots,\theta_{k})\otimes 1_1\right]^{\alpha_{n-k+r_1}\dots\alpha_n\alpha_1}_{\varepsilon_1\dots\varepsilon_{k-r_1+2}}\\
&&\hspace{.3cm}\times\left[1_{k-r_1-|C|+1}\otimes\prod_{j=1}^{|C|}S(\theta_{k+1}-\theta_{r_j})_{|C|+1,|C|+1-j}\right]^{\varepsilon_1\dots\varepsilon_{k-r_1+2}}_{\beta_{n-|C|-k+r_1}\dots\beta_{n-2|C|}\beta_1\overline{\xi}_{|C|+1}\dots\overline{\xi}_2}
\end{eqnarray*}
due to (\ref{crossing}). The last line, however, resembles $S^{\widetilde{\rho}}$, the tensor which results from a permutation $\pi_{\widetilde{\rho}}$ associated to a contraction $\widetilde{C}\in\hat{\hat{\mathscr{C}}}_{n,k+1}$, i.e. where $k+1$ appears as a noncontracted ``right'' index. More precisely, it follows from (\ref{PiContraction}) that $\pi_{\widetilde{\rho}}$ associated to $\widetilde{C}\in\hat{\hat{\mathscr{C}}}_{n,k+1}$ is given by
\begin{eqnarray*}
\pi_{\widetilde{\rho}}=\prod_{i=1}^{|\widetilde{C}|}\tau_{\widetilde{r}_i-i+1}\cdots\tau_{k+1-i}=\prod_{i=1}^{|\widetilde{C}|}\pi_{\widetilde{r}_i-i+1}^{k-i+1}\cdot \pi_{k+1}^{k+1-|\widetilde{C}|}.
\end{eqnarray*}
However, since $C\in\hat{\mathscr{C}}_{n,k}$ does not contract $k+1$, it can be regarded as an element of $\hat{\hat{\mathscr{C}}}_{n,k+1}$ too. In particular, we have $\delta_C=\delta_{\widetilde{C}}$, $|\widetilde{C}|=|C|$, $r_j=\widetilde{r}_j$ and $l_j=\widetilde{l}_j$, $j=1,\dots,|C|$. Moreover, with regard to (\ref{neueKontraktion}), it is obvious that $S^{\lambda\backslash\{k+1\}}=S^{\widetilde{\lambda}}$. Thus, we arrive at
\begin{multline*}
\mathcal{E}_{C}^{\alpha_1\dots\alpha_n}(\theta_1,\dots,\theta_{k+1}-i\pi,\dots,\theta_n)\\
=\delta_{\widetilde{C}}\left(S^{\widetilde{L}}(\boldsymbol{\theta})\cdot S^{\widetilde{R}}(\boldsymbol{\theta})\right)^{\alpha_2\quad\dots\quad\alpha_n\alpha_1}_{\xi_2\dots\xi_{|\widetilde{C}|+1}\beta_2\dots\beta_{n-2|\widetilde{C}|}\beta_1\overline{\xi}_{|\widetilde{C}|+1}\dots\overline{\xi}_{2}}\langle [A,z^\dagger_{k+1}]\rangle_{\widetilde{C}}^{\beta_2\dots\beta_{n-2|\widetilde{C}|}\beta_1}(\underline{\theta}),
\end{multline*}
with $\langle [A,z^\dagger_{k+1}]\rangle_{\widetilde{C}}$ given by (\ref{commu2}). Hence
\begin{eqnarray*}
\hspace{-.4cm}&&\hspace{-.4cm}\left(\langle A\rangle^{\text{con}}_{n,k}\right)^{\boldsymbol{\alpha}}(\theta_1,\dots,\theta_{k+1}-i\pi,\dots,\theta_n)\\
\hspace{-.4cm}&&\hspace{-.4cm}=\sum_{\widetilde{C}\in\hat{\hat{\mathscr{C}}}_{n,k+1}}(-1)^{|\widetilde{C}|}\,\delta_{\widetilde{C}}\,\left\{\text{Tr}\,_{1,\dots,|\widetilde{C}|}^{1,\dots,|\widetilde{C}|}\Big[S^{\widetilde{L}}(\boldsymbol{\theta})\cdot S^{\widetilde{R}}(\boldsymbol{\theta})\cdot \left( 1_{|\widetilde{C}|}\otimes\langle [A,z^\dagger_{k+1}]\rangle_{\widetilde{C}}(\underline{\theta})\otimes 1_{|\widetilde{C}|}\right)\Big]\right\}^{\alpha_2\dots\alpha_n\alpha_1}\\
\hspace{-.4cm}&&\hspace{-.4cm}=\left(\langle A\rangle^{\text{con}}_{n,k+1}\right)^{\alpha_2\dots\alpha_n\alpha_1}(\theta_1,\dots,\theta_{k+1},\dots,\theta_n),
\end{eqnarray*}
due to Lemma \ref{rewri}.
\end{proof}

Another important property of completely contracted matrix elements (\ref{contraction}) is the following.
\begin{lemma}\label{Lemma2}
Let $f_1,\dots,f_n\in\mathscr{S}(\mathbb{R})\otimes\mathcal{K}$ and $0\leq\lambda\leq\pi$, then
\begin{equation}
\left|\int d^n\boldsymbol{\theta}\left(\overline{\prod_{j=1}^{n}f_j(\theta_j)},\langle A\rangle^{\text{con}}_{n,k}(\theta_1,\dots,\theta_{k+1}-i\lambda,\dots,\theta_n)\right)\right|
\leq 2^{n}\sqrt{n!}\,\|A\|\cdot\prod_{j=1}^{n}\|f_j\|_2\,.
\end{equation}
\end{lemma}
\begin{proof}
We start by estimating $|\langle A\rangle^{\text{con}}_{n,k}(f_1\otimes\cdots\otimes f_n)|$ first. So we have
\begin{eqnarray*}
&&\hspace{-.8cm}\left|\langle A\rangle^{\text{con}}_{n,k}(f_1\otimes\cdots\otimes f_n)\right|\\
&=&\Bigg|\sum_{C\in\mathscr{C}_{n,k}}(-1)^{|C|}\int d^n\boldsymbol{\theta}\,
\delta_C\\&&\hspace{1.5cm}\times\Bigg\{\text{Tr}\,_{1,\dots,|C|}^{1,\dots,|C|}\Big[S^L(\boldsymbol{\theta})\cdot S^R(\boldsymbol{\theta})\cdot\left( 1_{|C|}\otimes\langle A\rangle_C(\underline{\theta})\otimes 1_{|C|}\right)\Big]\Bigg\}^{\alpha_{k+1}\dots\alpha_n\alpha_1\dots\alpha_k}\\&&\hspace{3cm}\times\prod_{i=1}^{k}f^{\alpha_i}_i(\theta_i)\prod_{j=k+1}^{n}f^{\alpha_j}_j(\theta_j)\Bigg|.
\end{eqnarray*}
Taking the notational convention of Definition (\ref{conmatrixelements}) into account, it follows
\begin{eqnarray*}
&&\hspace{-.8cm}\left|\langle A\rangle^{\text{con}}_{n,k}(f_1\otimes\cdots\otimes f_n)\right|\\
&=&\Bigg|\sum_{C\in\mathscr{C}_{n,k}}(-1)^{|C|}\int d^n\boldsymbol{\theta}\,\delta_C\,\\
&&\hspace{-.3cm}\times
\left<\prod_{j=k+1}^{n}\overline{f_j^{\alpha_j}(\theta_j)}\overline{\left(S^\lambda(\boldsymbol{\theta}_\lambda)\otimes 1_{n-l^C_{\text{max}}}\right)^{\alpha_{k+1}\quad\dots\quad\alpha_n}_{\gamma_1\dots\gamma_{|C|}\beta_{k+1}\dots\beta_{n-|C|}}}\,z^\dagger_{\beta_{k+1}}(\theta_{k+1})\stackrel{\widehat{l}}{\cdots}z^\dagger_{\beta_{n-|C|}}(\theta_{n})\Omega\Big|\right.\\
&&\times \left.\Big|A
\prod_{i=1}^{k}f_i^{\alpha_i}(\theta_i)\left( 1_{r_1-1}\otimes S^\rho(\boldsymbol{\theta}_\rho)\right)^{\alpha_1\quad\dots\quad\alpha_k}_{\beta_1\dots\beta_{k-|C|}\overline{\gamma}_{|C|}\dots\overline{\gamma}_{1}}\,z^\dagger_{\overline{\beta}_{k-|C|}}(\theta_k)\stackrel{\widehat{r}}{\cdots}z^\dagger_{\overline{\beta}_{1}}(\theta_1)\Omega
\right>\Bigg|,
\end{eqnarray*}
where $\boldsymbol{\theta}_\lambda:=(\theta_{k+1},\dots\theta_{l^C_{\text{max}}})$ and $\boldsymbol{\theta}_\rho:=(\theta_{r_1},\dots\theta_{k})$. The symbols $\widehat{l}$ and $\widehat{r}$ indicate the omission of $z^\dagger(\theta_{l_i})$ and $z^\dagger(\theta_{r_i})$ respectively with $i=1,\dots,|C|$. Let, moreover, $\boldsymbol{\theta}_{\widehat{\boldsymbol{l}}}:=(\theta_{k+1},\stackrel{\widehat{\theta}_{l_1},\dots,\widehat{\theta}_{l_{|C|}}}{\dots},\theta_n)$, $\boldsymbol{\theta}_{\widehat{\boldsymbol{r}}}:=(\theta_{1},\stackrel{\widehat{\theta}_{r_1},\dots,\widehat{\theta}_{r_{|C|}}}{\dots},\theta_k)$ and $\boldsymbol{\theta}_{\boldsymbol{r}}:=(\theta_{r_1},\dots,\theta_{r_{|C|}})$. Then, after integration over the delta distributions and by means of (\ref{PCT}), we obtain
\begin{eqnarray*}
\hspace{-.8cm}&&\hspace{-.8cm}\left|\langle A\rangle^{\text{con}}_{n,k}(f_1\otimes\cdots\otimes f_n)\right|\\
&&\leq\sum_{C\in\mathscr{C}_{n,k}}\Bigg|\int d^{|C|}\boldsymbol{\theta}_{\boldsymbol{r}}\\
&&\times
\left<\int d^{^{n-k-|C|}}\boldsymbol{\theta}_{\widehat{\boldsymbol{l}}}\,\left(T_l(\boldsymbol{\theta}_{\widehat{\boldsymbol{l}}},\boldsymbol{\theta}_{\boldsymbol{r}})\right)_{\gamma_1\dots\gamma_{|C|}\beta_{k+1}\dots\beta_{n-|C|}}\right. z^\dagger_{\beta_{k+1}}(\theta_{k+1})\stackrel{\widehat{l}}{\cdots}z^\dagger_{\beta_{n-|C|}}(\theta_{n})\Omega\Big|\\
&&\times \left.\Big|A\int d^{^{k-|C|}}\boldsymbol{\theta}_{\widehat{\boldsymbol{r}}}\,\left(T_r(\boldsymbol{\theta}_{\widehat{\boldsymbol{r}}},\boldsymbol{\theta}_{\boldsymbol{r}})\right)^{\gamma_{1}\dots\gamma_{|C|}\overline{\beta}_{k-|C|}\dots\overline{\beta}_1}
\,z^\dagger_{\overline{\beta}_{k-|C|}}(\theta_k)\stackrel{\widehat{r}}{\cdots}z^\dagger_{\overline{\beta}_{1}}(\theta_1)\Omega
\right>\Bigg|,
\end{eqnarray*}
where
\begin{multline*}
\left(T_l(\boldsymbol{\theta}_{\widehat{\boldsymbol{l}}},\boldsymbol{\theta}_{\boldsymbol{r}})\right)_{\gamma_1\dots\gamma_{|C|}\beta_{k+1}\dots\beta_{n-|C|}}\\:=\overline{\prod_{j=k+1}^{n}f_j^{\alpha_j}(\theta_j)\Big|_{\theta_{l_i}=\theta_{r_i}}\left(S^\lambda(\boldsymbol{\theta}_{\widehat{\boldsymbol{l}}},\boldsymbol{\theta}_{\boldsymbol{r}})\otimes 1_{n-l^C_{\text{max}}}\right)^{\alpha_{k+1}\quad\dots\quad\alpha_n}_{\gamma_1\dots\gamma_{|C|}\beta_{k+1}\dots\beta_{n-|C|}}},
\end{multline*}
and
\begin{equation*}
\left(T_r(\boldsymbol{\theta}_{\widehat{\boldsymbol{r}}},\boldsymbol{\theta}_{\boldsymbol{r}})\right)^{\gamma_{1}\dots\gamma_{|C|}\overline{\beta}_{k-|C|}\dots\overline{\beta}_1}:=\prod_{i=1}^{k}f_i^{\alpha_i}(\theta_i)\left(  \widetilde{S}^\rho(\boldsymbol{\theta}_\rho)\otimes 1_{r_1-1}\right)_{\overline{\alpha}_k\quad\dots\quad\overline{\alpha}_1}^{\gamma_{1}\dots\gamma_{|C|}\overline{\beta}_{k-|C|}\dots\overline{\beta}_1}
\end{equation*}
The tensor $\widetilde{S}^\rho$ denotes the PCT transformed version of $S^\rho$, in accordance with (\ref{PCT}). Due to the particle number bounds (\ref{numberBounds}) we further find
\begin{eqnarray*}
&&\hspace{-.8cm}\left|\langle A\rangle^{\text{con}}_{n,k}(f_1\otimes\cdots\otimes f_n)\right|\\
&\leq&\sum_{C\in\mathscr{C}_{n,k}}\|A\|\sqrt{(n-k-|C|)!}\sqrt{(k-|C|)!}\nonumber\\
&&\hspace{-.3cm}\times\sum_{\gamma_1,\dots,\gamma_{|C|}}\int d^{|C|}\boldsymbol{\theta}_{\boldsymbol{r}} \left\|(\boldsymbol{\theta}_{\widehat{\boldsymbol{l}}},\beta_{k+1},\dots\beta_{n-|C|})\mapsto \left(T_l(\boldsymbol{\theta}_{\widehat{\boldsymbol{l}}},\boldsymbol{\theta}_{\boldsymbol{r}})\right)_{\gamma_1\dots\gamma_{|C|}\beta_{k+1}\dots\beta_{n-|C|}}\right\|_{\mathscr{H}_{n-k-|C|}}\nonumber\\
&&\hspace{-.3cm}\times
\left\|(\boldsymbol{\theta}_{\widehat{\boldsymbol{r}}},\beta_{1},\dots,\beta_{k-|C|})\mapsto 
\left(T_r(\boldsymbol{\theta}_{\widehat{\boldsymbol{r}}},\boldsymbol{\theta}_{\boldsymbol{r}})\right)^{\gamma_{1}\dots\gamma_{|C|}\overline{\beta}_{k-|C|}\dots\overline{\beta}_1}\right\|_{\mathscr{H}_{k-|C|}}\nonumber\\
&\leq&\|A\|\,\prod_{i=1}^{n}\|f_i\|_2\sum_{C\in\mathscr{C}_{n,k}}\sqrt{(n-k-|C|)!}\sqrt{(k-|C|)!}\nonumber,
\end{eqnarray*}
where the last relation follows from the Cauchy-Schwarz inequality and the unitary of the S-matrix, namely
\begin{eqnarray*}
&&\left[\sum_{\gamma_1,\dots,\gamma_{|C|}}\int d^{|C|}\boldsymbol{\theta}_{\boldsymbol{r}}\left\|(\boldsymbol{\theta}_{\widehat{\boldsymbol{r}}},\beta_{1},\dots,\beta_{k-|C|})\mapsto 
\left(T_r(\boldsymbol{\theta}_{\widehat{\boldsymbol{r}}},\boldsymbol{\theta}_{\boldsymbol{r}})\right)^{\gamma_{1}\dots\gamma_{|C|}\overline{\beta}_{k-|C|}\dots\overline{\beta}_1}\right\|^2_{\mathscr{H}_{k-|C|}}\right]^{1/2}\\
&&=\left[\int d^{|C|}\boldsymbol{\theta}_{\boldsymbol{r}}\left(\int d^{k-|C|}\boldsymbol{\theta}_{\widehat{\boldsymbol{r}}}\sum_{\gamma_1,\dots,\gamma_{|C|}}\sum_{\beta_{1},\dots,\beta_{k-|C|}}\left|\left(T_r(\boldsymbol{\theta}_{\widehat{\boldsymbol{r}}},\boldsymbol{\theta}_{\boldsymbol{r}})\right)^{\gamma_{1}\dots\gamma_{|C|}\overline{\beta}_{k-|C|}\dots\overline{\beta}_1}
\right|^2\right)\right]^{1/2}\\
&&= \left(\int d^{|C|}\boldsymbol{\theta}_{\boldsymbol{r}}\int d^{k-|C|}\boldsymbol{\theta}_{\widehat{\boldsymbol{r}}}\left\|T_r(\boldsymbol{\theta}_{\widehat{\boldsymbol{r}}},\boldsymbol{\theta}_{\boldsymbol{r}})\right\|^2_{\mathcal{K}^{\otimes k}}\right)^{1/2}\\
&&= \left(\int d^{|C|}\boldsymbol{\theta}_{\boldsymbol{r}}\int d^{k-|C|}\boldsymbol{\theta}_{\widehat{\boldsymbol{r}}}\left\|f_1(\theta_1)\otimes \cdots\otimes f_k(\theta_k) \right\|^2_{\mathcal{K}^{\otimes k}}\right)^{1/2}=\prod_{i=1}^{k}\|f_i\|_2,
\end{eqnarray*}
and, in the same manner,
\begin{multline*}
\left[\sum_{\gamma_1,\dots,\gamma_{|C|}}\int d^{|C|}\boldsymbol{\theta}_{\boldsymbol{r}}\left\|(\boldsymbol{\theta}_{\widehat{\boldsymbol{l}}},\beta_{k+1},\dots,\beta_{n-|C|})\mapsto 
\left(T_l(\boldsymbol{\theta}_{\widehat{\boldsymbol{l}}},\boldsymbol{\theta}_{\boldsymbol{r}})\right)_{\gamma_1\dots\gamma_{|C|}\beta_{k+1}\dots\beta_{n-|C|}}\right\|^2_{\mathscr{H}_{n-k-|C|}}\right]^{1/2}\\
= \prod_{i=k+1}^{n}\|f_i\|_2.
\end{multline*}
The remaining part is to estimate the sum over all contractions. This was already done in \cite{DocL}. For the convenience of the reader, we carry out a detailed derivation. To this end, consider first contractions $C\in\mathscr{C}_{n,k}$ with length $|C|$. Due to the fact that to each such $C$ one associates the $|C|$-element sets $\{r_1,\dots,r_{|C|}\}\subset\{1,\dots,k\}$ and $\{l_1,\dots,l_{|C|}\}\subset\{k+1,\dots,n\}$, and a permutation of $\{1,\dots,|C|\}$, there are $|C|!$\scriptsize{$
\begin{pmatrix}
k\\
|C|
\end{pmatrix}
\begin{pmatrix}
n-k\\
|C|
\end{pmatrix}$} \normalsize {contractions} with this length. Furthermore, since $|C|\leq\min\{k,n-k\}$ and $u!\,v!\leq (u+v)!$ for $u,v\in\mathbb{N}$, we have
\begin{equation*}
\begin{aligned}
\sum_{C\in\mathscr{C}_{n,k}}\sqrt{(n-k-|C|)!(k-|C|)!}&=\hspace{-.3cm}\sum_{|C|=0}^{\min\{k,n-k\}}\hspace{-.3cm}\sqrt{(n-k-|C|)!(k-|C|)!}|C|!
\begin{pmatrix}
k\\
|C|
\end{pmatrix}
\begin{pmatrix}
n-k\\
|C|
\end{pmatrix}\\
&\leq\sqrt{n!}\sum_{|C|=0}^{\min\{k,n-k\}}\begin{pmatrix}
k\\
|C|
\end{pmatrix}
\begin{pmatrix}
n-k\\
|C|
\end{pmatrix}\\
&\leq\sqrt{n!}\sum_{|C|=0}^{k}\sum_{N=0}^{n-k}\begin{pmatrix}
k\\
|C|
\end{pmatrix}
\begin{pmatrix}
n-k\\
N
\end{pmatrix}\\
&=\sqrt{n!}\,2^k\cdot 2^{n-k}=\sqrt{n!}\,2^n.
\end{aligned}
\end{equation*}
Hence, we arrive at
\begin{equation}\label{estimate3}
\left|\langle A\rangle^{\text{con}}_{n,k}(f_1\otimes\cdots\otimes f_n)\right|
\leq 2^{n}\sqrt{n!}\,\|A\|\cdot\prod_{i=1}^{n}\|f_i\|_2\,.
\end{equation}
The bound (\ref{estimate3}) implies, in particular, that $$h^{\alpha_{k+1}}(\theta_{k+1}):=\int\left(\langle A\rangle^{\text{con}}_{n,k}\right)^{\boldsymbol{\alpha}}(\boldsymbol{\theta})\prod_{\stackrel{j=1}{j\neq k+1}}^{n}f_j^{\alpha_j}(\theta_j)d\theta_j$$ is square-integrable. On the other hand, in view of (\ref{rewrite}), the boundedness of $S$ on $\overline{S(0,\pi)}$ and the bounds found in Lemma \ref{basicLemma}, also $\theta_{k+1}\mapsto h_{-\lambda}^{\alpha_{k+1}}(\theta_{k+1})=h^{\alpha_{k+1}}(\theta_{k+1}-i\lambda)$ is in $L^2(\mathbb{R})$ for any $0\leq \lambda\leq \pi$ and, due to Lemma \ref{Lemma}, even analytic on $S(-\pi,0)$. By application of the three lines theorem, it follows that the bound (\ref{estimate3}) also holds for $|\int\,h^{\alpha_{k+1}}(\theta_{k+1}-i\lambda)f_{k+1}^{\alpha_{k+1}}(\theta_{k+1})d\theta_{k+1}|$.
\end{proof}

Lemmata \ref{Lemma} and \ref{Lemma2} yield important properties of the actual objects of interest, namely the functions $(A\Omega)_n^{\boldsymbol{\alpha}}$. They are connected to contracted matrix elements by
\begin{equation}\label{zusammenhang}
\begin{aligned}
\left(\langle A\rangle^{\text{con}}_{n,0}\right)^{\boldsymbol{\alpha}}(\boldsymbol{\theta})&=\langle z_{\alpha_1
}^\dagger(\theta_1)\cdots z_{\alpha_n
}^\dagger(\theta_n)\Omega,A\Omega\rangle=\sqrt{n!}\,(A\Omega)_n^{\boldsymbol{\alpha}}(\boldsymbol{\theta})\\
\left(\langle A\rangle^{\text{con}}_{n,n}\right)^{\boldsymbol{\alpha}}(\boldsymbol{\theta})&=\langle \Omega,Az_{\overline{\alpha}_n
}^\dagger(\theta_n)\cdots z_{\overline{\alpha}_1
}^\dagger(\theta_1)\Omega\rangle\\
&=\sqrt{n!}\,\overline{(A^*\Omega)_n^{\overline{\alpha_n}\dots\overline{\alpha_1}}(\theta_n,\dots,\theta_1)}=\sqrt{n!}\,(JA^*\Omega)_n^{\boldsymbol{\alpha}}(\boldsymbol{\theta}),
\end{aligned}
\end{equation}
since $\mathscr{C}_{n,0}=\mathscr{C}_{n,n}=\emptyset$. Obviously by reapplication of Lemma \ref{Lemma}, $(JA^*\Omega)_n^{\boldsymbol{\alpha}}(\boldsymbol{\theta})=(\Delta^{1/2}A\Omega)_n^{\boldsymbol{\alpha}}(\boldsymbol{\theta})$ is the boundary value which results from analytic continuation of $(A\Omega)_n^{\boldsymbol{\alpha}}(\boldsymbol{\theta})$ from $\mathbb{R}^n$ to $\mathbb{R}^n-i(\pi,\dots,\pi)$ along a certain path. Moreover, this boundary value is in agreement with modular theory which in fact implies analyticity of $(A\Omega)_n^{\boldsymbol{\alpha}}$ in the center of rapidity $(\theta_1+\cdots+\theta_n)/n$ in the strip $S(-\pi,0)$ due to the strong analyticity $\zeta\mapsto\Delta^{i\zeta}A\Omega$ in $S(-\tfrac{1}{2},0)$. However, it is possible to go beyond this analyticity property. Namely, by means of Lemma \ref{Lemma} one shows, as stated in Corollary \ref{Analyticity}, that $(A\Omega)_n^{\boldsymbol{\alpha}}(\boldsymbol{\theta})$ has even an analytic continuation into a certain tube domain, defined by
\begin{equation}
\mathcal{T}_n:=\mathbb{R}^n-i\mathcal{G}_n,\qquad \mathcal{G}_n:=\{\boldsymbol{\lambda}\in\mathbb{R}^n:\pi>\lambda_1>\lambda_2>\dots>\lambda_n>0\}.
\end{equation}
To simplify the notation, we shall denote both the distribution $(A\Omega)_n^{\boldsymbol{\alpha}}$ and the corresponding holomorphically continued function by the same symbol. 
\begin{corollary}\label{Analyticity}
Consider $A\in\mathcal{F}(W_R)$, then $(A\Omega)_n^{\boldsymbol{\alpha}}(\boldsymbol{\theta})$ can be analytically continued into the tube $\mathcal{T}_n$. Further, the distributional boundary value of this analytic function, found in the limit $\boldsymbol{\lambda}\rightarrow 0$ in $\mathcal{G}_n$, coincides with $(A\Omega)_n^{\boldsymbol{\alpha}}(\boldsymbol{\theta})$.
\end{corollary}
\begin{proof}Application of Lemma \ref{AllgemeinAnalyt} proves the claim. \end{proof}

Lemma \ref{Lemma2}, on the other hand, leads to the following result.
\begin{corollary}\label{boundCoro}
For $\boldsymbol{\theta}\in\mathbb{R}^n$ and $\boldsymbol{\lambda}\in\mathcal{G}_n$ we have
\begin{equation}\label{bound2}
|(A\Omega)_n^{\boldsymbol{\alpha}}(\boldsymbol{\theta}-i\boldsymbol{\lambda})|\leq\left(\frac{4}{\pi d(\boldsymbol{\lambda})}\right)^{n/2}\cdot\|A\|,
\end{equation} with
\begin{equation}\label{distance}
d(\boldsymbol{\lambda}):=\text{min}\{\pi-\lambda_1,\tfrac{1}{2}(\lambda_1-\lambda_2),\dots,\tfrac{1}{2}(\lambda_{n-1}-\lambda_n),\lambda_n\}.
\end{equation}
\end{corollary}
\begin{proof} This assertion can be proven along the same lines as in the scalar case \cite[Corollary 5.2.6 b)]{DocL}. Only minor changes are necessary. For the convenience of the reader we give a full proof.\par
Consider $\boldsymbol{f}:=f_1\otimes\cdots\otimes f_n$ with $f_i\in\mathscr{S}(\mathbb{R})\otimes\mathcal{K}$, $i=1,\dots,n$, then Lemma \ref{Lemma2} implies that at points $\boldsymbol{\theta}-i\boldsymbol{\lambda}\in\overline{\mathcal{T}}_n$ with $\boldsymbol{\lambda}=(\pi,\dots,\pi,\lambda_{k+1},0,\dots,0)$, $0\leq\lambda_{k+1}\leq\pi$, we have the following bound
\begin{equation}\label{L2-schranke}
\Big|\sum_{\boldsymbol{\alpha}}\big((A\Omega)_n^{\boldsymbol{\alpha}}\ast\boldsymbol{f}^{\boldsymbol{\alpha}}\big)\left(\boldsymbol{\theta}-i\boldsymbol{\lambda}\right)\Big|\leq 2^n\|A\|\prod_{j=1}^{n}\|f_j\|_2.
\end{equation}
It can be shown that this bound holds for arbitrary $\boldsymbol{\lambda}\in\overline{\mathcal{T}}_n$ \cite{Lech05} and extends to $f_i\in L^2(\mathbb{R})\otimes\mathcal{K}$ by continuity.\par
In order to proceed from (\ref{L2-schranke}) to pointwise bounds, we observe that the polydisc $D(z_1,r)\times\cdots\times D(z_n,r)$, with center $\boldsymbol{\zeta}=\boldsymbol{\theta}-i\boldsymbol{\lambda}\in\mathbb{C}^n$ and polyradius $\boldsymbol{r}=(r,\dots,r)$, is contained in the tube $\mathcal{T}_n$ if
\begin{equation*}
\pi>\lambda_1+r,\quad \lambda_n-r>0,\quad \lambda_k-r>\lambda_{k+1}+r,\quad  k=1,\dots,n-1.
\end{equation*}
Hence, $r<d(\boldsymbol{\lambda})$ (\ref{distance}) has to be chosen. By a variation on the mean value property of analytic functions we then have
\begin{eqnarray}\label{int2}
(A\Omega)_n^{\boldsymbol{\alpha}}(\boldsymbol{\zeta})&=&\frac{1}{(\pi r^2)^n}\int\limits_{D(\zeta_1,r)}d\theta_1'd\lambda_1'\cdots\int\limits_{D(\zeta_n,r)}d\theta_n'd\lambda_n'\,(A\Omega)^{\boldsymbol{\alpha}}_{n}(\boldsymbol{\theta}'+i\boldsymbol{\lambda}')\nonumber\\
&=&\frac{1}{(\pi r^2)^n}\prod_{k=1}^{n}\,\int\limits_{-r}^{r}d\lambda_k'\int\limits_{-r(\lambda_k')}^{r(\lambda_k')}d\theta_k'\,(A\Omega)_n^{\boldsymbol{\alpha}}(\boldsymbol{\theta}'+\boldsymbol{\theta}+i\boldsymbol{\lambda}'-i\boldsymbol{\lambda})\\
&=&\frac{1}{(\pi r^2)^n}\int\limits_{[-r,r]^{\times n}}d^n\boldsymbol{\lambda}'\,\Big((A\Omega)_n^{\boldsymbol{\alpha}}\ast(\chi_{r(\lambda_1')}\otimes\cdots\otimes\chi_{r(\lambda_n')})\Big)(\boldsymbol{\theta}-i\boldsymbol{\lambda}+i\boldsymbol{\lambda}'),\nonumber
\end{eqnarray}
with $r(\lambda_k'):=\sqrt{r^2-(\lambda_k')^2}$ and where $\chi_{r(\lambda_k')}$ is the characteristic function of $[-r(\lambda_k'),r(\lambda_k')]$. Define now $\boldsymbol{\chi}\in L^2(\mathbb{R}^n)\otimes\mathcal{K}^{\otimes n}$ such that $L^2(\mathbb{R}^n)\ni\boldsymbol{\chi}^{\boldsymbol{\alpha}}:=\delta^{\alpha_1}_{\beta_1}\cdots\delta^{\alpha_n}_{\beta_n}\boldsymbol{\chi}^{\boldsymbol{\beta}}$ (no summation), that is, only one component, namely $\boldsymbol{\chi}^{\boldsymbol{\alpha}}:=\chi_{r(\lambda_1')}\otimes\cdots\otimes\chi_{r(\lambda_n')}$, is different from zero. Hence, (\ref{int2}) becomes
\begin{multline}\label{projchi}
(A\Omega)_n^{\boldsymbol{\alpha}}(\boldsymbol{\theta}-i\boldsymbol{\lambda})
=\frac{1}{(\pi r^2)^n}\int\limits_{[-r,r]^{\times n}}d^n\boldsymbol{\lambda}'\,\sum_{\boldsymbol{\alpha}}\Big((A\Omega)_n^{\boldsymbol{\alpha}}\ast\boldsymbol{\chi}^{\boldsymbol{\alpha}}\Big)(\boldsymbol{\theta}-i\boldsymbol{\lambda}+i\boldsymbol{\lambda}').
\end{multline}
Then, with regard to (\ref{L2-schranke}), we obtain
\begin{equation}\label{bound1}
\left|(A\Omega)_n^{\boldsymbol{\alpha}}(\boldsymbol{\theta}-i\boldsymbol{\lambda})\right|\leq\frac{(2r)^n}{(\pi r^2)^n}2^n\|A\|\,\prod_{j=1}^{n}\sup_{r(\lambda_j')\leq r}\|\chi_{r(\lambda_j')}\|_{L^2(\mathbb{R})}=\left(\frac{32}{\pi^2 r}\right)^{n/2}\cdot\|A\|.
\end{equation}
By means of Cauchy's integral formula this bound can be slightly improved. Due to the boundedness of $(A\Omega)_n^{\boldsymbol{\alpha}}(\boldsymbol{\theta}-i\boldsymbol{\lambda})$ in $\boldsymbol{\theta}\in\mathbb{R}^n$ for fixed $\boldsymbol{\lambda}\in\mathcal{G}_n$ by (\ref{bound1}), namely, the corresponding integration contour $\partial D(z_1,r)\times\cdots\times\partial D(z_k,r)\times\cdots\times\partial D(z_n,r)$ can be deformed to \mbox{$\cdots\times\left((\mathbb{R}-i\lambda_k-r)\cup(\mathbb{R}-i\lambda_k+r)\right)\times\cdots$}. Hence, we have
\begin{equation}
(A\Omega)_n^{\boldsymbol{\alpha}}(\boldsymbol{\theta}-i\boldsymbol{\lambda})=\frac{1}{(2\pi i)^n}\sum_{\boldsymbol{\varepsilon}}\varepsilon_1\cdots\varepsilon_n\int\limits_{\mathbb{R}^n}d^n\boldsymbol{\theta}'\frac{(A\Omega)_n^{\boldsymbol{\alpha}}(\boldsymbol{\theta}'-i\boldsymbol{\lambda}-ir\boldsymbol{\varepsilon})}{\prod_{k=1}^{n}(\theta_k'-\theta_k-ir\varepsilon_k)},
\end{equation}
where the sum runs over $\varepsilon_k=\pm 1$. By the same methods used in the derivation of (\ref{projchi}) and with regard to (\ref{L2-schranke}), we arrive at
\begin{equation}
|(A\Omega)_n^{\boldsymbol{\alpha}}(\boldsymbol{\theta}-i\boldsymbol{\lambda})|\leq \frac{2^n}{(\pi r)^{n/2}}\|A\|,
\end{equation}
which yields the claimed bound for $r\rightarrow d(\boldsymbol{\lambda})$.
\end{proof}
So far we have shown that $(A\Omega)_n^{\boldsymbol{\alpha}}$ extends to a bounded and analytic function on the tube $\mathcal{T}_n$. Recalling the concrete action of the maps $\Xi_n(s)$ (\ref{concrete}), it is obvious that the points $\boldsymbol{\theta}-i(\tfrac{\pi}{2},\dots,\tfrac{\pi}{2})$, $\boldsymbol{\theta}\in\mathbb{R}^n$, are of particular interest. Unfortunately, the previous investigation does not yield the necessary information about the analyticity nor the boundedness of $(A\Omega)_n^{\boldsymbol{\alpha}}$ at these points, in order to prove the modular nuclearity condition analogously to the one particle case, see Lemma \ref{oneparticle}. This is due to the fact that these points are only contained in the boundary of $\mathcal{T}_n$. Therefore, an enlargement of this domain of analyticity to one involving $\boldsymbol{\theta}-i(\tfrac{\pi}{2},\dots,\tfrac{\pi}{2})$ is desirable. Indeed, this can be accomplished by taking into account that the S-matrix $S\in\mathcal{S}_0$ under consideration admits a continuation to the enlarged strip $S(-\kappa(S),\pi+\kappa(S))$, with $\kappa(S)>0$ as in Definition \ref{regularS}. The corresponding result on the extension of the analyticity domain is stated below in Proposition \ref{enlagement}. The expanded regions are defined by means of the following sets. Namely, for $a>0$ we introduce
\begin{equation}
\mathcal{B}_n(a):=\{\boldsymbol{\lambda}\in\mathbb{R}^n:0<\lambda_1,\dots,\lambda_n<\pi,\,-a<\lambda_k-\lambda_l<a,\,1\leq l<k\leq n\},
\end{equation}
as well as the cube $\mathcal{C}_n(a)+\boldsymbol{\lambda}_{\pi/2}\subset-\mathcal{B}_n(a)$,
\begin{equation}\label{l}
\boldsymbol{\lambda}_{\pi/2}:=(-\tfrac{\pi}{2},\dots,-\tfrac{\pi}{2}),\qquad \mathcal{C}_n(a):=(-\tfrac{a}{2},\tfrac{a}{2})^{\times n},
\end{equation}
and denote the tube based on this cube by
\begin{equation}\label{tube}
\mathcal{T}_n(a):=\mathbb{R}^n+i\left(\boldsymbol{\lambda}_{\pi/2}+\mathcal{C}_n(a)\right).
\end{equation}
\begin{figure}[h]
\begin{minipage}{3cm}
\begin{tikzpicture}[scale=1]
               \begin{scope}%[transparent]    %current
              \fill[gray!30!white] (0,0)--(-2.5,0)--(-2.5,-2.5) -- (0,0);
              
               \end{scope} 
        \begin{scope}%[transparent]    %current
                \draw[fill] (1.2,.2) node{Im$\,\zeta_1$};

\draw[fill] (.5,1) node{Im$\,\zeta_2$};

\draw[fill] (-2.5,.2) node{\footnotesize{$(-\pi,0)$}};
\draw[fill] (-2.8,-2.8) node{\footnotesize{$(-\pi,-\pi)$}};
\draw[fill] (.2,-2.8) node{\footnotesize{$(0,-\pi)$}};
\draw[fill] (.2,.2) node{\footnotesize{$0$}};
\fill[black] (-2.5,-2.5) circle (.3ex);
\fill[black] (-2.5,0) circle (.3ex);
\fill[black] (0,-2.5) circle (.3ex);
\fill[black] (0,0) circle (.3ex);
        \begin{scope}[->]
            \draw (0,-2.5) -- (0,1) node[anchor=north] {};
            \draw (-2.5,0) -- (1,0) node[anchor=east] {};

        \end{scope}
    \begin{scope}  
                \draw[fill] (-.9,-1.6) node{\scriptsize {$\boldsymbol{\lambda}_{\pi/2}$}} ;
       \draw[style=densely dashed] (-2.5,-2.5)--(0,0);
       \draw (-2.5,0) -- (-2.5,-2.5);
      \draw (-2.5,-2.5) -- (0,-2.5);
      \fill[black] (-1.25,-1.25) circle (.3ex);
        \end{scope}
        \end{scope}
\end{tikzpicture}
\end{minipage}\hspace{2cm}
\begin{minipage}{3cm}
\begin{tikzpicture}[scale=1]
               \begin{scope}%[transparent]    %current
              \fill[gray!30!white] (0,0)--(-.85,0)--(-2.5,-2.5+.85) -- (-2.5,-2.5)--(0,0);
              \fill[gray!50!white] (0,0)--(-2.5,-2.5)--(-1.65,-2.5) -- (0,-.85);
               \end{scope} 
        \begin{scope}%[transparent]    %current
                \draw[fill] (1.2,.2) node{Im$\,\zeta_1$};
\draw[fill] (.5,1) node{Im$\,\zeta_2$};
\draw[fill] (.3,-.45) node{$a$};
\draw[fill] (-2.5,.2) node{\footnotesize{$(-\pi,0)$}};
\draw[fill] (-2.8,-2.8) node{\footnotesize{$(-\pi,-\pi)$}};
\draw[fill] (.2,-2.8) node{\footnotesize{$(0,-\pi)$}};
\draw[fill] (.2,.2) node{\footnotesize{$0$}};
\fill[black] (-2.5,-2.5) circle (.3ex);
\fill[black] (-2.5,0) circle (.3ex);
\fill[black] (0,-2.5) circle (.3ex);
\fill[black] (0,0) circle (.3ex);
        \begin{scope}[->]
            \draw (0,-2.5) -- (0,1) node[anchor=north] {};
            \draw (-2.5,0) -- (1,0) node[anchor=east] {};
            \draw (0.1,-.05)--(.1,-.85);
               \draw (.1,-.85)--(0.1,-.05);
        \end{scope}
    \begin{scope}  
    \draw[style=densely dashed] (-0.85,0) --(-2.5,-1.65);
       
    \draw[style=densely dashed] (0,-0.85) --(-1.65,-2.5);
       \draw (-2.5,0) -- (-2.5,-2.5);
      \draw (-2.5,-2.5) -- (0,-2.5);
      \fill[black] (-1.25,-1.25) circle (.3ex);
        \end{scope}
        \end{scope}
\end{tikzpicture}
\end{minipage}\hspace{2cm}
\begin{minipage}{3cm}
\begin{tikzpicture}[scale=1]
               \begin{scope}%[transparent]    %current
              
               \end{scope} 
        \begin{scope}%[transparent]    %current
                \draw[fill] (1.2,.2) node{Im$\,\zeta_1$};
\draw[style=dashed] (0,0) -- (-2.5,-2.5);
  \fill[gray!80!white] (-.85,-.85)--(-1.65,-.85)--(-1.65,-.85) -- (-1.65,-1.65)--(-.85,-1.65);
\draw[fill] (.5,1) node{Im$\,\zeta_2$};
\draw[fill] (.3,-.45) node{$a$};
\draw[fill] (-2.5,.2) node{\footnotesize{$(-\pi,0)$}};
\draw[fill] (-2.8,-2.8) node{\footnotesize{$(-\pi,-\pi)$}};
\draw[fill] (.2,-2.8) node{\footnotesize{$(0,-\pi)$}};
\draw[fill] (.2,.2) node{\footnotesize{$0$}};
\fill[black] (-2.5,-2.5) circle (.3ex);
\fill[black] (-2.5,0) circle (.3ex);
\fill[black] (0,-2.5) circle (.3ex);
\fill[black] (0,0) circle (.3ex);
        \begin{scope}[->]
            \draw (0,-2.5) -- (0,1) node[anchor=north] {};
            \draw (-2.5,0) -- (1,0) node[anchor=east] {};
            \draw (0.1,-.05)--(.1,-.85);
               \draw (.1,-.85)--(0.1,-.05);
        \end{scope}
    \begin{scope}  
       \draw[style=densely dashed] (-.85,-.85)--(-1.65,-.85);
       \draw[style=densely dashed] (-1.65,-.85) -- (-1.65,-1.65);
       \draw[style=densely dashed] (-1.65,-1.65) --(-.85,-1.65);
  \draw[style=densely dashed] (-.85,-1.65) --(-.85,-.85);
    \draw[style=densely dashed] (0,-0.85) --(-1.65,-2.5);
       \draw (-2.5,0) -- (-2.5,-2.5);
      \draw (-2.5,-2.5) -- (0,-2.5);
      \fill[black] (-1.25,-1.25) circle (.3ex);
        \end{scope}
        \end{scope}
\end{tikzpicture}
\end{minipage}
\caption{Illustration of the two-dimensional bases $-\mathcal{G}_2$ (left), $-\mathcal{B}_2(a)$ (middle) and \mbox{$\boldsymbol{\lambda}_{\pi/2}+\mathcal{C}_2(a)$} (right).}
\end{figure}
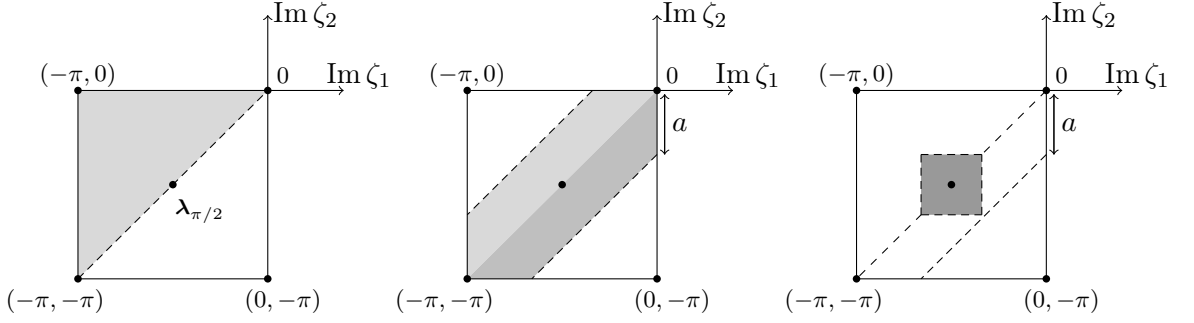

\begin{proposition}\label{enlagement}
Let $S\in\mathcal{S}_0$ and $A\in\mathcal{F}(W_R)$, then $(A\Omega)_n^{\boldsymbol{\alpha}}$ is analytic in the tube $\mathbb{R}^n-i\mathcal{B}_n(\kappa(S))$.
\end{proposition}
\begin{proof}
We follow the proof of the scalar case \cite[Proof of Prop. 5.2.7. a)]{DocL}. To this end, let $\sigma\in\mathfrak{S}_n$. Then due to Corollary \ref{Analyticity},
\begin{equation*}
\left(A\Omega\right)_n(\boldsymbol{\theta}^\sigma):=\left(A\Omega\right)_n(\theta_{\sigma(1)},\dots,\theta_{\sigma(n)})
\end{equation*}
is analytic in the permuted tubes $\mathbb{R}^n-i\mathcal{G}_n^\sigma$, where
\begin{equation*}
\mathcal{G}_n^\sigma:=\sigma\mathcal{G}_n=\{\boldsymbol{\lambda}\in\mathbb{R}^n:\pi>\lambda_{\sigma(1)}>\cdots>\lambda_{\sigma(n)}>0\}.
\end{equation*}
Since $\left(A\Omega\right)_n\in\mathscr{H}_n$, this vector is invariant under the representation $D_n$ of $\mathfrak{S}_n$ (\ref{reprpermugroup}). Hence, we have
\begin{equation}\label{permversion}
\left(A\Omega\right)_n(\boldsymbol{\theta})=\left(D_n(\sigma)\left(A\Omega\right)_n\right)(\boldsymbol{\theta})=S_n^\sigma(\boldsymbol{\theta})\cdot\left(A\Omega\right)_n(\boldsymbol{\theta}^\sigma).
\end{equation}
The tensor $S_n^\sigma(\boldsymbol{\theta})$ consists (up to trivial tensor factors $1$) of factors of the form $S(\theta_{k}-\theta_l)\in\mathcal{S}_0$, $1\leq k,l\leq n$, which are analytic in the strip $S(-\kappa(S),\pi+\kappa(S))$. Thus, all $S_n^\sigma(\boldsymbol{\theta})$, $\sigma\in\mathfrak{S}_n$, are analytic in the tube $\mathbb{R}^n+i\mathcal{B}_n'(\kappa(S))$ with
\begin{equation*}
\mathcal{B}_n'(\kappa(S)):=\{\boldsymbol{\lambda}\in\mathbb{R}^n:-\kappa(S)<\lambda_k-\lambda_l<\kappa(S),\,\,1\leq l<k\leq n\}.
\end{equation*}
Therefore, both sides of Equation (\ref{permversion}) admit an analytic continuation. The left hand side to $\mathcal{T}_n=\mathbb{R}^n-i\mathcal{G}_n$ as before and the right hand side to the tube based on $\mathcal{B}_n'(\kappa(S))\cap (-\mathcal{G}_n^\sigma)$. Since convergence to the boundary in the sense of distributions gives the same value on $\mathbb{R}^n$, Epstein's generalization of the Edge of the Wedge Theorem \cite{eps60} can be applied, yielding that $(A\Omega)_n$ can be analytically continued to the tube based on the convex closure of
\begin{equation*}
\bigcup\limits_{\sigma\in\mathfrak{S}_n}\mathcal{B}_n'(\kappa(S))\cap(-\mathcal{G}_n^\sigma).
\end{equation*}
As the convex closure of $\bigcup\limits_{\sigma\in\mathfrak{S}_n}(-\mathcal{G}_n^\sigma)$ is the cube $(-\pi,0)^{\times n}$, $(A\Omega)_n$ is analytic in the tube based on $\mathcal{B}_n'(\kappa(S))\cap(-\pi,0)^{\times n}=-\mathcal{B}_n(\kappa(S))$, cf. Figure \ref{fig1}, proving the claim.
\end{proof}

\begin{figure}[h]
\begin{minipage}{3cm}
\begin{tikzpicture}[scale=1]
               \begin{scope}%[transparent]    %current
              \fill[gray!30!white] (0,0)--(-2.5,0)--(-2.5,-2.5) -- (0,0);
   \fill[gray!60!white] (0,0)--(-2.5,-2.5)-- (0,-2.5) -- (0,0);
              
               \end{scope} 
        \begin{scope}%[transparent]    %current
                \draw[fill] (2.2,.4) node{Im$\,\zeta_1$};
\draw[fill] (.7,1.4) node{Im$\,\zeta_2$};

\draw[fill] (-2.5,.2) node{\footnotesize{$(-\pi,0)$}};
\draw[fill] (-2.8,-2.8) node{\footnotesize{$(-\pi,-\pi)$}};
\draw[fill] (.2,-2.8) node{\footnotesize{$(0,-\pi)$}};
\draw[fill] (.2,.2) node{\footnotesize{$0$}};
\fill[black] (-2.5,-2.5) circle (.3ex);
\fill[black] (-2.5,0) circle (.3ex);
\fill[black] (0,-2.5) circle (.3ex);
\fill[black] (0,0) circle (.3ex);
        \begin{scope}[->]
            \draw (0,-2.5) -- (0,1.3) node[anchor=north] {};
            \draw (-2.5,0) -- (1.5,0) node[anchor=east] {};
        \end{scope}
    \begin{scope}
       \draw[style=densely dashed] (-2.5,-2.5)--(0,0);
       \draw[fill] (-1.8,-.6) node{\footnotesize{$-\mathcal{G}^{id}_2$}};
 \draw[fill] (-.6,-1.8) node{\footnotesize{$-\mathcal{G}^{\tau_1}_2$}};
       \draw (-2.5,0) -- (-2.5,-2.5);
      \draw (-2.5,-2.5) -- (0,-2.5);
      \fill[black] (-1.25,-1.25) circle (.3ex);
        \end{scope}
        \end{scope}
\end{tikzpicture}
\end{minipage}\hspace{2.5cm}
\begin{minipage}{3cm}
\begin{tikzpicture}[scale=1]
               \begin{scope}%[transparent]    %current
\fill[gray!20!white] (1.5,0.65)--(-.2,.65)--(-3.85,-3)--(-2.15,-3);

              \fill[gray!50!white] (0,0)--(-.85,0)--(-2.5,-2.5+.85) -- (-2.5,-2.5)--(0,0);
              \fill[gray!50!white] (0,0)--(-2.5,-2.5)--(-1.65,-2.5) -- (0,-.85);

               \end{scope} 
        \begin{scope}%[transparent]    %current
                \draw[fill] (2.2,.4) node{Im$\,\zeta_1$};
\draw[fill] (.7,1.4) node{Im$\,\zeta_2$};
\draw[fill] (.5,-.85) node{\footnotesize{$-\kappa(S)$}};
\draw[fill] (-4.4,-2.3) node{\footnotesize{$\mathcal{B}_2'(\kappa(S))$}};
\draw[fill] (-3.65,-.6) node{\footnotesize{$-\mathcal{B}_2(\kappa(S))$}};
\draw[fill] (.2,-2.8) node{\footnotesize{$(0,-\pi)$}};
\draw[fill] (.2,.2) node{\footnotesize{$0$}};
\fill[black] (-2.5,-2.5) circle (.3ex);
\fill[black] (-2.5,0) circle (.3ex);
\fill[black] (0,-2.5) circle (.3ex);
\fill[black] (0,0) circle (.3ex);
        \begin{scope}[->]
            \draw (0,-2.5) -- (0,1.3) node[anchor=north] {};
            \draw (-2.5,0) -- (1.5,0) node[anchor=east] {};
        \end{scope}
    \begin{scope}
    \draw[help lines] (-2.85,-.615)--(-1.6,-1.05);
    \draw[help lines] (-3.75,-2.35)--(-2.7,-2.7);
    \draw[style=densely dashed] (1.5,0.65) --(-2.15,-3);
        \draw[style=densely dashed] (-.2,0.65) --(-3.85,-3);
       \draw (-2.5,0) -- (-2.5,-2.5);
      \draw (-2.5,-2.5) -- (0,-2.5);
      \fill[black] (-1.25,-1.25) circle (.3ex);
        \end{scope}
        \end{scope}
\end{tikzpicture}
\end{minipage}
\caption{Regions appearing in the proof of Proposition \ref{enlagement} for the case $n=2$.}
\label{fig1}
\end{figure}
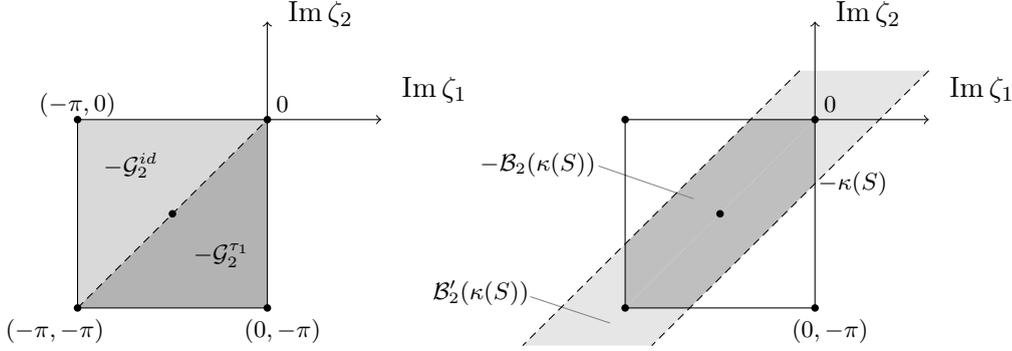
Since $S_n^\sigma(\boldsymbol{\theta})$ is a very complicated object, it is not straightforward to derive a formula for this tensor. For the scalar case, i.e. if $\dim\mathcal{K}=1$, however, the corresponding product of S-matrices is explicitly known \cite{L08}. This formula implies, in particular, that the domain of analyticity may be extended even further, namely to a tube with base $-\widetilde{\mathcal{B}}_n(a)$, where $$\widetilde{\mathcal{B}}_n(a):=\{\boldsymbol{\lambda}\in\mathbb{R}^n:0<\lambda_1,\dots,\lambda_n<\pi,\,\lambda_k-\lambda_l<a,\,1\leq l<k\leq n\}.$$
Although we conjecture that a similar formula can be found in the more general case of $\dim\mathcal{K}>1$, a larger domain of holomorphy is not relevant for the verification of the nuclearity of the maps $\Xi(s)$, since the main points of interest $\boldsymbol{\theta}+i\boldsymbol{\lambda}_{\pi/2}$, $\boldsymbol{\theta}\in\mathbb{R}^n$, are contained in $\mathbb{R}^n-i\mathcal{B}_n(\kappa(S))$.\par
With regard to the one particle case, cf. Lemma \ref{oneparticle}, it, therefore, remains to discuss the boundedness properties of $(A\Omega)_n^{\boldsymbol{\alpha}}$, in order to derive bounds on the nuclear norms of $\Xi_n(s)$, $n>1$. The results established so far yield, in fact, the existence of a map $\Upsilon_n(s,\kappa(S))$ from $\mathcal{F}(W_R)$ into a Hardy space, as the following proposition implies.
\begin{proposition}\label{Corrhardystructure}
Let $S\in\mathcal{S}_0$, $A\in\mathcal{F}(W_R)$, $0<\kappa(S)<\tfrac{\pi}{2}$ and $\undertilde{s}:=(0,s)$, $s>0$, then $(A(s)\Omega)_n:=\left(U(\undertilde{s})AU(\undertilde{s})\Omega\right)_n$ restricted to the tube $\mathcal{T}_n(\tfrac{\kappa(S)}{n})$ is in the Hardy space $H^2(\mathcal{T}_n(\tfrac{\kappa(S)}{n}))\otimes\mathcal{K}^{\otimes n}$ with Hardy norms bounded by
\begin{equation}\label{bound33}
\triplenorm(A(s)\Omega)_{n}\triplenorm\leq \upsilon(s,\kappa(S))^n\,\|A\|.
\end{equation}
An adequate choice for $\upsilon(s,\kappa(S))$ is given by
\begin{equation}\label{boundupsilon}
\upsilon(s,\kappa(S)):=\|S\|_{\kappa(S)}D^{1/2}\max\left\{1,\frac{\sqrt{2}e^{-sm_\circ\sin\kappa(S)}}{\sqrt{\kappa(S)}\cdot(\pi sm_\circ\sin\kappa(S))^{1/4}}\right\}>0.
\end{equation}

\end{proposition}
\begin{proof}
Similarly as in the proof of Lemma E.1 in \cite{erratum} we first show the validity of the bound (\ref{bound33}) in the tube based on the cube $\mathcal{C}_n^0:=(-\tfrac{\kappa(S)}{n},0)^{\times n}$. To this end, recall the correspondence (\ref{zusammenhang}), that is, $(A\Omega)_n^{\boldsymbol{\alpha}}(\boldsymbol{\theta})=\tfrac{1}{\sqrt{n!}}\langle z_{\alpha_2
}^\dagger(\theta_2)\cdots z_{\alpha_n
}^\dagger(\theta_n)\Omega,[z_{\alpha_1
}(\theta_1),A]\Omega\rangle$. Hence, by application of Lemma \ref{basicLemma} it follows that
$$h^{\alpha_1}_{-\lambda}:\theta_1\mapsto\int d\theta_2\cdots d\theta_n (A\Omega)_n^{\alpha_1\alpha_2\dots\alpha_n}(\theta_1-i\lambda,\theta_2,\dots,\theta_n)f^{\alpha_2\dots\alpha_n}(\theta_2,\dots,\theta_n)$$
is in $L^2(\mathbb{R},d\theta)$ for any $f\in L^2(\mathbb{R}^{n-1})\otimes\mathcal{K}^{\otimes n-1}$ and any $0\leq\lambda\leq \pi$. Moreover, we have $\|h_{-\lambda}\|_2\leq  \|f\|_2\|A\|$. By the mean value property
\begin{equation*}
h^{\alpha_1}(\zeta)=\frac{1}{\pi r^2}\int\limits_{D(\zeta,r)}d\theta\, d\lambda\, h^{\alpha_1}(\theta+i\lambda),
\end{equation*}
with $D(\zeta,r)\subset\mathbb{R}-i[0,\pi]$ a disc of radius $r$ and center $\zeta$. Hence,
\begin{eqnarray*}
\left|h^{\alpha_1}(\zeta)\right|^2&\leq&\frac{1}{\pi r^2}\int\limits_{D(\zeta,r)}d\theta\, d\lambda\, \left|h^{\alpha_1}(\theta+i\lambda)\right|^2\\
&\leq&\frac{1}{\pi r^2}\int\limits_{\mathbb{R}+\zeta+i[-r,r]}d\theta\, d\lambda\, \left|h^{\alpha_1}(\theta+i\lambda)\right|^2\\
&\leq&\frac{2r}{\pi r^2}\|f\|_2^2\|A\|^2,
\end{eqnarray*}
yielding $\left|h^{\alpha_1}(\theta_1-i\lambda)\right|\leq \left(\tfrac{2}{\pi \min\{\lambda,\pi-\lambda\}}\right)^{1/2} \|f\|_2\|A\|$, $\theta_1\in\mathbb{R}$. It, therefore, follows that for $\theta_1$ and $\lambda$ fixed the function $(\theta_2,\dots,\theta_n)\mapsto(A\Omega)_n^{\boldsymbol{\alpha}}(\theta_1-i\lambda,\theta_2,\dots,\theta_n)$ is in $L^2(\mathbb{R}^{n-1})$ with norm bounded by $\left(\tfrac{2}{\pi \min\{\lambda,\pi-\lambda\}}\right)^{1/2} \|A\|$. As a next step we may pass by means of the entire function $u_{n,s}^{[\boldsymbol{\alpha}]}(\boldsymbol{\zeta}):=\prod_{k=1}^{n}e^{-ism_{[\alpha_k]}\sinh \zeta_k}$ to the ``shifted'' function
$(A(s)\Omega)_n^{\boldsymbol{\alpha}}(\boldsymbol{\zeta})=u_{n,s}^{[\boldsymbol{\alpha}]}(\boldsymbol{\zeta})\cdot(A\Omega)_n^{\boldsymbol{\alpha}}(\boldsymbol{\zeta})$ with $s>0$. Then, in view of $\sinh(a+ib)=\sinh a\cdot\cos b+i\cosh a\cdot\sin b$, $a,b\in\mathbb{R}$ we find with $e_1=(1,0,\dots,0)$
\begin{eqnarray*}
\hspace{-.3cm}&&\hspace{-.3cm}\sum_{\boldsymbol{\alpha}}\int d^n\boldsymbol{\theta}\left|\left(A(s)\Omega\right)_n^{\boldsymbol{\alpha}}(\boldsymbol{\theta}-i\lambda e_1)\right|^2\\
&&= \sum_{\boldsymbol{\alpha}}\int d\theta_1 e^{-2sm_{[\alpha_1]}\sin\lambda\cosh\theta_1}\int d\theta_2\cdots d\theta_n\left|(A\Omega)_n^{\boldsymbol{\alpha}}(\theta_1-i\lambda,\theta_2,\dots,\theta_n)\right|^2\\
&&\leq \frac{2 }{\pi \min\{\lambda,\pi-\lambda\}}\|A\|^2\sum_{\boldsymbol{\alpha}}\int d\theta_1 e^{-2sm_{\circ}\sin\lambda\cosh\theta_1}\\
&&= \frac{2D^n}{\pi \min\{\lambda,\pi-\lambda\}}\|A\|^2\int d\theta_1 e^{-2sm_{\circ}\sin\lambda\cosh\theta_1}\\
&&\leq \frac{2D^n}{\pi \min\{\lambda,\pi-\lambda\}}\|A\|^2e^{-2sm_{\circ}\sin\lambda}\int d\theta_1 e^{-sm_{\circ}\sin\lambda\theta_1^2}\\
&&=\frac{2D^n\|A\|^2e^{-2sm_{\circ}\sin\lambda}\sqrt{\pi}}{\pi \min\{\lambda,\pi-\lambda\}\sqrt{sm_\circ\sin\lambda}}
=: a(s,\lambda)^2 D^n\|A\|^2,
\end{eqnarray*}
where we made use of $\cosh(\theta_1)\geq 1+\tfrac{1}{2}\theta_1^2$. For $0<\lambda<\pi$ we, therefore, arrive at $(A(s)\Omega)_{n,-\lambda e_1}\in L^2(\mathbb{R}^n)\otimes\mathcal{K}^{\otimes n}$. With regard to analytic continuation in the directions of the other standard basis vectors $e_2,\dots,e_n$ of $\mathbb{R}^n$ we take the S-symmetry of $(A(s)\Omega)_n$ into account. To this end, recall that for $0<\lambda\leq\kappa(S)$
\begin{equation*}
\sup\limits_{\theta\in\mathbb{R}}\|S(\theta)\|=1,\qquad \sup\limits_{\theta\in\mathbb{R}}\|S(\theta+i\lambda)\|\leq 1,\qquad\sup\limits_{\theta\in\mathbb{R}}\|S(\theta-i\lambda)\|\leq\|S\|_{\kappa(S)}.
\end{equation*}
Moreover, we consider the function
$$\zeta_k\mapsto S_n^\sigma(\theta_1,\dots,\theta_{k-1},\zeta_k,\theta_{k+1},\dots,\theta_n),\qquad\theta_k\in\mathbb{R},\qquad k=1,\dots,n,$$
which, clearly, is analytic in the strip $-\kappa(S)<\text{Im}\,\zeta_k<\kappa(S)$. An estimate on this function is obtained by determining the number of $\zeta_k$-dependent factors $S$ in the above tensor. To this end, recall the fact that any $\sigma\in\mathfrak{S}_n$ can be decomposed (non-uniquely) into a product of inv$(\sigma)$ transpositions $\tau_j\in\mathfrak{S}_n$, where inv$(\sigma)$ is the number of pairs $(i,j)$, $i,j=1,\dots,n$, with $i<j$ and $\sigma(i)>\sigma(j)$. Therefore, we count that the maximal possible number of transpositions which involve the element $k$ is $n-1$. Hence, the representing tensor $S_n^\sigma(\theta_1,\dots,\theta_{k-1},\zeta_k,\theta_{k+1},\dots,\theta_n)$ contains at most $n-1$ factors depending on the variable $\zeta_k$. Moreover, each of those factors is bounded by $\|S\|_{\kappa(S)}$. All the others are bounded by 1. Hence, it follows for $j=2,\dots,n$ and $0<\lambda\leq\kappa(S)$ that
\begin{eqnarray*}
\int d^n\boldsymbol{\theta}\,\|\left(A(s)\Omega\right)_n(\boldsymbol{\theta}-i\lambda e_j)\|_{\mathcal{K}^{\otimes n}}^2&\leq& \|S\|_{\kappa(S)}^{2(n-1)}\int d^n\boldsymbol{\theta}\,\|\left(A(s)\Omega\right)_n(\boldsymbol{\theta}-i\lambda e_1)\|_{\mathcal{K}^{\otimes n}}^2\\
&\leq&\|S\|_{\kappa(S)}^{2(n-1)}a(s,\lambda)^2  D^n\|A\|^2.
\end{eqnarray*}
That is, we have, in particular, $\|(A(s)\Omega)_{n,-\kappa(S)e_j}\|_2\leq\|S\|_{\kappa(S)}^na(s,\kappa(S))D^{n/2}\|A\|$, with $j=1,\dots,n$. On the other hand, $(A(s)\Omega)_n\in L^2(\mathbb{R}^n)\otimes\mathcal{K}^{\otimes n}$ with norm bounded by $\|A\|$. Since, moreover, the latter function is obtained from its analytic continuation as a strong limit for imaginary part $\boldsymbol{\lambda}\rightarrow 0$, we conclude that for all $\boldsymbol{\lambda}$ in the convex closure of the points $0,-\kappa(S)e_1,\dots,-\kappa(S)e_n$ we have the following bound
\begin{equation}\label{b}
\|(A(s)\Omega)_{n,\boldsymbol{\lambda}}\|_2\leq \upsilon(s,\kappa(S))^n\|A\|,\quad \upsilon(s,\kappa(S)):=\|S\|_{\kappa(S)}D^{1/2}\max\{1,a(s,\kappa(S))\}.
\end{equation}
As $\mathcal{C}_n^0$ is contained in the convex closure of $(0,-\kappa(S)e_1,\dots,-\kappa(S)e_n)$, the bound $(\ref{b})$ holds, clearly, for all $\boldsymbol{\lambda}\in\mathcal{C}_n^0$, establishing the claim for the tube based on $\mathcal{C}_n^0$.\par
We proceed by considering the tube based on $\mathcal{C}_n^{-\pi}:=(-\pi,\dots,-\pi)+(0,\tfrac{\kappa(S)}{n})^{\times n}$. Since $(A(s)\Omega)^{\boldsymbol{\alpha}}_{n,(-\pi,\dots,-\pi)}=(JA^*(s)\Omega)_n^{\boldsymbol{\alpha}}$, cf. (\ref{zusammenhang}), it follows, due to $\|A^*\|=\|A\|$, immediately that the bound (\ref{b}) also holds for $\boldsymbol{\lambda}\in\mathcal{C}_n^{-\pi}$.\par
It remains to show the validity of (\ref{b}) in the tube $\mathcal{T}_n(\tfrac{\kappa(S)}{n})$. To this end, note that the cubes $\mathcal{C}_n^0$ and $\mathcal{C}_n^{-\pi}$ (both contained in the analyticity domain of $(A(s)\Omega)_n^{\boldsymbol{\alpha}}$) are connected by the line segment $l$ from the point $(0,\dots,0)$ to $(-\pi,\dots,-\pi)$, cf. Figure \ref{fig}.
\begin{figure}[h]
\begin{center}
\begin{tikzpicture}[scale=1.2]
               \begin{scope}%[transparent]    %current
              \fill[gray!30!white] (0,0)--(-.85,0)--(-2.5,-2.5+.85) -- (-2.5,-2.5)--(0,0);
              \fill[gray!30!white] (0,0)--(-2.5,-2.5)--(-1.65,-2.5) -- (0,-.85);
               \end{scope} 
        \begin{scope}%[transparent]    %current
                \draw[fill] (1.2,.2) node{Im$\,\zeta_1$};
\draw[fill] (.5,1) node{Im$\,\zeta_2$};
\draw[fill] (.4,-.45) node{\footnotesize{$\kappa(S)$}};
\draw[fill] (-2.5,.2) node{\footnotesize{$(-\pi,0)$}};
\draw[fill] (-2.8,-2.8) node{\footnotesize{$(-\pi,-\pi)$}};
\draw[fill] (.2,-2.8) node{\footnotesize{$(0,-\pi)$}};
\draw[fill] (.2,.2) node{\footnotesize{$0$}};
\fill[black] (-2.5,-2.5) circle (.3ex);
\fill[black] (-2.5,0) circle (.3ex);
\fill[black] (0,-2.5) circle (.3ex);
\fill[black] (0,0) circle (.3ex);
        \begin{scope}[->]
            \draw (0,-2.5) -- (0,1) node[anchor=north] {};
            \draw (-2.5,0) -- (1,0) node[anchor=east] {};
            \draw (0.1,-.05)--(.1,-.85);
               \draw (.1,-.85)--(0.1,-.05);
        \end{scope}
    \begin{scope}  
    \draw[style=densely dashed] (-0.85,0) --(-2.5,-1.65);
       
    \draw[style=densely dashed] (0,-0.85) --(-1.65,-2.5);
       \draw (-2.5,0) -- (-2.5,-2.5);
      \draw (-2.5,-2.5) -- (0,-2.5);
        \draw[style=densely dashed](0,0)--(-2.5,-2.5);
      \fill[black] (-1.25,-1.25) circle (.3ex);
             \draw[style=densely dashed] (-1.25+.2,-1.25+.2)--(-1.25+.2,-1.25-.2);
             \draw[style=densely dashed] (-1.25-.2,-1.25+.2) -- (-1.25+.2,-1.25+.2);
             \draw[style=densely dashed] (-1.25-.2,-1.25+.2) --(-1.25-.2,-1.25-.2);
               \draw[style=densely dashed] (-1.25-.2,-1.25-.2) --(-1.25+.2,-1.25-.2);
       \draw[style=densely dashed] (0,0)--(-.2,-.2);
   \draw[style=densely dashed] (-.4,0)--(-.4,-.4);
    \draw[style=densely dashed] (-.4,0)--(-.4,-.4);
     \draw[style=densely dashed] (0,-.4)--(-.4,-.4);
 \draw[style=densely dashed] (-2.5+.4,-2.5+.4)--(-2.5,-2.5+.4);
  \draw[style=densely dashed] (-2.5+.4,-2.5)--(-2.5+.4,-2.5+.4);
        \end{scope}
               \end{scope}
\end{tikzpicture}
\end{center}
\caption{Cubes appearing in the proof of Proposition \ref{Corrhardystructure} for the case $n=2$.}\label{fig}
\end{figure}
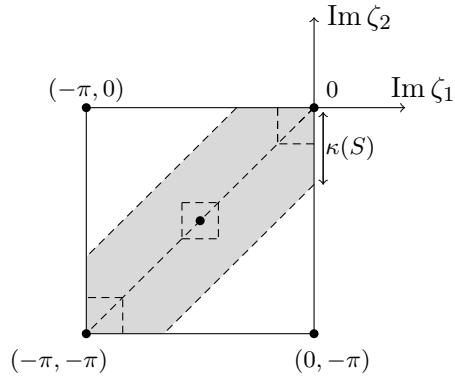
By modular theory we have that $\|(A(s)\Omega)_{n,\boldsymbol{\lambda}}\|_2\leq\|A\|$ for all $\boldsymbol{\lambda}\in l$. Consequently, the bound (\ref{b}) also holds for $\boldsymbol{\lambda}$ in the convex closure of $\mathcal{C}_n^0\cup\mathcal{C}_n^{-\pi}\cup l$. But the cube $\boldsymbol{\lambda}_{\pi/2}+\mathcal{C}_n(\tfrac{\kappa(S)}{n})$ is contained in this region, yielding what is claimed.
\end{proof}
The preceding proposition establishes a generalization of the Hardy space structure found in the case $n=1$. In particular, define
\begin{equation}\label{upsilon}
\Upsilon_n(s,\kappa(S)):\mathcal{F}(W_R)\rightarrow H^2(\mathcal{T}_n(\tfrac{\kappa(S)}{n}))\otimes\mathcal{K}^{\otimes n},\qquad \Upsilon_n(s,\kappa)A:=\left(A(\tfrac{s}{2})\Omega\right)_n,
\end{equation}
then, an immediate consequence of Proposition \ref{Corrhardystructure} is the following.
\begin{corollary}
Let $s>0$, then, the map $\Upsilon_n(s,\kappa)$, as defined in (\ref{upsilon}), is a bounded operator between the Banach spaces $(\mathcal{F}(W_R),\|\cdot\|_{\mathcal{B}(\mathscr{H})})$ and \mbox{$(H^2(\mathcal{T}_n(\tfrac{\kappa(S)}{n}))\otimes\mathcal{K}^{\otimes n},\triplenorm\cdot\triplenorm)$}, with operator norm complying with
\begin{equation}\label{opb}
\|\Upsilon_n(s,\kappa(S))\|\leq\, \upsilon(\tfrac{s}{2},\kappa(S))^n,
\end{equation}
where $\upsilon(s,\kappa(S))$ is given by (\ref{boundupsilon}).
\end{corollary}
\begin{proof}
The bound (\ref{opb}) on the operator norm of the map $\Upsilon_n(s,\kappa(S))$ is a direct consequence of (\ref{bound33}). 
\end{proof}
Since we aim at proceeding as in the case $n=1$, we consider the maps $\Xi_n(s)$, $n> 1$, as the concatenation of $\Upsilon_n(s,\kappa(S))$ and maps $X_n(s,\kappa(S))$, as illustrated by the following commutative diagram.
    \begin{center}
        \begin{tikzpicture}
        \begin{scope}%[transparent]    %current
                \draw[fill] (0.3,2.2) node{$\mathcal{F}(W_R)$};
\draw[fill] (-1,0) node{$H^2(\mathcal{T}_n(\tfrac{\kappa(S)}{n}))\otimes\mathcal{K}^{\otimes n}$};
\draw[fill] (2.8,0) node{$\mathscr{H}_n$};
\draw[fill] (2.2,1.3) node{$\Xi_n(s)$};
\draw[fill] (-1.6,1.3) node{$\Upsilon_n(s,\kappa(S))$};
\draw[fill] (1.7,-.4) node{$X_n(s,\kappa(S))$};
        \begin{scope}[->]
            \draw (0,2) -- (-1,0.5) node[anchor=south] {};
            \draw (.8,0) -- (2.4,0) node[anchor=east] {};
            \draw (0.7,2) -- (2.4,0.2) node[anchor=east] {};
        \end{scope}
        \end{scope}
        \end{tikzpicture} 
        \end{center}
The map $X_n(s,\kappa(S)):H^2(\mathcal{T}_n(\tfrac{\kappa(S)}{n}))\otimes\mathcal{K}^{\otimes n}\rightarrow\mathscr{H}_n$ is defined by
\begin{equation}\label{xn}
\left(X_n(s,\kappa(S))h\right)^{\boldsymbol{\alpha}}(\boldsymbol{\theta}):=\prod_{k=1}^{n}e^{-\tfrac{s}{2}\,m_{[\alpha_k]}\cosh\theta_k}\cdot h^{\boldsymbol{\alpha}}(\boldsymbol{\theta}+i\boldsymbol{\lambda}_{\pi/2}),
\end{equation}
which, with regard to (\ref{concrete}), justifies
\begin{equation}\label{concatenation}
\Xi_n(s)A=\left(X_n(s,\kappa(S))\,\circ\,\Upsilon_n(s,\kappa(S))\right)A,\qquad A\in\mathcal{F}(W_R).
\end{equation}
It, thus, remains to investigate the properties of the maps $X_n(s,\kappa(S))$, $n>1$. Since $\Upsilon_n(s,\kappa(S))$ is a bounded map, nuclearity of $\Xi_n(s)$ would be proven if $X_n(s,\kappa(S))$ was nuclear, cf. Lemma \ref{propertiesNuclearMaps}. Indeed, nuclearity holds for $X_n(s,\kappa(S))$ as stated in the subsequent lemma.
\begin{lemma}\label{Xn}
For $s>0$ and $0<\kappa(S)<\tfrac{\pi}{2}$ the map $X_n(s,\kappa(S))$, as defined in (\ref{xn}), is a nuclear mapping between the Banach spaces $(H^2(\mathcal{T}_n(\tfrac{\kappa(S)}{n}))\otimes\mathcal{K}^{\otimes n},\triplenorm\cdot\triplenorm)$ and $(\mathscr{H}_n,\|\cdot\|)$. Its nuclear norm is bounded by
\begin{equation}\label{boundx}
\|X_n(s,\kappa(S))\|_1\leq  x(s,\kappa(S))^n\cdot n^n,
\end{equation}
where $x(s,\kappa(S))>0$ is monotonously decreasing function of $s$ for fixed $\kappa(S)$, with limits $x(s,\kappa(S))\rightarrow 0$ for $s\rightarrow\infty$ and $x(s,\kappa(S))\rightarrow \infty$ for $s\rightarrow 0$. 
\end{lemma}
\begin{proof}
Consider a closed polydisc $\overline{D_n(\boldsymbol{\theta}+i\boldsymbol{\lambda}_{\pi/2})}\subset\mathcal{T}_n(\tfrac{\kappa(S)}{n})$ with center $\boldsymbol{\theta}+i\boldsymbol{\lambda}_{\pi/2}$, $\boldsymbol{\theta}\in\mathbb{R}^n$. For $h\in H^2(\mathcal{T}_n(\tfrac{\kappa(S)}{n}))\otimes\mathcal{K}^{\otimes n}$ we then have by Cauchy's integral formula
\begin{equation*}
h^{\boldsymbol{\alpha}}(\boldsymbol{\theta}+i\boldsymbol{\lambda}_{\pi/2})=\frac{1}{(2\pi i)^n}\oint_{T_n(\boldsymbol{\theta}+i\boldsymbol{\lambda}_{\pi/2})}d^n\boldsymbol{\zeta}'\frac{h^{\boldsymbol{\alpha}}(\boldsymbol{\zeta}')}{\prod_{k=1}^{n}\left(\zeta'_k-\theta_k+i\tfrac{\pi}{2}\right)},
\end{equation*}
where $T_n(\boldsymbol{\theta}+i\boldsymbol{\lambda}_{\pi/2})$ denotes the distinguished boundary of $D_n(\boldsymbol{\theta}+i\boldsymbol{\lambda}_{\pi/2})$. Since $\mathcal{C}_n(\tfrac{\kappa(S)}{n})$ is a polyhedron, the last two properties listed in Proposition \ref{hardyeigenschaften} can be applied and, thus, the contour of integration may be deformed to the boundary of $\mathcal{T}_n(\tfrac{\kappa(S)}{n})$. Hence, we have
\begin{equation*}
h^{\boldsymbol{\alpha}}(\boldsymbol{\theta}+i\boldsymbol{\lambda}_{\pi/2})=\frac{1}{(2\pi i)^n}\sum_{\boldsymbol{\varepsilon}}\int_{\mathbb{R}^n}d^n\boldsymbol{\theta}'\left(\prod_{k=1}^{n}\frac{\varepsilon_k}{\theta_k'-\theta_k-i\tfrac{\kappa(S)}{2n}\varepsilon_k}\right)h^{\boldsymbol{\alpha}}(\boldsymbol{\theta}'+i(\boldsymbol{\lambda}_{\pi/2}-\tfrac{\kappa(S)}{2n}\boldsymbol{\varepsilon})),
\end{equation*}
where $\boldsymbol{\varepsilon}=(\varepsilon_1,\dots,\varepsilon_n)$, with $\varepsilon_k=\pm 1$, $k=1,\dots,n$.\par
Writing
\begin{equation*}
\left(X_n(s,\kappa(S))h\right)^{\boldsymbol{\alpha}}(\boldsymbol{\theta})=\widehat{u}_{n,s/4}^{[\boldsymbol{\alpha}]}(\boldsymbol{\theta})\cdot\left( u_{n,s/4}^{[\boldsymbol{\alpha}]}\cdot 
h\right)^{\boldsymbol{\alpha}}(\boldsymbol{\theta}+i\boldsymbol{\lambda}_{\pi/2}),
\end{equation*}
where $u_{n,s}^{[\boldsymbol{\alpha}]}(\boldsymbol{\zeta}):=\prod_{k=1}^{n}e^{-ism_{[\alpha_k]}\sinh \zeta_k}$ and $\widehat{u}_{n,s}^{[\boldsymbol{\alpha}]}(\boldsymbol{\theta}):=u_{n,s}^{[\boldsymbol{\alpha}]}(\boldsymbol{\theta}+i\boldsymbol{\lambda}_{\pi/2})$, we have
\begin{eqnarray*}
&&(X_n(s,\kappa(S))h)^{\boldsymbol{\alpha}}(\boldsymbol{\theta})\\
&&\,\,=\tfrac{\widehat{u}_{n,s/4}^{[\boldsymbol{\alpha}]}(\boldsymbol{\theta})}{(2\pi i)^n}\sum_{\boldsymbol{\varepsilon}}\int\limits_{\mathbb{R}^n}d^n\boldsymbol{\theta}'\left(\prod_{k=1}^{n}\tfrac{\varepsilon_k}{\theta_k'-\theta_k-i\frac{\kappa(S)}{2n}\varepsilon_k}\right)\left( u_{n,s/4}^{[\boldsymbol{\alpha}]}\cdot 
h\right)^{\boldsymbol{\alpha}}(\boldsymbol{\theta}'+i(\boldsymbol{\lambda}_{\pi/2}-\tfrac{\kappa(S)}{2n}\boldsymbol{\varepsilon}))\\
&&\,\,=\widehat{u}_{n,s/4[1-\cos(\kappa(S)/(2n))]}^{[\boldsymbol{\alpha}]}(\boldsymbol{\theta})\\
&&\qquad\times\sum_{\boldsymbol{\varepsilon}}\int\limits_{\mathbb{R}^n}d^n\boldsymbol{\theta}'\left(\prod_{k=1}^{n}\tfrac{\varepsilon_k \exp[-\frac{1}{4}sm_{[\alpha_k]}\cos\frac{\kappa(S)}{2n}(\cosh\theta_k+\cosh\theta_k')]}{2\pi i(\theta_k'-\theta_k-i\frac{\kappa(S)}{2n}\varepsilon_k)}\right)u_{n,-s/4\sin(\kappa(S)/(2n)),\boldsymbol{\varepsilon}}^{[\boldsymbol{\alpha}]}(\boldsymbol{\theta}')\\
&&\qquad\qquad\times\, h^{\boldsymbol{\alpha}}(\boldsymbol{\theta}'+i(\boldsymbol{\lambda}_{\pi/2}-\tfrac{\kappa(S)}{2n}\boldsymbol{\varepsilon})),
\end{eqnarray*}
with $u_{n,s,\boldsymbol{\varepsilon}}^{[\boldsymbol{\alpha}]}(\boldsymbol{\theta}):=u_{n,s}^{[\boldsymbol{\alpha}]}(\varepsilon_1\theta_1,\dots,\varepsilon_n\theta_n)$. Putting $g^{[\alpha]}(\theta):=\exp[-\frac{1}{4}sm_{[\alpha]}\cos\frac{\kappa(S)}{2n}\cosh\theta]\in L^2(\mathbb{R})$ and considering an integral operator $R_{g,b}=\bigoplus_{\alpha=1}^D R_{g,b}^{[\alpha]}$ on $L^2(\mathbb{R})\otimes\mathcal{K}$, defined by the kernels
\begin{equation*}
R^{[\alpha]}_{g,b}(\theta,\theta')=\frac{-\text{sign}\,(b)}{2\pi i}\frac{\overline{g^{[\alpha]}(\theta)}g^{[\alpha]}(\theta')}{\theta'-\theta+ib},\qquad b\in\mathbb{R}\backslash\{0\},
\end{equation*}
we arrive at
\begin{eqnarray}\label{x}
X_n(s,\kappa(S))&=& \widehat{u}_{n,s/4[1-\cos(\kappa(S)/(2n))]}\sum_{\boldsymbol{\varepsilon}} \left(R_{g,-\varepsilon_1\frac{\kappa(S)}{2n}}\otimes\cdots\otimes R_{g,-\varepsilon_n\frac{\kappa(S)}{2n}}\right)\nonumber\\
&&\qquad\times\,u_{n,-s/4\sin(\kappa(S)/(2n)),\boldsymbol{\varepsilon}}\, E_{\boldsymbol{\lambda}_{\pi/2}-\frac{\kappa(S)}{2n}\boldsymbol{\varepsilon}},
\end{eqnarray}
with$(\widehat{u}_{n,s}h)^{\boldsymbol{\alpha}}(\boldsymbol{\zeta}):=\widehat{u}_{n,s}^{[\boldsymbol{\alpha}]}(\boldsymbol{\zeta})h^{\boldsymbol{\alpha}}(\boldsymbol{\zeta})$ and analogously  $(u_{n,s,\boldsymbol{\varepsilon}}h)^{\boldsymbol{\alpha}}(\boldsymbol{\zeta}):=u_{n,s,\boldsymbol{\varepsilon}}^{[\boldsymbol{\alpha}]}(\boldsymbol{\zeta})h^{\boldsymbol{\alpha}}(\boldsymbol{\zeta})$. The operator norms of the maps $E_{\boldsymbol{\lambda}_{\pi/2}-\frac{\kappa(S)}{2n}\boldsymbol{\varepsilon}}:h\mapsto h_{\boldsymbol{\lambda}_{\pi/2}-\frac{\kappa(S)}{2n}\boldsymbol{\varepsilon}}$ from $H^2(\mathcal{T}_n(\tfrac{\kappa(S)}{n}))\otimes\mathcal{K}^{\otimes n}$ to $L^2(\mathbb{R}^n)\otimes\mathcal{K}^{\otimes n}$ are bounded for any $\boldsymbol{\varepsilon}$ by one, as follows directly from $\|h_{\boldsymbol{\lambda}_{\pi/2}-\frac{\kappa(S)}{2n}\boldsymbol{\varepsilon}}\|\leq\triplenorm h\triplenorm$. It is, moreover, obvious that the operator norms of the multiplication operators $u_{n,s}$ and $\widehat{u}_{n,s,\boldsymbol{\varepsilon}}$ appearing in (\ref{x}) are also bounded by one. Since $R_{g,\pm \frac{\kappa(S)}{2n}}$ are trace class operators on $L^2(\mathbb{R})\otimes\mathcal{K}$, cf. Lemma \ref{traceclassneu}, we have that $R_{g,-\varepsilon_1\frac{\kappa(S)}{2n}}\otimes\cdots\otimes R_{g,-\varepsilon_n\frac{\kappa(S)}{2n}}$ are of trace class on $L^2(\mathbb{R}^n)\otimes\mathcal{K}^{\otimes n}$ for any $\varepsilon_1,\dots,\varepsilon_n=\pm 1$. Thus, the nuclearity of $X_n(s,\kappa(S))$ follows.\par
In order to prove the statement on the nuclear norm, note that $R_{g,b}=\widehat{U}R_{g_-,-b}\widehat{U}$, with unitary $(\widehat{U}f)(\theta):=i\cdot f(-\theta)$, $f\in L^2(\mathbb{R})\otimes\mathcal{K}$, and $g^{[\alpha]}_-(\theta):=g^{[\alpha]}(-\theta)$. Hence, we have
\begin{equation*}
\|R_{g,-\varepsilon_1\frac{\kappa(S)}{2n}}\otimes\cdots\otimes R_{g,-\varepsilon_n\frac{\kappa(S)}{2n}}\|_1=\|R_{g,\frac{\kappa(S)}{2n}}\|_1^n.
\end{equation*}
Taking into account that the sum $\sum_{\boldsymbol{\varepsilon}}$ runs over $2^n$ terms, it follows with regard to (\ref{tracenormneu}) that
\begin{eqnarray*}
\|X_n(s,\kappa(S))\|_1&\leq& 2^n\|R_{g,\frac{\kappa(S)}{2n}}\|_1^n\leq 2^nD^n\left(\frac{2n}{2\pi\kappa(S)}\int d\theta e^{-\frac{sm_\circ}{2}\cos \frac{\kappa(S)}{2n}\cosh\theta}\right)^n\\
&\leq&2^nD^n\left(\frac{2n}{2\pi\kappa(S)}\int d\theta e^{-\frac{sm_\circ}{2}\cosh\theta}\right)^n\\
&\leq&n^n\left(\frac{2D}{\pi\kappa(S)}\right)^ne^{-\frac{nsm_\circ}{2}}\left(\frac{4\pi}{sm_\circ}\right)^{n/2},
\end{eqnarray*}
where we made use of $\cosh\theta\geq 1+\theta^2/2$. Hence, the bound on $\|X_n(s,\kappa(S))\|_1$ is of the form (\ref{boundx}).
\end{proof}
The results established by the preceding Lemmata, therefore, yield that $\Xi_n(s):\mathcal{F}(W_R)\rightarrow\mathscr{H}_n$ given by (\ref{concatenation}) is nuclear. The nuclearity of $\Xi(s)=\sum_{n=0}^{\infty}\Xi_n(s)$, (\ref{xi}), on the other hand, does not yet follow since $\|\Xi_n(s)\|_1\leq n^{n}x(s,\kappa(S))^n\upsilon(\tfrac{s}{2},\kappa(S))^n$ is not summable over $n$. However, in the special case of $\dim\mathcal{K}=1$, discussed in \cite{L08, erratum}, the nuclearity of $\Xi(s)$ can be established for a certain class of S-matrices and for values of $s$ above a certain minimal threshold $s_{\text{min}}$. This finding relies on a mapping $I:\mathscr{H}\rightarrow\mathscr{H}^-$ from the S-symmetric Fock space to the totally antisymmetric Fermi Fock space. In case the S-matrices are scalar valued such a map can be found quite easily. In fact, one has
\begin{lemma}[\cite{L08}]
Let $\dim\mathcal{K}=1$ and $S\in\mathcal{S}_0^-:=\{S\in\mathcal{S}_0:S(0)=-1\}$. Then, there exists an analytic function $\delta:S(-\kappa(S),\kappa(S))\rightarrow\mathbb{C}$ (the phase shift) such that $S(\zeta)=S(0)e^{2i\delta(\zeta)}$, $\zeta\in S(-\kappa(S),\kappa(S))$, fixed uniquely by $\delta(0)=0$. Then, the multiplication operators corresponding to the functions
\begin{equation}
I_0:=1,\qquad I_1(\zeta):=1,\qquad I_n(\boldsymbol{\zeta}):=\prod_{1\leq k<l\leq n}\left(-e^{i\delta(\zeta_k-\zeta_l)}\right),\quad n\geq 2,
\end{equation}
denoted by the same symbols $I_n$, fulfill the following.
\begin{itemize}
\item[a)] Viewed as an operator on $H^2(\mathcal{T}_n(\tfrac{\kappa(S)}{n}))$, $I_n$ is a bounded map with operator norm $\|I_n\|\leq\|S\|^{n/2}_{\kappa(S)}$\footnote{This is, indeed, the correct bound with regard to the tube $\mathcal{T}_n(\tfrac{\kappa(S)}{n})$ based on the ``shrinking'' cube $\boldsymbol{\lambda}_{\pi/2}+\mathcal{C}_n(\tfrac{\kappa(S)}{n})$, as can be shown in a straightforward manner by means of the Malgrange Zerner theorem \cite{eps66}. It does, however, not apply when considering the tube based on $\boldsymbol{\lambda}_{\pi/2}+\mathcal{C}_n(\kappa(S))$, as claimed in \cite{L08}.}.
\item[b)] Viewed as an operator on $L^2(\mathbb{R}^n)$, $I_n$ is a unitary intertwining the representations $D_n$ and $D_n^-$ (corresponding to $S(\theta)=-1$, for all $\theta\in\mathbb{R}$) of $\mathfrak{S}_n$. Hence, it maps the S-symmetric subspace $\mathscr{H}_n\subset L^2(\mathbb{R}^n)$ onto the totally antisymmetric subspace $\mathscr{H}_n^-\subset L^2(\mathbb{R}^n)$.
\end{itemize}
\end{lemma}
In the more general situation considered here the derivation of such intertwiners is a more complicated task, in particular, due to the general noncommutativity of matrices. Nevertheless, we conjecture that a similar mapping can be obtained in case $\dim\mathcal{K}>1$. If, namely, $0\leq n\leq 2$ a direct generalization of the above lemma can be stated for S-matrices of a certain type. More precisely, let $F$ denote the \textit{flip} operator, i.e.
\begin{equation}
F:\mathcal{K}\otimes\mathcal{K}\rightarrow\mathcal{K}\otimes\mathcal{K},\qquad F(u\otimes v):=v\otimes u,
\end{equation}
and consider $S\in\mathcal{S}_0^-:=\{S\in\mathcal{S}_0:S(0)=-F\}$ which are of the form
\begin{equation}\label{form}
S(\zeta)=\sum_{j=1}^{k}h_j(\zeta)M_j,\qquad k\in\mathbb{N},
\end{equation}
with analytic functions $h_j:S(-\kappa(S),\kappa(S))\rightarrow\mathbb{C}$ and constant, pairwise commuting matrices $M_j$ on $\mathcal{K}\otimes\mathcal{K}$, appropriately chosen such that $S\in\mathcal{S}_0^-$. In particular, we have $[S(\zeta_1),S(\zeta_2)]=0$ for $S(\zeta_i)$ of the form (\ref{form}). Then, the following can be shown.
\begin{lemma}\label{lem}
Let $\dim\mathcal{K}\geq 1$ and $S\in\mathcal{S}_{0}^-$, being of the form (\ref{form}). Then, there exists an analytic function $\rho:S(-\kappa(S),\kappa(S))\rightarrow\mathcal{L}(\mathcal{K}\otimes\mathcal{K})$\footnote{Linear operators on $\mathcal{K}\otimes\mathcal{K}$.} (the phase shift matrix) such that $S(\zeta)=-F\cdot e^{2i\rho(\zeta)}$, \mbox{$\zeta\in S(-\kappa(S),\kappa(S))$}, which we fix uniquely by $\rho(0)=0$. Moreover, the multiplication operators corresponding to the functions
\begin{equation}\label{I}
\mathcal{I}_0:=1,\qquad\mathcal{I}_1(\zeta):=1,\qquad\mathcal{I}_2(\zeta_1,\zeta_2):=-e^{i\rho(\zeta_1-\zeta_2)},
\end{equation}
denoted by the same symbols $\mathcal{I}_n$, $n=0,1,2$, fulfill the following.
\begin{itemize}
\item[a)]  Viewed as an operator on $H^2(\mathcal{T}_n(\tfrac{\kappa(S)}{n}))\otimes\mathcal{K}^{\otimes n}$, $\mathcal{I}_n$, $n=0,1,2$, is a bounded map with operator norm $\|\mathcal{I}_n\|\leq\|S\|_{\kappa(S)}^{n/2}$.
\item[b)]  Viewed as an operator on $L^2(\mathbb{R}^n)\otimes\mathcal{K}^{\otimes n}$, $\mathcal{I}_n$, $n=0,1,2$, is a unitary intertwining the representations $D_n$ and $D_n^-$ (corresponding to $S(\theta)=-F$, for all $\theta\in\mathbb{R}$) of $\mathfrak{S}_2$. Hence, it maps the S-symmetric subspace $\mathscr{H}_n\subset L^2(\mathbb{R}^n)\otimes\mathcal{K}^{\otimes n}$ onto the totally antisymmetric subspace $\mathscr{H}_n^-\subset L^2(\mathbb{R}^n)\otimes\mathcal{K}^{\otimes n}$.
\end{itemize}
\end{lemma}
\begin{proof}
We first prove the statement on the analyticity of $\rho$. Due to the assumed form of the S-matrices, this can be done by means of an argument stated for scalar valued functions \cite[Corollary 6.17]{conway1973functions}. Namely, since $S\in\mathcal{S}_{0}^-$ is analytic and invertible in the strip $S(-\kappa(S),\kappa(S))$, also $S^{-1}\cdot S'$ is analytic on this strip. Hence, it has a primitive on $S(-\kappa(S),\kappa(S))$ which we denote by $\rho_1$. To prove this, one may proceed exactly as in the proof of \cite[Corollary 6.16]{conway1973functions} formulated for scalar valued functions. The main ingredient thereby is Cauchy's theorem which implies that $\rho_1(\zeta):=\int_{\gamma}(S^{-1}\cdot S')$, where $\gamma$ is a rectifiable curve from some fixed point $a\in S(-\kappa(S),\kappa(S))$ to $\zeta\in S(-\kappa(S),\kappa(S))$, is a well-defined function from $S(-\kappa(S),\kappa(S))$ to $\mathcal{L}(\mathcal{K}\otimes\mathcal{K})$. The fact that $\rho_1$ is a primitive of $S^{-1}\cdot S'$ follows then simply by determining $\rho_1'$.\par
Defining, next, $H(\zeta):=\exp \rho_1(\zeta)$, then $H$ is analytic and invertible. Therefore, $S\cdot H^{-1}$ is analytic and has the following derivative $S'\cdot H^{-1} +S\cdot (H^{-1})'$. But in view of the special choice (\ref{form}), $(H^{-1})'=-S^{-1}\cdot S'\cdot H^{-1}$, yielding that $S\cdot H^{-1}$ is a constant matrix $c$ for all $\zeta\in S(-\kappa(S),\kappa(S))$. That is, $S(\zeta)=c\cdot \exp\rho_1(\zeta)$. By letting $2i\rho(\zeta)=\rho_1(\zeta)+c'$ with some suitable $c'$, we may choose $\rho(0)=0$, yielding the first claim.\par
Item $a)$ is a direct consequence of the analyticity of $\rho$ and the definition of $\mathcal{I}_n$, $n=0,1,2$, (\ref{I}).\par
To show $b)$, we first prove the unitarity feature. Hence, we have to show that $\rho(\theta)^*=\rho(\theta)$. To this end, note that from the properties of the S-matrix $S(\theta)=-F\cdot e^{2i\rho(\theta)}$ we obtain $\rho(\theta)^*=\rho(\theta)+\pi k\, 1_2$, $k\in\mathbb{Z}$. But $\rho(0)^*=\rho(0)=0$ and both $\rho^*$ and $\rho$ are analytic. Hence, their difference is also an analytic function, implying $k=0$.\par
To see the intertwining property, we consider $\Psi_2\in L^2(\mathbb{R}^2)\otimes\mathcal{K}^{\otimes 2}$ and compute
\begin{eqnarray*}
\left(\mathcal{I}_2D_2(\tau_1)\Psi_2\right)(\theta_1,\theta_2)&=& -e^{i\rho(\theta_1-\theta_2)}\cdot S(\theta_2-\theta_1)\cdot\Psi_2(\theta_2,\theta_1)\\
&=& e^{i\rho(\theta_1-\theta_2)}\cdot F\cdot e^{2i\rho(\theta_2-\theta_1)}\cdot\Psi_2(\theta_2,\theta_1).
\end{eqnarray*}
With regard to the properties of $S$, we have $\rho(-\theta)=-\rho(\theta)$. Moreover, $[F,\rho(\theta)^k]=0$ for all $k\geq 0$ due to the special form (\ref{form}) of $S(\theta)=-Fe^{2i\rho(\theta)}$. Hence,
\begin{eqnarray*}
\left(\mathcal{I}_2D_2(\tau_1)\Psi_2\right)(\theta_1,\theta_2)&=& e^{i\rho(\theta_1-\theta_2)}\cdot F\cdot e^{2i\rho(\theta_2-\theta_1)}\cdot\Psi_2(\theta_2,\theta_1)\\
&=& (-F)\cdot (-e^{i\rho(\theta_2-\theta_1)})\cdot\Psi_2(\theta_2,\theta_1)\\
&=&\left(D_2^-(\tau_1)\mathcal{I}_2\Psi_2\right)(\theta_1,\theta_2),
\end{eqnarray*}
proving the intertwining property, and $\mathcal{I}_2:\mathscr{H}_2\rightarrow\mathscr{H}_2^-$.
\end{proof}
Although this result can at the moment be verified only for $0\leq n\leq 2$ there is evidence that Lemma \ref{lem} can be generalized to $n\in\mathbb{N}_0$. Considering, namely, the simplices $$\Sigma_\pi:=\{\boldsymbol{\theta}\in\mathbb{R}^n:\theta_{\pi(1)}\leq\dots\leq\theta_{\pi(n)}\}, \qquad\pi\in\mathfrak{S}_n,$$
then, the multiplication operators corresponding to the functions
\begin{equation}\label{nichtanalytisch}
\widehat{\mathcal{I}}_n(\boldsymbol{\theta}):=(-1)^{\text{sign}(\pi)}F_n^\pi\cdot S_n^\pi(\boldsymbol{\theta})^{-1},\qquad \boldsymbol{\theta}\in\Sigma_\pi,
\end{equation}
are easily seen to be unitaries intertwining the representations $D_n$ and $D_n^-$ of $\mathfrak{S}_n$ on $\bigcup_\pi\Sigma_\pi$ by taking into account $\boldsymbol{\theta}\in\Sigma_\pi$ if and only if $\boldsymbol{\theta}^\pi\in\Sigma_{id}$. Here the tensors $F_n^\pi$ have the components $\left(F_n^\pi\right)^{\boldsymbol{\alpha}}_{\boldsymbol{\beta}}=\delta^{\alpha_1}_{\beta_{\pi(1)}}\cdots\delta^{\alpha_n}_{\beta_{\pi(n)}}$. These intertwiners are, however, not suited for our purposes since they are in general not continuous. Nevertheless, due the possibility to construct such objects and with regard to Lemma \ref{lem} there is strong indication for the validity of the following conjecture which restricts to the minimal assumptions necessary with regard to the nuclearity of the map $\Xi(s)$, (\ref{xi}).
\begin{conjecture}\label{conj}
Let $S\in\mathcal{S}_0^-$. Then, there exist multiplication operators $\mathcal{I}_n$, $n\in\mathbb{N}_0$, which intertwine the representations $D_n$ and $D_n^-$ of $\mathfrak{S}_n$, and, moreover, have the following properties.
\begin{itemize}
\item[a)] $\mathcal{I}_n$ is a bounded operator on $H^2(\mathcal{T}_n(\tfrac{\kappa(S)}{n}))\otimes\mathcal{K}^{\otimes n}$. For the operator norm we further have $\|\mathcal{I}_n\|\leq \gamma^n$ with some constant $\gamma>0$ not depending on $n$.
\item[b)] $(\mathcal{I}_n)^{-1}:\mathscr{H}^-_n\ni f_{\boldsymbol{\lambda}_{\pi/2}}\mapsto (\mathcal{I}_n)^{-1}f_{\boldsymbol{\lambda}_{\pi/2}}\in\mathscr{H}$ is bounded with operator norm $\|(\mathcal{I}_n)^{-1}\|\leq \gamma'^n$, where the constant $\gamma'>0$ does not depend on $n$.
\end{itemize}
\end{conjecture}
In spite of the strong indications speaking for the validity of Conjecture \ref{conj}, there is still the possibility to circumvent this assumption by improving the nuclearity estimates on the maps $\Xi_n(s)$. For instance, in case one can show that there exists a bounded operator from the wedge algebras into a Hardy space based on a cube with \textit{fixed} side length independent of $n$ (in contrast to the present situation where we are dealing with cubes shrinking like $\tfrac{1}{n}$) with operator norm growing at most exponentially in $n$\footnote{Presently, we can prove the existence of such an operator with operator norm bounded by $\alpha^{n^2}$ for some constant $\alpha>1$.}, then the verification of the modular nuclearity condition would be possible by the applied techniques even for the whole class $\mathcal{S}_0$ of regular S-matrices. Other improvements may be achieved by taking the underlying S-symmetry of the functions in the Hilbert space $\mathscr{H}$ into account. The nuclearity estimates are in principle, namely, derived on the unsymmetrized Fock space $\widehat{\mathscr{H}}$.\par
Note, moreover, that the above discussion resulting in Conjecture $\ref{conj}$ applies analogously to S-matrices $S\in\mathcal{S}_0^+:=\{S\in\mathcal{S}_0:S(0)=+F\}$, hence, corresponding to maps $\mathcal{I}^+:\mathscr{H}\rightarrow\mathscr{H}^+$. Since in Lagrangian field theory no interacting model is known to fall into this class and since we do not yet have an argument for verifying the modular nuclearity condition in this case (not even for $\dim\mathcal{K}=1$), we specialized to the family $\mathcal{S}_0^-$. In particular, the $O(N)$-invariant nonlinear sigma-models fall into that class, cf. Chapter \ref{Chapter4}.
\begin{theorem}\label{mainTheorem}
Consider a model theory associated with a regular S-matrix $S\in\mathcal{S}_0^-$. Suppose further Conjecture \ref{conj} holds true. Then, there exists a value $s_{\text{min}}<\infty$ such that $\Xi(s):\mathcal{F}(W_R)\rightarrow\mathscr{H}$ is nuclear for all $s>s_{\text{min}}$. Thus, in these models the corresponding local algebras $\mathcal{F}(\mathcal{O}_{x,y})=\mathcal{F}(W_R+x)\cap\mathcal{F}(W_L+y)$ have $\Omega$ as a cyclic vector, for each double cone $\mathcal{O}_{x,y}=(W_R+x)\cap(W_L+y)$ with $y-x\in W_R$ and $-(y-x)^2>s_{\text{min}}^2$.
\end{theorem}
\begin{proof}
Since $\Xi_n(s)=X_n(s,\kappa(S))\,\circ\,\Upsilon_n(s,\kappa(S))$, with $0<\kappa(S)<\tfrac{\pi}{2}$, we have that $\Xi_n(s)$ is nuclear with nuclear norm bounded by
\begin{equation}\label{b1}
\|\Xi_n(s)\|_1\leq \|\Upsilon_n(s,\kappa)\|\cdot\|X_n(s,\kappa)\|_1\leq n^{n}\upsilon(\tfrac{s}{2},\kappa(S))^nx(s,\kappa(S))^n.
\end{equation}
To show nuclearity of $\Xi(s)=\sum_{n=0}^{\infty}\Xi_n(s)$, we make use of the map $\mathcal{I}:=\bigoplus_{n=0}^\infty\mathcal{I}_n:\mathscr{H}\rightarrow\mathscr{H}^-$ and consider $\Xi_n(s)$ as the concatenation
\begin{equation}\label{concat}
\mathcal{F}(W_R)\stackrel{\Upsilon_n(s,\kappa(S))}{\longrightarrow}H^2(\mathcal{T}_n(\tfrac{\kappa(S)}{n}))\otimes\mathcal{K}^{\otimes n}\stackrel{\mathcal{I}_n}{\longrightarrow}H^2_-(\mathcal{T}_n(\tfrac{\kappa(S)}{n}))\otimes\mathcal{K}^{\otimes n}\stackrel{X^-_n(s,\kappa(S))}{\longrightarrow}\mathscr{H}_n^-\stackrel{(\mathcal{I}_n)^{-1}}{\longrightarrow}\mathscr{H}_n,
\end{equation}
where $X^-_n(s,\kappa(S))$ acts in the same way as $X_n(s,\kappa(S))$. It further can be expressed as $X_n(s,\kappa(S))$ in (\ref{x}) and is, thus, nuclear too. Hence, we have with $h^-\in H^2_-(\mathcal{T}_n(\tfrac{\kappa(S)}{n}))\otimes\mathcal{K}^{\otimes n}$ and the notation used in the proof of Lemma \ref{Xn}
\begin{eqnarray*}
X_n^-(s,\kappa(S))h^-&=& \widehat{u}_{n,s/4[1-\cos(\kappa(S)/(2n))]}\sum_{\boldsymbol{\varepsilon}} \left(R_{g,-\varepsilon_1\frac{\kappa(S)}{2n}}\otimes\cdots\otimes R_{g,-\varepsilon_n\frac{\kappa(S)}{2n}}\right)\nonumber\\
&&\qquad\times\,u_{n,-s/4\sin(\kappa(S)/(2n)),\boldsymbol{\varepsilon}}\, h^-_{\boldsymbol{\lambda}_{\pi/2}-\frac{\kappa(S)}{2n}\boldsymbol{\varepsilon}}\\
&=&\sum_{\boldsymbol{\varepsilon}}\left(Z_{g,-\varepsilon_1\frac{\kappa(S)}{2n}}\otimes\cdots\otimes Z_{g,-\varepsilon_n\frac{\kappa(S)}{2n}}\right)h^-_{\boldsymbol{\lambda}_{\pi/2}-\frac{\kappa(S)}{2n}\boldsymbol{\varepsilon}},
\end{eqnarray*}
where $$Z_{g,-\varepsilon\frac{\kappa(S)}{2n}}:=\widehat{u}^{(1)}_{n,s/4[1-\cos(\kappa(S)/(2n))]}R_{g,-\varepsilon\frac{\kappa(S)}{2n}}u^{(1)}_{n,-s/4\sin(\kappa(S)/(2n)),\varepsilon}.$$
Here $\widehat{u}^{(1)}_{n,s}$ and $u^{(1)}_{n,s,\varepsilon}$ are the multiplication operators $\widehat{u}_{n,s}$, respectively, $u_{n,s,\varepsilon}$ restricted to $L^2(\mathbb{R})\otimes\mathcal{K}$. For the nuclear norm of $X^-_n(s,\kappa(S))$, on the other hand, we find a sharper bound, due to the Pauli principle. Namely, consider the positive trace class operator $\widehat{Z}_{g,b}:=\left(Z_{g,b}Z_{g,b}^*+Z_{g,-b}Z_{g,-b}^*\right)^{1/2}$, $b:=\frac{\kappa(S)}{2n}$, with $\|\widehat{Z}_{g,b}\|_1\leq 2\|Z_{g,b}\|_1\leq 2 \|R_{g,b}\|_1$ \cite{kosaki1984continuity}, for $g$ as above, and choose an orthonormal basis $\{\psi_k\}_k$ of $L^2(\mathbb{R})\otimes\mathcal{K}$ which consists of eigenvectors $\psi_k$ of $\widehat{Z}_{g,b}$. Denoting the corresponding eigenvalues by $\lambda_k\geq 0$, we have $\widehat{Z}_{g,b}f=\sum_{k=1}^{\infty}\lambda_k\langle\psi_k,f\rangle\psi_k$, $f\in L^2(\mathbb{R})\otimes\mathcal{K}$, and $\|\widehat{Z}_{g,b}\|_1=\sum_{k=1}^{\infty}\lambda_k<\infty$. On $\mathscr{H}_n^-$, the vectors
$$\Psi_{\boldsymbol{k}}^-:=\frac{1}{\sqrt{n!}}\sum_{\sigma\in\mathfrak{S}_n}(-1)^{\text{sign}(\sigma)}\psi_{\sigma(k_1)}\otimes\cdots\otimes\psi_{\sigma(k_n)}$$
form an orthonormal basis if $k_1<k_2<\cdots <k_n$, $k_1,\dots,k_n\in\mathbb{N}$. Hence, it follows
\begin{equation*}
\begin{aligned}
X_n^-(s,\kappa(S))h^-&=\sum_{\boldsymbol{\varepsilon}}\left(Z_{g,-\varepsilon_1\frac{\kappa(S)}{2n}}\otimes\cdots\otimes Z_{g,-\varepsilon_n\frac{\kappa(S)}{2n}}\right)h^-_{\boldsymbol{\lambda}_{\pi/2}-\frac{\kappa(S)}{2n}\boldsymbol{\varepsilon}}\\
&= \sum_{\boldsymbol{\varepsilon}}\sum_{k_1<\cdots <k_n}\langle Z^*_{g,-\varepsilon_1\frac{\kappa(S)}{2n}}\otimes\cdots\otimes Z^*_{g,-\varepsilon_n\frac{\kappa(S)}{2n}}\Psi_{\boldsymbol{k}}^-,h^-_{\boldsymbol{\lambda}_{\pi/2}-\frac{\kappa(S)}{2n}\boldsymbol{\varepsilon}}\rangle\Psi_{\boldsymbol{k}}^-
\end{aligned}
\end{equation*}
For the nuclear norm of $X_n^-(s,\kappa(S))$ we, therefore, find with $\|h^-_{\boldsymbol{\lambda}_{\pi/2}-\frac{\kappa(S)}{2n}\boldsymbol{\varepsilon}}\|\leq \triplenorm h^-\triplenorm$
\begin{equation*}
\begin{aligned}
\|X_n^-(s,\kappa(S))\|_1\leq\sum_{\boldsymbol{\varepsilon}}\sum_{k_1<\cdots <k_n}\|Z^*_{g,-\varepsilon_1\frac{\kappa(S)}{2n}}\otimes\cdots\otimes Z^*_{g,-\varepsilon_n\frac{\kappa(S)}{2n}}\Psi_{\boldsymbol{k}}^-\|.
\end{aligned}
\end{equation*}
Taking into account that
\begin{equation*}
\begin{aligned}
\|Z^*_{g,-\varepsilon_1\frac{\kappa(S)}{2n}}\otimes\cdots&\otimes Z^*_{g,-\varepsilon_n\frac{\kappa(S)}{2n}}\Psi_{\boldsymbol{k}}^-\|^2\\
&=\langle\Psi_{\boldsymbol{k}}^-,Z_{g,-\varepsilon_1\frac{\kappa(S)}{2n}}Z^*_{g,-\varepsilon_1\frac{\kappa(S)}{2n}}\otimes\cdots\otimes Z_{g,-\varepsilon_n\frac{\kappa(S)}{2n}}Z^*_{g,-\varepsilon_n\frac{\kappa(S)}{2n}}\Psi_{\boldsymbol{k}}^-\rangle\\
&\leq \langle\Psi_{\boldsymbol{k}}^-,\left|\widehat{Z}_{g,\frac{\kappa(S)}{2n}}\right|^2\otimes\cdots\otimes \left|\widehat{Z}_{g,\frac{\kappa(S)}{2n}}\right|^2\Psi_{\boldsymbol{k}}^-\rangle\\
&= \|\widehat{Z}_{g,\frac{\kappa(S)}{2n}}^{\otimes n}\Psi_{\boldsymbol{k}}^-\|^2,
\end{aligned}
\end{equation*}
we arrive at
\begin{equation*}
\begin{aligned}
\|X_n^-(s,\kappa(S))\|_1\leq 2^n\sum_{k_1<\cdots <k_n}\lambda_{k_1}\cdots\lambda_{k_n}&\leq\frac{2^n}{n!}\sum_{k_1,\dots,k_n}\lambda_{k_1}\cdots\lambda_{k_n}\\
&=\frac{2^n}{n!}\|\widehat{Z}_{g,\frac{\kappa(S)}{2n}}\|_1^n\leq\frac{4^n}{n!}\|R_{g,\frac{\kappa(S)}{2n}}\|_1^n.
\end{aligned}
\end{equation*}
With regard to (\ref{concat}) and the assumptions stated in Conjecture $\ref{conj}$, therefore, the following sharper bound compared to (\ref{b1}) holds true
\begin{equation}
\|\Xi_n(s)\|_1\leq\frac{n^{n}}{n!} \Big(2\,\gamma\,\gamma'\,\upsilon(\tfrac{s}{2},\kappa(S))\,x(s,\kappa(S))\Big)^n\leq \frac{1}{\sqrt{2\pi n}} \Big(2\,e\,\gamma\,\gamma'\,\upsilon(\tfrac{s}{2},\kappa(S))\,x(s,\kappa(S))\Big)^n.
\end{equation}
where we made use of Stirling's formula, i.e. $n!\geq \sqrt{2\pi n}\,n^{n}e^{-n}$. Due to the monotonous convergence of $\upsilon(\tfrac{s}{2},\kappa(S))\cdot x(s,\kappa(S))$ to zero as $s\rightarrow\infty$, there exists a value $s_{\text{min}}<\infty$ such that
\begin{equation}
\sum_{n=0}^{\infty}\|\Xi_n(s)\|_1\leq\sum_{n=0}^{\infty}\frac{1}{\sqrt{2\pi n}}\Big(2\,e\,\gamma\,\gamma'\,\upsilon(\tfrac{s}{2},\kappa(S))\,x(s,\kappa(S))\Big)^n<\infty,
\end{equation}
for all $s>s_{\text{min}}$. Therefore, the series $\sum_{n=0}^{\infty}\Xi_n(s)$ converges in nuclear norm to $\Xi(s)$, implying the nuclearity of the map $\Xi(s)$, since $\left(\mathcal{N}(\mathcal{F}(W_R),\mathscr{H}),\|\cdot\|_1\right)$ is a Banach space. Correspondingly, it follows from Theorem \ref{NuclearityCondition} that the algebras $\mathcal{F}(\mathcal{O}_{0,\undertilde{s}})$, $\undertilde{s}:=(0,s)$, $s>s_{\text{min}}$, have $\Omega$ as a cyclic vector, i.e. the Reeh-Schlieder property holds. The respective statement for double cones $\mathcal{O}_{x,y}$, with $y-x\in W_R$ and $-(y-x)^2>s_{\text{min}}^2$, is obtained by covariance.
\end{proof}
Starting from a certain family of factorizing S-matrices, the corresponding model theories were constructed. Theorem \ref{mainTheorem} states the existence of local quantum fields in these models under a very plausible conjecture. The fields are localized in bounded regions which, however, cannot be chosen arbitrarily small but have of a minimal ``relativistic size''. This minimal diameter depends on the S-matrix at hand as well as the mass gap $m_\circ$ of the theory, and manifests itself through the length $s_{\text{min}}$. There is reason to believe that a minimal localization length, namely the Planck length $l_P\approx 10^{-35}$m, exists in theories combining general relativity and quantum theory. However, the occurrence of the length $s_{\text{min}}$ in our approach has no physical motivation and is assumed to be a by-product of the various estimates. Indeed, considering as an instructive example the well-studied model of the scalar free Bose field with constant S-matrix $S=+1$, cf. \cite{Jost, araki1963lattice}, with regard to the modular nuclearity condition, then this requirement can be verified by an alternative argument for arbitrary $s>0$ as shown in \cite{BL4}. The alternative strategy yielding this result relies to a large extent on the analysis of nuclear maps on the Bose Fock space developed by Buchholz and Wichmann \cite{BW86} and can, unfortunately, not be carried over to more general situations.\par
Nevertheless, the result established in Theorem \ref{mainTheorem}, namely the existence of compactly localized quantum fields\footnote{Recall that these operators are not constructed explicitly but are characterized indirectly as elements of the local field algebras $\mathcal{F}(\mathcal{O})$.} complying with the Reeh-Schlieder property, is already sufficient to investigate collision states. The fact that the localization regions cannot be chosen arbitrarily small but above a minimal diameter does not give rise to any complications, since double cones of any size can be spacelike separated simply by translation. Thus, Haag-Ruelle-Hepp scattering theory can be applied \cite{Araki99,hepp1965} and one may prove that the models constructed by the methods presented in this chapter solve the inverse scattering problem. Indeed as shown in  \cite{LS} under the assumption that the maps $\Xi(s)$ are nuclear, the S-matrices of the constructed theories coincide with those initially considered as starting point, and, moreover, asymptotic completeness can be proven to hold. We discuss these results, established in \cite{LS}, in the next section.

\section{Scattering Operator and Asymptotic Completeness}\label{AC}
In this section we summarize the results found in \cite[Section 5]{LS}, concerning the physical properties of the models previously constructed. So far we showed\footnote{Under Conjecture \ref{conj}.} that our construction gives rise to a large class of models which comply with all fundamental concepts of quantum field theory if $S\in\mathcal{S}_0^-$. Beyond that one is particularly interested in the interaction taking place in these models. At the moment it is by no means clear how the S-matrix $S$, which served as the starting point in our approach, is related to the scattering operator of the respective model. It is, therefore, necessary to derive explicit formulae for scattering states in order to compute the scattering matrix of the theory at hand. It will turn out, though, that this operator, which maps outgoing into incoming scattering states, coincides with the tensor corresponding to $S$. That is, the construction solves the inverse scattering problem. Furthermore, the emerging models are asymptotically complete.\par
Note that in Section \ref{Sectionfactorizing} we agreed to refer to the scattering amplitude $S$ as S-matrix for short. Therefore, the operator $\textbf{S}:\mathscr{H}^+\rightarrow\mathscr{H}^+$ defined in (\ref{S}) and usually declared as S-matrix shall be referred to as scattering operator.\\

Since double cones of any size can be spacelike separated by translation, cyclicity of the vacuum vector $\Omega$ for the field algebra $\mathcal{F}(\mathcal{O})$, with $\mathcal{O}$ above a minimal diameter, cf. Theorem \ref{mainTheorem}, is sufficient in order to compute $n$-particle scattering states. Lechner and Sch\"utzenhofer, however, \textit{assume} in their analysis \cite[Section 5]{LS} the nuclearity of the maps $\Xi(s)$ for arbitrary $s>0$. This is clearly an unnecessarily strong requirement. Nevertheless, keeping in mind that spacelike separation can be achieved in any case, their results still hold true in this slightly different situation involving the minimal length $s_{\text{min}}$.\par
For the task of calculating collision states one uses the methods of Haag-Ruelle-Hepp scattering theory \cite{Araki99,hepp1965} reconciled with the wedge-locality of the fields $\phi$ \cite{BBS01} and the charge structure being present \cite{doplicher1974local}. Let us recall the basic ingredients. First af all, one usually considers quasi-local operators of the form
\begin{equation}
A_q(g_{t,m(q)})=\int d^2 x\,g_{t,m(q)}(x)U(x)A_qU(x)^{-1},
\end{equation}
where $A_q\in\mathcal{F}(\mathcal{O})$ is a field operator of definite charge $q\in\mathcal{Q}$ and $g_{t,m(q)}\in\mathscr{S}(\mathbb{R}^2)$, $t\in\mathbb{R}$, is defined in terms of $\widetilde{g}\in\mathscr{S}(\mathbb{R}^2)$ by
\begin{equation}\label{tAbh}
g_{t,m(q)}(x):=\frac{1}{2\pi}\int d^2 p\,\widetilde{g}(p)\,e^{it\left[p_0-\left(p^2_1+m(q)^2\right)^{1/2}\right]}\,e^{-ip\cdot x}.
\end{equation}
If $\widetilde{g}$ has support in a sufficiently small neighborhood of some point on the mass shell $H^+_{m(q)}$ in the sector $q$, $A_q(g_{t,m(q)})\Omega$ is an element of $L^2(\mathbb{R},d\theta)\otimes\mathcal{K}_q$, not depending on $t$. That is, in this case $A_q(g_{t,m(q)})$ creates one particle states of charge $q$ from the vacuum. Moreover,
\begin{equation}
\lim\limits_{t\rightarrow\pm}A_q(g_{t,m(q)})\Psi=A_q(g)_{\text{out}\atop\text{in}}\Psi,\qquad\lim\limits_{t\rightarrow\pm}A_q(g_{t,m(q)})^*\Psi=A_q(g)_{\text{out}\atop\text{in}}\,\hspace{-.2cm}^*\,\Psi,
\end{equation}
where $A_q(g)_{\text{in}}$ and $A_q(g)_{\text{out}}$ are the asymptotic creation operators of an incoming, respectively outgoing particle of charge $q$, being in the state $A_q(g)\Omega$. The corresponding annihilation operators are their adjoints $A_q(g)_{\text{in}}\,\hspace{-.2cm}^*$ and $A_q(g)_{\text{out}}\,\hspace{-.2cm}^*$. These  asymptotic relations hold for a certain dense set of scattering states $\Psi$ \cite{Araki99, hepp1965}. However, they can be extended to all states $\Psi$ with compact energy momentum support due to a result by Buchholz \cite{buchholz1990harmonic}.\par
In order to compute collision states, one further considers the velocity support of mass $m$ of $g\in\mathscr{S}(\mathbb{R}^2)$ defined as follows
\begin{equation}
\mathcal{V}_{m}(g):=\left\{\left(1,\frac{p_1}{\left(p^2_1+m^2\right)^{1/2}}\right):(p_0,p_1)\in\,\text{supp}\,\widetilde{g}\right\}.
\end{equation}
In an analogous manner the velocity support of a vector $\Psi_{1,q}\in L^2(\mathbb{R},d\theta)\otimes\mathcal{K}_q$ is given by
\begin{equation}
\mathcal{V}(\Psi_{1,q}):=\left\{\left(1,\frac{p_1}{\left(p^2_1+m(q)^2\right)^{1/2}}\right):(p_0,p_1)\in\,\text{supp}\,\Psi_{1,q}\right\},
\end{equation}
where $\text{supp}\,\Psi_{1,q}$ is the spectral support of $\Psi_{1,q}$. Considering such a single particle vector $\Psi_{1,q}$ of charge $q$ as given, then there is always a $A\in\mathcal{F}(\mathcal{O})$ and a test function $g\in\mathscr{S}(\mathbb{R}^2)$, whose velocity support is contained in an arbitrary small neighborhood of $\mathcal{V}(\Psi_{1,q})$, such that $\|A_q(g)\Omega-\Psi_{1,q}\|<\varepsilon$ for some $\varepsilon>0$. This is due to the cyclicity of the vacuum $\Omega$ for $\mathcal{F}(\mathcal{O})$, with $\mathcal{O}$ above a minimal size, implying that there are sufficiently many quasi local creation operators.\par
Recall further that the support of a function $g_{t,m}\in\mathscr{S}(\mathbb{R}^2)$, cf. (\ref{tAbh}), is essentially contained in $t\mathcal{V}_m(g)$ for asymptotic times, i.e. for $t\rightarrow\pm\infty$, \cite{hepp1965}. Namely, let $\chi_m$ be a smooth function which is equal to $1$ on $\mathcal{V}_m(g)$ and vanishes in the complement of an $\varepsilon$-neighborhood $\mathcal{V}^\varepsilon_m(g)$ of $\mathcal{V}_m(g)$. Then, $\hat{g}_{t,m}(x):=\chi_m(x/t)g_{t,m}(x)$ is the asymptotically dominant part of $g_{t,m}$ with support in $t\mathcal{V}^\varepsilon_m(g)$. That is, for any $N\in\mathbb{N}$, we have \mbox{$|t|^N(g_{t,m}-\hat{g}_{t,m})\rightarrow 0$} in the topology of $\mathscr{S}(\mathbb{R}^2)$ as $t\rightarrow\pm\infty$.\par
Having recalled the basic ingredients and facts about scattering theory, we may come now to the main result of Haag-Ruelle-Hepp collision theory adjusted to our setting \cite{Araki99, doplicher1974local, hepp1965}. Namely, let field operators $A_1,\dots, A_n\in\mathcal{F}(\mathcal{O})$, charges $q_1,\dots,q_n\in\mathcal{Q}$ and test functions $g_1,\dots,g_n\in\mathscr{S}(\mathbb{R}^2)$, having disjoint velocity supports $\mathcal{V}_{m(q_j)}(g_j)$, i.e. supp$\,\widetilde{g}_j\,\cap\,$supp$\,\widetilde{g}_k=\emptyset$ for $j\neq k$, be given, then
\begin{equation}\label{collision}
\lim\limits_{t\rightarrow\pm}A_{1,q_1}(g_{1,t,m(q_1)})\cdots A_{n,q_n}(g_{n,t,m(q_n)})\Omega=:(\psi_1\times\cdots\times\psi_n)_{\text{out}\atop\text{in}}
\end{equation}
exists and, moreover, depends only on $\psi_j:=A_{j,q_j}(g_j)\Omega\in L^2(\mathbb{R},d\theta)\otimes\mathcal{K}_{q_j}$ in a linear and continuous manner. In (\ref{collision}) we have used the standard notation for collision states.\par
To derive explicit formulae for these scattering states, the auxiliary field $\phi$ may be taken into account. With regard to its wedge-locality and its affiliation with $\mathcal{F}(W_L)$, one may follow the analysis of \cite{BBS01}. To this end, we introduce the notion of ordered velocity supports of say $g,g'\in\mathscr{S}(\mathbb{R}^2)$. Namely, $\mathcal{V}_m(g)$ is called a precursor of $\mathcal{V}_m(g')$, $g\prec_m g'$ in symbols, if the set of all difference vectors $\mathcal{V}_m(g')-\mathcal{V}_m(g)$ is contained in $\{0\}\times\mathbb{R}_+$.\\
The notations collected so far generalize to test functions $f\in\mathscr{S}(\mathbb{R}^2)\otimes\mathcal{K}=\bigoplus_{q\in\mathcal{Q}}\mathscr{S}(\mathbb{R}^2)\otimes\mathcal{K}_q$ in a straightforward manner by the decomposition $f=\bigoplus_{q\in\mathcal{Q}}f_q$. In particular, the velocity support $\mathcal{V}(f_q)$ is understood as the union over the velocity supports of mass $m(q)$ of all the components of $f_q$.\par
These additional ingredients allow for the statement of the following results which were established in \cite{LS}. For a proof we, therefore, refer the reader to this particular reference.
\begin{proposition}[\textbf{\cite{LS}}]\label{propLechner}
Let $f_1,\dots,f_n\in\mathscr{S}(\mathbb{R}^2)\otimes\mathcal{K}$ be test functions whose Fourier transforms $\widetilde{f}_j$ are supported in the forward light cone, and whose velocity supports are ordered, i.e. $f_1\prec\dots\prec f_n$. Then,
\begin{equation}\label{propEqnLechner}
\begin{aligned}
&\left(f_1^+\times\cdots\times f_n^+\right)_{\text{out}}=\phi(f_1)\cdots\phi(f_n)\Omega=\sqrt{n!}\,P_n\left(f_1^+\otimes\cdots\otimes f_n^+\right),\\
&\left(f_1^+\times\cdots\times f_n^+\right)_{\text{in}}=\phi(f_n)\cdots\phi(f_1)\Omega=\sqrt{n!}\,P_n\left(f_n^+\otimes\cdots\otimes f_1^+\right).
\end{aligned}
\end{equation}
Moreover, the sets of incoming and outgoing $n$-particle scattering states are total in $\mathscr{H}_n$. That is, asymptotic completeness holds.
\end{proposition}
By means of these explicit formulae for the collision states, the scattering matrix $\mathbf{S}$ can be computed. It is defined as the product of
the M\o ller operators $W_{\text{in/out}}$ which map the totally symmetrized Bose Fock space $\mathscr{H}^+=\bigoplus_{n=0}^{\infty}\mathscr{H}_n^+$ over $\mathscr{H}_1$, spanned by the asymptotic scattering states, to $\mathscr{H}$. According to (\ref{propEqnLechner}), these operators are given by
\begin{equation}\label{moller}
\begin{aligned}
W_{\text{out}}P_n^+\left(f_1^+\otimes\cdots\otimes f_n^+\right)&=P_n\left(f_1^+\otimes\cdots\otimes f_n^+\right),\\
W_{\text{in}}P_n^+\left(f_n^+\otimes\cdots\otimes f_1^+\right)&=P_n\left(f_n^+\otimes\cdots\otimes f_1^+\right),
\end{aligned}
\end{equation}
where $P_n^+$ is the orthogonal projection onto $\mathscr{H}_n^+$ and $f_1\prec\dots\prec f_n$. $W_{\text{out}}$ and $W_{\text{in}}$ are obviously well-defined linear operators with dense domains and ranges. Moreover, they extend to unitaries because
\begin{equation*}
\|P_n^+\left(f_1^+\otimes\cdots\otimes f_n^+\right)\|=\|P_n\left(f_1^+\otimes\cdots\otimes f_n^+\right)\|=\frac{1}{\sqrt{n!}}\|f_1^+\|\cdots\|f_n^+\|,
\end{equation*}
due to the ordering of the supports of the $f_j^+$. The product of these M\o ller operators gives the S-Matrix
\begin{equation*}
\mathbf{S}:=W_{\text{out}}\,^*W_{\text{in}}:\mathscr{H}^+\rightarrow\mathscr{H}^+,
\end{equation*}
considered as an operator on the Bose Fock space.
\begin{theorem}[\cite{LS}]\label{TheoremLS}
The model with S-Matrix $S\in\mathcal{S}_0^-$, defined in Section \ref{ModSection}, solves the inverse scattering problem. That is, the scattering operator $\mathbf{S}$ of the model coincides with the tensor corresponding to $S$.
\end{theorem}
This theorem was proven in \cite{LS} under the assumption of cyclicity of $\Omega$ for $\mathcal{F}(\mathcal{O})$. We do not repeat the proof here. However, we emphasize that both this result and the previous Proposition \ref{propLechner} essentially rely on the conclusions which can be drawn from, respectively are stated in Theorem \ref{mainTheorem} of the previous section.\par
Theorem \ref{TheoremLS} does not only assure that the above construction solves the inverse scattering problem but also justifies the heuristic approach to Zamolodchikov's algebra in a rigorous way. Namely, by starting from the wedge-local field $\phi$, in terms of which the Zamolodchikov operators $z^{\#}(\theta)$ obtain a spacetime interpretation, a family of asymptotically complete models emerged, complying with the concept of factorized scattering as proven by Haag-Ruelle-Hepp scattering theory.

\section{Examples}\label{Examples}
To conclude this chapter, we want to give concrete examples of S-matrices $S$ contained in the set $\mathcal{S}_0^-$. The results obtained above do not rely on an explicit form of $S$. However, in order to relate our approach of constructing models within the framework of Algebraic Quantum Field Theory to the more traditional Lagrangian one, concrete examples are of particular interest. Of course, the construction presented here gives rise to a large class of models to which a Lagrangian description or a classical counterpart is not known. On the other hand, there are indeed important integrable models to which an exact S-matrix is available also in the Lagrangian setting. These S-matrices are then obtained by exploiting conservation laws, by taking symmetries of the quantized theory and analyticity assumptions into account and by comparing with perturbative results \cite{AAR, dorey1997exact, mussardo1992off}.\\

The first class of examples we want to discuss is the simplest one, namely that of scalar-valued S-matrices, also referred to as scattering functions. In this case $\dim\mathcal{K}=1$ and the mass spectrum consists of a single mass $m>0$ only. That is, we are dealing with theories which describe a single species of neutral massive
particles. For this specific setting some of the constraints on $S\in\mathcal{S}_0$ are trivially fulfilled. These are conditions $3.)$, $5.)$, $6.)$ and $7.)$ appearing in Definition \ref{S-matrixDefinition}. In particular, the non-appearance of the Yang-Baxter equation simplifies the structure of $S$ significantly. In fact, its general form can be worked out. Having, first of all, regard to the family $\mathcal{S}$ of S-matrices as defined in \ref{S-matrixDefinition}, then we have \cite[Proposition 3.2.2]{DocL}
\begin{equation}\label{SetScalar}
\mathcal{S}=\{\zeta\mapsto\varepsilon\cdot e^{ia\sinh\zeta}\cdot\prod_{j}\frac{\sinh\beta_j-\sinh\zeta}{\sinh\beta_j+\sinh\zeta}:\,\varepsilon=\pm 1,\,a\geq 0,\,\{\beta_j\}\in\mathfrak{B}\},
\end{equation}
where the set $\mathfrak{B}$ consists of finite or infinite sequences $\{\beta_j\}\subset\mathbb{C}$ which satisfy
\begin{itemize}
\item[1.)] $0<\,\text{Im}\,\beta_j\leq\tfrac{\pi}{2}$,
\item[2.)] $\beta_j$ and $-\overline{\beta_j}$ appear the same (finite) number of times in the sequence $\{\beta_j\}$,
\item[3.)] no subsequence of $\{\beta_j\}$ has a finite limit,
\item[4.)] $\sum_k\,\text{Im}\,\tfrac{1}{\sinh\beta_j}<\infty$.
\end{itemize}
In case of $a=0$ and $\{\beta_j\}$ being a finite sequence, we obtain the subfamily of regular scattering functions $\mathcal{S}_0$ and by choosing the correct sign at zero rapidity we arrive at $\mathcal{S}_0^-$. A prominent element in the latter set is the scattering function of the Sinh-Gordon model which is the integrable model defined by the Lagrangian
\begin{equation}\label{ShGLagrangian}
\mathcal{L}_{\text{ShG}}=\frac{1}{2}\partial_\mu\phi(x)\partial^\mu\phi(x)-\frac{m^2}{g^2}\cosh\left(g\phi(x)\right),
\end{equation}
where $g\in\mathbb{R}$ is the coupling constant. Due to results obtained within perturbation
theory, the scattering function of this model is expected to be \cite{Arinshtein1979389, braden1991s}
\begin{equation}\label{ShGSMatrix}
S_{\text{ShG}}(\theta)=\frac{\sinh\theta-i\sin\tfrac{\pi g^2}{4\pi+g^2}}{\sinh\theta+i\sin\tfrac{\pi g^2}{4\pi+g^2}}.
\end{equation}
\\

Proceeding to the matrix-valued case, i.e. $\dim\mathcal{K}>1$, another class of examples is of particular interest, namely that of so-called diagonal solutions\footnote{See \cite{jimbo1986quantumr} for similar solutions emerging in the context of Toda systems.} also regarded in e.g. \cite{liguoriLetters}. In the corresponding models a spectrum of $N\in\mathbb{N}$ neutral particles of the same mass is considered. More precisely, we have $\mathcal{K}=\mathbb{C}^N$, $\overline{\alpha}=\alpha$ and $m_{[\alpha]}=m$, $\alpha=1,\dots,N$. Introducing continuous bounded functions $\omega_{\alpha\beta}:\overline{S(0,\pi)}\rightarrow\mathbb{C}$ which are analytic in $S(0,\pi)$ and satisfy
\begin{equation}
\overline{\omega_{\alpha\beta}(\theta)}=\omega_{\alpha\beta}(\theta)^{-1}=\omega_{\beta\alpha}(-\theta)=\omega_{\beta\alpha}(i\pi+\theta),
\end{equation}
the S-matrix defined by
\begin{equation}\label{diagS}
S_D(\theta)^{\alpha\beta}_{\gamma\eta}:=\omega_{\alpha\beta}(\theta)\delta^\alpha_\eta\delta^\beta_\gamma,\qquad\text{(no summation over $\alpha,\beta$)},
\end{equation}
meets all the requirements of Definition \ref{S-matrixDefinition} and belongs to the set $\mathcal{S}_0$. This follows, on the one hand, from the constraints on the functions $\omega_{\alpha\beta}$, which ensure the analytic properties as well as those of unitarity, hermitian analyticity and crossing symmetry. Due to its special form (\ref{diagS}), on the other hand, $S$ is a solution of the Yang-Baxter equation and complies with the conditions of translational-, PCT- and gauge invariance. Choosing, additionally, $\omega_{\alpha\beta}(0)=-1$, we have $S_D\in\mathcal{S}_0^-$.\par
Finally, note that the scattering functions, i.e. the scalar-valued S-matrices, discussed above constitute explicit examples for the functions $\omega_{\alpha\beta}$ if we put $\omega_{\alpha\beta}=\omega_{\beta\alpha}$.\\

The most prominent matrix-valued S-matrices, belonging to the family $\mathcal{S}_0^-$ of regular S-matrices satisfying $S(0)=-F$, are those corresponding to the $O(N)$-invariant nonlinear sigma-models in two spacetime dimensions. The above construction which takes place within the framework of Algebraic Quantum Field Theory is, in fact, the first rigorous one yielding these interesting and multifaceted models up to a very plausible conjecture. Due to their great importance, particularly in connection with four-dimensional non-Abelian gauge theories, we want to stress their accessibility within our approach by devoting to them the next chapter.

% Chapter4

\chapter{$O(N)$-Invariant Nonlinear $\sigma$-Models} % Main chapter title

\label{Chapter4} % Change X to a consecutive number; for referencing this chapter elsewhere, use \ref{ChapterX}
\fancyhead[LE,RO]{\thepage}
\fancyhead[LO]{\thesection. \emph{\rightmark}}
\fancyhead[RE]{Chapter 4. \emph{$O(N)$-Invariant Nonlinear $\sigma$-Models}}
\renewcommand{\chaptermark}[1]{ \markboth{#1}{} }
\renewcommand{\sectionmark}[1]{ \markright{#1}{} }
%\lhead{Chapter 4. \emph{$O(N)$-Invariant Nonlinear $\sigma$-Models}} % Change X to a consecutive number; this is for the header on each page - perhaps a shortened title

In the previous chapter we introduced a method to construct a large class of integrable quantum field theories on two-dimensional Minkowski space  by means of operator-algebraic techniques. The main input into the construction is a factorizing S-matrix. Concrete models arising within this inverse scattering approach are $O(N)$-invariant nonlinear $\sigma$-models to which we devote this chapter.
\section{General Overview}
$O(N)$-invariant nonlinear sigma-models in two dimensions can be viewed as theoretical laboratories for studying more realistic theories. They have extensive applications in experimentally-realizable condensed matter systems due to their integrability, and they share many common features with four-dimensional non-Abelian gauge theories, such as (conjectured) asymptotic freedom, instanton solutions or renormalizability. Therefore, there has been a lot of interest in these models and they have been analyzed from various points of view.\par
Classically, the models describe the interaction of (spin) fields $\phi=(\phi_1,\dots,\phi_N)$ which take values in the $(N-1)$-dimensional unit sphere $S^{N-1}$, i.e. $\phi\cdot\phi=1$. The dynamics is governed by the action
\begin{equation}
\widehat{S}=\frac{1}{2g^2}\int d^2x \,\partial_\mu \phi\cdot\partial_\mu \phi,
\end{equation}
where $g$ is a dimensionless coupling constant. Indeed, the Lagrangian is that of the free field, but the mere presence of the constraint $\phi\cdot\phi=1$ implies interaction. A remarkable property of these models is that an infinite number of local \cite{P77} and nonlocal \cite{L78} classical conservation laws survive quantization. The existence of such conserved charges in the quantum theory imply the absence of particle production in scattering processes. Moreover, under the assumption that the theory has a mass gap and the spectrum consists of one stable $O(N)$–vector multiplet it is shown in \cite{L78} that the two-particle scattering matrix is (up to CDD ambiguities) the one previously proposed by the Zamolodchikov brothers \cite{Zam78} for general $N\geq 3$. The postulate of a mass gap, however, is
supported by the well-known fact that, in contrast to the $O(N)$-symmetry, the scale invariance of the classical theory is broken after quantization by the conformal anomaly. The expressions for the factorizing S-matrix have also been verified to $O(1/N^2)$ in the $1/N$-expansion \cite{Berg:1978}. In the special case of the $O(3)$-model numerical simulations exist \cite{LW90}, confirming Zamolodchikovs' result for low energies. In particular, the data was in accordance with the nonperturbative property that at zero energy the S-matrix is repulsive, i.e. $S(0)=-F$. Furthermore, the absence of bound states, assumed in the construction of the S-matrix, was shown to hold again by means of the $1/N$-expansion \cite{BaLS76} and also in a semi-classical approach \cite{BZ77}.\par
In 1976 Brézin, Zinn-Justin and Le Guillou proved in the framework of dimensional regularization that the $O(N)$ sigma-models are perturbatively renormalizable \cite{BZG76}. However, this result relies on the introduction of a symmetry breaking term which makes the theory infrared finite. Nevertheless, Elitzur's conjecture \cite{Eli83}, that $O(N)$-invariant correlation functions are infrared finite order by order in perturbation theory, was proven later on by David \cite{Dav81}.\par
Regarding $O(N)$-invariant $\sigma$-models as toy analogues of QCD, one is interested in the property of asymptotic freedom. Again perturbative results \cite{P75} state that the models under consideration indeed exhibit this feature. However, these perturbative findings have faced serious objections in the last decades. The existence of superinstanton solutions in these models were shown to be closely related to the failure of perturbation theory to produce unique results \cite{Sei1P95}. This is, of course, only one point of criticism. In fact, since no rigorous proof for or against asymptotic freedom exists, there is an on-going controversy concerning this topic and we refer to \cite{Sei03} for an overview about this discussion.\par
The obvious necessity of a nonperturbative approach to the $O(N)$ sigma-models led to the attempt to compute quantities, such as $2$-point functions of local operators, within the form factor program \cite{Smir92, KW78}. There one starts from the knowledge of the stable particle spectrum and their S-matrix. This inverse scattering point of view is indeed a more convenient concept as it bypasses all the problems related to the quantization of the classical Lagrangian, perturbation theory and renormalization. For the case $N=3$ several form factors are explicitly known \cite{BN97}. However, the $n$-point functions are given as an infinite series of integrals over form factors and one has to investigate its convergence. Presently, only extrapolating results, based on the explicitly known form factors, suggest the convergence of this series.\par
Despite the extensive analysis performed for these QCD toy models, no mathematically sound description can be found in the literature. However, the Zamolodchikov S-matrix can be shown to comply with the requirements of Definition \ref{S-matrixDefinition}. Hence, a rigorous construction of $O(N\geq 3)$-invariant nonlinear $\sigma$-models on two-dimensional Minkowski space can be achieved by means of operator-algebraic techniques.
\section{Construction in AQFT}
In order to make contact with the construction of models carried out in Chapter \ref{Chapter3}, we first clarify the particle spectrum of the $O(N)$ nonlinear $\sigma$-models. To this end, we shall use the notation introduced previously.\par
The theory describes a single species of neutral massive particles with an internal degree of freedom. The global gauge group $G=O(N)$ acts on $\mathcal{K}=\mathbb{C}^N$, $N\geq 3$, via its defining self-conjugate irreducible representation. Hence, we have, in particular, $\overline{\alpha}=\alpha$ and, moreover, $m_{[\alpha]}=m$, $\alpha=1,\dots,N$.\par
The derivation of the $O(N)$ nonlinear $\sigma$-model S-matrix relies on the existence of a stable $O(N)$-vector multiplet of massive particles with equal masses $m$. As shown by the Zamolodchikov brothers \cite{Zam78}, by exploiting the $O(N)$-symmetry, the corresponding S-matrices can been determined up to CDD ambiguities and the maximal analytic solutions are of the form
\begin{equation}\label{s-matrixsigma}
\begin{aligned}
S_N(\theta)^{\alpha\beta}_{\gamma\eta}&=\sigma_1(\theta)\delta^{\alpha}_\beta\delta^\gamma_\eta+\sigma_2(\theta)\delta^\alpha_\gamma\delta^\beta_\eta+\sigma_3(\theta)\delta^\alpha_\eta\delta^\beta_\gamma,\\
&=\sigma_1(\theta)P^{\alpha\beta}_{\gamma\eta}+\sigma_2(\theta)\left(1_2\right)^{\alpha\beta}_{\gamma\eta}+\sigma_3(\theta)F^{\alpha\beta}_{\gamma\eta},
\end{aligned}
\end{equation}
with functions $\sigma_k:\mathbb{R}\rightarrow\mathbb{C}$, $k=1,2,3$, given by
\begin{equation}
\begin{aligned}
\sigma_2(\theta)&=\frac{\Gamma\left(\tfrac{1}{N-2}-i\tfrac{\theta}{2\pi}\right)\Gamma\left(\tfrac{1}{2}-i\tfrac{\theta}{2\pi}\right)\Gamma\left(\tfrac{1}{2}+\tfrac{1}{N-2}+i\tfrac{\theta}{2\pi}\right)\Gamma\left(1+i\tfrac{\theta}{2\pi}\right)}{\Gamma\left(\tfrac{1}{2}+\tfrac{1}{N-2}-i\tfrac{\theta}{2\pi}\right)\Gamma\left(-i\frac{\theta}{2\pi}\right)\Gamma\left(1+\tfrac{1}{N-2}+i\tfrac{\theta}{2\pi}\right)\Gamma\left(\tfrac{1}{2}+i\tfrac{\theta}{2\pi}\right)},\\
\sigma_1(\theta)&=-\frac{2\pi i}{N-2}\cdot\frac{\sigma_2(\theta)}{i\pi-\theta},\\
\sigma_3(\theta)&=\sigma_1(i\pi-\theta).
\end{aligned}
\end{equation}
The operator $P$ is a projection on the one-dimensional $O(N)$-invariant space $\left(\mathbb{C}^N\otimes\mathbb{C}^N\right)^{O(N)}$ given by
\begin{eqnarray}
P:\mathbb{C}^N\otimes\mathbb{C}^N\rightarrow\mathbb{C}^N\otimes\mathbb{C}^N,\qquad P(u\otimes v):=\frac{1}{N}\left(u,v\right)\sum_{k=1}^{N}e_k\otimes e_k,
\end{eqnarray}
where $\{e_k\}_{k=1,\dots,N}$ is some orthonormal basis of $\mathbb{C}^N$. The operator $F$, on the other hand, is the flip operator introduced earlier in Chapter \ref{Chapter3}, i.e.
\begin{equation}
F:\mathbb{C}^N\otimes\mathbb{C}^N\rightarrow\mathbb{C}^N\otimes\mathbb{C}^N,\qquad F(u\otimes v):=v\otimes u.
\end{equation}
It is straightforward to verify that the operators $P,F$ and $1_2$ commute and are linearly independent. In fact, by Theorem 10.1.6 in \cite{goodman2009symmetry} these operators span the endomorphism algebra End$_{O(N)}(\mathbb{C}^N\otimes\mathbb{C}^N)$, consisting of linear transformations on $\mathbb{C}^N\otimes\mathbb{C}^N$ which commute with the group action.\par
Besides this gauge symmetry the S-matrix (\ref{s-matrixsigma}) does further comply with the properties of unitarity, hermitian analyticity, crossing symmetry, translational-, TCP-invariance and is a solution of the Yang-Baxter equation. Moreover, it extends to a bounded and analytic function on the strip $\{\zeta\in\mathbb{C}:-\kappa<\text{Im}\,\zeta<\pi+\kappa\}$, where $\kappa=\frac{2\pi}{N-2}-\varepsilon$ with $\varepsilon>0$. That is, we have
\begin{proposition}[\cite{LS}]
The S-matrix $S_N$ defined in (\ref{s-matrixsigma}) is an S-matrix in the sense of Definition \ref{S-matrixDefinition} for the particle spectrum given by $G=O(N)$, $V_1=id$ and $m>0$. It further belongs to the class $\mathcal{S}_0^-$ of regular S-matrices, satisfying $S_N(0)=-F$.
\end{proposition}
The significance of the previous proposition lies in the fact that the results of Theorem \ref{mainTheorem} apply also to the very interesting and multifaceted models considered in this chapter. In addition, since $S_N$ (\ref{s-matrixsigma}) is of the form (\ref{form}), Lemma \ref{lem} holds, in particular, for the models at hand. Thus, in view of the special structure of the S-matrix $S_N$ there is strong indication for the validity of Conjecture \ref{conj} in the present setting. We further expect that by exploiting the underlying $O(N)$-symmetry not only Conjecture \ref{conj} may be shown to hold true but also a suitable generalization of Lemma \ref{lem} to $n>2$, involving \textit{unitary} intertwiners, is possible.\par
Concretely, one ansatz for constructing intertwiners is given by reducing the $O(N)$-action. This can be done in a straightforward manner. However, already $\mathbb{C}^3\otimes\mathbb{C}^3\otimes\mathbb{C}^3$ is not multiplicity free, complicating the diagonalization process of the tensors $S_n^\pi$, cf. (\ref{tensor}), $n\geq 3$.\par
On the other hand, one can take advantage of the well studied structure of the endomorphism algebra $\text{End}_{O(N)}((\mathbb{C}^N)^{\otimes n})$ \cite{goodman2009symmetry} and look for multiplicative intertwiners which are elements of this algebra. More precisely, for the corresponding functions one could put
\begin{equation}\label{11}
\mathcal{I}_n(\boldsymbol{\theta}):=\sum_{i=1}^{d}\alpha_i(\boldsymbol{\theta})B_i,\qquad\boldsymbol{\theta}\in\mathbb{R}^n,
\end{equation}
with scalar-valued $\alpha_i$, $i=1,\dots,d=\tfrac{(2n)!}{2^nn!}$ which have to be determined. The operators $B_i$ span $\text{End}_{O(N)}((\mathbb{C}^N)^{\otimes n})$ and are explicitly known by Theorem 10.1.6 in \cite{goodman2009symmetry}. Since $S_N$ is invariant under the action of the group $O(N)$, the tensors $S_n^\pi$ are also of the form (\ref{11}) with appropriate coefficient functions. Moreover, $F_n^\pi$ is an element of $\{B_1,\dots,B_d\}$. Hence, one can easily write down a set of conditions on the functions $\alpha_i$ to ensure the intertwining property
\begin{equation*}
\mathcal{I}_n(\boldsymbol{\theta})\, S_n^\pi (\boldsymbol{\theta})=(-1)^{\text{sign}(\pi)}\,F_n^\pi\,\mathcal{I}_n(\boldsymbol{\theta}^\pi)
\end{equation*}
and additional features such as, for instance, unitarity. In view of (\ref{nichtanalytisch}) nontrivial solutions  to these requirements exist. What remains open is to show the existence of an \textit{analytic} solution which is not clear from the outset. Nevertheless, we expect that this result can be achieved as we have reduced the problem of investigating matrix-valued functions to one analyzing scalar-valued functions. In particular, in the case $n=2$ the functions defined by (\ref{I}) can be realized as in (\ref{11}) since the matrix $\rho$, related to the S-matrix by $S_N(\theta)=-F\cdot e^{2i\rho(\theta)}$, is invariant under the action of $O(N)$. Hence, for $n=2$ there do exist analytic coefficient functions $\alpha_i$ in agreement with Lemma \ref{lem}.\par
In conclusion, due to the underlying $O(N)$-symmetry there are concrete options to verify Conjecture \ref{conj}, which are expected to be effective. Hence, up to this very plausible assumption, our analysis gives rise to the $O(N)$-invariant nonlinear $\sigma$-models in a rigorous way apart from perturbation theory and renormalization. It, moreover, demonstrates the great potential behind the algebraic approach to quantum field theory.

% Chapter4

\chapter{Deformations of Quantum Field Theories} % Main chapter title

\label{Chapter5} % Change X to a consecutive number; for referencing this chapter elsewhere, use \ref{ChapterX}
\fancyhead[LE,RO]{\thepage}
\fancyhead[LO]{\thesection. \emph{\rightmark}}
\fancyhead[RE]{Chapter 5. \emph{Deformations of QFTs}}
\renewcommand{\chaptermark}[1]{ \markboth{#1}{} }
\renewcommand{\sectionmark}[1]{ \markright{#1}{} }
%\lhead{Chapter 5. \emph{Deformations of QFTs}} % Change X to a consecutive number; this is for the header on each page - perhaps a shortened title

In Chapter \ref{Chapter3} we discussed the construction of quantum field theoretic models from an inverse scattering point of view. Thereby, we restricted our attention to two spacetime dimensions. The main intention behind this limitation was the possibility to work with a rather simple class of initial S-matrices, namely those of factorizing type, cf. Section \ref{Sectionfactorizing}. The properties of the S-matrix, in particular, the crossing symmetry, readily yield wedge-local auxiliary fields and one may pass to the von Neumann algebras they generate. Moreover, these fields are examples of temperate polarization-free generators \cite{schroer1999modular, BBS01} and, hence, can be used to calculate two-particle scattering amplitudes in Haag-Ruelle collision theory. One may ask at this point the natural question how to generalize the described procedure to higher spacetime dimensions. Taking a closer look at the construction of the wedge-local objects, one notices that the key ingredient giving rise to the nontrivial models reflects in a ``deformed'' algebraic structure when compared with the free field case. This suggests that a potential generalization to $d>2$ is achieved by a certain modification of the free field theory. These considerations benefit from recently developed construction procedures \cite{grosse2007wedge, buchholz2008warped, buchholz2011warped}, relying on deformation techniques. Thereby, one starts from a well-known model which is subjected to a certain modification. Building on these first examples of deformations as to be understood here, it was shown in \cite{GL, alazzawi2013deformations} that the above considerations can, indeed, be implemented. The applied deformation techniques allow for the nonperturbative construction of new quantum field theoretic models with nontrivial scattering in $d\geq 2$ spacetime dimensions. The corresponding field operators are not localized in bounded regions but are wedge-local. Due to this remaining localization property, scattering theory can be applied and the two-particle scattering matrix can be determined \cite{BBS01}. The form of the resulting scattering operator is very simple. It does not allow for particle creation or momentum transfer in collision processes of particles. Due to the weakened locality requirements this simple structure of the scattering operator does not contradict the no-go theorems \cite{aks1965proof, coleman1967all} discussed in Section \ref{Sectionfactorizing}, which are stated for \textit{local} theories.\par
The fact that deformation procedures give rise to wedge-local, nontrivial models in a nonperturbative way attracted attention. After their first appearance in the framework of quantum field theories on noncommutative Minkowski space \cite{grosse2007wedge}, generalizations to an operator-algebraic setting \cite{buchholz2008warped} were developed and go under the name of \textit{warped convolutions}.  These particular deformations were considered with regard to Wightman field theories and related to a modification of the tensor product on the underlying Borchers-Uhlmann algebra \cite{grosse2008noncommutative}. Further, their connection to Rieffel's deformation quantization \cite{rieffel1993deformation} was investigated in \cite{buchholz2011warped}. Deformations by warped convolutions were also applied to conformal field theories \cite{dybalski2011asymptotic} and, moreover, to quantum field theories on curved spacetimes \cite{dappiaggi2011deformations}. At their basis is the action of the translation group. However, the method of warped convolutions can be extended to actions of the special conformal group as was shown in \cite{much2012wedge} and can also be applied in a nonrelativistic setting \cite{much2014quantum}.\par
Deformations of a more general nature were later on developed in \cite{GL}, providing a generalization of the inverse scattering approach discussed in Chapter \ref{Chapter3} to higher spacetime dimensions. The considered techniques give rise to nontrivial wedge-local theories with particle spectra consisting of only a single species of massive neutral particles. In two spacetime dimensions the two-particle scattering operators of the models obtained by these methods fall into a certain subclass of possible S-matrices $S\in\mathcal{S}$ satisfying the properties \ref{S-matrixDefinition} required for the inverse scattering approach. This subclass involves only scattering \textit{functions} $S$ with value $+1$ at zero rapidity parameter, i.e. $S(0)=1$. We shall show in this and the subsequent chapter that the more interesting class of models, corresponding to S-matrices with $S(0)=-1$ and not fitting into the deformation scheme of \cite{GL}, can be included into the deformation context. Our construction gives, for instance, rise to the Sinh-Gordon model. The respective results were published in \cite{alazzawi2013deformations}.\par
Similar deformation methods were considered in \cite{plaschke2013wedge} for constructing wedge-local fields with anyonic statistics. Further contributions being concerned with deformation procedures in different contexts are \cite{bostelmann2013operator, soloviev2014wedge, much2014relativistic}.\par
In this chapter, we shall be concerned with the deformation of the model of a scalar massive Fermion. Our investigation complements the results of \cite{GL} as explained above. In contrast to the situation existent in \cite{GL}, the starting point of our analysis, the undeformed model, is, except in two spacetime dimensions, nonlocal from the outset. However, remnants of locality can be found as shown in \cite{buchholz2007string}.\par
Moreover, we also comment on the bosonic case. In particular, we point out that the analysis in \cite{GL} relies on unnecessarily restrictive assumptions and that a slightly more general result can be obtained (see Lemma \ref{lem2} below).\par
This chapter is organized as follows. We start by reviewing the first examples of deformations as to be understood here. Then, the model of a scalar massive Fermion is introduced and its main features are collected. The deformation process is carried out in Section \ref{3}, including an analysis of the properties of the deformed model. In particular, we show that it exhibits nontrivial scattering by computing the two-particle collision states.\par
We postpone the discussion of the two-dimensional case to the subsequent Chapter \ref{Chapter6} where we associate the deformed model with integrable models and establish a connection with the construction carried out in Chapter \ref{Chapter3}.

\section{Development and First Examples}\label{Develop}
In this section we introduce the method of deforming quantum field theoretic models as an approach to construct new nontrivial models in a nonperturbative way and review its development.\par
This recent concept is motivated by the study of quantum field theories on noncommutative Minkowski space \cite{doplicher1995quantum, szabo2003quantum}, where the coordinate operators $X_0,\dots,X_{d-1}$, $d\geq 2$, satisfy commutation relations of the form
\begin{equation*}
\left[X_\mu,X_\nu\right] =iQ_{\mu\nu}\cdot 1,\qquad \mu,\nu=0,\dots,d-1,
\end{equation*}
with $Q$ being a numerical, skew-symmetric $(d\times d)$-matrix. The possibility of transferring a scalar free field $\phi_0$ from ordinary Minkowski space to this noncommutative space, as noticed by Grosse and Lechner \cite{grosse2007wedge}, constitutes the first example of a deformation as to be understood in this thesis. The new quantum field $\phi_Q$ arises from a deformation of the creation and annihilation operators representing the canonical commutation relations on Fock space. This procedure is $Q$-dependent and changes the underlying algebraic structure in a nontrivial way. More precisely, in case of the deformed creation operators $a^{*}_Q(p)$, for instance, we have the following exchange relations
$$a^{*}_Q(p)a^{*}_Q(q)=e^{-ipQq}a^{*}_Q(q)a^{*}_Q(p).$$
The fields $\phi_Q$ are a priori neither local nor covariant. However, as shown by Grosse and Lechner in \cite{grosse2007wedge}, under certain requirements on the matrix $Q$ there does exist a family of fields $\phi_Q$ being fully Poincaré covariant and fulfilling locality in a weakened form. The localization regions are in that case wedges, cf. Section \ref{wedges}. Due to this remnant of locality scattering theory can be applied and the two-particle S-matrix can consistently be computed \cite{BBS01}. Explicitly, for the model corresponding to the field $\phi_Q$ one finds a nontrivial, $Q$-dependent scattering matrix of the simple form $e^{ip^{\mu}Q_{\mu\nu}q^{\nu}}$ \cite{grosse2007wedge}. Thus, the models emerging from deformation of the scalar free field can be interpreted as wedge-local, nontrivial quantum field theories on ordinary $d\geq 2$-dimensional Minkowski space.\par
These compelling results attracted the attention of Buchholz and Summers who generalized the deformation concept to that effect that a deformation scheme was developed which can be applied to \textit{any} quantum field theory in its vacuum representation \cite{buchholz2008warped}. Their approach by so-called ``warped convolutions'' gives rise to a covariant, wedge-local quantum field theory whose two-particle S-matrix differs likewise from the initial one by phase factors. Despite the analogue results to \cite{grosse2007wedge} the deformation method by warped convolutions does not originate from noncommutative spacetime. The starting point of the analysis is rather constituted by warped convolution integrals of the form
\begin{equation*}
A_Q:=\int T(Qp)A_0T(Qp)^{-1}dE(p)
=(2\pi)^{-d}\int\int dp\,dx\, e^{-ip\cdot x}T(Qp)AT(Qp)^{-1}T(x).
\end{equation*}
The definition of this oscillatory integral is nontrivial and demands certain regularity properties of the operators $A_0$ acting on a Hilbert space. On the latter there is a unitary strongly continuous representation $T$ of the translation group $(\mathbb{R}^d,+)$. $dE$ denotes the corresponding spectral measure and $Q$ again a skew-symmetric $(d\times d)$-matrix which takes the role of the deformation parameter.\par
In \cite{grosse2008noncommutative}, in a concrete Wightman setting this generalized deformation procedure was shown to be related to a twisted tensor product on the underlying Borchers-Uhlmann algebra $\underline{\mathscr{S}}$ of test functions \cite{borchers1962structure, uhlmann1962definition}. Introducing, namely, a new product on this algebra is equivalent to modifying the representing field operators, i.e. deforming $\phi_0\rightarrow\phi_Q$.\par
Shortly after these insights it was shown in \cite{buchholz2011warped} that deformations by warped convolutions always constitute a representation of Rieffel's deformed algebra. Thereby, Rieffel-deformations are certain alterations of the product in a $C^*$-algebra with a strongly continuous automorphic action of the translations \cite{rieffel1993deformation}.\par
Further progress in the direction of more general deformation procedures beyond warped convolutions and Rieffel-deformations respectively was made by Lechner in \cite{GL}. The considered starting point there is the Borchers-Uhlmann algebra $\underline{\mathscr{S}}$ which is subjected to a deformation. The latter is based on linear homeomorphisms $\rho:\underline{\mathscr{S}}\rightarrow\underline{\mathscr{S}}$ with $\rho(1)=1$ and $\rho(f)^*=\rho(f^*)$, $f\in\underline{\mathscr{S}}$, which endow $\underline{\mathscr{S}}$ with a family of new products $\otimes_{\rho}$ via
$$f\otimes_{\rho}g:=\rho^{-1}(\rho(f)\otimes\rho(g)),\qquad f,g\in\underline{\mathscr{S}}.$$
Requiring a certain compatibility between $\rho$ and a state $\omega$ on $\underline{\mathscr{S}}$, namely
\begin{equation}\label{komp1}
\omega(f\otimes_{\rho}g)=\omega(f\otimes g),\qquad f,g\in\underline{\mathscr{S}},
\end{equation}
the representation spaces arising from GNS construction are identical for the deformed and undeformed case, simplifying the analysis. Moreover, assuming that the deformation maps $\rho$ act multiplicatively in momentum space, i.e.
$$\widetilde{\rho(f)}_n(p_1,\dots,p_n):=\rho_n(p_1,\dots,p_n)\cdot \widetilde{f}_n(p_1,\dots,p_n),$$
the compatibility requirement (\ref{komp1}) for quasi-free, translationally invariant states $\omega$ yields explicit conditions on the functions $\rho_n\in C^\infty(\mathbb{R}^{nd})$, $n\in\mathbb{N}_0$. The representing operators arising from GNS construction are in the deformed case again wedge-local and covariant. The deformed model admits a particle spectrum consisting of only a single species of massive neutral particles. Moreover, there is nontrivial interaction with two-particle S-matrix also being of a simple form and not allowing for particle production or momentum transfer in collision processes. Hence, the constructed models are not physically realistic in higher spacetime dimensions. However, in case $d=2$ this deformation method gives rise to a large class of wedge-local models discussed in Chapter \ref{Chapter3}. In fact, if the construction carried out in \cite{L08} can fully be repaired, see also footnote \ref{footnote}, then the described deformation procedure would yield \textit{local} integrable quantum field theories with scalar-valued S-matrices fulfilling the properties of Definition \ref{S-matrixDefinition} and \ref{regularS}, \cite[Theorem 6.1]{GL}. However, one arrives only at a certain subclass, namely $\mathcal{S}^+_0=\{S\in\mathcal{S}_0:S(0)=+1\}$, of the explicitly known set $\mathcal{S}_0\subset\mathcal{S}$ (\ref{SetScalar}).\par
By combining the presented techniques, in particular, those of \cite{grosse2007wedge} and \cite{GL}, we shall show in this and the following chapter that the more interesting class of models, corresponding to S-matrices in the family $\mathcal{S}^-_0=\{S\in\mathcal{S}_0:S(0)=-1\}$, can also be obtained by deformation methods.

\section{The Model of a Scalar Massive Fermion} \label{fermi}
\subsection{The Model}\label{model}
This section is devoted to the specification of the model which describes a scalar massive Fermion. The model at least goes back to the 1960's and can be found in R. Jost's book \cite[p. 103]{Jost} in connection with weak local commutativity of field operators.\par
In \cite{buchholz2007string}, D. Buchholz and S. Summers studied this model in more detail. In particular, they were interested in the degree of nonlocality of the model and investigated if there are any remnants of locality which have physical significance. We recall those findings which are of particular interest for our purposes.\par
To set the stage, let $\mathscr{H}^-$ denote the antisymmetric Fock space over the one-particle space $\mathscr{H}_{1}$ of a scalar particle of mass $m>0$,
\begin{equation}\label{rc}
\mathscr{H}^-=\bigoplus_{n=0}^{\infty}\mathscr{H}^-_{n},\qquad \mathscr{H}^-_{n}=\mathscr{H}_{1}\wedge\dots\wedge\mathscr{H}_{1},\nonumber
\end{equation}
and $\mathscr{H}_{0}=\mathbb{C}$ consisting of multiples of the vacuum state $\Omega$. Here, we use the following convention
\begin{equation*}
\varphi_1\wedge\dots\wedge\varphi_n:=
\frac{1}{n!}\sum_{\pi\in\mathfrak{S}_n}\sigma(\pi)\varphi_{\pi(1)}\otimes\cdots\otimes\varphi_{\pi(n)},\quad\varphi_i\in\mathscr{H}_{1},
\end{equation*}
where $\sigma(\pi)$ is $+1$ if the permutation $\pi\in\mathfrak{S}_n$ is even and $-1$ if $\pi$ is odd. Furthermore, we use the notation $$\Psi_n(\varphi_1,\dots,\varphi_n)=\sqrt{n!}\,\varphi_1\wedge\dots\wedge\varphi_n,\qquad \varphi_i\in\mathscr{H}_{1}.$$
As usual, we introduce creation and annihilation operators $a^{\#}(\varphi)$ representing the CAR algebra on the Fock space $\mathscr{H}^-$, i.e. for $\varphi,\psi\in \mathscr{H}_1$ we have
\begin{eqnarray*}
\{a^{*}(\varphi),a^{*}(\psi)\}&=&0\\
\{a(\varphi),a(\psi)\}&=&0\\
\{a(\varphi),a^{*}(\psi)\}&=&\langle \varphi,\psi\rangle\cdot 1.
\end{eqnarray*}
In the following, we shall identify the one-particle space $\mathscr{H}_{1}$ with $L^{2}(\mathbb{R}^{d},d\mu(p))$ where $d\geq 2$ and
\begin{equation*}
d\mu(p):=\omega(\vec{p})^{-1}\delta(p^0-\omega(\vec{p}))dp,\quad \omega(\vec{p})=\sqrt{\vec{p}^2+m^2},\quad m>0,\quad p=(p^0,\vec{p})\in\mathbb{R}^d.
\end{equation*}
In this setting, the Fourier transform $\widetilde{f}$ of a function $f\in\mathscr{S}(\mathbb{R}^d)$,
\begin{equation*}
\widetilde{f}(p):=\frac{1}{(2\pi)^{d/2}}\int d^dx f(x) e^{ip\cdot x},
\end{equation*}
restricted to the positive mass shell ${H_{m}^{+}}=\{p=(p^0,\vec{p})\in\mathbb{R}^d: p^0=\omega(\vec{p})\}$ is an element of $\mathscr{H}_{1}$, i.e. $\widetilde{f}|_{H_{m}^{+}}\in\mathscr{H}_{1}$. We shall, further, use the notation
\begin{equation*}
f^{\pm}(p):=\widetilde{f}(\pm p)=\frac{1}{(2\pi)^{d/2}}\int d^dx f(x) e^{\pm ip\cdot x},\qquad p\in H_{m}^{+}.
\end{equation*}
The scalar product in $\mathscr{H}^-_{n}$ is given by
\begin{equation*}
\langle\varphi_n|\psi_n\rangle=\int d\mu(p_1)\cdots d\mu(p_n)\overline{\varphi_n(p_1,\dots,p_n)}\psi_n(p_1,\dots,p_n).
\end{equation*}
The action of the annihilation and creation operators is defined by
\begin{eqnarray*}
\left(a(\varphi)\Psi\right)_{n}(p_{1},\dots,p_{n})&:=&\sqrt{n+1}\int d\mu(p) \overline{\varphi(p)}\Psi_{n+1}(p,p_{1},\dots,p_{n}),\\
\left(a^{*}(\varphi)\Psi\right)_{n}(p_{1},\dots,p_{n})&:= &\frac{1}{\sqrt{n}}\sum_{k=1}^{n}(-1)^{k+1}\varphi(p_{k}) \Psi_{n-1}(p_{1},\dots,\hat{p}_{k},\dots,p_{n}),\\
a^{*}(\varphi)\Omega:=\varphi,&\qquad& a(\varphi)\Omega:=0,
\end{eqnarray*}
where $\varphi\in\mathscr{H}_{1}$, $\Psi\in\mathscr{H}^-$ and $\hat{p}_{k}$ denotes the omission of the variable $p_k$. For further purposes we introduce the operator-valued distributions $a^{\#}(p)$ such that
\begin{equation*}
a(\varphi)=\int d\mu(p) \overline{\varphi(p)}a(p), \qquad a^{*}(\varphi)=\int d\mu(p) \varphi(p)a^{*}(p).
\end{equation*}
Their action is given by
\begin{subequations}\label{a}
\begin{eqnarray}
\left(a(p)\Psi\right)_{n}(p_{1},\dots,p_{n})&=&\sqrt{n+1}\Psi_{n+1}(p,p_{1},\dots,p_{n}),\\
 \left(a^{*}(p)\Psi\right)_{n}(p_{1},\dots,p_{n})&=&\frac{1}{\sqrt{n}}\sum_{k=1}^{n}(-1)^{k+1}\omega(\vec{p}) \delta(\vec{p}-\vec{p_k})\Psi_{n-1}(p_{1},\dots,\hat{p}_{k},\dots,p_{n}).\nonumber\\
\end{eqnarray}
\end{subequations}
The Poincaré group $\mathcal{P}_+^{\uparrow}$ is represented on $\mathscr{H}^-$ in the usual manner by the second quantized continuous unitary representation $U^-$ which leaves $\Omega$ invariant and acts according to
\begin{equation}\label{unitary}
\left(U^-(a,\Lambda)\Psi\right)_n(p_1,\dots,p_n)=e^{i\sum_{k=1}^{n}p_k\cdot a}\,\Psi_n(\Lambda^{-1}p_1,\dots,\Lambda^{-1}p_n),\qquad (a,\Lambda)\in\mathcal{P}_+^{\uparrow}.
\end{equation}
The joint spectrum of the generators of the translation group $U^-(a,1)$ is a subset of the forward light cone $\overline{V_+}$, i.e. the translation group satisfies the spectral condition. In the sequel we shall for simplicity omit the index ``$-$'' and write $U$ instead.\par
Proceeding in a standard way, we introduce an operator-valued distribution $\underline{\phi}:\mathscr{S}(\mathbb{R}^{d})\rightarrow \mathcal{B}(\mathscr{H}^-)$ which is defined by
\begin{equation}\label{feld}
\underline{\phi}(f):=a^{*}(f^{+})+a(\overline{f^{-}}).
\end{equation}
This field operator obviously satisfies canonical anticommutation relations, in particular,
\begin{equation*}
\{\underline{\phi}(f),\underline{\phi}(g)\}=\left(\langle \left(\overline{f}\right)^+|g^+\rangle+\langle \left(\overline{g}\right)^+|f^+\rangle\right)\cdot 1.
\end{equation*}
Furthermore, $\underline{\phi}$ is a weak solution of the Klein-Gordon equation, $\underline{\phi}(f)^{*}\supset\underline{\phi}(\overline{f})$ and it transforms covariantly under the adjoint action of the unitary representation $U$ of the Poincaré group, i.e.
$$U(a,\Lambda)\underline{\phi}(f) U(a,\Lambda)^{-1}=\underline{\phi}(f_{(a,\Lambda)}),\qquad (a,\Lambda)\in\mathcal{P}_+^{\uparrow}, $$
where $f_{(a,\Lambda)}(x):=f\left(\Lambda^{-1}(x-a)\right)$. Note that neither the anticommutator nor the commutator of two field operators $\underline{\phi}(f)$ and $\underline{\phi}(g)$ vanishes for spacelike separated supports of $f$ and $g$. This circumstance is consistent with the spin-statistics theorem \cite{Jost, streater2000pct} and expresses the nonlocality of the field $\underline{\phi}$. In fact, considering the corresponding net of von Neumann algebras $\{\mathcal{R}(\mathcal{O})\}_{\mathcal{O}\subset\mathbb{R}^d}$, with $\mathcal{O}$ being nonempty and open, and where $\mathcal{R}(\mathcal{O})$ is generated by the set of operators
\begin{equation*}
\{\underline{\phi}(f):f\in\mathscr{S}(\mathbb{R}^d),\, \text{supp}f\subset \mathcal{O}\},
\end{equation*}
we have
\begin{equation*}
\mathcal{R}(W)'\cap \mathcal{R}(W')=\mathbb{C}\cdot 1,\qquad W\in \mathcal{W},
\end{equation*}
by \cite[Proposition 3.5]{buchholz2007string}. That is, the model at hand is nonlocal in a strong sense.\\

We now introduce an auxiliary field $\widehat{\phi}:\mathscr{S}(\mathbb{R}^{d})\rightarrow \mathcal{B}(\mathscr{H}^-)$
\begin{equation}\label{hat}
\widehat{\phi}(f):=(-1)^{N(N-1)/2}\underline{\phi}(f)(-1)^{N(N-1)/2}=(a^{*}(f^{+})-a(\overline{f^{-}}))(-1)^{N},
\end{equation}
where $N$ is the particle number operator acting on $\mathscr{H}^-_{n}$ according to $N|_{\mathscr{H}^-_n}=n\cdot 1$. The field $\widehat{\phi}$ has the same properties as $\underline{\phi}$, in particular, it is also nonlocal. It turns out, however, that the fields $\phi$ and $\widehat{\phi}$ are relatively local, i.e. the commutator
\begin{equation}\label{relloc}
[\widehat{\phi}(f),\underline{\phi}(g)]=\left(\langle (\overline{f})^+|g^+\rangle-\langle (\overline{g})^+|f^+\rangle\right)(-1)^{N},
\end{equation}
vanishes for spacelike separated supports of the test functions $f$ and $g$ \cite{buchholz2007string}. More precisely, $\langle (\overline{f})^+|g^+\rangle-\langle (\overline{g})^+|f^+\rangle$ equals zero for spacelike separation of the supports of $f$ and $g$.\par
Moreover, we denote by $\widehat{\mathcal{R}}:\mathcal{O}\mapsto\widehat{\mathcal{R}}(\mathcal{O})$, $\mathcal{O}\subset\mathbb{R}^d$, the net generated by the field $\widehat{\phi}$. In view of the properties of the operators $\underline{\phi}$ and $\widehat{\phi}$, we have
\begin{lemma}[\cite{buchholz2007string}]
The nets $\{\mathcal{R}(\mathcal{O})\}_{\mathcal{O}\subset\mathbb{R}^d}$ and $\{\widehat{\mathcal{R}}(\mathcal{O})\}_{\mathcal{O}\subset\mathbb{R}^d}$ are $\mathcal{P}_+^{\uparrow}$-covariant, nonlocal nets which are relatively local in the sense that
\begin{equation*}
\mathcal{R}(\mathcal{O})\subset\widehat{\mathcal{R}}(\mathcal{O}')'=(-1)^{N(N-1)/2}\mathcal{R}(\mathcal{O}')'(-1)^{N(N-1)/2},
\end{equation*}
for any open $\mathcal{O}\subset\mathbb{R}^d$.
\end{lemma}

\subsection{Modular Structure}
The analysis in \cite{buchholz2007string} revealed that the vacuum vector $\Omega$ is cyclic and separating for the algebras $\mathcal{R}(W)$ and $\widehat{\mathcal{R}}(W)$, where $W$ is any wedge region. Thus, it is possible to determine the modular objects associated with the pairs $(\mathcal{R}(W),\Omega)$ and $(\widehat{\mathcal{R}}(W),\Omega)$, $W\in\mathcal{W}$. It turns out that the modular objects corresponding to $(\mathcal{R}(W),\Omega)$ and those corresponding to $(\widehat{\mathcal{R}}(W),\Omega)$ coincide. The modular operator and conjugation are given by
\begin{equation}\label{modulardata}
\Delta_W=U(\Lambda_W(i))\quad \mathrm{and}\quad J_W=U(j_W)
\end{equation}
respectively, where $\Lambda_W(t)$, $t\in\mathbb{R}$, is the one-parameter group of Lorentz boosts which leave the wedge $W$ invariant and $j_W$ is the reflection across the edge of the wedge $W$. The operator $U(j_W)$ acts according to
\begin{equation}\label{anti}
\left(U(j_W)\Psi\right)_n(p_1,\dots,p_n):=\overline{\Psi_n(-j_Wp_n,\dots,-j_Wp_1)}
\end{equation}
and extends the representation $U$ of $\mathcal{P}_+^{\uparrow}$ to a representation of $\mathcal{P}_+$. Moreover, we have $\mathcal{R}(W)'=\widehat{\mathcal{R}}(W')$.\par
In this setting the modular groups act geometrically correctly as expected from the Bisognano-Wichmann theorem, but as the model is not local the condition of geometric modular action \cite{buchholz2000geometric} is not satisfied, i.e. the modular conjugations do not act geometrically correctly.

\subsection{The 2-Dimensional Case}\label{twodimundef}
As discussed in Section \ref{ModSection} there are powerful tools available in Algebraic Quantum Field Theory which allow for analyzing the content of compactly localized operators in a model under consideration. The split property of wedge algebras is, thereby, of central significance and of crucial importance with regard to the modular nuclearity condition. However, this property cannot hold in more than two spacetime dimensions as can be inferred from an argument by Araki, cf. \cite[Section 2]{buchholz1974product}. Checking locality properties in a similar manner by means of spectral features of the corresponding modular operators without having to rely on the split property for wedges is desirable, but such a method has not been found yet.\par
By restricting the spacetime dimension to two, it is possible to show that for the model at hand the modular nuclearity condition can be verified \cite{lechner2005existence}. To this end, first note that in $d=2$ wedge-locality can be implemented by defining
\begin{subequations}\label{net}
\begin{eqnarray}
\widetilde{\mathcal{R}}(W_R+x)&:=&\{\underline{\phi}(f):f\in\mathscr{S}(W_R+x)\}'',\\
\widetilde{\mathcal{R}}(W_R'+x)&:=& \{\widehat{\phi}(f):f\in\mathscr{S}(W_R'+x)\}'',
\end{eqnarray}
\end{subequations}
where $x\in\mathbb{R}^2$. Due to the properties of the fields $\underline{\phi}$ and $\widehat{\phi}$ it is clear from this definition that the resulting net $\{\widetilde{\mathcal{R}}(W)\}_{W\in\mathcal{W}}$ is wedge-local and transforms covariantly under Poincaré transformations. In more than two dimensions, however, this approach is not meaningful because one could rotate $W_R$ into $W_R'$ and obtain by covariance an algebra $\widetilde{\mathcal{R}}(W_R')$ generated by the field $\underline{\phi}$. But as already discussed above $[\underline{\phi}(f),\underline{\phi}(g)]$ does not vanish at spacelike distances. The 2-dimensional case is special because there are no rotations mapping $W_R$ to $W_R'$.\par
In fact, within the 2-dimensional setting induced by Definition (\ref{net}) both the modular groups and the modular conjugation $J$ act geometrically correctly. Moreover, Haag duality holds, i.e. $\widetilde{\mathcal{R}}(W)'=\widetilde{\mathcal{R}}(W')$, $W\in\mathcal{W}$, \cite{BL4, DocL}.\par
As already mentioned above, it can be shown that the local algebras $\mathcal{A}(\mathcal{O})=\widetilde{\mathcal{R}}(W')\cap \widetilde{\mathcal{R}}(W+x)$, $\mathcal{O}=W'\cap(W+x)$, $x\in W'$, are nontrivial and have the vacuum $\Omega$ as a cyclic vector \cite{ lechner2005existence}. Moreover, by application of Haag-Ruelle-Hepp scattering theory it turns out that the covariant and local net $\mathcal{A}$ describes a Boson with nontrivial scattering operator $\textbf{S}=(-1)^{N(N-1)/2}$. In particular, $\textbf{S}$ is factorizing and corresponds to the two-particle scattering function $S=-1$ \cite{DocL, L08}.

\section{The Deformed Fermionic Model}\label{3}

\subsection{The Deformation Procedure}\label{deformiert}
We shall work within the framework introduced in Section \ref{model} and shall consider any spacetime dimension $d\geq 2$. Motivated by the deformation methods presented in \cite{grosse2007wedge} and \cite[Chap. 4]{GL}, our deformation approach involves first of all an operator-valued function $T_R:\mathbb{R}^{d}\rightarrow \mathcal{B}(\mathscr{H}^-)$ which is defined by
\begin{equation}\label{TR}
\left(T_{R}(x)\Psi\right)_{n}(p_{1},\dots,p_{n}):=\prod_{k=1}^{n}R(x\cdot p_{k})\Psi_{n}(p_{1},\dots,p_{n}),
\end{equation}
with $\Psi\in\mathscr{H}^-$. The function $R$, hereinafter referred to as the deformation function, should satisfy the following conditions
\begin{definition}\label{R}
A deformation function is a continuous function $R:\mathbb{R}\rightarrow\mathbb{C}$ such that the following properties hold:
\begin{description}
    \item[i)]
\begin{equation*}
R(a)^{-1}=\overline{R(a)}
\end{equation*}
\item[ii)] The Fourier transform $\widetilde{R}$ of $R$ is a tempered distribution, i.e. $\widetilde{R}\in\mathscr{S}'$, and has support in $\mathbb{R}_+$, implying that $R$ extends to an analytic function on the upper half plane.
\item[iii)] The extension of $R$ to an analytic function on the upper half plane is continuous on the closure of the latter.
\end{description}
\end{definition}
Note that the first property in Definition \ref{R} yields that $R(a)$ is a phase factor, that is, $\left|R(a)\right|=1$. Therefore, $T_R(x)$ is a unitary operator, i.e. $T_R(x)^{*}=T_R(x)^{-1}$, since by Definition (\ref{TR}) we have
\begin{equation*}
T_R(x)^{*}=T_{\overline{R}}(x),\qquad T_R(x)^{-1}=T_{R^{-1}}(x).
\end{equation*}
The requirements ii) on the Fourier transform $\widetilde{R}$ of $R$ in Definition \ref{R} imply that $R$ extends to an analytic function on the upper half plane due to Theorem IX.16 in \cite{RS2}. In particular, it follows from condition ii) that $R$ is the boundary value in the sense of $\mathscr{S}'$ of a function which is holomorphic in the upper half plane and satisfies polynomial bounds at infinity and at the real boundary. Condition iii) requires that the boundary value is even obtained in the sense of continuous functions.\par
Definition (\ref{TR}) further leads to the conclusion that for arbitrary deformation functions $R$ and $R'$ we have
\begin{equation}\label{multiplikativ}
T_{R}(x)T_{R'}(x)=T_{RR'}(x).
\end{equation}

In addition, we introduce a $(d\times d)$-matrix $Q$ which is skew-symmetric with respect to the Minkowski inner product on $\mathbb{R}^{d}$ and satisfies
\begin{equation}\label{Q}
\Lambda Q\Lambda^{-1}=\left\{\begin{array}{ccc}
Q&\mathit{for}&\Lambda\in\mathcal{L}_+^{\uparrow}\,\,\mathit{with}\,\, \Lambda W_R=W_R\\
-Q&\mathit{for}&\Lambda\in\mathcal{L}_+^{\downarrow}\,\,\mathit{with}\,\, \Lambda W_R=W_R.
\end{array}\right.
\end{equation}
The most general $Q$ satisfying $(\ref{Q})$ is known to be of the form \cite{grosse2007wedge}
\begin{equation}\label{Qform}
Q = 
\left( \begin{array}{cccc}
0 & \kappa &0&0  \\
\kappa & 0  &0&0\\
0&0&0&\kappa '\\
0&0&-\kappa '&0
\end{array} \right),\qquad 
Q = 
\left( \begin{array}{ccccc}
0 & \kappa &0&\cdots&0  \\
\kappa & 0  &0&\cdots&0\\
\vdots&\vdots&\vdots &\ddots&\vdots\\
0&0&0&\cdots &0
\end{array} \right),
\end{equation}
for $d=4$ and $d\neq 4$ respectively and with $\kappa,\kappa '\in \mathbb{R}$. Moreover, we have
\begin{equation}\label{Q'}
\Lambda Q\Lambda^{-1}=\left\{\begin{array}{ccc}
-Q&\mathit{for}&\Lambda\in\mathcal{L}_+^{\uparrow}\,\,\mathit{with}\,\, \Lambda W_R=W_R'\\
Q&\mathit{for}&\Lambda\in\mathcal{L}_+^{\downarrow}\,\,\mathit{with}\,\, \Lambda W_R=W_R'.
\end{array}\right.
\end{equation}
Having introduced the necessary notation, we may now define deformed versions of the operator-valued distributions $a^{\#}(p)$ by
\begin{equation}\label{def1}
a^{*}_{R,Q}(p):=a^{*}(p)T_{R}(Qp)^{*}, \qquad a_{R,Q}(p):=a^{*}_{R,Q}(p)^{*}.
\end{equation}
We shall need the commutation relations of $a^{\#}(p)$ and $T_{R}(x)$, which can be computed very easily. First,
\begin{equation}\label{comm}
a(p)T_{R}(x)=R(x\cdot p)T_{R}(x)a(p),
\end{equation}
which for $x=Qp$ yields that $a(p)T_{R}(Qp)=R(0)T_{R}(Qp)a(p)$ due to the antisymmetry of the matrix $Q$. Taking adjoints, we find from equation (\ref{comm})
\begin{equation*}
a^{*}(p)T_{R}(x)^{*}=\overline{R(x\cdot p)}^{-1}T_{R}(x)^{*}a^{*}(p),
\end{equation*}
respectively
\begin{equation}\label{comm2}
a^{*}(p)T_{R}(x)=R(x\cdot p)^{-1}T_{R}(x)a^{*}(p).
\end{equation}
The deformed creation and annihilation operators, therefore, satisfy the following exchange relations for arbitrary $Q$ and $Q'$
\begin{subequations}\label{exchange1}
\begin{eqnarray}
a^{*}_{R,Q}(p)a^{*}_{R,Q'}(q)&=&-\frac{R(Q'q\cdot p)}{R(Qp\cdot q)}a^{*}_{R,Q'}(q)a^{*}_{R,Q}(p),\\
a_{R,Q}(p)a_{R,Q'}(q)&=&-\frac{R(Q'q\cdot p)}{R(Qp\cdot q)}a_{R,Q'}(q)a_{R,Q}(p),
\end{eqnarray}
\begin{multline}
a_{R,Q}(p)a^{*}_{R,Q'}(q)\\
=\omega(\vec{p})\delta(\vec{p}-\vec{q}) T_{R}(Qp)T_{R}(Q'p)^{*}-
\frac{R(Qp\cdot q)}{R(Q'q\cdot p)}a^{*}_{R,Q'}(q)a_{R,Q}(p).
\end{multline}
\end{subequations}
Thus, as expected, the deformation has changed the underlying algebraic structure.\par
We may now introduce as usual corresponding field operators using the deformed creation and annihilation operators. These deformed field operators $\phi_{R,Q}(f)$ are defined by
\begin{equation}\label{def}
\phi_{R,Q}(f):=a^{*}_{R,Q}(f^{+})+a_{R,Q}(\overline{f^{-}}), \qquad f\in \mathscr{S}(\mathbb{R}^{d}),
\end{equation}
where for $\varphi\in\mathscr{H}_1$
\begin{equation*}
a_{R,Q}(\varphi)=\int d\mu(p) \overline{\varphi(p)}a_{R,Q}(p), \qquad a^{*}_{R,Q}(\varphi)=\int d\mu(p) \varphi(p)a^{*}_{R,Q}(p).
\end{equation*}
Note that if we set the deformation function $R(a)=-1$ for all $a\in\mathbb{R}$, the correspondingly deformed field operators are equal to the auxiliary fields given by Equation (\ref{hat}), i.e.
\begin{equation}\label{R-1}
\phi_{-1}(f)=\widehat{\phi}(f).
\end{equation}
For $R(a)=1$ for all $a\in\mathbb{R}$, one recovers the undeformed field $\underline{\phi}$ given by (\ref{feld}), i.e. $\phi_{1}(f)=\underline{\phi}(f)$.\par
In the same way as in the undeformed case, see Equation (\ref{hat}), we may also consider the auxiliary fields
\begin{equation}\label{phihut}
\widehat{\phi}_{R,Q}(f):=
(-1)^{N(N-1)/2}\phi_{R,Q}(f)(-1)^{N(N-1)/2}.
\end{equation}
Due to (\ref{multiplikativ}) and (\ref{R-1}), however, we have
\begin{equation}\label{phihut2}
\widehat{\phi}_{R,Q}(f)=\phi_{-R,Q}(f).
\end{equation}
In particular, in analogy to (\ref{R-1}) it follows
\begin{equation}\label{R-2}
\widehat{\phi}_{-1}(f)=\phi_{1}(f)=\underline{\phi}(f).
\end{equation}
Due to the unitary equivalence
\begin{equation*}
a^{\#}_{-1}(\varphi)=(-1)^{N(N-1)/2}a^{\#}(\varphi)(-1)^{N(N-1)/2},\qquad \varphi\in\mathscr{H}_1,
\end{equation*}
the operator-valued distributions $a^{\#}_{-1}(p)$ also satisfy canonical anticommutation relations. Furthermore, it is straightforward to check that
\begin{subequations}\label{1}
\begin{equation}
[a(p), a_{-1}(q)] = 0,\qquad [a^{*}(p), a^{*}_{-1}(q)] = 0
\end{equation}
\begin{equation}
[a(p), a^{*}_{-1}(q)] = \{a(p), a^{*}(q)\}(-1)^{N}.
\end{equation}
\end{subequations}
\subsection{Properties of the Deformed Model in $d\geq 2$}\label{resultate}
In the following discussion we are interested in the features of the deformed field operators $\phi_{R,Q}(f)$. To begin with, we investigate domain and hermiticity properties, the Reeh-Schlieder property and the Klein-Gordon equation. Our results are stated in the following proposition.
\begin{proposition}\label{prop1}
Let $R$ be a deformation function in the sense of Definition \ref{R} and let $Q$ be a $(d\times d)$-matrix which is antisymmetric w.r.t. the Minkowski inner product on $\mathbb{R}^{d}$ and satisfies (\ref{Q}) and (\ref{Q'}). Then the deformed field operators $\phi_{R,Q}(f)$, $f\in\mathscr{S}(\mathbb{R}^{d})$, have the following properties:
\begin{description}
    \item[a)] The dense subspace $\mathcal{D}^-\subset\mathscr{H^-}$ of vectors of finite particle number is contained in the domain $\mathcal{D}_{\phi}$ of any $\phi_{R,Q}(f)$. Moreover, $\phi_{R,Q}(f)\mathcal{D}^-\subset\mathcal{D}^-$ and $\phi_{R,Q}(f)\Omega=\underline{\phi}(f)\Omega$.
    \item[b)] For $\Psi\in\mathcal{D}^-$ we have
\begin{equation}
\phi_{R,Q}(f)^{*}\Psi=\phi_{R,Q}(\overline{f})\Psi,
\end{equation}
and $\phi_{R,Q}(f)$ is essentially selfadjoint on $\mathcal{D}^-$ for real $f\in\mathscr{S}(\mathbb{R}^{d})$.
    \item[c)] $\phi_{R,Q}$ is a weak solution of the Klein-Gordon equation, i.e.
\begin{equation}
\phi_{R,Q}\left(\left(\Box+m^2\right)f\right)=0.
\end{equation}
    \item[d)] The Reeh-Schlieder property holds: For any nonempty open $\mathcal{O}\subset\mathbb{R}^d$ the set 
\begin{equation}
\mathcal{D}_{R,Q}(\mathcal{O}):=span\{\phi_{R,Q}(f_1)\cdots\phi_{R,Q}(f_n)\Omega:n\in\mathbb{N}_0, f_1,\dots,f_n\in\mathscr{S}(\mathcal{O})\}
\end{equation}
is dense in $\mathscr{H}^-$.
\end{description}
\end{proposition}
\begin{proof}
a) These statements are a direct consequence of the definition of $\phi_{R,Q}$ (\ref{def}).\par
b) Since $\left(\overline{f}\right)^{\pm}=\overline{f^{\mp}}$ we have  $\phi_{R,Q}(f)^{*}\Psi=\phi_{R,Q}(\overline{f})\Psi$, $\Psi\in\mathcal{D}^-$. Along the same lines as \cite[Prop. 5.2.3]{robinson1997operator} one can show the essential selfadjointness for real $f$. In particular, due to $R$ being a phase factor, we find for $\Psi_n\in\mathscr{H}^-_n$ the estimate
\begin{equation}\label{Feldesti}
\|\phi_{R,Q}(f)\Psi_n\|\leq \left(\|f^+\|+\|f^-\|\right)\|(N+1)^{1/2}\Psi_n\|.
\end{equation}
Therefore, for $k\in\mathbb{N}$
\begin{multline*}
\|\phi_{R,Q}(f)^k\Psi_n\|\leq(n+k)^{1/2}\left(\|f^+\|+\|f^-\|\right)\|\phi_{R,Q}(f)^{k-1}\Psi_n\|\leq\\
(n+k)^{1/2}\cdots(n+1)^{1/2}\left(\|f^+\|+\|f^-\|\right)^k\|\Psi_n\|.
\end{multline*}
This yields for arbitrary $t\in\mathbb{C}$ that
$$\sum_{k=0}^{\infty}|t|^k\frac{\|\phi_{R,Q}(f)^k\Psi_n\|}{k!}\leq\sum_{k=0}^{\infty}
\left(\frac{(n+k)!}{n!}\right)^{1/2}\frac{|t|^k}{k!}
\left(\|f^+\|+\|f^-\|\right)^k\|\Psi_n\|<\infty,$$
implying that every $\Psi\in\mathcal{D}^-$ is an analytic vector for $\phi_{R,Q}(f)$. Since $\mathcal{D}^-$ is dense in $\mathscr{H}^-$ and $\phi_{R,Q}(f)$ is hermitian for real $f$ one can apply Nelson's theorem \cite[Thm. X.39]{RS2} and conclude that for real $f$, $\phi_{R,Q}(f)$ is essentially selfadjoint on $\mathcal{D}^-$.\par
c) This follows directly from $\left((\Box+m^2)f\right)^{\pm}=0$.\par
d) In order to prove this statement we want to make use of the spectrum condition and show in a standard manner \cite{streater2000pct} that $\mathcal{D}_{R,Q}(\mathcal{O})$ is dense in $\mathscr{H}^-$ if and only if $\mathcal{D}_{R,Q}(\mathbb{R}^d)\subset\mathscr{H}^-$ is dense. Thus, let $f_i\in\mathscr{S}(\mathbb{R}^d)$, $i=1,\dots,n$, with supp$\widetilde{f}_i\subset V_+$, then $\mathcal{D}_{R,Q}(\mathbb{R}^d)$ contains the vectors
$$\phi_{R,Q}(f_1)\cdots\phi_{R,Q}(f_n)\Omega=a_{R,Q}^*(f^+_1)\cdots a_{R,Q}^*(f^+_n)\Omega=\sqrt{n!}P_n^-\left(D_{n,R}(f_1^+\otimes\cdots\otimes f_n^+)\right),$$
where $P_n^-$ is the orthogonal projection from the unsymmetrized $\mathscr{H}_1^{\otimes n}$ onto its totally antisymmetric subspace $\mathscr{H}^-_n$, and $D_{n,R}\in\mathcal{B}(\mathscr{H}_1^{\otimes n})$ is the unitary operator multiplying with
$$D_{n,R}(p_1,\dots,p_n)=\prod_{1\leq k<l\leq n}R(Qp_k\cdot p_l)^{-1}.$$
By varying the test functions $f_i\in\mathscr{S}(\mathbb{R}^d)$ within this setting we obtain dense sets of $f^+_i$ in $\mathscr{H}_1$. Moreover, due to the unitary of $D_{n,R}$ this also leads to a total set of vectors $D_{n,R}(f_1^+\otimes\cdots\otimes f_n^+)$ in $\mathscr{H}_1^{\otimes n}$, implying that under the projection $P_n^-$ this set is total in $\mathscr{H}_n^-$. Hence it follows that $\mathcal{D}_{R,Q}(\mathbb{R}^d)$ is dense in $\mathscr{H}^-$. Application of the standard Reeh-Schlieder argument, which makes use of the
spectrum condition \cite{streater2000pct}, finishes the proof.
\end{proof}

Furthermore, we are interested in the transformation behavior of the deformed fields $\phi_{R,Q}$ under the adjoint action of the representation $U$ of the Poincaré group $\mathcal{P}_+$. We find the following results.
\begin{lemma}\label{lem1}
The operator-valued function $T_R(Qp)$ defined by (\ref{TR}) transforms under the adjoint action of the representation $U$ of $\mathcal{P}_+$ (\ref{unitary}), (\ref{anti}) according to
\begin{subequations}
\begin{eqnarray}
U(a,\Lambda)T_{R}(Qp)U(a,\Lambda)^{-1}&=&T_{R}\left(\left(\Lambda Q\Lambda^{-1}\right)\Lambda p\right),\quad(a,\Lambda)\in \mathcal{P}_{+}^{\uparrow}\\
U(a,\Lambda)T_{R}(Qp)U(a,\Lambda)^{-1}&=&T_{R}\left(-\left(\Lambda Q\Lambda^{-1}\right)\Lambda p\right)^{*},\quad(a,\Lambda)\in \mathcal{P}_{+}^{\downarrow},
\end{eqnarray}
\end{subequations}
where $Q$ is a $(d\times d)$-matrix which is antisymmetric w.r.t. the Minkowski inner product on $\mathbb{R}^{d}$, satisfying (\ref{Q}) and (\ref{Q'}). Correspondingly, the operator-valued distributions $a^{\#}_{R,Q}(p)$ transform as follows
\begin{subequations}\label{UaU}
\begin{eqnarray}
U(a,\Lambda)a^{*}_{R,Q}(p)U(a,\Lambda)^{-1}&=&e^{i\Lambda p\cdot a}a^{*}_{R,\Lambda Q\Lambda^{-1}}(\Lambda p),\quad(a,\Lambda)\in \mathcal{P}_{+}^{\uparrow},\\
U(a,\Lambda)a^{*}_{R,Q}(p)U(a,\Lambda)^{-1}&=&e^{-i\Lambda p\cdot a}a^{*}_{-\overline{R},\Lambda Q\Lambda^{-1}}(-\Lambda p),\quad(a,\Lambda)\in \mathcal{P}_{+}^{\downarrow},
\end{eqnarray}
\end{subequations}
\begin{subequations}\label{UaU2}
\begin{eqnarray}
U(a,\Lambda)a_{R,Q}(p)U(a,\Lambda)^{-1}&=&e^{-i\Lambda p\cdot a}a_{R,\Lambda Q\Lambda^{-1}}(\Lambda p),\quad(a,\Lambda)\in \mathcal{P}_{+}^{\uparrow},\\
U(a,\Lambda)a_{R,Q}(p)U(a,\Lambda)^{-1}&=&e^{i\Lambda p\cdot a}a_{-\overline{R},\Lambda Q\Lambda^{-1}}(-\Lambda p),\quad(a,\Lambda)\in \mathcal{P}_{+}^{\downarrow}.
\end{eqnarray}
\end{subequations}
The smeared field operators $\phi_{R,Q}(f)$, $f\in\mathscr{S}(\mathbb{R}^{d})$, (\ref{def}) therefore satisfy
\begin{subequations}\label{UphiU}
\begin{eqnarray}
U(a,\Lambda)\phi_{R,Q}(f)U(a,\Lambda)^{-1}&=&\phi_{R,\Lambda Q\Lambda^{-1}}(f_{(a,\Lambda)}),\quad(a,\Lambda)\in \mathcal{P}_{+}^{\uparrow}\\
U(a,\Lambda)\phi_{R,Q}(f)U(a,\Lambda)^{-1}&=&\phi_{-\overline{R},\Lambda Q\Lambda^{-1}}(\overline{f}_{(a,\Lambda)}),\quad(a,\Lambda)\in \mathcal{P}_{+}^{\downarrow},
\end{eqnarray}
\end{subequations}
where $f_{(a,\Lambda)}(x)=f(\Lambda^{-1}(x-a))$.
\end{lemma}
\begin{proof}
If $(a,\Lambda)\in\mathcal{P}_{+}^{\uparrow}$ and $\Psi\in\mathscr{H}^-$, then
\begin{eqnarray*}
\left(U(a,\Lambda)T_{R}(Qp)U(a,\Lambda)^{-1}\Psi\right)_n(p_1,\dots,p_n)&=&\prod_{k=1}^{n}R(Qp\cdot \Lambda^{-1} p_k)\Psi_n(p_1,\dots,p_n)\\
&=&\prod_{k=1}^{n}R(\Lambda Q\Lambda^{-1}\Lambda p\cdot  p_k)\Psi_n(p_1,\dots,p_n)\\
&=&\left(T_R(\Lambda Q\Lambda^{-1}\Lambda p)\Psi\right)_n(p_1,\dots,p_n),
\end{eqnarray*}
proving the first statement. Since $U(a,\Lambda)a(p)U(a,\Lambda)^{-1}=e^{-i\Lambda p\cdot a}a(\Lambda p)$, it follows for $a_{R,Q}(p)$
\begin{eqnarray*}
U(a,\Lambda)a_{R,Q}(p)U(a,\Lambda)^{-1}&=&e^{-i\Lambda p\cdot a}T_R(\Lambda Q\Lambda^{-1}\Lambda p)a(\Lambda p)\\
&=&e^{-i\Lambda p\cdot a}a_{R,\Lambda Q\Lambda^{-1}}(\Lambda p).
\end{eqnarray*}
Analogously, one shows the corresponding statement for $a_{R,Q}^{*}(p)$. For $(a,\Lambda)\in\mathcal{P}_{+}^{\downarrow}$ one finds
\begin{eqnarray*}
\left(U(a,\Lambda)T_{R}(Qp)U(a,\Lambda)^{-1}\Psi\right)_n(p_1,\dots,p_n)&=&\prod_{k=1}^{n}\overline{R(-Qp\cdot\Lambda^{-1} p_k)}\Psi_n(p_1,\dots,p_n)\\
&=&\left(T_R(-\Lambda Q\Lambda^{-1}\Lambda p)^{*}\Psi\right)_n(p_1,\dots,p_n).
\end{eqnarray*}
Hence, with $U(a,\Lambda)a(p)U(a,\Lambda)^{-1}=e^{i\Lambda p\cdot a}a_{-1}(-\Lambda p)$ it follows
\begin{eqnarray*}
U(a,\Lambda)a_{R,Q}(p)U(a,\Lambda)^{-1}&=&e^{i\Lambda p\cdot a}T_{\overline{R}}(-\Lambda Q\Lambda^{-1}\Lambda p)a_{-1}(-\Lambda p)\\
&=&e^{i\Lambda p\cdot a}a_{-\overline{R},\Lambda Q\Lambda^{-1}}(-\Lambda p).
\end{eqnarray*}
For $a_{R,Q}^{*}(p)$ one proceeds in the same way. The transformation behavior (\ref{UphiU}) of the field $\phi_{R,Q}$ is a direct consequence of Equations (\ref{UaU}) and (\ref{UaU2}).\end{proof}

The previous lemma shows that in the deformed model $\mathcal{P}_+$-covariance is violated. The property of $\mathcal{P}_+^{\uparrow}$-covariance is, however, preserved. To this end, let $\mathscr{P}_{R}(W_R)$ denote the polynomial algebra of fields generated by all $\phi_{R,Q}(f)$ with $f\in\mathscr{S}(W_R)$. It follows from the transformation behavior (\ref{UphiU}) that the algebra
\begin{equation}\label{netdefi}
\mathscr{P}_{R}(\Lambda W_R+a):=U(a,\Lambda)\mathscr{P}_{R}(W_R)U(a,\Lambda)^{-1}, \qquad (a,\Lambda)\in\mathcal{P}_+^{\uparrow}
\end{equation}
is generated by the fields $\phi_{R,\Lambda Q\Lambda^{-1}}(f)$ with $f\in\mathscr{S}(\Lambda W_R+a)$. Moreover, the map $W\mapsto \mathscr{P}_{R}(W)$, $W\in\mathcal W$, is clearly inclusion preserving and hence defines a \textit{net} which, in addition, is $\mathcal{P}_+^{\uparrow}$-covariant. Note, moreover, that each $\mathscr{P}_{R}(W)$ is a $*$-algebra by property b) in Proposition \ref{prop1}.\par
In two spacetime dimensions the deformed theory admits a $\mathcal{P}_+$-covariant net if an additional condition is imposed on the deformation function $R$, see Section \ref{2}.\par
Note that there is a connection between the set of wedges $\mathcal{W}$ and the orbit $\mathcal{Q}:=\{\Lambda Q\Lambda^{-1}:\Lambda\in\mathcal{L}_+\}$. Namely, $\mathcal{Q}$ is in one-to-one correspondence with wedges whose edges contain the origin \cite{grosse2007wedge}. The deformation function $R$, on the other hand, specifies the kind of deformation that is used.\\

It is clear that, in general, the properties of the deformed field $\phi_{R,Q}$ differ from those of the undeformed field $\phi$. In particular, $\phi$ is a bounded operator, whereas $\phi_{R,Q}$ is in general not as the exchange relations (\ref{exchange1}) imply. $\phi$ and $\phi_{R,Q}$, however, have in common that they are both nonlocal fields. The nonlocality of $\phi_{R,Q}$ can be explicitly seen by computing the two-particle contribution of the field commutator $[\phi_{R,Q}(f),\phi_{R,Q}(g)]$ applied to the vacuum $\Omega$, which yields
\begin{equation}\label{nloc}
\int d\mu(p)d\mu(q)f^{+}(p)g^{+}(q)\left(\overline{R(q\cdot Qp)}a^{*}(p)a^{*}(q)+\overline{R(p\cdot Qq)}a^{*}(p)a^{*}(q)\right)\Omega.
\end{equation}
This expression, however, only vanishes if $R(a)=-R(-a)$, $\forall a\in\mathbb{R}$. This requirement may be true for a function that fulfills $R(0)=0$, but with regard to Definition \ref{R} that requires $|R(a)|=1$ such a deformation function is inadmissible.\par
Note that in contrast to the deformation of a bosonic model \cite{grosse2007wedge, GL}, where $\phi_{R,Q}^{\rm{CCR}}(f)$ is relatively local to $\phi_{R,-Q}^{\rm{CCR}}(g)$, i.e. $[\phi_{R,Q}^{\rm{CCR}}(f),\phi_{R,-Q}^{\rm{CCR}}(g)]=0$ for supp$\,f\subset W_R$ and supp$\,g\subset W_R'$, $\phi_{R,Q}(f)$ is not relatively local to $\phi_{R,-Q}(g)$ for supp$\,f\subset W_R$ and supp$\,g\subset W_R'$. In particular, the two-particle contribution of $[\phi_{R,Q}(f),\phi_{R,-Q}(g)]$ applied to the vacuum reads
\begin{equation*}
2\int d\mu(p)d\mu(q)f^{+}(p)g^{+}(q)\overline{R(q\cdot Qp)}a^{*}(p)a^{*}(q)\Omega
\end{equation*}
which because of Definition \ref{R} does not vanish.\par
Along the lines of the undeformed case, see Equation (\ref{relloc}), we may consider the field commutator $[\phi_{R,Q}(f),\widehat{\phi}_{R,-Q}(g)]=[\phi_{R,Q}(f),\phi_{-R,-Q}(g)]$ for supp$\,f\subset W_R$ and supp$\,g\subset W_R'$. For the investigation of this commutator it is necessary to compute the corresponding commutation relations of the operators $a_{R,Q}^{\#}$ with $a_{-R,-Q}^{\#}$. A simple calculation shows that
\begin{subequations}\label{12}
\begin{equation}
[a_{R,Q}(p),a_{-R,-Q}(q)]=0,\qquad [a^{*}_{R,Q}(p),a^{*}_{-R,-Q}(q)]=0,
\end{equation}
\begin{equation}
[a_{R,Q}(p),a_{-R,-Q}^{*}(q)]
=\omega(\vec{p})\delta(\vec{p}-\vec{q})
(-1)^N T_{R}(Qp)T_{R}(-Qp)^{*},
\end{equation}
\begin{equation}
[a_{R,Q}^{*}(p),a_{-R,-Q}(q)] 
=\omega(\vec{p})\delta(\vec{p}-\vec{q})
(-1)^{N+1}T_{R}(Qp)^{*}T_{R}(-Qp).
\end{equation}
\end{subequations}
\begin{proposition}\label{co}
Let $R$ be a deformation function in the sense of Definition \ref{R} and $Q$ a $(d\times d)$-matrix which is antisymmetric w.r.t. the Minkowski inner product on $\mathbb{R}^{d}$, satisfying (\ref{Q}) and (\ref{Q'}). If $\kappa\geq 0$ in (\ref{Qform}), then the field operators $\phi_{R,Q}(f)$ (\ref{def}) and $\phi_{-R,-Q}(g)$ are relatively wedge-local to each other, i.e. for $f\in\mathscr{S}(W_R)$, $g\in\mathscr{S}(W_R')$
\begin{equation}\label{commuu}
[\phi_{R,Q}(f),\phi_{-R,-Q}(g)]\Psi=0,\qquad\Psi\in\mathcal{D}^-,
\end{equation}
holds.
\end{proposition}
\begin{proof}
Since $\langle\Phi,[\phi_{R,Q}(f),\phi_{-R,-Q}(g)]\Psi\rangle$, $\Phi, \Psi\in\mathcal{D}^-$, is a tempered distribution in $f$ and $g$, vanishing on $C_{0}^{\infty}(W_R)\times C_{0}^{\infty}(W_R')$ implies vanishing on $\mathscr{S}(W_R)\times\mathscr{S}(W_R')$. Making use of that property it thus suffices to prove (\ref{commuu}) for $(f,g)\in C_{0}^{\infty}(W_R)\times C_{0}^{\infty}(W_R')$.\par
Due to the commutation relations (\ref{12}) we have
\begin{multline*}
[\phi_{R,Q}(f),\phi_{-R,-Q}(g)]\Psi=\left([a_{R,Q}(\overline{f^-}),a_{-R,-Q}^*(g^+)]+
[a_{R,Q}^*(f^+),a_{-R,-Q}(\overline{g^-})]\right)\Psi,
\end{multline*}
which together with Definition \ref{R} yields the following n-particle contribution of this vector
\begin{multline}\label{int}
\left([\phi_{R,Q}(f),\phi_{-R,-Q}(g)]\Psi\right)_{n}(p_{1},\dots,p_{n})
\\=(-1)^{n}\int d\mu(p)\left(f^{-}(p)g^{+}(p)\prod_{k=1}^{n}\frac{R(p_{k}\cdot Qp)}{R(-p_{k}\cdot Qp)}-
f^{+}(p)g^{-}(p)\prod_{k=1}^{n}\frac{R(-p_{k}\cdot Qp)}{R(p_{k}\cdot Qp)}\right)\\
\times\Psi_{n}(p_{1},\dots,p_{n}).
\end{multline}
Our task is now to show that this expression vanishes for all $p_k$, $k=1,\dots,n$. Following the proof of Proposition 3.4 in \cite{grosse2007wedge} we may introduce new coordinates:
\begin{equation*}
m_{\perp}:=\sqrt{m^{2}+p_{\perp}^{2}},\qquad p_{\perp}:=(p_{2},\dots,p_{d-1}),\qquad \vartheta:=\mathrm{Arsinh} \frac{p_{1}}{m_{\perp}}.
\end{equation*}
Thus, in the coordinates $(\vartheta,p_{\perp})$ we have
\begin{equation*}
d\mu(p)=\frac{d^{d-1}\vec{p}}{\omega(\vec{p})}=d\vartheta d^{d-2}p_{\perp},
 \qquad p=p(\vartheta):=\left( \begin{array}{c}
m_{\perp}\rm{cosh}\,\vartheta   \\
m_{\perp}\rm{sinh}\,\vartheta \\
p_{\perp}
\end{array} \right).
\end{equation*}
Correspondingly, we use the following notation
$$f^{\pm}(\vartheta,p_{\perp}):=\widetilde{f}(\pm p(\vartheta)).$$
According to \cite{grosse2007wedge}, $f^{-}(\vartheta+i\lambda,p_{\perp})$ is bounded on the strip $0\leq\lambda\leq\pi$, $\vartheta\in\mathbb{R}$, due to supp$f\subset W_R$ and analyticity properties of $\widetilde{f}$, $f\in C_{0}^{\infty}(W_R)$. In particular, $\widetilde{f}$ is an entire analytic function because $f$ has compact support. Moreover, also $g^{+}(\vartheta+i\lambda,p_{\perp})$, supp$\,g\subset W_R'=-W_R$, is bounded on the strip $0\leq\lambda\leq\pi$, $\vartheta\in\mathbb{R}$, and the boundary values at $\lambda=\pi$ are given by
\begin{equation}\label{perp}
f^{-}(\vartheta+i\pi,p_{\perp})=f^{+}(\vartheta,-p_{\perp}),\qquad
g^{+}(\vartheta+i\pi,p_{\perp})=g^{-}(\vartheta,-p_{\perp}).
\end{equation}
It remains to study the properties of the functions $\vartheta\mapsto R(Qp(\vartheta)\cdot p_k)\overline{R(-Qp(\vartheta)\cdot p_k)}$, $k=1,\dots,n$, which appear in (\ref{int}). It follows for $0\leq\lambda\leq\pi$ that
$$\mathrm{Im}\left(p(\vartheta+i\lambda)Q\cdot p_k\right)=\kappa m_\perp \mathrm{sin}\,\lambda\left(\begin{array}{c}
\rm{cosh}\,\vartheta   \\
\rm{sinh}\,\vartheta 
\end{array} \right)\cdot \left(\begin{array}{c}
p_k^0   \\
p_k^1
\end{array} \right)\geq 0$$
because $\kappa\geq 0$ and both $(\rm{cosh}\,\vartheta,\rm{sinh}\,\vartheta)$ and $(p_k^0,p_k^1)$ are in the two-dimensional forward lightcone.
Due to Definition \ref{R}, this implies that the functions $$z\mapsto R(Qp(z)\cdot p_k)\overline{R(-Qp(z)\cdot p_k)},\qquad k=1,\dots,n,$$are analytic on the strip $S(0,\pi)=\{z=\vartheta+i\lambda\in\mathbb{C}:0<\lambda<\pi\}$. In addition, it also follows from Definition \ref{R} that these functions are continuous on the closure $\overline{S(0,\pi)}$ of $S(0,\pi)$, implying that $|R(Qp(z)\cdot p_k)\overline{R(-Qp(z)\cdot p_k)}|\leq 1$ for $z\in\overline{S(0,\pi)}$ \cite[Thm. 12.9]{rudin1987real}.
Hence, together with the previous discussion it is possible to shift the $\vartheta$-integration in (\ref{int}) from $\mathbb{R}$ to $\mathbb{R}+i\pi$. Making use of (\ref{perp}), we have
\begin{eqnarray*}\label{rel}
&&\int d\mu(p)f^{-}(p)g^{+}(p)\prod_{k=1}^{n}\frac{R(p_{k}\cdot Qp)}{R(-p_{k}\cdot Qp)}\\
&=&\int d^{d-2}p_{\perp}\int d\vartheta f^{-}(\vartheta,p_{\perp})g^{+}(\vartheta,p_{\perp})
\prod_{k=1}^{n}\frac{R(p_{k}\cdot Qp(\vartheta))}{R(-p_{k}\cdot Qp(\vartheta))}\\
&=&\int d^{d-2}p_{\perp}\int d\vartheta f^{+}(\vartheta,-p_{\perp})g^{-}(\vartheta,-p_{\perp})\prod_{k=1}^{n}\frac{R(p_{k}\cdot Qp(\vartheta+i\pi))}{R(-p_{k}\cdot Qp(\vartheta+i\pi))}\\
&=&\int d\mu(p)f^{+}(p)g^{-}(p)\prod_{k=1}^{n}\frac{R(-p_{k}\cdot Qp)}{R(p_{k}\cdot Qp)}.
\end{eqnarray*}
Thus,
\begin{equation*}
\left([\phi_{R,Q}(f),\phi_{-R,-Q}(g)]\Psi\right)_{n}(p_{1},\dots,p_{n})=0
\end{equation*}
for supp$\,f\subset W_R$ and supp$\,g\subset W_R'$.
\end{proof}
\begin{corollary}
Let $R$ be a deformation function in the sense of Definition \ref{R} and $Q$ a $(d\times d)$-matrix which is antisymmetric w.r.t. the Minkowski inner product on $\mathbb{R}^{d}$, satisfying (\ref{Q}) and (\ref{Q'}). If $\kappa<0$ in (\ref{Qform}), then the field operators $\phi_{R,-Q}(f)$ (\ref{def}) and $\phi_{-R,Q}(g)$ are relatively wedge-local to each other, i.e. for $f\in\mathscr{S}(W_R)$, $g\in\mathscr{S}(W_R')$
\begin{equation}\label{commu1}
[\phi_{R,-Q}(f),\phi_{-R,Q}(g)]\Psi=0,\qquad\Psi\in\mathcal{D}^-,
\end{equation}
holds.
\end{corollary}
The proof of this statement is analogous to the one of Proposition \ref{co}.\\

Thus, due to the previous results, the properties of the deformed fields $\phi_{R,Q}$ listed in Proposition \ref{prop1}, particularly feature d), and their transformation behavior, Lemma \ref{lem1}, we have established the following result.
\begin{proposition} Let $R$ be a deformation function and $Q$ an admissible matrix of the form (\ref{Qform}) with $\kappa\geq 0$. Then, the nets $W\mapsto\mathscr{P}_{R}(W)$ and $W\mapsto\mathscr{P}_{-R}(W)$, $W\in\mathcal{W}$, defined by (\ref{netdefi}), are $\mathcal{P}_+^\uparrow$-covariant, nonlocal nets which are relatively wedge-local in the sense that
\begin{equation}\label{relW}
\mathscr{P}_{R}(W)\subset\mathscr{P}_{-R}(W')',\qquad W\in\mathcal{W}.
\end{equation}
Moreover, the vacuum vector $\Omega$ is cyclic and separating for each $\mathscr{P}_{\pm R}(W)$, $W\in\mathcal{W}$.
\end{proposition}
\begin{proof}
All statements are obvious by the above remarks except for the separability of the vacuum which is implied by its cyclicity and (\ref{relW}).
\end{proof}

Note that in contrast to the results in \cite{GL} where the deformation of a bosonic model is investigated, we did not require that the deformation function $R$ satisfies $R(a)^{-1}=R(-a)$ and $R(0)=1$. In particular, the condition $R(0)=1$ in \cite{GL} results from the deformation of the underlying Borchers-Uhlmann algebra $\underline{\mathscr{S}}$. More precisely, recall from Section \ref{Develop} that in \cite{GL} the deformation is based on linear homeomorphisms $\rho:\underline{\mathscr{S}}\rightarrow\underline{\mathscr{S}}$ which act multiplicatively in momentum space, i.e.
$$\widetilde{\rho(f)}_n(p_1,\dots,p_n):=\rho_n(p_1,\dots,p_n)\cdot \widetilde{f}_n(p_1,\dots,p_n).$$
Explicit conditions on the functions $\rho_n\in C^\infty(\mathbb{R}^{nd})$, $n\in\mathbb{N}_0$ arise from the compatibility requirement (\ref{komp1}) for quasi-free, translationally invariant states $\omega$. In particular, it turns out that the functions $\rho_n$ are determined by the functions $\rho_2$. The connection to our deformation approach is then given by setting $$\rho_2(p,q):=R(-p\cdot Qq).$$
The requirements on $\rho_2$ yield, inter alia, $R(a)^{-1}=R(-a)$ and $R(0)=1$.\par 
In the approach taken in \cite{GL} the fulfillment of the property $R(a)^{-1}=R(-a)$ is necessary for obtaining wedge-locality and covariance. In contrast to this, we do not obtain the same result for our deformed fermionic model by imposing this relation, except in two spacetime dimensions, see Section \ref{2}. Nevertheless, the requirement $R(0)=1$ is redundant for establishing wedge-locality and covariance properties for the deformed model in both the deformed bosonic and the deformed fermionic case. In particular, one can perform a deformation as presented in Section \ref{deformiert} of a bosonic model involving field operators $\phi^{CCR}$ and arrive at a covariant and wedge-local deformed model involving deformed field operators $\phi_{R,Q}^{CCR}$ with deformation functions $R$ not necessarily satisfying $R(0)=1$ as is the case in \cite{GL}.\par
In fact, considering a deformation function $R$, the correspondingly deformed net $\mathscr{P}_{R}$ is unitarily equivalent to the net $\mathscr{P}_{-R}$, implying that deformations involving $R$ and those involving $-R$ are equivalent. In other words, a model resulting from deformation associated with $R$ is physically indistinguishable from a model arising from deformation with $-R$. We summarize this result in the following lemma.
\begin{lemma}\label{lem2}
The net $\mathscr{P}_{R}$ is unitarily equivalent to the net $\mathscr{P}_{-R}$. The unitary $Z$ relating these two nets is given by
\begin{equation*}
Z:=(-1)^{N(N-1)/2},
\end{equation*}
where $N$ is the particle number operator, i.e. $N|_{\mathscr{H}^-_n}=n\cdot 1$.
\end{lemma}
\begin{proof}
Since the unitary $Z$ commutes with all Poincaré transformations, i.e. 
$$[Z,U(g)]=0,\qquad\forall\, g\in\mathcal{P}_+,$$
and satisfies $Z\Omega=\Omega$, the unitary equivalence $Z\mathscr{P}_{R}Z^{-1}=\mathscr{P}_{-R}$ follows.\end{proof}

\subsection{Two-Particle Scattering Operator}\label{AC1}
In this section we show that the models arising by the deformation procedure presented in this chapter do indeed differ from the initial one. This circumstance becomes apparent by analyzing the two-particle scattering processes of the deformed models. To this end, we shall employ similar techniques to those discussed in Section \ref{AC}. We rely, in particular, on the model-independent analysis of \cite{BBS01} which adapts Haag-Ruelle-Hepp scattering theory to certain wedge-localized operators, so-called temperate polarization-free generators.
\begin{lemma} Let $R$ be a deformation function and $Q$ an admissible matrix (\ref{Qform}) with $\kappa\geq 0$. Then, the fields $\phi_{R,Q}(f)$ and $\phi_{-R,-Q}(f)$, $f\in\mathscr{S}(\mathbb{R}^d)$, are temperate polarization-free generators, that is,
\begin{itemize}
\item[i)] $\phi_{\pm R,\pm Q}(f)$ are closable and $\Omega$ is contained in the domains of $\phi_{\pm R,\pm Q}(f)$ and $\phi_{\pm R,\pm Q}(f)^*$,
\item[ii)] $\phi_{\pm R,\pm Q}(f)\Omega$ and $\phi_{\pm R,\pm Q}(f)^*\Omega$ are elements of $\mathscr{H}_1$,
\item[iii)] $\phi_{R,Q}(f)$ is localized in the wedge $(W_R+\text{supp}\,f)''$ and $\phi_{-R,-Q}(f)$ is localized in the wedge $(W_L+\text{supp}\,f)''$
\item[iv)] $\phi_{\pm R,\pm Q}(f)$ are temperate in the sense that the functions $x\mapsto\phi_{\pm R,\pm Q}(f)^{\#}U(x)\Psi$, $\Psi\in\mathcal{D}^-$, are strongly continuous and bounded in norm for large $x$, where e.g. $\phi_{ R, Q}(f)^{\#}\in\{\phi_{ R, Q}(f),\phi_{ R, Q}(f)^{*}\}$.
\end{itemize}
\end{lemma}
\begin{proof}
Statements $i)-iii)$ are direct consequences of the definition of the deformed field operators $\phi_{\pm R,\pm Q}(f)$, Propositions \ref{prop1} and \ref{co}. The temperedness property $iv)$ follows, since $\|\phi_{\pm R,\pm Q}(f)\left(U(x)\Psi-\Psi\right)\|\leq \left(\|f^+\|+\|f^-\|\right)\|(N+1)^{1/2}\left(U(x)\Psi-\Psi\right)\|\rightarrow 0$ for $x\rightarrow 0$ with regard to the bound (\ref{Feldesti}) and with $\Psi\in\mathcal{D}^-$. Analogously, one proceeds for the case of the adjoints.
\end{proof}
Due to these properties of the fields $\phi_{R,Q}(f)$ and $\phi_{-R,-Q}(f)$ we are in a position to construct two-particle scattering states. Let us recall how this can be done in the present situation.\par
As in Section \ref{AC} we introduce for $t\in\mathbb{R}$ and $f\in\mathscr{S}(\mathbb{R}^d)$ the functions
\begin{equation}\label{tAbh1}
f_{t}(x):=\frac{1}{(2\pi)^{d/2}}\int d^d p\,\widetilde{f}(p)\,e^{it\left[p_0-\omega(\vec{p})\right]}\,e^{-ip\cdot x}.
\end{equation}
The support of $f_t$ is essentially contained in $t\mathcal{V}(f)$ for $t\rightarrow\pm\infty$ \cite{hepp1965}. The set
\begin{equation}
\mathcal{V}(f):=\{\left(1,\vec{p}/\omega(\vec{p})\right): p\in\text{supp}\,\widetilde{f}\}
\end{equation}
is the velocity support of $f$. Furthermore, we choose such test functions $f,g\in\mathscr{S}(\mathbb{R}^d)$ whose Fourier transforms have compact supports concentrated around points on $H_m^+$. Analogously as in Section \ref{AC}, shall write $f\prec g$ whenever $\mathcal{V}(g)-\mathcal{V}(f)$ is contained in the interior of $W_R$. If $f\prec g$, respectively $g\prec f$, then the regions $W_R+t\mathcal{V}(g)$ and $W_R'+t\mathcal{V}(f)$ are spacelike separated for $t>0$, respectively $t<0$.\par
For asymptotic times the operators $\phi_{R,Q}(g_t)$ and $\phi_{-R,-Q}(f_t)$ are essentially localized in $W_R+t\mathcal{V}(g)$ and $W_R'+t\mathcal{V}(f)$ respectively. In view of the analysis of \cite{BBS01} the following strong limits, stated in (\ref{lim}), exist and constitute asymptotic two-particle collision states. More precisely, with regard to $\phi_{\pm R,\pm Q}(f)\Omega=f^+$, we have
\begin{eqnarray}\label{lim}
\lim\limits_{t\rightarrow\infty}\phi_{R,Q}(g_t)\phi_{-R,-Q}(f_t)\Omega&=:&(g^+\times_R f^+)_{\text{out}},\quad\text{for}\quad f\prec g,\\
\lim\limits_{t\rightarrow -\infty}\phi_{R,Q}(g_t)\phi_{-R,-Q}(f_t)\Omega&=:&(g^+\times_R f^+)_{\text{in}},\quad\text{for}\quad g\prec f.
\end{eqnarray}
Due to the commutativity of the operators $\phi_{R,Q}(g_t)$ and $\phi_{-R,-Q}(f_t)$ for asymptotic times the scattering states are symmetric, $(g^+\times_R f^+)_{\text{in/out}}=(f^+\times_R g^+)_{\text{in/out}}$, as is the case for a Boson.\par
In the present setting it is straightforward to compute the above limits. It follows, namely, from the support properties of $\widetilde{f}$ and $\widetilde{g}$ that $f_t^+=f^+$, $g_t^+=g^+$ and $f_t^-=0$, $g_t^-=0$. Hence, we arrive at
\begin{eqnarray}
(g^+\times_R f^+)_{\text{out}}&=&a^*_{R,Q}(g^+)a^*(f^+)\Omega,\quad\text{for}\quad f\prec g,\\
(g^+\times_R f^+)_{\text{in}}&=&a^*_{R,Q}(g^+)a^*(f^+)\Omega,\quad\text{for}\quad g\prec f.
\end{eqnarray}
Note the dependence on the ordering of the test functions. Considering explicitly $f\prec g$, then, we have
\begin{equation}
\begin{aligned}
(g^+\times_R f^+)_{\text{out}}&=\int d\mu(p)\int d\mu(q) g^+(p)f^+(q)\overline{R(Qp\cdot q)}a^*(p)a^*(q)\Omega,\\
(g^+\times_R f^+)_{\text{in}}&=\int d\mu(p)\int d\mu(q) g^+(p)f^+(q)\left(-\overline{R(-Qp\cdot q)}\right)a^*(p)a^*(q)\Omega.
\end{aligned}
\end{equation}
In order to obtain the S-matrix elements, we consider additional test functions $h,k\in\mathscr{S}(\mathbb{R}^d)$ with the same properties as $f$ and $g$, in particular, $f\prec g$ and $h\prec k$, which yields
\begin{multline}\label{2parS}
\langle (g^+\times_R f^+)_{\text{out}},(h^+\times_R k^+)_{\text{in}}\rangle=\\
-\int d\mu(p)\int d\mu(q)R(Qp\cdot q)\overline{R(-Qp\cdot q)}\,\overline{g^+(p)}\,\overline{f^+(q)}h^+(p)k^+(q).
\end{multline}
Hence, the two-particle S-matrix elements of the deformed model differ from the undeformed theory, corresponding to $R(a)=1$ for all $a\in\mathbb{R}$, in a nontrivial way. They depend, moreover, on the deformation. We, therefore, conclude that the deformation procedure carried out in this chapter gives rise to models which are \textit{inequivalent} from the initial, undeformed theory. However, since the deformation function $R$ is merely a phase factor, effects like momentum transfer or particle production in scattering processes cannot be expected for the present model. On the other hand, specific arrangements such as time delay experiments can uncover the effects of the deformation. Similar results have, for instance, been found in the construction approaches by warped convolutions \cite{buchholz2008warped, grosse2008noncommutative} and the more general one pursued in \cite{GL}.\par
Due to the $Q$-dependence of the two-particle S-matrix elements, cf. (\ref{2parS}), they are not fully Lorentz invariant in spacetime dimensions $d>2$. This is a consequence of the properties of the matrix $Q$ which is only invariant under boosts that preserve the wedge $W_R$, cf. (\ref{Q}) and (\ref{Q'}). This indicates that when passing on to algebras corresponding to \textit{bounded} regions the Reeh-Schlieder property cannot be expected to hold.\par
The situation is, however, different in $d=1+1$ dimensions where the restricted Lorentz group $\mathcal{L}_+^\uparrow$ leaves $W_R$ invariant. As a consequence the deformed S–matrix does not break the Lorentz symmetry. Since many more interesting features can be observed for the model at hand in two spacetime dimensions, we devote the next chapter to their discussion.

% Chapter6

\chapter{Deformations and Integrable Models} % Main chapter title

\label{Chapter6} % Change X to a consecutive number; for referencing this chapter elsewhere, use \ref{ChapterX}
\fancyhead[LE,RO]{\thepage}
\fancyhead[LO]{\thesection. \emph{\rightmark}}
\fancyhead[RE]{Chapter 6. \emph{Deformations and Integrable Models}}
\renewcommand{\chaptermark}[1]{ \markboth{#1}{} }
\renewcommand{\sectionmark}[1]{ \markright{#1}{} }
%\lhead{Chapter 6. \emph{Deformations and Integrable Models}} % Change X to a consecutive number; this is for the header on each page - perhaps a shortened title

So far we have presented two different methods of constructing nontrivial quantum field theoretical models. In Chapter \ref{Chapter3} we proceeded from an inverse scattering point of view where a factorizing S-matrix constituted the starting point of our journey. Due to the simple nature of these scattering matrices we had to restrict the spacetime dimension to two. In Chapter \ref{Chapter5}, on the other hand, the deformation procedure served as a useful tool in constructing nontrivial models in any spacetime dimension $d\geq 2$.\par
In the following sections we associate the models arising in Chapter \ref{Chapter5} from deformation of the model of a scalar massive Fermion with integrable model theories in $d=1+1$ dimensions. To this end, we first specialize the framework of the previous chapter to this special case and establish afterwards a connection with the models constructed in Chapter \ref{Chapter3}.

\section{Deformed Models in $1+1$ Dimensions}\label{2}
In this section we formulate the deformed model of Section \ref{3} on two-dimensional Minkowski space. For that purpose, recall from Section \ref{wedges1+1} that in $d=2$ the set of wedges $\mathcal{W}$ consists of two disjoint subsets, namely the translates of $W_R$ (\ref{wedge}) and the translates of $W_R'=-W_R=W_L$. The matrix $Q$ is, moreover, of the form
\begin{equation}\label{Qfor}
Q = \lambda
\left( \begin{array}{cc}
0 & 1  \\
1 & 0  
\end{array} \right),\qquad \lambda\in \mathbb{R}.
\end{equation}
In analogy to the undeformed case (\ref{net}) we may define for a fixed deformation function $R$ given by Definition \ref{R} and an admissible and fixed matrix $Q$ (\ref{Qfor}) the von Neumann algebras generated by $\phi_{R,Q}$ and $\phi_{-\overline{R},Q}$, i.e. with $x\in\mathbb{R}^2$
\begin{subequations}\label{algebras1}
\begin{eqnarray}
\mathcal{M}(W_R+x)&:=&\{e^{i\overline{\phi_{R,Q} (f)}}:f=\overline{f}\in\mathscr{S}(W_R+x)\}'',\\
\mathcal{M}(W_L+x)&:=&\{e^{i\overline{\phi_{-\overline{R},Q}(f)}}:f=\overline{f}\in\mathscr{S}(W_L+x)\}''.
\end{eqnarray}
\end{subequations}
This definition, however, only produces a wedge-local and covariant net $W\mapsto\mathcal{M}(W)$, $W\in\mathcal{W}$, if the deformation function $R$ additionally fulfills the property
\begin{equation}
\overline{R(a)}=R(-a),\qquad\forall\, a\in\mathbb{R}.
\end{equation}
If, namely, this property holds, the fields $\phi_{-\overline{R},Q}(f)$ are equal to $\phi_{-R,-Q}(f)$ and wedge-locality can be inferred from Proposition \ref{co}. We shall, therefore, assume this relation in what follows. Due to Lemma \ref{lem1} the net $W\mapsto\mathcal{M}(W)$, $W\in\mathcal{W}$, also transforms covariantly under the adjoint action of the representation $U$ of $\mathcal{P}_+$. In fact, we have
\begin{proposition}\label{PropWedgeAlgebra1}
Let $R$ be a deformation function, $Q$ an admissible matrix (\ref{Qfor}) with $\lambda\geq 0$ and $\mathcal{M}(W)$, \mbox{$W\in\mathcal{W}$}, be defined as in (\ref{algebras1}). Then, $\{\mathcal{M}(W)\}_{W\in\mathcal{W}}$ is a local net, cf. Section \ref{AQFT}, of von Neumann algebras, transforming covariantly under the adjoint action of the representation $U$ of the proper Poincaré group $\mathcal{P}_+$. Furthermore, locality is fulfilled (without twisting) and the vacuum vector $\Omega$ is cyclic and separating for each $\mathcal{M}(W)$, $W\in\mathcal{W}$.
\end{proposition}
\begin{proof}
The properties of isotony and covariance are obvious by the above remarks. With regard to locality, we only have to prove that the commutation relation of Proposition \ref{co} holds in the stronger sense that also the unitaries $e^{i\overline{\phi_{R,Q}(f)}}$ and $e^{i\overline{\phi_{-\overline{R},Q}(g)}}$, for real $f\in\mathscr{S}(W_R)$ and $g\in\mathscr{S}(W_L)$, commute. To this end, one shows in an analogous manner as in the proof of Proposition \ref{prop1} part b) that $e^{i\overline{\phi_{-\overline{R},Q}(g)}}\Psi$, $\Psi\in\mathcal{D}^-$, is an analytic vector for $\overline{\phi_{R,Q}(f)}$. Thus, for $\Psi\in\mathcal{D}^-$ we can compute the commutator of the unitaries as follows
\begin{equation*}
\left[e^{i\overline{\phi_{R,Q}(f)}},e^{i\overline{\phi_{-\overline{R},Q}(g)}}\right]\Psi=\sum_{j,k=0}^{\infty}\frac{i^{j+k}}{j! k!}\left[\phi_{R,Q}(f)^j,\phi_{-\overline{R},Q}(g)^k\right]\Psi,
\end{equation*}
which due to $\phi_{-\overline{R},Q}(g)=\phi_{-R,-Q}(g)$, Proposition \ref{co} and covariance yields the locality claim.\par
The cyclicity of the vacuum vector $\Omega$ for the wedge algebras can be shown by the following standard argument. Namely, consider real $f_1,\dots,f_n\in\mathscr{S}(W_R)$ and denote by $E_j(t)$ the spectral projection of the selfadjoint operator $\overline{\phi_{R,Q}(f_j)}$, which corresponds to spectral values in $[-t,t]$. Then, $G_j(t):=E_j(t)\overline{\phi_{R,Q}(f_j)}\in\mathcal{M}(W_R)$ for all $t\in\mathbb{R}$. Moreover, $G_j(t)$ converges strongly to $\overline{\phi_{R,Q}(f_j)}$ on $\mathcal{D}^-$ as $t\rightarrow\infty$. Consequently, $G_1(t)\cdots G_n(t)\Omega\rightarrow\phi_{R,Q}(f_1)\cdots\phi_{R,Q}(f_n)\Omega$ for $t\rightarrow\infty$. Hence, the cyclicity of $\Omega$ for $\mathcal{M}(W_R)$ can be inferred from Proposition \ref{prop1} d). By an analogous argument involving the field $\phi_{-\overline{R},Q}$ the cyclicity of $\Omega$ for $\mathcal{M}(W_L)$ follows. Due to the commutativity of these algebras $\Omega$  is also separating for them. The statement is then proven by covariance.
\end{proof}

\section{Integrable Models from Deformation Theory}
The properties of the net $W\mapsto\mathcal{M}(W)$, $W\in\mathcal{W}$, stated in Proposition \ref{PropWedgeAlgebra1} together with the simple structure of the underlying scattering operator uncovered in Section \ref{AC1} suggests that there is a close connection of the models at hand with the model theories constructed in Chapter \ref{Chapter3}. In order to analyze this relation, we recall the possibility of parameterizing $H_{m}^{+}$ with the help of the rapidity $\theta\in\mathbb{R}$, i.e. $p(\theta)=m(\cosh\theta,\sinh\theta)$, in case $d=2$. Making use of this notation and (\ref{Qfor}), we have
\begin{equation*}
-p(\theta_{1})\cdot Qp(\theta_{2})=\lambda m^{2}\sinh(\theta_{1}-\theta_{2}),\qquad \theta_1,\theta_2\in\mathbb{R}.
\end{equation*}
We further define
\begin{equation}\label{scatteringfunction}
S_{\lambda}:\mathbb{R}\rightarrow\mathbb{C},\qquad S_{\lambda}(\theta):=-R(\lambda m^{2}\sinh\theta)^2.
\end{equation}
Since the entire analytic function $\sinh$ maps the strip $S(0,\pi):=\{z\in\mathbb{C}:0<\mathrm{Im}\,z<\pi\}$ onto the upper half plane and since by Definition \ref{R} $R$ has an analytic continuation to the upper half plane, the function $S_\lambda$, $\lambda\geq 0$, extends to an analytic function on the strip $S(0,\pi)$. Moreover, it follows from the requirements on the function $R$ by Definition \ref{R} and the properties of $\sinh$ that
\begin{equation*}
S_{\lambda}(0)=-1,\qquad \overline{S_{\lambda}(\theta)}=S_{\lambda}(-\theta)=S_{\lambda}(\theta)^{-1}=
S_{\lambda}(\theta+i\pi),\qquad \lambda,\theta\in\mathbb{R}.
\end{equation*}
These properties of the function $S_\lambda$ are familiar from the context of factorizing S-matrices, cf. Section \ref{Sectionfactorizing}, and express the unitarity, hermitian analyticity and crossing symmetry of the scattering operator $\textbf{S}$ associated with $S_\lambda$ \cite{iagolnitzer1993scattering, Smir92}. In addition, the Yang-Baxter equation is trivially fulfilled because we are considering here only a single species of particles. Due to these properties the scattering operator $\textbf{S}$ associated with $S_\lambda$ agrees with an S-matrix of a completely integrable relativistic quantum field theory \cite{Smir92}. In particular, we conclude
\begin{lemma}
The function $S_\lambda$ defined in (\ref{scatteringfunction}) constitutes an S-matrix in the sense of Definition \ref{S-matrixDefinition} for the particle spectrum given by $G=\{e\}$ and $m>0$.
\end{lemma}
The connection of the model theories constructed in Chapter \ref{Chapter3} to the deformation procedure carried out in Section \ref{deformiert} and, therefore, to the net $\{\mathcal{M}(W)\}_{W\in\mathcal{W}}$ may be clarified by introducing
\begin{equation*}
z_{\lambda}(\theta):=a_{R,Q}(p(\theta)),\qquad z^{\dagger}_{\lambda}(\theta):=a^{*}_{R,Q}(p(\theta)).
\end{equation*}
The exchange relations (\ref{exchange1}) for $Q=Q'$ and $Q$ given by (\ref{Qfor}) then read
\begin{eqnarray*}
z_{\lambda}(\theta_{1})z_{\lambda}(\theta_{2})&=&S_{\lambda}(\theta_{2}-\theta_{1})
z_{\lambda}(\theta_{2})z_{\lambda}(\theta_{1})\\
z^{\dagger}_{\lambda}(\theta_{1})z^{\dagger}_{\lambda}(\theta_{2})&=&S_{\lambda}(\theta_{2}-\theta_{1})
z^{\dagger}_{\lambda}(\theta_{2})z^{\dagger}_{\lambda}(\theta_{1})\\
z_{\lambda}(\theta_{1})z^{\dagger}_{\lambda}(\theta_{2})&=&S_{\lambda}(\theta_{1}-\theta_{2}) z^{\dagger}_{\lambda}(\theta_{2})z_{\lambda}(\theta_{1})+\delta(\theta_{1}-\theta_{2})\cdot 1.
\end{eqnarray*}
That is, $z_{\lambda}(\theta)$ and $z^{\dagger}_{\lambda}(\theta)$ form a representation of the Zamolodchikov-Faddeev algebra \cite{faddeev1980quantum, zamolodchikov1979factorized} with scattering function $S_{\lambda}(\theta)$. But this algebraic structure was the starting point for the construction of models with factorizing S-matrices in Section \ref{Construction}. Thereby, the auxiliary fields $\phi_{\lambda}(x):=\int d\theta(e^{ip(\theta)\cdot x}z_\lambda^{\dagger}(\theta)+e^{-ip(\theta)\cdot x}z_\lambda(\theta))$ associated with $z^{\#}_{\lambda}(\theta)$, which are wedge-local polarization-free generators, cf. Section \ref{SecWedgeLocal}, are at the basis of the approach. The interesting point here is that  these fields appear in the present setting as a consequence of the deformation of the model given in Section \ref{model}. More precisely, the fields $\phi_{\lambda}$ coincide with the deformed fields $\phi_{R,Q}$.\par
For the case of $S_\lambda$ being a regular scattering function in the sense of Definition \ref{regularS}, it was shown in \cite{L08, erratum} that Theorem \ref{mainTheorem} holds in the present setting without reference to any conjecture. Hence, the quantum field theory arising from $\phi_{\lambda}$ contains nontrivial observables localized in bounded open regions $\mathcal{O}\subset\mathbb{R}^2$ above a minimal size. Moreover, besides other standard properties of quantum field theory also the Reeh-Schlieder property holds. In addition, the S-matrix of the model is found to be the one determined by the two-particle scattering function $S_\lambda$ \cite{L08}. The following theorem demonstrates the connection of the deformation of a scalar massive Fermion to integrable models.
\begin{theorem}\label{theorem}
Every integrable quantum field theory on two-dimensional Minkowski space with scattering function $S_\lambda$ of the form (\ref{scatteringfunction}) can be obtained by deformation of a scalar massive Fermion in the sense of Section \ref{deformiert}. If further $S_\lambda$ is regular, then in the deformed theory there exist observables localized in double cones above a minimal diameter, and the Reeh-Schlieder property holds \cite[Theorem E.4]{L08}.
\end{theorem}
Note that, since the deformation function $R$ appears quadratically in the definition of the scattering function $S_\lambda$ (\ref{scatteringfunction}), $S_\lambda$ does not depend on the sign of $R$. This circumstance implies physical indistinguishability of correspondingly deformed models, i.e. models arising from deformation with $R$ and $-R$, which is in agreement with Lemma \ref{lem2}. Concrete examples of deformation functions $R$ for which Theorem \ref{theorem} applies constitute the following
\begin{equation*}
R(a)=\pm\prod_{k=1}^n\frac{\zeta_k-a}{\zeta_k+a},\qquad \mathrm{Im}\,\zeta_k> 0,
\end{equation*}
where for each $\zeta_k$ also $-\overline{\zeta_k}$ is contained in the set of zeros $\{\zeta_1,\dots,\zeta_n\}$.\par
In conclusion, we summarize that, when specializing to two spacetime dimensions, the construction approach by deformation of a scalar massive Fermion gives rise to the model theories constructed by inverse scattering methods in Chapter \ref{Chapter3}. The particle spectrum consists thereby only of a single species of neutral massive particles. Our analysis includes further the class of integrable models with scattering functions satisfying $S_\lambda(0)=-1$ into the deformation framework. A prominent model arising in this context is the Sinh-Gordon model, cf. Section \ref{Examples}.
 
% Chapter7

\chapter{Conclusions and Outlook} % Main chapter title

\label{Chapter7} % Change X to a consecutive number; for referencing this chapter elsewhere, use \ref{ChapterX}
\fancyhead[LE,RO]{\thepage}
\fancyhead[LO]{}
\fancyhead[RE]{Chapter 7. \emph{Conclusions and Outlook}}
\renewcommand{\chaptermark}[1]{ \markboth{#1}{} }
\renewcommand{\sectionmark}[1]{ \markright{#1}{} }
%\lhead{Chapter 7. \emph{Conclusions and Outlook}} % Change X to a consecutive number; this is for the header on each page - perhaps a shortened title

The present thesis establishes the construction of nontrivial quantum field theoretical models in $d\geq 2$-dimensional Minkowski space within the framework of Local Quantum Physics.\par
Two different approaches are introduced. On the one hand, inverse scattering methods are applied which give rise to interacting theories in $d=2$ spacetime dimensions. On the other hand, deformation procedures which can be viewed as generalizations of the latter technique yield nontrivial models in higher, i.e. $d\geq 2$, dimensions. In both cases the simple nature of the underlying scattering operator does not allow for particle production or momentum transfer in collision processes of particles, however, effects like time delays can appear.\par
In the case of the inverse scattering approach we started from a prescribed factorizing, i.e. purely elastic, collision operator which is determined by its two-particle scattering amplitude $S$, called S-matrix for short. The considered particle spectrum involves an arbitrary finite number of massive particle species which transform under some global gauge group. Building on recent results concerning the successful construction of a covariant field net of wedge algebras in this setting, we investigated appropriate intersections of such wedge algebras. Our main interest was to show that the emerging local algebras, associated with bounded regions in Minkowski space, comply with all basic assumptions of relativistic quantum physics. For this purpose we looked into the important question whether certain methods, applied earlier for the construction of models with \textit{scalar}-valued S-matrices, can be generalized effectively to the more involved framework at hand. By verifying the modular nuclearity condition we not only succeeded in proving the existence of a meaningful local theory related to an initial S-matrix $S$ in a certain family $\mathcal{S}_0^-$ of regular S-matrices, but also answered the above question in a positive way. The small gap, which however remains in our construction, arises from problems discovered very recently for the scalar case. In view of plausible arguments presented here, we are certain that a bridge can be built, giving, in particular, rise to the $O(N)$-invariant nonlinear $\sigma$-models in a rigorous way, independent of perturbation theory and renormalization, for the first time. Moreover, by our analysis a large class of integrable models emerges to which no Lagrangian formulation is known.\par
These results open up the possibility to generalize the applied methods further to a larger family of initial S-matrices. For instance, it would be particularly interesting to include factorizing S-matrices with poles in the physical sheet in the framework. A famous representative corresponding to such an S-matrix is the Sine-Gordon model. In fact, there is an ongoing project by Y. Tanimoto and D. Cadamuro in that direction and preliminary results concerning the construction of wedge-local fields do already exist. However, there are still many open questions with regard to the verification of the modular nuclearity condition.\par
Generalizations of the inverse scattering techniques to $d\geq 2$ spacetime dimensions can be achieved by deformation procedures. We concretely investigated the deformation a scalar massive Fermion. The emerging models were shown to be based on wedge-local temperate polarization-free generators which admit the consistent computation of the two-particle S-matrix. The latter is shown to depend on the deformation and to differ from the two-particle scattering operator of the undeformed theory in a nontrivial way. This implies that the deformed theory is \textit{not} equivalent to the initial, undeformed one.\par
Our analysis thereby includes a certain class of integrable quantum field theory models into the deformation framework in \textit{two}-dimensional Minkowski space. Namely, those integrable models whose factorizing S-matrices are completely determined by scattering \textit{functions} $S$ satisfying $S(0)=-1$. For example, the scattering function of the Sinh-Gordon model belongs to this class.\par
The establishment of locality properties of the deformed model turns out to be a difficult task as the undeformed model is already nonlocal. In two dimensions, however, it is possible to achieve wedge-locality by imposing an additional condition on the deformation function $R$. Moreover, we show with reference to the inverse scattering approach that the deformed theory also admits local observables in $d=2$ which comply with localization in bounded regions above a minimal size. Analogous results for higher dimensions have not been achieved up to now. This problem emerges from the fact that the split property for wedges enters the stage and is of crucial importance with regard to the modular nuclearity condition. This property does, however, not hold in more than two spacetime dimensions and a direct generalization of the applied methods to higher dimensions is not possible.\par
On the other hand, relative commutants of wedge algebras are not necessarily trivial if the split property for wedges does not hold. In fact, this condition is rather strong with regard to a mere interest in the existence of local fields.\par
One can also ask if some of the conditions on the deformation function $R$ can be relaxed. As part of our analysis it turned out that from the physical point of view the deformed theory does not depend on the sign of the function $R$. In particular, two nets arising from deformation with deformation functions $R$ and $-R$ respectively are unitarily equivalent. This result generalizes the deformation procedure of \cite{GL} since in that work one requires that the function $R$ satisfies the condition $R(0)=1$ which by our result is redundant. Since the latter condition is a consequence of the deformation of the underlying Borchers-Uhlmann algebra \cite{GL}, the deformation approach presented here extends the possibilities for obtaining new models by deformation techniques.\par
Due to the simple form of the deformation the encountered interaction in the deformed theory does not involve momentum transfer or particle production. In order to realize these interactions, a generalization of the deformation techniques themselves is desirable. Concrete realizations of this task constitute suitable integral operators which can be considered in replacement of the multiplication operators $T_R$ used to deform the model of a scalar massive Fermion. Such operators are, however, much more difficult to cope with. Thus, there are no satisfactory results available so far.\par
A successful establishment of such a deformation procedure or the discovery of a suitable generalization of the modular nuclearity condition to four spacetime dimensions would finally solve the long-standing problem of rigorously constructing interacting quantum field theories in physical spacetime.

%----------------------------------------------------------------------------------------
%	THESIS CONTENT - APPENDICES
%----------------------------------------------------------------------------------------

\addtocontents{toc}{\vspace{2em}} % Add a gap in the Contents, for aesthetics

\appendix % Cue to tell LaTeX that the following 'chapters' are Appendices

% Include the appendices of the thesis as separate files from the Appendices folder
% Uncomment the lines as you write the Appendices

% Appendix A

\chapter{Auxiliary Results} % Main appendix title

\label{AppendixA} % For referencing this appendix elsewhere, use \ref{AppendixA}
\fancyhead[LE,RO]{\thepage}
\fancyhead[LO]{\thesection. \emph{\rightmark}}
\fancyhead[RE]{Appendix A. \emph{Auxiliary Results}}
\renewcommand{\chaptermark}[1]{ \markboth{#1}{} }
\renewcommand{\sectionmark}[1]{ \markright{#1}{} }
%\lhead{Appendix A. \emph{Auxiliary Results}} % This is for the header on each page - perhaps a shortened title
In this appendix we state some general results which have concrete applications, in particular, in Chapter \ref{Chapter3}.

\section{General Results}

\begin{lemma}\label{AllgemeinAnalyt}
Let a sequence $\{\mathfrak{h}_k\}_{0\leq k\leq n-1}$, of tempered distributions on $\mathscr{S}(\mathbb{R}^n)\otimes\mathcal{K}^{\otimes n}$ be given. Moreover, let $\mathfrak{h}_k$ have an analytic continuation in the variable $x_{k+1}$ to the strip $S(0,a)$, $a\in\mathbb{R}^+$, in the sense that
\begin{equation}
\lim\limits_{\lambda\rightarrow 0}\int d^n\boldsymbol{x}\,\,\mathfrak{h}^{\boldsymbol{\alpha}}_k(x_1,\dots,x_{k+1}+i\lambda,\dots,x_n) g^{\boldsymbol{\alpha}}(\boldsymbol{x})=\mathfrak{h}_k(g),\qquad g\in\mathscr{S}(\mathbb{R}^n)\otimes\mathcal{K}^{\otimes n},
\end{equation}
with $0<\lambda<a$. Further, its boundary value at Im$(x_{k+1})=a$ is given by
\begin{equation}
\mathfrak{h}^{\alpha_1\dots\alpha_n}_k(x_1,\dots,x_{k+1}+ia,\dots,x_n)=\mathfrak{h}^{\alpha_2\dots\alpha_n\alpha_1}_{k+1}(x_1,\dots,x_{k+1},\dots,x_n),
\end{equation}
in the sense of distributions. Then, $\mathfrak{h}_0$ is the distributional boundary value of a function analytic in the tube
\begin{equation}
\mathfrak{T}_n:=\mathbb{R}^n+i\{\boldsymbol{y}\in\mathbb{R}^n:0<y_n<\cdots<y_2<y_1<a\}.
\end{equation}
\end{lemma}\begin{proof}
We start by considering $f\in\mathscr{S}(\mathbb{R}^n)$. Proceeding by induction in $k\in\{1,\dots,n\}$, we next prove that $\mathfrak{h}_0^{\boldsymbol{\alpha}}\ast f$, regarded as a function of $x_1,\dots,x_k$ and $x_{k+1},\dots,x_n\in\mathbb{R}$ held fixed, is analytic in the tube $\mathfrak{T}_k$. Since $\mathfrak{T}_1=S(0,a)$, the assertion for $k=1$ follows directly from the assumptions of the Lemma. For the inductive step $k\rightarrow k+1$ note first that the boundary value of $(x_1,\dots,x_k)\mapsto\left(\mathfrak{h}_0^{\boldsymbol{\alpha}}\ast f\right)(x_1,\dots,x_n)$ at Im$(x_1)=\dots=\,\,$Im$(x_k)=a$ is given by $\mathfrak{h}_{k}^{\alpha_{k+1}\dots\alpha_n\alpha_1\dots,\alpha_k}(\boldsymbol{x})$ which again has an analytic continuation in $x_{k+1}$ to $S(0,a)$. Thus, we may apply the Malgrange Zerner theorem \cite{eps66} and conclude that $(x_1,\dots,x_{k+1})\mapsto\left(\mathfrak{h}_0^{\boldsymbol{\alpha}}\ast f\right)(x_1,\dots,x_n)$, can be analytically continued to the convex closure of
\begin{equation*}
\mathbb{R}^{k+1}+i\Big(\{(y_1,\dots,y_k,0):0<y_k<\dots<y_1<a\}\cup\{(a,\dots,a,y_{k+1}):0<y_{k+1}<a\}\Big),
\end{equation*}
which agrees with $\mathfrak{T}_{k+1}$. Therefore, it follows that $\mathfrak{h}_0^{\boldsymbol{\alpha}}\ast f$ is analytic in $\mathfrak{T}_n$ and, since $f$ is arbitrary, $\mathfrak{h}_0$ is the distributional boundary value of a function analytic in $\mathfrak{T}_n$ and denoted by the same symbol.
\end{proof}
\begin{lemma}\label{traceclassneu}
Let $R_{g,b}=\bigoplus_{\alpha=1}^D R_{g,b}^{[\alpha]}$ be an integral operator on $L^2(\mathbb{R})\otimes\mathcal{K}$, $\dim\mathcal{K}=D<\infty$, which is defined by the kernels
\begin{equation}\label{kernels1}
R^{[\alpha]}_{g,b}(\theta,\theta')=\frac{-\text{sign}\,(b)}{2\pi i}\frac{\overline{g(\theta)}g(\theta')}{\theta'-\theta+ib},\qquad g\in L^2(\mathbb{R}),\,\,b\in\mathbb{R}\backslash\{0\}.
\end{equation}
Then, $R_{g,b}$ is a positive trace class operator with trace norm bounded by
\begin{equation}\label{tracenormneu}
\|R_{g,b}\|_1\leq D\|R_{g,b}^{[\alpha]}\|_1=D\frac{\|g\|_2^2}{2\pi |b|}.
\end{equation}
\end{lemma}
\begin{proof}
The assertion can be proven along the same lines as in the scalar case \cite[Lemma E.2]{erratum}. For the convenience of the reader, and since reference \cite{erratum} has not been published yet, we give a full proof.\par
Note first that $R_{g,b}=\widehat{U}R_{g_-,-b}\widehat{U}$, with unitary $(\widehat{U}f)(\theta):=i\cdot f(-\theta)$, $f\in L^2(\mathbb{R})\otimes\mathcal{K}$, and $g_-(\theta):=g(-\theta)$. Due to this unitary equivalence it suffices to consider $b>0$.\par
To show that $R_{g,b}$ is positive, define $K_b(\theta):=-(2\pi i)^{-1}(\theta+ib)^{-1}$ which has positive Fourier transform $\widetilde{K}_b(\eta)=\Theta(\eta)e^{-b\eta}$. Then, we have with $f\in L^2(\mathbb{R})\otimes\mathcal{K}$
$$\langle f,R_{g,b}f\rangle=\sum_{\alpha=1}^{D}\langle(g\cdot f^\alpha),K_b\ast (g\cdot f^\alpha)\rangle_{L^2(\mathbb{R})}=\sqrt{2\pi}\sum_{\alpha=1}^{D}\langle\widetilde{(g\cdot f^\alpha)},\widetilde{K_b}\cdot \widetilde{(g\cdot f^\alpha)}\rangle_{L^2(\mathbb{R})}\geq 0,$$
yielding positivity. Since
$$-\frac{1}{2\pi i}\int d\theta\frac{\overline{g(\theta)}g(\theta)}{\theta-\theta+ib}=\frac{\|g\|_2^2}{2\pi b},$$
it follows by \cite[Lemma on p.65]{RS3} that $R_{g,b}^{[\alpha]}$ is trace class with trace $\frac{\|g\|_2^2}{2\pi b}$, proving the claim.
\end{proof}

\section{An Alternative Proof of Lemma \ref{LemmaVorbereitung01}}\label{LemmaVor01Appendix}
We want to present an alternative proof of Lemma \ref{LemmaVorbereitung01} with regard to the explicit realization of permutations in terms of transpositions. For the convenience of the reader we restate the Lemma.
\begin{lemma*}
Let $C\in\hat{\mathscr{C}}_{n,k}$, where $\hat{\mathscr{C}}_{n,k}\subset\mathscr{C}_{n,k}$ denotes the set of contractions $C\in\mathscr{C}_{n,k}$ for which $k+1\notin\{ l_1,\dots,l_{|C|}\}$, and let $C'\in\check{\mathscr{C}}_{n,k}$, with $\check{\mathscr{C}}_{n,k}\subset\mathscr{C}_{n,k}$ the set of contractions $C'\in\mathscr{C}_{n,k}$ for which $k+1\in\{ l'_1,\dots,l'_{|C'|}\}$, such that $C'=C\cup\{(k+1,r)\}$ with $r\notin\{r_1,\dots,r_{|C|}\}$, then we have
\begin{eqnarray*}
\pi_{\rho'}&=&\pi_\rho\cdot\pi_{r-v_r}^{k-v_r},\\
\pi_{\lambda'}&=&\pi_\lambda\cdot \pi_{k+1+|C|}^{k+1+v_r},
\end{eqnarray*}
hence,
\begin{equation*}
\pi_{C'}=\pi_C\cdot\pi_{r-v_r}^{k-v_r}\cdot\pi_{k+1+|C|}^{k+1+v_r},
\end{equation*}
with $v_r:=$card$\left(\{r_i\in\{r_1,\dots,r_{|C|}\}:r_i<r\}\right)$.
\end{lemma*}
\begin{proof}
Note first that $|C'|=|C|+1$, $\{r'_1,\dots,r'_{|C'|}\}=\{r_1,\dots,r_{v_r},r,r_{v_r+1}\dots,r_{|C|}\}$ and $\{l'_1,\dots,l'_{|C'|}\}=\{l_1,\dots,l_{v_r},k+1,l_{v_r+1},\dots,l_{|C|}\}$, with $0\leq v_r\leq |C|$. Using these relations, we find, with $\pi_{C'}=\pi_{\rho'}\cdot\pi_{\lambda'}$,
\begin{eqnarray*}
\pi_{\rho'}&=& \left(\prod_{i=1}^{v_r}\tau_{r_i-i+1}\cdots\tau_{k-i}\right)\cdot\left(\tau_{r-v_r}\cdots\tau_{k-1-v_r}\right)\cdot\left(\prod_{j=v_r+1}^{|C|}\tau_{r_j-j}\cdots\tau_{k-1-j}\right)\nonumber\\
\pi_{\lambda'}&=&\left(\prod_{s=1}^{v_r}\tau_{l_s+u_s-1}\cdots\tau_{k+s}\right)\cdot\left(\prod_{t=v_r+1}^{|C|}\tau_{l_t+u_t-1}\cdots\tau_{k+1+t}\right).\nonumber
\end{eqnarray*}
In order to prove what is claimed, we make use of the following properties. Namely, for $\tau_i\in\mathfrak{S}_n$, $i=1,\dots,n-1$, we have
\begin{equation}\label{permutationgroup}
\begin{aligned}
\tau_i^2&=id,\\
\tau_i\cdot\tau_j&=\tau_j\cdot\tau_i,\qquad\text{for}\qquad |i-j|>1,\\
\tau_{i}\cdot\tau_{i+1}\cdot\tau_{i}&=\tau_{i+1}\cdot\tau_{i}\cdot\tau_{i+1},\qquad i=1,\dots,n-2.
\end{aligned}
\end{equation}
Employing these relations, in particular the last two ones, and due to the fact that $r<r_{v_r+1}$, we reformulate the following partial product appearing for $\pi_{\rho'}$
\begin{eqnarray*}
\hspace{-0.8cm}&&\hspace{-0.8cm} \left(\tau_{r-v_r}\cdots\tau_{k-1-v_r}\right)\cdot\left(\tau_{r_{v_r+1}-(v_r+1)}\dots\tau_{k-1-(v_r+1)}\right)\nonumber\\
&=& \left(\tau_{r-v_r}\cdots\tau_{r_{v_r+1}-v_r-2}\cdot\textcolor{blue}{\tau_{r_{v_r+1}-v_r-1}\cdot\tau_{r_{v_r+1}-v_r}}\cdot\tau_{r_{v_r+1}-v_r+1}\cdots\tau_{k-1-v_r}\right)\nonumber\\
\hspace{.3cm}&&\hspace{.3cm}\times \left(\textcolor{blue}{\tau_{r_{v_r+1}-v_r-1}}\cdot\tau_{r_{v_r+1}-v_r}\dots\tau_{k-2-v_r}\right)\nonumber\\
&=& \tau_{r_{v_r+1}-v_r}\cdot(\tau_{r-v_r}\cdots\tau_{r_{v_r+1}-v_r-2})\cdot(\tau_{r_{v_r+1}-v_r-1}\cdot\textcolor{blue}{\tau_{r_{v_r+1}-v_r}\cdot\tau_{r_{v_r+1}-v_r+1}})\nonumber\\
\hspace{.3cm}&&\hspace{.3cm}\times
(\tau_{r_{v_r+1}-v_r+2}\cdots\tau_{k-1-v_r})\cdot\left(\textcolor{blue}{\tau_{r_{v_r+1}-v_r}}\cdot\tau_{r_{v_r+1}-v_r+1}\dots\tau_{k-2-v_r}\right)\nonumber\\
&=& (\tau_{r_{v_r+1}-v_r}\cdot\tau_{r_{v_r+1}-v_r+1})\cdot(\tau_{r-v_r}\cdots\tau_{r_{v_r+1}-v_r-1})\nonumber\\
\hspace{.3cm}&&\hspace{.3cm}\times
(\tau_{r_{v_r+1}-v_r}\cdot\textcolor{blue}{\tau_{r_{v_r+1}-v_r+1}\cdot\tau_{r_{v_r+1}-v_r+2}}\cdots\tau_{k-1-v_r})\cdot\left(\textcolor{blue}{\tau_{r_{v_r+1}-v_r+1}}\dots\tau_{k-2-v_r}\right)\nonumber\\
&=&\dots\nonumber\\
&=&\left(\tau_{r_{v_r+1}-(v_r+1)+1}\dots\tau_{k-(v_r+1)}\right)\cdot\left(\tau_{r-v_r}\cdots\tau_{k-1-v_r}\right),
\end{eqnarray*}
where to the blue highlighted transpositions the third relation listed in (\ref{permutationgroup}) was applied. Proceeding in this way, we arrive at
\begin{eqnarray*}
\pi_{\rho'}&=& \left(\prod_{i=1}^{v_r}\tau_{r_i-i+1}\cdots\tau_{k-i}\right)\cdot\left(\prod_{j=v_r+1}^{|C|}\tau_{r_j-j+1}\cdots\tau_{k-j}\right)\cdot\left(\tau_{r-v_r}\cdots\tau_{k-1-v_r}\right)\nonumber\\
&=&\pi_{\rho}\cdot\left(\prod_{i=r}^{k-1}\tau_{i-v_r}\right).
\end{eqnarray*}
It remains to establish the corresponding result for $\pi_{\lambda'}$. To this end, we take a closer look at the second factor appearing in $\pi_{\lambda'}$, namely
\begin{eqnarray*}
\hspace{-0.8cm}&&\hspace{-0.8cm} \prod_{i=v_r+1}^{|C|}\tau_{l_i+u_i-1}\cdots\tau_{k+1+i}\\
&&=(\tau_{l_{v_r+1}+u_{v_r+1}-1}\cdots\tau_{k+2+v_r})\cdot(\tau_{k+1+v_r}\cdot\tau_{k+1+v_r})\\
\hspace{.3cm}&&\hspace{.3cm}\times
(\tau_{l_{v_r+2}+u_{v_r+2}-1}\cdots\tau_{k+3+v_r})\cdot(\tau_{k+2+v_r}\cdot\tau_{k+2+v_r})\cdots (\tau_{l_{|C|}+u_{|C|}-1}\cdots\tau_{k+1+|C|})\\
\hspace{.3cm}&&\hspace{.3cm}\times
(\tau_{k+|C|}\cdot\tau_{k+|C|})\\
&&= \left(\prod_{i=v_r+1}^{|C|}\tau_{l_i+u_i-1}\cdots\tau_{k+i}\right)\cdot\left(\prod_{j=1}^{|C|-v_r}\tau_{k+1+|C|-j}\right),
\end{eqnarray*}
where we made use of the fact that $\tau_i^2=id$. Hence, we have
\begin{equation*}
\pi_{\lambda'}=\pi_\lambda\cdot\left(\prod_{j=1}^{|C|-v_r}\tau_{k+1+|C|-j}\right),
\end{equation*}
and
\begin{eqnarray*}
\pi_{C'}&=&\pi_{\rho}\cdot\left(\prod_{i=r}^{k-1}\tau_{i-v_r}\right)\cdot\pi_\lambda\cdot\left(\prod_{j=1}^{|C|-v_r}\tau_{k+1+|C|-j}\right)\\
&=&\pi_{\rho}\cdot\pi_\lambda\cdot\left(\prod_{i=r}^{k-1}\tau_{i-v_r}\right)\cdot\left(\prod_{j=1}^{|C|-v_r}\tau_{k+1+|C|-j}\right)=\pi_C\cdot\pi_{r-v_r}^{k-v_r}\cdot\pi_{k+1+|C|}^{k+1+v_r},
\end{eqnarray*}
due to the commutativity of $\pi_{\rho'}$ with $\pi_{\lambda'}$ and by means of (\ref{schiebePermut}).\end{proof}

% Appendix B

\chapter{Background Material} % Main appendix title

\label{AppendixB} % Change X to a consecutive letter; for referencing this appendix elsewhere, use \ref{AppendixX}
\fancyhead[LE,RO]{\thepage}
\fancyhead[LO]{\thesection. \emph{\rightmark}}
\fancyhead[RE]{Appendix B. \emph{Background Material}}
\renewcommand{\chaptermark}[1]{ \markboth{#1}{} }
\renewcommand{\sectionmark}[1]{ \markright{#1}{} }
%\lhead{Appendix B. \emph{Background Material}} % Change X to a consecutive letter; this is for the header on each page - perhaps a shortened title

This appendix serves as a glossary for several mathematical topics, relied on in the main text. We shall refrain from giving proofs, however refer to the respective literature.\\

\section{Nuclear Maps}
This section provides results concerning nuclear maps which are of crucial importance in Chapter \ref{Chapter3}. Complementary material and proofs to the following statements can be found in e.g. \cite{jarchow,pietsch}.
\begin{definition}
Let $E$ and $F$ be Banach spaces. A mapping $T\in\mathcal{L}(E,F)$ is called nuclear if there is a sequence of linear functionals $\{a_n\}_{n\in\mathbb{N}}\subset E^*$ and a sequence of vectors \mbox{$\{y_n\}_{n\in\mathbb{N}}\subset F$} with
\begin{equation}
\sum_{n=1}^{\infty}\|a_n\|_{E^*}\,\|y_n\|_F<\infty,
\end{equation}
such that
\begin{equation}\label{nucleareAbb}
T(x)=\sum_{n=1}^{\infty}a_n(x)\,y_n,\qquad x\in E.
\end{equation}
The nuclear norm of such a linear map is defined by
\begin{equation}\label{normnuc}
\|T\|_1:=\inf\sum_{n=1}^{\infty}\|a_n\|_{E^*}\,\|y_n\|_F,
\end{equation}
where the infimum is taken over all possible representations (\ref{nucleareAbb}) of $T$.
\end{definition}
The sets of nuclear, compact and bounded maps between two Banach spaces $E$ and $F$ are, further, denoted by $\mathcal{N}(E,F)$, $\mathscr{K}(E,F)$ and $\mathcal{B}(E,F)$ respectively. The properties collected in Lemma \ref{propertiesNuclearMaps} are frequently used in Section \ref{VerifyingNuclearity}.
\begin{lemma}\label{propertiesNuclearMaps}
Let $E,F,G,H$ be Banach spaces. Then, we have
\begin{itemize}
\item[i)] $\|T\|\leq\|T\|_1$ with $T\in\mathcal{N}(E,F)$,
\vspace{.2cm}
\item[ii)] $\mathcal{N}(E,F)\subset\mathscr{K}(E,F)$,
\vspace{.2cm}
\item[iii)] $\left(\mathcal{N}(E,F),\|\cdot\|_1\right)$ is a Banach space,
\vspace{.2cm}
\item[iv)] for $T\in\mathcal{N}(E,F)$, $B_1\in\mathcal{B}(F,G)$ and $B_2\in\mathcal{B}(H,E)$ that $B_1TB_2\in\mathcal{N}(H,G)$ and
\begin{equation}
\|B_1TB_2\|_1\leq\|B_1\|\,\|T\|_1\,\|B_2\|.
\end{equation}
\end{itemize}
Moreover, let $\mathscr{H}$ be a separable Hilbert space. Then,
\begin{itemize}
\item[v)] $\mathcal{N}(\mathscr{H},\mathscr{H})$ agrees with the set of trace class operators on $\mathscr{H}$, and further
\begin{equation*}
\|T\|_1=\text{Tr}\,\,|T|,\qquad T\in\mathcal{N}(\mathscr{H},\mathscr{H}).
\end{equation*}
\end{itemize}
\end{lemma}
By means of a standard argument \cite[Theorem XI.21]{RS3}, the following statement can be proven.
\begin{lemma}\label{traceclass}
Let $T_{a,b}=\bigoplus_{\alpha=1}^D T_{a,b}^{[\alpha]}$ be an integral operator on $L^2(\mathbb{R})\otimes\mathcal{K}$, $\dim\mathcal{K}=D<\infty$, which is defined by the kernels
\begin{equation}\label{kernels}
T^{[\alpha]}_{a,b}(\theta,\theta')=\frac{e^{-a\cosh \theta}}{\theta'-\theta+ib},\qquad a>0,\,\,b\in\mathbb{R}\backslash\{0\}.
\end{equation}
Then, $T_{a,b}$ is of trace class for any $a>0$ and $b\in\mathbb{R}\backslash\{0\}$ with trace norm bounded by
\begin{equation}\label{tracenorm}
\|T_{a,b}\|_1\leq D \cdot\|T^{[\alpha]}_{a,b}\|_1\leq D\cdot 2^{1/4}\cdot\pi^{5/4}\cdot\frac{e^{-a}}{a^{1/4}}\cdot\left[\left(\sqrt{\frac{\pi}{2}}+\frac{1}{4a}\right)\cdot\frac{b^4+4b^2+24}{|b|^5}\right]^{1/2}.
\end{equation}
\end{lemma}
This assertion can be proven along the same lines as in the case $D=1$ which was carried out in \cite[Appendix B.2]{DocL}. We only mention that the estimate on the trace norm of $T_{a,b}$ is a direct consequence of the fact that this integral operator may be expressed as the direct sum $\bigoplus_{\alpha=1}^D T_{a,b}^{[\alpha]}$, where $T_{a,b}^{[\alpha]}$ acts on $L^2(\mathbb{R})$ and is of the form (\ref{kernels}).
\section{The Hardy Space $H^2(\mathcal{T}_{\mathcal{C}},\mathcal{K}^{\otimes n})$}\label{hardyAppendix}
In this section we collect general results on the Hardy space $H^2$ of analytic functions on tube domains $\mathcal{T}_{\mathcal{C}}$ with values in a finite dimensional Hilbert space $\mathcal{K}^{\otimes n}$ as considered in Chapter \ref{Chapter3}. In the literature  mainly the case $\dim\mathcal{K}=1$ is treated. However, the properties we are interested in generalize directly to vector-valued functions. We, thus, refer the reader to \cite{stein1971introduction} for a thorough discussion of this topic and further to \cite[Appendix C]{DocL} for a similar compilation including proofs.
\begin{definition}\label{d}
Consider the tube
\begin{equation}
\mathcal{T}_{\mathcal{C}}:=\mathbb{R}^n+i\mathcal{C},
\end{equation}
with base $\mathcal{C}\subset\mathbb{R}^n$ being open and convex. Further, let $\mathcal{K}^{\otimes n}$ be a finite dimensional Hilbert space. Then, we denote by $H^2(\mathcal{T}_{\mathcal{C}},\mathcal{K}^{\otimes n})$ the Hardy space consisting of functions $h:\mathcal{T}_{\mathcal{C}}\rightarrow\mathcal{K}^{\otimes n}$ which are analytic in $\mathcal{T}_{\mathcal{C}}$ and, moreover, satisfy
\begin{equation}
h_{\boldsymbol{\lambda}}\in L^2(\mathbb{R}^n,d^n\boldsymbol{\theta};\mathcal{K}^{\otimes n}),\qquad h_{\boldsymbol{\lambda}}:\boldsymbol{\theta}\mapsto h(\boldsymbol{\theta}+i\boldsymbol{\lambda}),\qquad\boldsymbol{\lambda}\in\mathcal{C},
\end{equation}
as well as
\begin{equation}\label{hardynor}
\triplenorm h\triplenorm:=\sup\limits_{\boldsymbol{\lambda}\in\mathcal{C}}\left(\int_{\mathbb{R}^n}d^n\boldsymbol{\theta}\,\|h(\boldsymbol{\theta}+i\boldsymbol{\lambda})\|^2_{\mathcal{K}^{\otimes n}}\right)^{1/2}<\infty.
\end{equation}
\end{definition}
Noting that $H^2(\mathcal{T}_{\mathcal{C}},\mathcal{K}^{\otimes n})\simeq H^2(\mathcal{T}_{\mathcal{C}})\otimes \mathcal{K}^{\otimes n}$ and considering an orthonormal basis $\{e^\alpha:\alpha=1,\dots,D=\dim\mathcal{K}\}$ in $\mathcal{K}$ we state the following results.
\begin{proposition}\label{hardyeigenschaften}
\begin{itemize}
\item[]
\item[i)] $\left(H^2(\mathcal{T}_{\mathcal{C}},\mathcal{K}^{\otimes n}),\triplenorm\cdot\triplenorm\right)$ is a Banach space.
\vspace{.2cm}
\item[ii)] Let $h\in H^2(\mathcal{T}_{\mathcal{C}})\otimes\mathcal{K}^{\otimes n}$, $K\subset \mathcal{C}$ compact and $k=1,\dots,n$, then
\begin{equation}
\lim\limits_{|\theta_k|\rightarrow\infty}\sup\limits_{\boldsymbol{\lambda}\in K}|h^{\boldsymbol{\alpha}}(\boldsymbol{\theta}+i\boldsymbol{\lambda})|=0,
\end{equation}
with $\theta_1,\dots,\theta_{k-1},\theta_{k+1},\dots,\theta_n\in\mathbb{R}$.
\item[iii)] Let $h\in H^2(\mathcal{T}_{\mathcal{P}},\mathcal{K}^{\otimes n})$ and $\mathcal{P}$ be an open polyhedron, that is, the interior of the convex hull of a finite subset of $\mathbb{R}^n$. Then, $h$ can be extended to $\mathcal{T}_{\overline{\mathcal{P}}}$ such that the mapping $\overline{\mathcal{P}}\ni \boldsymbol{\lambda}\mapsto h_{\boldsymbol{\lambda}}\in L^2(\mathbb{R}^n,\mathcal{K}^{\otimes n})$ is continuous.
\end{itemize}
\end{proposition}

\cleardoublepage

%----------------------------------------------------------------------------------------
%	LIST OF CONTENTS/FIGURES/TABLES PAGES
%----------------------------------------------------------------------------------------

%\pagestyle{fancy} % The page style headers have been "empty" all this time, now use the "fancy" headers as defined before to bring them back

\addtocontents{toc}{\vspace{2em}} % Add a gap in the Contents, for aesthetics

%\lhead{\emph{List of Figures}} % Set the left side page header to "List of Figures"
\listoffigures % Write out the List of Figures

%\lhead{\emph{List of Tables}} % Set the left side page header to "List of Tables"
%\listoftables % Write out the List of Tables
\cleardoublepage

%----------------------------------------------------------------------------------------
%	SYMBOLS
%----------------------------------------------------------------------------------------
\rhead{\thepage}
%\clearpage % Start a new page
%\pagestyle{empty}
\lhead{\emph{Frequently Used Symbols}} % Set the left side page header to "Symbols"
\addtotoc{Frequently Used Symbols}
\chapter*{Frequently Used Symbols}
\begin{tabular}[hc]{c c r}
\multirow{1}{3cm}{Symbol} & \multirow{1}{8.5cm}{Description} & \multirow{1}{2cm}{Page}\\
  \hline
\multirow{1}{3cm}{$\|\cdot\|_1$} & \multirow{1}{8.5cm}{nuclear norm} & \multirow{1}{2cm}{\pageref{normnuc}}\\
\multirow{1}{3cm}{$\triplenorm\cdot\triplenorm$} & \multirow{1}{8.5cm}{Hardy norm} & \multirow{1}{2cm}{\pageref{hardynor}}\\
\multirow{1}{3cm}{$\widehat{a}^\dagger$, $\widehat{a}$} & \multirow{1}{8.5cm}{creation and annihilation operators} & \multirow{1}{2cm}{\pageref{cr}}\\
\multirow{1}{3cm}{$a^*$, $a$} & \multirow{1}{8.5cm}{fermionic creation and annihilation operators} & \multirow{1}{2cm}{\pageref{rc}}\\
\multirow{1}{3cm}{$a_{R,Q}^*$, $a_{R,Q}$} & \multirow{1}{8.5cm}{deformed creation and annihilation operators} & \multirow{1}{2cm}{\pageref{def1}}\\
\multirow{1}{3cm}{$C$, $|C|$} & \multirow{1}{8.5cm}{contraction and its length} & \multirow{1}{2cm}{\pageref{oneparticle}}\\
\multirow{1}{3cm}{$\mathscr{C}_{n,k}$} & \multirow{1}{8.5cm}{the set of contractions} & \multirow{1}{2cm}{\pageref{oneparticle}}\\
\multirow{1}{3cm}{$D$} & \multirow{1}{8.5cm}{dimension of $\mathcal{K}$} & \multirow{1}{2cm}{\pageref{hilbert1}}\\
\multirow{1}{3cm}{$D_n$} & \multirow{1}{8.5cm}{representation of $\mathfrak{S}_n$} & \multirow{1}{2cm}{\pageref{reprpermugroup}}\\
\multirow{1}{3cm}{$\mathcal{D}$} & \multirow{1}{8.5cm}{subspace of $\mathscr{H}$ of finite particle number} & \multirow{1}{2cm}{\pageref{opN}}\\
\multirow{1}{3cm}{$\mathcal{D}^-$} & \multirow{1}{8.5cm}{subspace of $\mathscr{H}^-$ of finite particle number} & \multirow{1}{2cm}{\pageref{prop1}}\\
\multirow{1}{3cm}{$\Delta$} & \multirow{1}{8.5cm}{modular operator of $(\mathcal{F}(W_R),\Omega)$} & \multirow{1}{2cm}{\pageref{oneparticletomita}}\\
\multirow{1}{3cm}{$\mathcal{F}$} & \multirow{1}{8.5cm}{field net of local algebras} & \multirow{1}{2cm}{\pageref{AQFT}}\\
\multirow{1}{3cm}{$\phi$, $\phi'$} & \multirow{1}{8.5cm}{wedge-local quantum field} & \multirow{1}{2cm}{\pageref{SecWedgeLocal}}\\
\multirow{1}{3cm}{$\underline{\phi}$} & \multirow{1}{8.5cm}{wedge-local quantum field} & \multirow{1}{2cm}{\pageref{feld}}\\
\multirow{1}{3cm}{$\widehat{\phi}$} & \multirow{1}{8.5cm}{wedge-local quantum field} & \multirow{1}{2cm}{\pageref{hat}}\\
\multirow{1}{3cm}{$G$} & \multirow{1}{8.5cm}{gauge group} & \multirow{1}{2cm}{}\\
\multirow{1}{3cm}{$\widehat{\mathscr{H}}$} & \multirow{1}{8.5cm}{unsymmetrized Fock space over $\mathscr{H}_1$} & \multirow{1}{2cm}{\pageref{fockUnsy}}\\
\multirow{1}{3cm}{$\mathscr{H}$} & \multirow{1}{8.5cm}{S-symmetric Fock space} & \multirow{1}{2cm}{\pageref{s}}\\
\multirow{1}{3cm}{$\mathscr{H}_n$} & \multirow{1}{8.5cm}{n-particle spaces} & \multirow{1}{2cm}{\pageref{s}}\\
\multirow{1}{3cm}{$\mathscr{H}^\pm$} & \multirow{1}{8.5cm}{Bose/Fermi Fock spaces over $\mathscr{H}_1$} & \multirow{1}{2cm}{\pageref{fockUnsy}}\\
\multirow{1}{3cm}{$H^2(\mathcal{T}_{\mathcal{C}},\mathcal{K}^{\otimes n})$} & \multirow{1}{8.5cm}{Hardy space over the tube $\mathcal{T}_{\mathcal{C}}$} & \multirow{1}{2cm}{\pageref{d}}\\
\multirow{1}{3cm}{$J$} & \multirow{1}{8.5cm}{modular conjugation of $(\mathcal{F}(W_R),\Omega)$} & \multirow{1}{2cm}{\pageref{oneparticletomita}}\\
\multirow{1}{3cm}{$\mathcal{K}$} & \multirow{1}{8.5cm}{finite dimensional Hilbert space} & \multirow{1}{2cm}{\pageref{hilbert1}}\\
\multirow{1}{3cm}{$\kappa(S)$} & \multirow{1}{8.5cm}{distance of singularities of $S$ to the real line} & \multirow{1}{2cm}{\pageref{regularS}}\\
\multirow{1}{3cm}{$\boldsymbol{\lambda}_{\pi/2}$} & \multirow{1}{8.5cm}{certain vector in $\mathbb{C}^n$} & \multirow{1}{2cm}{\pageref{l}}\\
\multirow{1}{3cm}{$m_\circ$} & \multirow{1}{8.5cm}{mass gap of the theory} & \multirow{1}{2cm}{\pageref{hilbert1}}\\
\multirow{1}{3cm}{$\mathcal{O}$} & \multirow{1}{8.5cm}{double cone} & \multirow{1}{2cm}{\pageref{o}}\\
\multirow{1}{3cm}{$\Omega$} & \multirow{1}{8.5cm}{the vacuum vector} & \multirow{1}{2cm}{}\\
\multirow{1}{3cm}{$P_n$} & \multirow{1}{8.5cm}{projections onto S-symmetric spaces} & \multirow{1}{2cm}{\pageref{Pn}}\\
\multirow{1}{3cm}{$P_n^-$} &
\multirow{1}{8.5cm}{projections onto totally antisymmetric spaces} & \multirow{1}{2cm}{\pageref{lem1}}\\
\end{tabular}
\begin{tabular}[hc]{c c r}
\multirow{1}{3cm}{$\mathcal{P}$, $\mathcal{P}_+$, $\mathcal{P}_+^\uparrow$} & \multirow{1}{8.5cm}{Poincar\'{e} group and subgroups} & \multirow{1}{2cm}{\pageref{poinc}}\\
\multirow{1}{3cm}{$Q$} & \multirow{1}{8.5cm}{skew-symmetric matrix} & \multirow{1}{2cm}{\pageref{Q}}\\
\multirow{1}{3cm}{$R$} & \multirow{1}{8.5cm}{deformation function} & \multirow{1}{2cm}{\pageref{R}}\\
\multirow{1}{3cm}{$\textbf{S}$} &
\multirow{1}{8.5cm}{scattering operator on $\mathscr{H}^+$} & \multirow{1}{2cm}{\pageref{S}}\\
\multirow{1}{3cm}{$S$} & \multirow{1}{8.5cm}{S-matrix} & \multirow{1}{2cm}{\pageref{S-matrixDefinition}}\\
\multirow{1}{3cm}{$\mathcal{S}$} &
\multirow{1}{8.5cm}{set of all S-matrices} & \multirow{1}{2cm}{\pageref{S-matrixDefinition}}\\
\multirow{1}{3cm}{$\mathcal{S}_0$} &
\multirow{1}{8.5cm}{set of regular S-matrices} & \multirow{1}{2cm}{\pageref{regularS}}\\
\multirow{1}{3cm}{$\mathcal{S}_0^-$} &
\multirow{1}{8.5cm}{subsets of $\mathcal{S}_0$} & \multirow{1}{2cm}{\pageref{form}}\\
\multirow{1}{3cm}{$S(a,b)$} & \multirow{1}{8.5cm}{strip region in $\mathbb{C}$} & \multirow{1}{2cm}{\pageref{strip}}\\
\multirow{1}{3cm}{$\|S\|_{\kappa}$} & \multirow{1}{8.5cm}{supremum norm of $S$} & \multirow{1}{2cm}{\pageref{regularS}}\\
\multirow{1}{3cm}{$S_n^\pi$} & \multirow{1}{8.5cm}{tensor corresponding to $D_n$} & \multirow{1}{2cm}{\pageref{tensor}}\\
\multirow{1}{3cm}{$\mathfrak{S}_n$} & \multirow{1}{8.5cm}{the permutation group of $n$ elements} & \multirow{1}{2cm}{\pageref{reprpermugroup}}\\
\multirow{1}{3cm}{$T_R$} & \multirow{1}{8.5cm}{operator related to $R$} & \multirow{1}{2cm}{\pageref{TR}}\\
\multirow{1}{3cm}{$U$} & \multirow{1}{8.5cm}{representation of $\mathcal{P}_+$} & \multirow{1}{2cm}{\pageref{rep}}\\
\multirow{1}{3cm}{$V$} & \multirow{1}{8.5cm}{representation of $G$} & \multirow{1}{2cm}{\pageref{rep}}\\
\multirow{1}{3cm}{$W_R$, $W_L$} & \multirow{1}{8.5cm}{the right and left wedge} & \multirow{1}{2cm}{\pageref{wedge}}\\
\multirow{1}{3cm}{$\mathcal{W}$} & \multirow{1}{8.5cm}{set of all wedges} & \multirow{1}{2cm}{\pageref{wedges}}\\
\multirow{1}{3cm}{$\Xi(s),\Xi_n(s)$} & \multirow{1}{8.5cm}{maps related to the modular nuclearity condition} & \multirow{1}{2cm}{\pageref{xi}}\\
\multirow{1}{3cm}{$z^\dagger,z$} & \multirow{1}{8.5cm}{Zamolodchikov creation and annihilation operators} & \multirow{1}{2cm}{\pageref{erzVern}}\\
\end{tabular}
%----------------------------------------------------------------------------------------
%	BIBLIOGRAPHY
%----------------------------------------------------------------------------------------
\backmatter
\label{Bibliography}
\fancyhead[LE,RO]{\thepage}
\fancyhead[RE,LO]{\emph{Bibliography}}
\renewcommand{\chaptermark}[1]{ \markboth{#1}{} }
\renewcommand{\sectionmark}[1]{ \markright{#1}{} }
%\lhead{\emph{Bibliography}} % Change the page header to say "Bibliography"

\bibliographystyle{unsrtnat} % Use the "unsrtnat" BibTeX style for formatting the Bibliography

\bibliography{Literatur} % The references (bibliography) information are stored in the file named "Bibliography.bib"
\nocite{*}
%----------------------------------------------------------------------------------------
%	ACKNOWLEDGEMENTS
%----------------------------------------------------------------------------------------
\clearpage
%\setstretch{1.3} % Reset the line-spacing to 1.3 for body text (if it has changed)

\acknowledgements{\addtocontents{toc}{\vspace{1em}} % Add a gap in the Contents, for aesthetics
I would like to express my deep gratitude to Professor Jakob Yngvason for awaking my interest in this interesting field of mathematical physics and for giving me the opportunity to work on this particular topic of the thesis.\par
Especially appreciated is Dr. Gandalf Lechner's kind willingness to co-supervise this thesis. I would like to thank both Prof. J. Yngvason and Dr. G. Lechner for their constant support, for many useful hints, their advise and careful reading of the final manuscript.\par
I am further grateful for numerous discussions with Jan Schlemmer, Christian K\"ohler, Matthias Plaschke and Albert Much. Also Daniela Cadamuro, Yoh Tanimoto and Prof. H. Grosse have to be thanked for interesting discussions.\par
Many thanks are due to Wojciech Dybalski for giving me the opportunity to finalize my thesis in the inspiring atmosphere of the TU Munich.\par
Finally, financial support from the FWF-project P22929-N16 ``Deformations of Quantum
Field Theories'' and travel grants from the University of Vienna are gratefully acknowledged.

}
\clearpage % Start a new page

\end{document}